\documentclass{article}
\usepackage[utf8]{inputenc}
\usepackage{amsthm,amsmath,amssymb,amsfonts,tikz,pgfplots,pdfcomment}
\usepackage{color}
\usepackage{wrapfig}
\usepackage{array}
\usepackage{authblk}
\usepackage{graphicx}

\makeatletter
\newcommand{\dashunder}[2][\mathop]{#1{\mathpalette\df@under{{\dashfill}{#2}}}}
\newcommand{\df@under}[2]{\df@@under#1#2}
\newcommand\df@@under[3]{%
   \vbox{
    \offinterlineskip
    \ialign{
     ##\cr
     #2{#1}\cr
      \noalign{\kern-9pt}
      $\m@th#1#3$\cr
    }
  }%
}
\newcommand{\dashover}[2][\mathop]{#1{\mathpalette\df@over{{\dashfill}{#2}}}}
\newcommand{\df@over}[2]{\df@@over#1#2}
\newcommand\df@@over[3]{%
  \vbox{
    \offinterlineskip
    \ialign{
     ##\cr
     #2{#1}\cr
      \noalign{\kern2pt}
      $\m@th#1#3$\cr
    }
  }%
}
\newcommand{\dashfill}[1]{%
  \kern-.5pt
  \xleaders\hbox{\kern.5pt \vrule height.4pt width \dash@width{#1}\kern .5pt}\hfill
  \kern-.5pt
}
\newcommand{\dash@width}[1]{%
  \ifx#1\displaystyle
    2pt
  \else
    \ifx#1\textstyle
      1.5pt
    \else
      \ifx#1\scriptstyle
        1.25pt
      \else
        \ifx#1\scriptscriptstyle
          1pt
        \fi
      \fi
    \fi
  \fi
}
\makeatother

\usetikzlibrary{matrix,graphs,arrows,positioning,cd,calc,decorations.markings,decorations.pathreplacing,shapes.symbols}
\tikzset{
	on each segment/.style={
		decorate,
		decoration={
			show path construction,
			moveto code={},
			lineto code={
				\path [#1]
				(\tikzinputsegmentfirst) -- (\tikzinputsegmentlast);
			},
			curveto code={
				\path [#1] (\tikzinputsegmentfirst)
				.. controls
				(\tikzinputsegmentsupporta) and (\tikzinputsegmentsupportb)
				..
				(\tikzinputsegmentlast);
			},
			closepath code={
				\path [#1]
				(\tikzinputsegmentfirst) -- (\tikzinputsegmentlast);
			},
		},
	},
	end arrow/.style={postaction={decorate,decoration={
				markings,
				mark=at position 1.0 with {\arrow[#1]{stealth}}
	}}},
}
\tikzset{
	>=stealth',
	pil/.style={
		->,
		shorten <=0pt,
		shorten >=0pt,},
}

\def\a{2.5}
\def\s{0.9}
\def\l{0.1}

\definecolor{dgreen}{rgb}{0.0, 0.5, 0.0}

\newcommand{\td}{\tilde}
\newcommand{\f}{\mathsf{F}}
\newcommand{\g}{\mathsf{G}}
\newcommand{\A}{\mathsf{a}}
\newcommand{\h}{\mathcal{H}}
\newcommand{\e}{\mathcal{E}}
\newcommand{\ri}{\mathrm{i}}
\newcommand{\nod}{\mathrm{nod}}
\newcommand{\symm}{\mathrm{sym}}
\newcommand{\Aut}{\mathrm{Aut}}
\newcommand{\Auti}{\mathsf{Aut}}
\newcommand{\Dih}{\mathrm{Dih}}

\pgfplotsset{compat=1.13}

\newcolumntype{M}[1]{>{\centering\arraybackslash}m{#1}}

\numberwithin{equation}{section}

\title{Folding transformations for $q$-Painlev\'e equations}
\author[1,2,3]{M. Bershtein \thanks{mbersht@gmail.com}}
\author[2,3]{A. Shchechkin \thanks{A.Shchechkin@skoltech.ru}}
\affil[1]{Landau Institute for Theoretical Physics, Chernogolovka, Russia}
\affil[2]{Skolkovo Institute of Science and Technology, Moscow, Russia}
\affil[3]{National Research University Higher School of Economics, Moscow, Russia}

\date{}

\newtheorem{thm}{Theorem}
\newtheorem{prop}{Proposition}[section]
\newtheorem{lemma}{Lemma}[section]

\newtheorem{corol}[lemma]{Corollary}

\theoremstyle{definition}
\newtheorem{Remark}{Remark}[section]
\newtheorem{defin}{Definition}[section]
\newtheorem{Example}{Example}[section]

\begin{document}
	
	\maketitle

\begin{abstract}
	Folding transformation of the Painlev\'e  equations is an algebraic (of degree greater than 1) transformation between solutions of 
	different equations. In 2005 Tsuda, Okamoto and Sakai classified folding transformations of differential Painlev\'e equations.
	These transformations are in correspondence with automorphisms of affine Dynkin diagrams. 

	We give a complete classification of folding transformations of the $q$-difference Painlev\'e equations,
	these transformations are in correspondence with certain subdiagrams of the affine Dynkin diagrams (possibly with automorphism).
	The method is based on Sakai's approach to Painlev\'e equations through rational surfaces.
\end{abstract}
	
	\tableofcontents

       \newpage

\section{Introduction}

The paper is devoted to the classification of folding transformations between solutions of 
\(q\)-difference Painlev\'e equations. 
We first explain the problem and illustrate it by a simple instructive example. 
 
\paragraph{$q$-Painlev\'e equations.}
The \(q\)-difference Painlev\'e equations are systems of two difference equations on two functions which we denote by \(F, G\).
The celebrated Sakai's classification \cite{Sakai01} assigns to each \(q\)-difference Painlev\'e equation two affine root systems.
These systems are called the symmetry and the surface root systems, and are denoted by \(\Phi^a_{\text{sym}}\) and \(\Phi^a_{\text{surf}}\)
respectively. The list of all pairs, together with arrows of degenerations is given in the Fig. \ref{Fig:qPainleve}.
\begin{figure}[!h]
 	\includegraphics[scale=0.78]{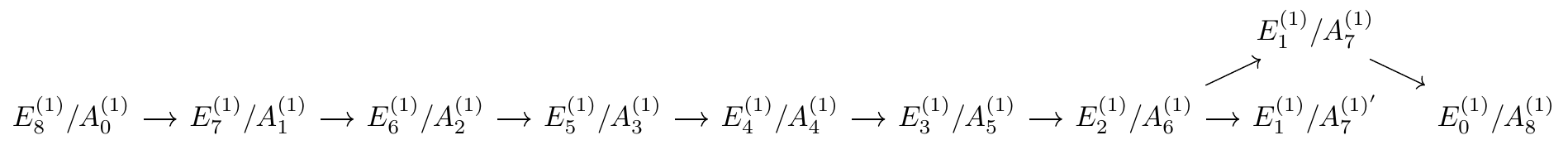}
	\caption{Symmetry/surface classification scheme for \(q\)-difference Painlev\'e equations }
	\label{Fig:qPainleve}
\end{figure}

\noindent Here we use the notation $E_5^{(1)}=D_5^{(1)},
E_4^{(1)}=A_4^{(1)},E_3^{(1)}=(A_2+A_1)^{(1)},E_2^{(1)}=(A_1+A_1)^{(1)},E_1^{(1)}=A_1^{(1)}$. 

The symmetry group of a Painlev\'e equation is an extended affine Weyl group \(W^{ae}\) of the root system \(\Phi^a_{\text{sym}}\),
translations in $W^{ae}$ generate $q$-difference dynamics. 
The \(F,G\) are then coordinates on (some chart of) the surface \(\mathcal{X}_{\vec{a}}\) called the space of initial conditions  \cite{Sakai01}.
This surface depends on set of parameters \(a_0,\dots,a_r\) which are called multiplicative root variables; here \(r+1\) is the rank of
\(\Phi^a_{\text{sym}}\). The Weyl group acts on $a_j$ as 
\begin{equation} \label{eq:intro:a_act}
	s_i (a_j)=a_j a_i^{-C_{ij}}, \quad 	\pi (a_i)=a_{\pi(i)},
\end{equation}
where \(s_0,\dots, s_r\) are simple reflections, \(C\) is a Cartan matrix of \(\Phi^a_{\text{sym}}\) and \(\pi\) is an external automorphism.

\begin{Example}\label{ex:Intr1} The dynamics of the \(q\)-Painlev\'e equation is usually denoted by an underline and an overline,
$\ldots \mapsto \underline{X}\mapsto X \mapsto \overline{X}\mapsto \ldots$. The standard $q$-Painlev\'e equation of the
symmetry/surface type $D_5^{(1)}/A_3^{(1)}$ together with symmetry type Dynkin diagram and transformation of root variables
are given below
\begin{tabular}{m{3.5cm}m{7cm}m{4.5cm}}
\begin{center}
	\begin{tikzpicture}[elt/.style={circle,draw=black!100,thick, inner sep=0pt,minimum size=2mm},scale=1]
		\path 	(-cos{60},sin{60}) 	node 	(a1) [elt] {}
		(-cos{60},-sin{60}) 	node 	(a0) [elt] {}
		( 0,0) node  	(a2) [elt] {}
		( 1,0) 	node  	(a3) [elt] {}
		( 1+cos{60},-sin{60}) 	node 	(a4) [elt] {}
		( 1+cos{60},sin{60})	node 	(a5) [elt] {};
		\draw [black,line width=1pt] (a1) -- (a2) -- (a3) -- (a4) (a3) -- (a5) (a2) -- (a0);
		\node at ($(a0.west) + (-0.2,0)$) 	{$\underline{\alpha_{0}}$};			    
		\node at ($(a1.west) + (-0.2,0)$) 	{$\alpha_{1}$};
		\node at ($(a2.south east) + (0.1,-0.2)$)  {$\alpha_{2}$};
		\node at ($(a3.south west) + (-0.1,-0.2)$)  {$\alpha_{3}$};
		\node at ($(a4.east) + (0.2,0)$)  {$\alpha_{4}$};	
		\node at ($(a5.south east) + (0.2,0)$)  {$\alpha_{5}$};		
		\draw[->] (a0) edge[bend right=60] node[]{} (a5);
			\draw[->] (a1) edge[bend right=60] node[]{} (a4);
			\draw[->] (a4) edge[bend left=60] node[fill=white]{$\pi$} (a0);
			\draw[->] (a5) edge[bend right=60] node[fill=white]{$\pi$} (a1);
			\draw[<->] (a2) edge[bend left=30] node[]{} (a3);
		
		\end{tikzpicture}
		\end{center}
	&
\begin{equation}\label{eq:qPVIF}
\begin{aligned}
	&F\underline{F}=a_1^{-1}\frac{(G-a_3^{-1})(G-a_5^{-1}a_3^{-1})}{(G-1)(G-a_4)}, \\	
	&G\overline{G}=a_4 \frac{(F-a_2 )(F-a_0a_2)}{(F-1)(F-a_1^{-1})},
\end{aligned}
\end{equation}
&
\begin{equation}\label{eq:intro:bar a}
\begin{aligned}
	&\overline{a_{0,1,4,5}}=a_{0,1,4,5}, \\ &\overline{a_2}=q a_2, \, \overline{a_3}=q^{-1}a_3, \\ &q=a_0a_1a_2^2a_3^2a_4a_5.
\end{aligned}
\end{equation}
\\
\end{tabular}

This equation is a \(q\)-analog of the Painlev\'e VI equation. Here \(a_2\) can be considered as a time, $a_3$ could be replaced by $q$,
while \(a_0,a_1,a_4,a_5\) are the four parameters of the equation.

The symmetry group is $W^{ae}_{D_5}$, it is generated by the simple reflections \(s_0,\dots,s_5\) and an order $4$ external automorphism
\(\pi\) which permutes the \(s_i\) (arrows on above Dynkin diagram). The variables \(a_i\) are multiplicative root variables and
the action of $W^{ae}_{D_5}$ on them is given by
\eqref{eq:intro:a_act}. The action of this group on the coordinates \(F,G\) is given by
\begin{center}
	\begin{tabular}{|c|c|c | c | c | c | c | c|}
		\hline
		coord. & $s_0$ & $s_1$ &
		$s_2$ & $s_3$ &
		$s_4$ & $s_5$ & $\pi$ \\
		\hline
		$F$ &$F$ & $a_1 F$ & $a_2^{-1}F$ & $F \frac{G-1}{G-a_3^{-1}}$ &
		$F$ & $F$ & $1/G$\\
		\hline
		$G$& $G$ & $G$ & $G \frac{F-1}{F-a_2}$ & $a_3G$ & $a_4^{-1}G$ & $G$ & $F/a_2$\\
		\hline
	\end{tabular}
\end{center}
\end{Example}

\paragraph{Foldings of Painlev\'e equations.}
A folding transformation \cite{TOS05} is an algebraic map between solutions of Painlev\'e equations which goes through
the quotient of the space of initial conditions. In other words, there is a certain \(\vec{a}\), subgroup \(H \in \mathrm{Aut}(\mathcal{X}_{\vec{a}})\) and birational map between the quotient \(\mathcal{X}_{\vec{a}}/H\) and space of initial conditions of another Painlev\'e equation \(\mathcal{Y}_{\vec{\A}}\). 
Such \(H\) should then be a subgroup of \(W^{ae}\) commuting with certain Painlev\'e dynamics, we call it folding subgroup.
       
\begin{Example}[Continuation of Example \ref{ex:Intr1}] \label{ex:Intr2}

	Consider an element \(w=s_0s_1s_4s_5\) of order \(2\). It acts as 
	\begin{equation}
		w\colon a_{0,1,4,5}\mapsto a_{0,1,4,5}^{-1},  \, a_2\mapsto a_2 a_0 a_1, \, a_3\mapsto a_3 a_4 a_5, \qquad
		F\mapsto a_1 F, \;\; G\mapsto a_4^{-1}G.
	\end{equation}
	It is easy to see that \(w\) commutes with the $q$-Painlev\'e dynamics \eqref{eq:qPVIF}.
	
	If $a_0=a_1=a_4=a_5=-1$, then \(w\) preserves the multiplicative root variables. 
	Equations \eqref{eq:qPVIF} become       
	\begin{equation}\label{qPVIex}
		F\underline{F}=-\frac{G^2-q^{-1}a_2^2}{G^2-1}, \qquad
		G\overline{G}=-\frac{F^2-a_2^2}{F^2-1},
	\end{equation}
	where we used that $q=a_2^2a_3^2$ in this case. Two equations in \eqref{qPVIex} now look almost the same and this leads
	to the possibility of creating a half-step (or ``extracting a root'') of this dynamics. We denote this square root
	by dashed over-(under-) line and define it by
	\begin{equation}\label{qPVIprr a}
		\dashover{a_2}=q^{1/2}a_2=a_2^2a_3, \qquad \dashover{G}=F, \dashover{F}=\overline{G}.
	\end{equation}
	Then equations \eqref{qPVIex} can be rewritten as
	\begin{align}\label{qPVIprr}
		\dashover{F}\dashunder{F}=-\frac{F^2-a_2^2}{F^2-1}, \qquad \dashover{G}=F.
	\end{align}

	Let us introduce coordinates $\f=F^2$, $\g=FG$, which are invariant under $w$ and give a map of degree $2$
	from \(\mathcal{X}_{\vec{a}}\).
	Then from \eqref{qPVIprr} we have
	\begin{align}\label{qPIIID8}
		\dashover{\g}\g=-\frac{\f(\f-a_2^2)}{\f-1}, \qquad \f\dashunder{\f}=\g^2,
	\end{align}
	which is a $q$-Painlev\'e equation with symmetry/surface type $A_1^{(1)}/A_7^{(1)'}$ (or parameterless $q$-Painlev\'e III)
	with $\A_1=a_2^2, \A_0=a_3^2$.
\end{Example}       


\paragraph{Motivation.}

Folding transformations lead to algebraic relations between solutions of Painlev\'e equations. For example, one has a relation
between special solutions. The points fixed under the action of \(H\) on the surface \(\mathcal{X}_{\vec{a}}\) correspond to certain
algebraic solutions of the Painlev\'e equations. On the other hand, in the quotient space \(\mathcal{X}_{\vec{a}}/H\) these fixed points
become singularities which, after the resolution, give exceptional divisors. Such divisors correspond to Ricatti solutions, so the folding
transformation relates the equation having an algebraic solution to that having Riccati solution. This was the original observation in \cite{TOS05}.
	
On the other hand, there was huge progress in the last 10 years in the Painlev\'e theory in relation to the formulas for generic solutions.
Namely, it was proposed in \cite{GIL12}, \cite{GIL13} that Painlev\'e equations can be solved in terms of certain special functions that
appeared in mathematical physics. These functions themselves have two definitions: conformal blocks in 2d conformal field theory 
\cite{BPZ:1984} and partition functions in 4d \(\mathcal{N}=2\) supersymmetric gauge theory \cite{Nekrasov:2003}. The equivalence
between these two definitions is famous AGT correspondence \cite{AGT:2010}. Later the formulas for generic solutions were extended
to the case of \(q\)-difference Painlev\'e equations for \(E_r^{(1)}\), \(r \leq 4\) symmetry \cite{BS16}, \cite{JNS17}, \cite{BGT:2019}, \cite{MN18}.

Therefore, folding transformations lead to nontrivial algebraic relations between partition functions. And the folding transformation
discussed in Example \ref{ex:Intr2} indeed has such meaning. The corresponding relation is (roughly speaking) a combination of
the identity found in the conformal field theory \cite{Poghossian:2009},\cite{FLNO:2009},\cite{HJS:2010} and bilinear relations from the
gauge theory \cite{NY1:2005},\cite{NY2:2005}. This was (again, roughly speaking) observed in \cite{BS18}, \cite{Shch2020}, and
our original motivation was to find more folding transformations in order to find new relations for partition functions. For the \(E_6^{(1)}, E_7^{(1)}, E_8^{(1)}\) we do not have formulas for generic solutions, so the folding transformation gives particular family of solutions.

\medskip

The classification of folding transformations is a problem also interesting in itself. From the first glance, the task seems
to be wild, there are many vectors \(\vec{a}\) for which the stabilizer \(W^{ae}_{\vec{a}}\) is nontrivial, this stabilizer could
act on the surface of initial conditions \(\mathcal{X}_{\vec{a}}\). But there is a huge restriction --- namely, for given \(\vec{a}\) the subgroup of stabilizer 
\((W^{ae}_{\vec{a}})_{\vec{a}}\) generated by reflections acts trivially on the surface \(\mathcal{X}_{\vec{a}}\)
(with some exceptions for small \(\Phi^a_{\text{sym}}\), namely for the cases \(E_1^{(1)},E_2^{(1)},E_3^{(1)}\)). 
It was shown in \cite{TOS05}, using this observation, that all folding transformations of \emph{differential} Painlev\'e equations come
from the subgroups of external automorphisms \(\widehat{\Omega}\subset W^{ae}\). More geometrically, \(\widehat{\Omega}\) is a subgroup
of \(W^{ae}\) preserving fundamental alcove. It is well known that  \(\widehat{\Omega}\simeq P/Q\), where \(P\) and \(Q\)
are weight and root lattices for the finite root system \(\Phi_{\text{sym}}\).
	
Recall that differential Painlev\'e equations correspond to the additive cases in Sakai's classification \cite{Sakai01}, contrary to
\(q\)-difference equations, which correspond to the multiplicative cases. 
In a multiplicative case situation becomes richer, it appears that nontrivial folding transformations could come
from the product of reflections from $W^a$, and even from the finite Weyl group \(W\). This is similar to the known fact in the theory of Lie groups
that the centralizer of a semisimple element of the adjoint group can be non-connected, contrary to simply connected case,
or Lie algebra case (see e.g. \cite{Hum95}). And the group of connected components can be identified with a subgroup of the group \(\Omega \), where \(\Omega\) is the subgroup of \(W\) which maps fundamental alcove to a parallel one. It is easy to see that the groups
\(\Omega\) and \(\widehat{\Omega}\) are isomorphic.

Remarkably, it appears that any folding subgroup \(H\) can be embedded as
\begin{equation}\label{eq:intro:H subset}
	H \subset \widehat{\Omega} \ltimes \prod_{j} \Omega_{\Phi_j},
\end{equation} 
\vspace{-3mm}

where a reducible root system \(\sqcup_j \Phi_j\) corresponds to a subdiagram \(I \subset \Delta^{a}\) in the affine Dynkin diagram of \(\Phi^a\).
Only special subdiagrams and special subgroups of \eqref{eq:intro:H subset} give (essentially different) foldings.
We proceed now to the more precise statement.

\paragraph{Main result.}
	
It follows from formula \eqref{eq:intro:H subset} that any folding subgroup \(H\) is solvable. So it is convenient to classify
first the cyclic folding subgroups and then to construct general folding subgroups from the cyclic ones.

Clearly it is sufficient to classify folding subgroups in \(H\subset W^{ae}\) up to conjugacy. Moreover, since the map from \(W^{ae}
\rightarrow \mathrm{Aut}\mathcal{X}_{\vec{a}}\) can have a kernel, it is sufficient to classify the images of \(H\)
in \(\mathrm{Aut}\mathcal{X}_{\vec{a}}\).

\begin{thm}

	a) There are 24 nonequivalent cyclic folding subgroups. They are given in Table \ref{Tab:intro}.
	
	b) There are also 9 nonequivalent non-cyclic folding subgroups. They are given in Table \ref{Tab:intro:2}.
\end{thm}
        
\begin{table}

	\begin{tabular}{|m{1.5cm}|m{4.5cm}m{4.5cm}m{4.5cm}|}
	\hline
	
	\begin{center}
	\Large $E_8^{(1)}$
 	\end{center}		
 	
 	&
 	
 	\vspace{0.2cm}
 	
 	$3A_2: E_8^{(1)}\xrightarrow{{}\quad 3 \quad {}} E_6^{(1)}$
 	
 	\includegraphics[scale=0.9]{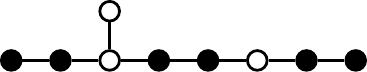}
 	
 	&
 	
 	\vspace{0.2cm}
 	
 	$4A_1: E_8^{(1)}\xrightarrow{{}\quad 2 \quad {}} E_7^{(1)}$
 	
 	\includegraphics[scale=0.9]{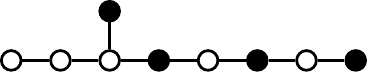}

 	&
 	
 	\vspace{0.2cm}
 	
 	$A_1+2A_3: E_8^{(1)}\xrightarrow{{}\quad 4 \quad {}} D_5^{(1)}$
 	
 	\includegraphics[scale=0.9]{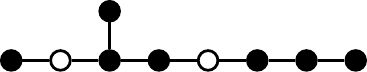}
 	
 	\\
 	\hline

         &
         
         \vspace{0.2cm}
        
        $\pi: E_7^{(1)}\xrightarrow{{}\quad 2 \quad {}} E_8^{(1)}$
      
        \includegraphics[scale=0.9]{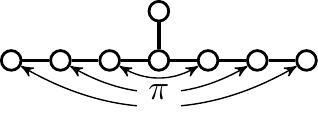}
        
         &
         
         $3A_1: E_7^{(1)}\xrightarrow{{}\quad 2 \quad {}} E_8^{(1)}$
         
        \includegraphics[scale=0.9]{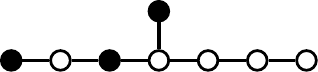}
         
         &
        
         $4A_1: E_7^{(1)}\xrightarrow{{}\quad 2 \quad {}} D_5^{(1)}$
        
         \includegraphics[scale=0.9]{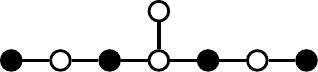}
        
        \\
         
         \vspace{-3cm}
        \begin{center}
	\Large $E_7^{(1)}$
 	\end{center}
         
         & 
         
         $3A_2: E_7^{(1)}\xrightarrow{{}\quad 3 \quad {}} E_3^{(1)}$
         
      \includegraphics[scale=0.9]{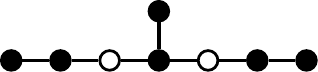}
         
         & 
         
         $2A_3: E_7^{(1)}\xrightarrow{{}\quad 4 \quad {}} E_7^{(1)}$
         
       \includegraphics[scale=0.9]{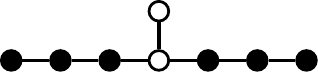}

         &
         
         \vspace{0.2cm}
         
         $\pi \ltimes 5A_1: E_7^{(1)}\xrightarrow{{}\quad 4 \quad {}} E_7^{(1)}$
         
        \includegraphics[scale=0.9]{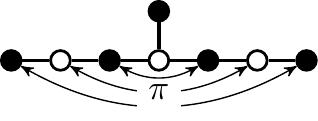}
        
         \\
         
         \hline
 	
 	 \begin{center}
	\Large $E_6^{(1)}$
 	\end{center}     
 	
 	&
 	
 	\begin{tikzpicture}
 	
 	\node at (-2.25,0.75) { $\pi: E_6^{(1)}\xrightarrow{{}\, 3 \, {}} E_8^{(1)}$};
 	
 	\node at (0,0) {\includegraphics[scale=0.9]{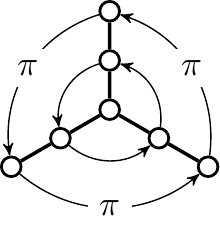}};
 	
 	\end{tikzpicture}
 	
 	&
 	
 	\begin{tikzpicture}
 	
 	\node at (-1.75,0.75) { $2A_2: E_6^{(1)}\xrightarrow{{}\, 3 \, {}} E_8^{(1)}$};
 	
 	\node at (0,0) {\includegraphics[scale=0.9]{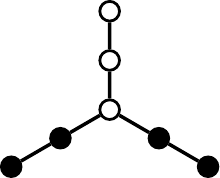}};
 	
 	\end{tikzpicture}
 	
 	&
 	
 	\begin{tikzpicture}
 	
 	\node at (-1.75,0.75) { $4A_1: E_6^{(1)}\xrightarrow{{}\, 2 \, {}} E_3^{(1)}$};
 	
 	\node at (0,0) {\includegraphics[scale=0.9]{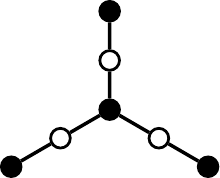}};
 	
 	\end{tikzpicture}
 	
 	\\
 	
 	\hline
 	
 	&
 	
 	 $\pi^2: D_5^{(1)}\xrightarrow{{}\, 2 \, {}} E_7^{(1)}$
 	
 	\includegraphics[scale=0.75]{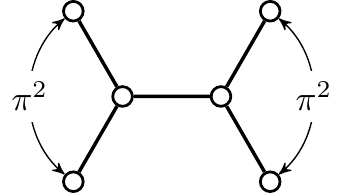}

 	&
 	
 	$2A_1: D_5^{(1)}\xrightarrow{{}\, 2 \, {}} E_7^{(1)}$
 	
 	\includegraphics[scale=0.75]{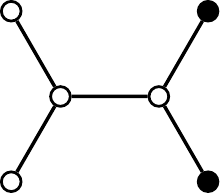}
 	
 	&
 	
 	\hspace{-1cm}
 	\begin{tikzpicture}
 	
 	\node at (-2.25,0) { $\pi: D_5^{(1)}\xrightarrow{{}\, 4 \, {}} E_8^{(1)}$};
 	
 	\node at (0,0) {\includegraphics[scale=0.75]{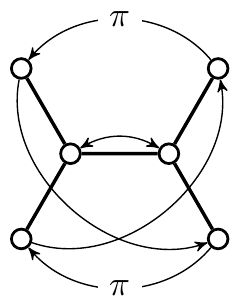}};
 	
 	\end{tikzpicture}
 	
 	\\
 	
 	\vspace{-4cm}
 	
 	\begin{center}
	\Large $D_5^{(1)}$
 	\end{center}	
 	
 	&
 	
 	 $A_1+A_3: D_5^{(1)}\xrightarrow{{}\, 4 \, {}} E_8^{(1)}$
 	
 	\includegraphics[scale=0.75]{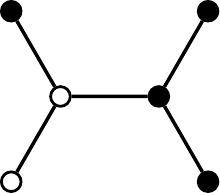}
 	
 	&
 	
 	$\pi^2\ltimes 3A_1: D_5^{(1)}\xrightarrow{{}\, 4 \, {}} E_8^{(1)}$
 	
 	\includegraphics[scale=0.75]{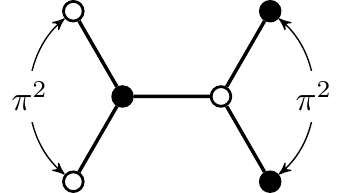}
 	
 	&
 	
 	 $4A_1: D_5^{(1)}\xrightarrow{{}\, 2 \, {}} A_1^{(1)}$
 	 
 	 \includegraphics[scale=0.75]{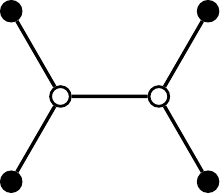}
 	
 	\\
 	
 	\hline
 	
 	  &
  	    
  	    &  
  	    
  	    \hspace{-4cm}
  	    \begin{tikzpicture}
 	
 	\node at (-2.5,0) { $\pi^3: E_3^{(1)}\xrightarrow{{}\, 2 \, {}} E_6^{(1)}$};
 	
 	\node at (0,0) {\includegraphics[scale=0.75]{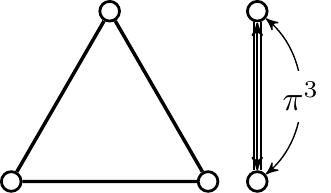}};
 	
 	\end{tikzpicture}
 	
 	&
 	
 	\hspace{-2cm}
 	\begin{tikzpicture}
 	
 	\node at (-2.5,0) { $A_1: E_3^{(1)}\xrightarrow{{}\, 2 \, {}} E_6^{(1)}$};
 	
 	\node at (0,0) {\includegraphics[scale=0.75]{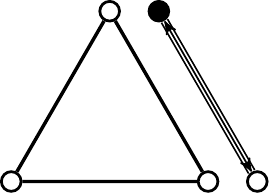}};
 	
 	\end{tikzpicture}
 	
 	\\
 	
 	\vspace{-3cm}
 	
 	\begin{center}
	\Large $E_3^{(1)}$
 	\end{center}
 	
 	&
 	
 	&
 
         \hspace{-4cm}
         \begin{tikzpicture}
 	
 	\node at (-2.5,0) { $\pi^2: E_3^{(1)}\xrightarrow{{}\, 3 \, {}} E_7^{(1)}$};
 	
 	\node at (0,0) {\includegraphics[scale=0.75]{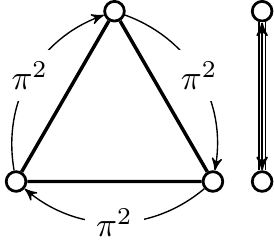}};
 	
 	\end{tikzpicture}
 	
 	& 
 	
 	\hspace{-2cm}
 	\begin{tikzpicture}
 	
 	\node at (-2.5,0) { $A_2: E_3^{(1)}\xrightarrow{{}\, 3 \, {}} E_7^{(1)}$};
 	
 	\node at (0,0) {\includegraphics[scale=0.75]{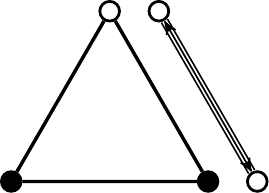}};
 	
 	\end{tikzpicture}
 	
 	\\
 	
 	\hline
 	
 	\begin{center}
	\Large $A_1^{(1)}$
 	\end{center}
 	
 	&
 	
 	&  
 	
 	\vspace{2mm}
 	
 	{$\pi^2: A_1^{(1)}\xrightarrow{{}\quad 2 \quad {}} D_5^{(1)}$
 	
 	\includegraphics[scale=1]{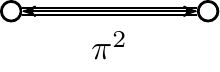}}

 	&
 	
 	\vspace{2mm}
 	
 	{$\sigma: A_1^{(1)}\xrightarrow{{}\quad 2 \quad {}} D_5^{(1)}$
 	
 	\includegraphics[scale=1]{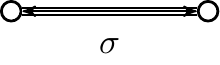}}

 	\\
 	
 	\hline
 	
	\end{tabular}
\caption{Cyclic folding subgroups \label{Tab:intro}}
\end{table}

\begin{table}

	\begin{tabular}{|m{1.5cm}|m{7cm}m{6.5cm}|}
	\hline
&   
\begin{tikzpicture}


\node at (-2.25,0) {\parbox{6cm}{$C_2 \ltimes C_4$\\[2mm]$E_7^{(1)}\longrightarrow E_8^{(1)}$}};

\node at (0,0)  {\includegraphics[scale=1]{e7_p11111.pdf}};
\end{tikzpicture}

&  
\begin{tikzpicture}


\node at (-2,0) {\parbox{4cm}{$C_2 \ltimes C_4$\\[2mm] $E_7^{(1)}\longrightarrow E_8^{(1)}$}};

\node at (0,0) {\includegraphics[scale=1]{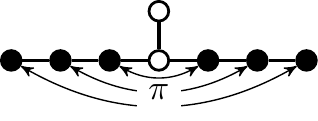}};         

\end{tikzpicture}

\\	
	
\vspace{-2cm}	
\begin{center}
\Large $E_7^{(1)}$ 
\end{center}
&  
\begin{tikzpicture}


\node at (-2.25,0)
{\parbox{6cm}{$C_2\times C_2$\\[2mm]
	$E_7^{(1)}\longrightarrow E_7^{(1)}$}};

\node at (0,0) {\includegraphics[scale=1]{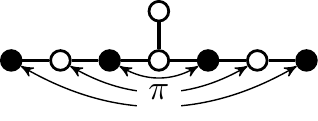}};

\end{tikzpicture}
&      

\begin{tikzpicture}


\node at (-2,0) {\parbox{4cm}{$C_2\times 4A_1$\\[2mm] $E_7^{(1)}\longrightarrow E_7^{(1)}$}};

\node at (0,0) {\includegraphics[scale=1]{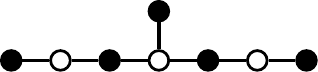}}; 
 
\end{tikzpicture} 
 
 \\	
\hline

& 
\begin{tikzpicture}

\node at (-2.25,0)
{\parbox{4cm}{$C_2\times C_2 \times C_2$\\[2mm] $D_5^{(1)}\longrightarrow E_7^{(1)}$}};

\node at (0,0) {\includegraphics[scale=0.75]{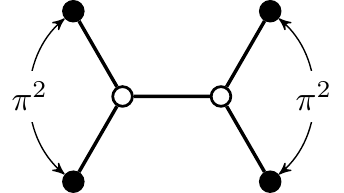}};

\end{tikzpicture}

&
\begin{tikzpicture}

\node at (-2,0)
{\parbox{4cm}{$C_2 \times C_2$\\[2mm] $D_5^{(1)}\longrightarrow D_5^{1)}$}};

\node at (0,0) {\includegraphics[scale=0.75]{d5_1111.pdf}};
\end{tikzpicture}
\\
\vspace{-2cm}	
\begin{center}
\Large $D_5^{(1)}$ 
\end{center}
&      

\begin{tikzpicture}

\node at (-2.25,0)
{\parbox{4cm}{$C_2\times C_2 $\\[2mm] $D_5^{(1)}\longrightarrow D_5^{(1)}$}};

\node at (0,0)
{\includegraphics[scale=0.75]{d5_pp1111.pdf}};
\end{tikzpicture}
 	& 
\begin{tikzpicture}

\node at (-2,0)
{\parbox{4cm}{$C_2\times C_2$\\[2mm] $D_5^{(1)}\longrightarrow E_8^{(1)}$}};

\node  at (0,0)
{\includegraphics[scale=0.75]{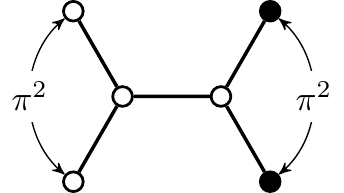}};
\end{tikzpicture} 
 \\
\hline

\begin{center}
\Large $A_1^{(1)}$ 
\end{center}
&             &      
\hspace{-4cm}\begin{tikzpicture}

\node at (-2,0)
{\parbox{4cm}{$C_2\times C_2$\\[2mm] $A_1^{(1)}\longrightarrow E_7^{(1)}$}};

\node at (0,0)
{\includegraphics[scale=1]{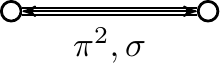}}; 
\end{tikzpicture}
\\
\hline
	
\end{tabular}

\caption{Non-cyclic folding subgroups \label{Tab:intro:2}}
\end{table}
       
Let us explain the notation used in the tables. For each folding subgroup, we draw nodes of the corresponding subset in black.
We label folding transformation via types of the root system generated by simple roots corresponding to \(I\).
If \(H\) contains (at least as a factor) an external automorphism \(\pi\in\widehat{\Omega}\), we also write its action
on the diagram and add \(\pi \ltimes\) in the notation of the folding transformation. In particular, if the folding subgroup is a
subgroup of \(\widehat{\Omega}\), then we draw all vertices in white and denote the corresponding folding transformation by \(\pi\).

For each folding subgroup, we draw an arrow whose source is the symmetry of the initial Painlev\'e equation and whose target is the Painlev\'e
equation after the folding. We write the degree of the transformation above the arrow.

In Table \ref{Tab:intro:2} we also write the isomorphism class of the group \(H\). Note that contrary to differential foldings, in this
case the group \(H\) can be non-commutative.

\begin{Example}
	The folding transformation used in the Example \ref{ex:Intr2} in this notation corresponds to the subdiagram \(4A_1 \subset D_5^{(1)}\).
	The corresponding folding subgroup is generated by \(w=s_{0}s_1s_4s_5 \in \Omega_{A_1}^4\).
\end{Example}

\begin{Remark}\label{rem:intro:selection}
The diagrams represented in Table \ref{Tab:intro} are specified by the element $\pi\in \widehat{\Omega}$ and the coloring which satisfies
the following combinatorial properties. Let \(I\) denote subgraph of \(\Delta^{a}\) with black vertices, \(\{\Delta_j\}\) be the set
of connected components of \(I\), and \(n_j\) be a number of vertices in \(\Delta_j\). Then:

\begin{itemize}
	\item Each \(\Delta_j\) is a Dynkin diagram of type \(A\). Coloring is invariant with respect to $\pi$,
	if some \(\Delta_j\) is $\pi$-invariant as a set, then \(\pi\) acts trivially on it.
	\item Number of $\pi$ orbits on the set of vertices $\Delta^a \setminus I$ is at least $2$.
	\item In each $\pi$-orbit choose one component $\Delta_j$, called distinguished. Number the vertices of each distinguished 
	component successively from $1$ to $n_j$. Then there is a choice of integer numbers \(m_j, \, 0<m_j\leq n_j,\quad
	\gcd (m_j,n_j+1)=1\), for each distinguished component such that, for any white vertex $c$,
	\begin{equation}
		\sum_{b\in I_D\cap N(c)} \frac{i_b m_{j(b)}}{n_{j(b)}+1}\in \mathbb{Z}, 
	\end{equation}
    where the sum runs over the set of black vertices belonging to the set of distinguished components $I_D$ and adjacent to \(c\), $j(b)$ is such that
    \(b\in \Delta_{j(b)}\), and $i_b$ is the number of $b$ in the component $\Delta_{j(b)}$.
\end{itemize}
\end{Remark}

In the main text we present more details about folding transformations given in tables \ref{Tab:intro},\ref{Tab:intro:2}.
In particular, we present the corresponding subgroups \(H \subset W^{ae}\), give explicit formulas for the coordinate transformation
\(\f=\f(F,G), \g=\g(F,G)\), and compute the automorphism group of the image.

\paragraph{Content of the paper.}

The main text of the paper consists of two parts: the first one is devoted to the algebra and combinatorics of the affine Weyl group
(Secs. \ref{sec:class_gp}, \ref{sec:class_answ}) and the second one is devoted to the geometry of rational surfaces (Secs. \ref{sec:geom_gs},
\ref{sec:geom_answ}). In each of the two parts, we first explain the results and give some examples, and then we give some tables with
the explicit data.

Section \ref{sec:class_gp} is devoted to the classification. Theorem \ref{thm:stab} reduces the problem to the subgroups
of \ref{eq:intro:H subset}. In Section \ref{ssec:selection} we obtain selection rules as in Remark \ref{rem:intro:selection},
and in Section \ref{ssec:answers algebraic} we give the results. 

In Section \ref{ssec:invlat} we discuss the normalizer subgroup \(N=N(H,W^{ae}) \subset W^{ae}\) of \(H\). The group \(N\)
is important since it acts after the folding. In particular, it should include the Painlev\'e dynamics. 
Here we restrict ourselves to the case of cyclic \(H\) and denote its generator by $w$. 
Roughly speaking the answer is the following. 
Let \(Q^{a}_{\text{fold sym}}\) be the sublattice in \(Q^a_{\text{sym}}\) invariant under \(w\). 
Then it appears that \(Q^{a}_{\text{fold sym}}\)  is the lattice for the affine root system \(\Phi^{a}_{\text{fold sym}}\) and the group \(N\) on the image is a finite extension of the affine Weyl group \(W^a_{\text{fold sym}}\). 
Details on the computation of \(N\) in each case are given in Sec.~\ref{sec:class_answ}. 

As we saw in Example \ref{ex:Intr2}, for special values of the parameters it may be possible to extract a root of the dynamics.
This phenomenon is called a projective reduction of the Painlev\'e equation \cite{KNT11}, \cite{KN13}. Projective reductions quite
often appear in examples of folding transformations, we also discuss this Sec.~\ref{ssec:invlat}.

In the geometrical part, we find the type of a Painlev\'e equation after the folding\footnote{It is natural to believe that resulting dynamics should also be of a discrete Painlev\'e type.
Indeed by Castelnuovo theorem (see e.g. \cite[Cor. V.5]{Beauville:1996}) the quotient 
\(\mathcal{X}_{\vec{a}}/H\) is rational surface. Moreover, the standard Painlev\'e  detectors \cite{Grammaticos:2004}: singularity confinement and algebraic entropy and are also preserved after the folding.}
and compute explicit formulas for the coordinates \(\f=\f(F,G), \g=\g(F,G)\). To do this we have to take the quotient
\(\mathcal{X}_{\vec{a}}/H\), resolve its singularities, and then (probably after some blowdown) identify the result with
the space \(\mathcal{Y}_{\vec{\A}}\) of initial conditions of the resulting Painlev\'e equation. This scheme is described
in Sec.~\ref{sec:geom_gs} and details for each example are given in Sec.~\ref{sec:geom_answ}. For the convenience of the reader,
note that tables with the main results are given in Sec.~\ref{ssec:geom answers}. In Sec.~\ref{ssec:nodal} (again, with
details in Sec.~\ref{sec:geom_answ}) we compute the symmetry group from the geometry of \(\mathcal{Y}_{\vec{\A}}\) and observe that it is in
complete agreement with the computation of \(N\) mentioned above.

Finally, in Section \ref{Sec:Further} we shortly discuss two questions that are natural after completion of the classification.
The first question is the continuous limit. Unfortunately, we do not have a complete answer, so we will limit ourselves to one example.
The second question is a comparison of our classification with previously known examples. There are 4 folding transformations of
\(q\)-difference equations in \cite{RGT00} (see also \cite{RNGT09}), we found them in our classification.
The fifth example appeared in \cite{BS18}.

\paragraph{Notation.} Here we list some notation and abbreviations which we used through the main text. Of course this not the full
list, but rather a list necessary to avoid confusion.
\begin{itemize}
	\item \(C_n\): the cyclic group of order \(n\),
	\item  \(\zeta_n\): some primitive root of \(1\) of order \(n\),
	\item \(\ri=\sqrt{-1}\), \(\zeta=\zeta_3\),
	\item \(s_{i_1\dots i_k}=s_{i_1}\cdot\ldots \cdot s_{i_k}\), where \(s_{i_1}\dots s_{i_k}\) are simple reflections.
	\item $t_{\lambda}$ are translations in affine Weyl group, corresponding to weight $\lambda\in P$.
	\item \(\alpha_{i_1\dots i_k}=\alpha_{i_1}+\ldots+\alpha_{i_k}\), where \(\alpha_{i_j}\) are simple roots. 
	\item $\omega_i$, $\omega_{\alpha_i}$: fundamental weight, corresponding to the simple root $\alpha_i$. 
	\item \(F,G,a_i,q, E_i\): coordinates, root variables, imaginary root variable, exceptional divisors
	on the surface from which folding transformation acts.
	\item \(\f,\g,\A_i, \mathsf{q}, \e_i\): coordinates, root variables, imaginary root variable, 
	exceptional divisors on the surface obtained after folding transformation. 
	\item \(E_{i_1\dots i_k}=E_{i_1}+\dots+E_{i_k}\), \(H_{FG}=H_F+H_G\). 
\end{itemize}

\paragraph{Acknowledgements.}
We are grateful to Oleg Lisovyy who drew our attention to folding transformations, Konstantin Shramov for explanations on
the intersection theory and Pavlo Gavrylenko for many very useful discussions. 

We are especially grateful to Anton Dzhamay, who participated in the early stage of this project. Our understanding of Sakai's
approach is due to him, he also suggested that nodal curves and projective reduction play an important role in the study
of folding transformations.

M.B. was partially supported by the HSE University Basic Research Program.
	
\newpage

\section{Classification of folding transformations} \label{sec:class_gp}

\subsection{Main problems} \label{ssec:main_probl}
	
\paragraph{Sakai's geometric approach.}
	
Celebrated Sakai's classification \cite{Sakai01} relates $q$-Painlev\'e equations to the families of rational surfaces $\mathcal{X}$.
These surfaces are blowups of $\mathbb{P}^1\times \mathbb{P}^1$ in $8$ points (or, equivalently, $\mathbb{P}^2$ in $9$ points).
We will usually denote coordinates on \(\mathbb{P}^1\times \mathbb{P}^1\) by \(F\) and \(G\).
	
Each family is parametrized by a collection of root variables $\vec{a}=(a_0,\dots,a_r)\in \mathcal{A} \subset (\mathbb{C^*})^{r+1}$
\footnote{For $E_2^{(1)}$ there is dependence on square roots of these variables.}
, which correspond to the set $\{\alpha_0,\dots,\alpha_r\}=\Delta^a$ of simple roots of certain affine root system $\Phi^a$. The positions of 8 blowup
	points are given in terms of \(\vec{a}\). Symmetries of the family $\mathcal{X}$ are given by affine extended Weyl group $W^{ae}$
	of root system $\Phi^a$, thus type of this system in called symmetry type\footnote{\label{footnote:A1/A7}
	For the symmetry/surface type $E_1^{(1)}/A_7^{(1)}$ the symmetry group is slightly bigger
	and for $E_2^{(1)}/A_6^{(1)}$ --- smaller, these cases will require more care below.}.
	All these $\Phi^a$ are of $E$-type: $\Phi^a=E_r^{(1)},\, r=1\ldots 8$, see Fig. \ref{Fig:qPainleve} above.
	
	Group $W^{ae}$ is a semidirect product \(\widehat{\Omega}\ltimes W^a\) of affine Coxeter group $W^a$ and group of external
	automorphisms $\widehat{\Omega}$. Hence $W^{ae}$ is generated by simple reflections \(s_i\), \(i=0,\dots,r\) and permutations
	of simple roots \(\pi\). Action of $w$ on root variables is multiplicative, namely 
	\begin{equation}
		s_i(a_j) = a_j a_i^{-C_{ij}}, \, s_i\in W^a,  \qquad \pi(a_i)=a_{\pi(i)}, \, \pi\in\widehat{\Omega},
	\end{equation}
	where $C_{ij}$ is Cartan matrix of $\Phi^a$.
	
 	Group $W^{ae}$ acts between different members of the family $\mathcal{X}$, nontrivially acting on $\mathcal{X}_{\vec{a}}$ 
	\begin{equation}
	W^{ae} \ni w: \quad \mathcal{X}_{\vec{a}} \rightarrow \mathcal{X}_{w(\vec{a})}, \quad (F,G)\in \mathcal{X}_{\vec{a}} \mapsto  
	(w(F),w(G))\in \mathcal{X}_{w(\vec{a})}.
	\end{equation}
	Explicit formulas for the action are given in Section \ref{sec:geom_answ}. More details about Sakai's approach can be found in
	review \cite{KNY15}, we usually follow \emph{loc. cit.} in notations and conventions.
        
\paragraph{Folding: definition.}

We see that generally, the symmetry changes $\vec{a}$. However, under restrictions on parameters $\vec{a}$ symmetry can act as
an automorphism of corresponding $\mathcal{X}_{\vec{a}}$ and such automorphism can be nontrivial.
This phenomenon is what we call folding
        
\begin{defin} 
	\emph{Folding transformation} of surface is an element $w\in W^{ae}$ and family $\mathcal{A}_w\subset \mathcal{A}$ of root variable
	tuples such that for any (generic) $\vec{a}\in \mathcal{A}_w$  we have
	\begin{equation}
		w(\vec{a})=(\vec{a}),\quad w((F,G))\neq (F,G).
	\end{equation}	
		
	\emph{Folding subgroup} is a subgroup $H \subset W^{ae}$ and family $\mathcal{A}_H$ of root variable tuples
	such that for any (generic) $\vec{a}\in \mathcal{A}_H$  we have an embedding $H \rightarrow \mathrm{Aut}(\mathcal{X}_{\vec{a}})$.
\end{defin} 
Without loss of generality, we can take $\mathcal{A}_w$ ($\mathcal{A}_H$) to be an irreducible component of the invariant
subset $\mathcal{A}^{w}$ ($\mathcal{A}^H$). We will see below that for given \(w\) there is often only one component
\(\mathcal{A}_w\) which gives nontrivial folding transformation. If there are several such components, then they give folding transformations,
equivalent under a certain symmetry. Hence we will label folding transformations by group element \(w \in W^{ae}\).
	
Informally we can say that the folding subgroup is a subgroup in which any nontrivial element is folding transformation.
Individual folding transformations correspond to the cyclic folding subgroups \(H\). As already mentioned in the Introduction we
concentrate first on these cyclic groups and then construct all folding subgroups from cyclic ones.
	
For the folding subgroup, we can take the quotient $\mathcal{X}_{\vec{a}}/H$. The normalizer subgroup $N$ of $H$ in $W^{ae}$ acts on
the family $\mathcal{A}_H$. Therefore it is natural to ask for the birational model $Y$ of the quotient $\mathcal{X}_{\vec{a}}/H$
and map from $N$ to the corresponding affine Weyl group. If $\mathcal{X}$ corresponds to root system $\Phi^a$ and $\mathcal{Y}$
corresponds to root system $\Phi'^a$ we say that we get folding transformation from symmetry type $\Phi^a$ to symmetry type $\Phi'^a$.

The first key point in the classification of folding transformations is the following lemma and corollary. 
\begin{lemma}\label{lemma:reflection}
	        For $a_i=1$ we have $s_i(F,G)=(F,G)$. 
\end{lemma}
\begin{proof}
	One can just check the statement of the lemma case-by-case, using formulas for action of simple reflections from Section
	\ref{sec:geom_answ} or \cite{KNY15}
\end{proof}

\begin{Remark}\label{rem:aexcept}
	According to formulas from Section \ref{sec:geom_answ} the simple reflections $s_0$ for $E_7^{(1)}/A_1^{(1)}$
	and $s_1$ for $E_8^{(1)}/A_0^{(1)}$ naively violate the statement of the lemma. However, this is not the case,
	because value $1$ for root variables
	$a_0$ for $E_7^{(1)}/A_1^{(1)}$ and $a_1$ for $E_8^{(1)}/A_0^{(1)}$ is excluded from the family $\mathcal{X}$,
	since it leads to another geometry, see defining equations \eqref{E7curve} for $E_7^{(1)}/A_1^{(1)}$  case
	and \eqref{E8curve} for $E_8^{(1)}/A_0^{(1)}$ case. 	
\end{Remark}

	There is also a geometric proof of the Lemma \ref{lemma:reflection}, see Section \ref{ssec:Picard}.

\begin{corol}\label{corol:reflection}
	Let root \(\alpha\) belong to component of $\Phi^a$ different from \(A_1^{(1)}\), let \(s_\alpha\) be the corresponding reflection. If 
	$s_{\alpha}(\vec{a})=\vec{a}$ then $s_{\alpha}((F,G))=(F,G)$.
\end{corol}	
In other words the condition of the lemma holds if \(\Phi^a\) is different from \(E_1^{(1)}\), \(E_2^{(1)}\) and \(E_3^{(1)}\).

\begin{proof}[Proof of the Colorrary]
	Since any roots is conjugated to simple one it is sufficient to consider case \(\alpha=\alpha_i\). It follows from the
	assumptions of the corollary that there exists simple root \(\alpha_j\) such that \(C_{ij}=-1\). Hence \(s_i(a_j)=a_j\)
	implies \(a_i=1\) and it remains to use Lemma \ref{lemma:reflection}.
\end{proof}

        \paragraph{Folding on equations.}
        
    $q$-Painlev\'e equations correspond to translations in affine Weyl group. Recall that affine Weyl groups have a structure of semidirect product 
	\begin{equation}\label{eq:Waff semidirect}
	W^a=Q \rtimes W, \quad W^{ae}=P\rtimes W. 
	\end{equation}
	Here \(W\) is (finite) Weyl group corresponding to (finite) root system \(\Phi\), $Q$ is a root lattice of \(\Phi\) and
	$P$ is a  weight lattice of \(\Phi\). We denote by $\bar{w}\in W$ the finite part of $w\in W^{a(e)}$. We denote fundamental weight corresponding to simple root \(\alpha_i\) by \(\omega_i\).

    For any \(\lambda \in P\) let \(t_\lambda\) be the corresponding element of \(W^{ae}\). 
    Its action on root variables is given by 
    \( t_\lambda a_i= q^{(\lambda,\alpha_i)}a_i\).
Here \(q\) corresponds to imaginary root \(\delta \in \Phi^a\), namely 
\begin{equation}\label{eq:marks def}
	\delta=\sum_{i=0}^r \mathfrak{m}_i \alpha_i, \quad q=\prod_{i=0}^r a_i^{\mathfrak{m}_i}.
\end{equation}
Coefficients \(\mathfrak{m}_i\) of $\delta$ decomposition into sum of simple roots $\alpha_i$ we call marks of $\alpha_i$.
        
The action of translation \(t_\lambda\) on coordinates \(F,G\) is very nontrivial.
The $q$-Painlev\'e equation is an equation on the section of the bundle \(\mathcal{X}\rightarrow \mathcal{A}\)
equivariant under the action of  \(t_\lambda\). Namely, $q$-Painlev\'e equations are equations of the form
\begin{equation}
	t_\lambda(F(\vec{a}),G(\vec{a}))=(F(t_\lambda(\vec{a})),G(t_\lambda(\vec{a})))
\end{equation}
on functions \((F(\vec{a}),G(\vec{a}))\) (which actually take values in surface \(\mathcal{X}_{\vec{a}}\)). In a sense, the Painlev\'e equation is defined by the family of surfaces and chosen translation \(t_\lambda\).

\begin{Remark}\label{rem:projective reduction}
	There is a slight extension of this construction. Namely, for certain  values of root variables \(\vec{a}\), there is 
	certain \(w\in W^{ae}, \, w^k \in P\) which act as a translation by \(q^{1/k}\) for certain \(k\in \mathbb{Z}_{>0}\).
	Such phenomenon is called projective reduction, see e.g. \cite{KNT11}, \cite{KN13}.
	Projective reductions quite often appear in the study of the folding transformations, see end of Sec.~\ref{ssec:invlat} below. 	
\end{Remark}

In order to have \emph{folding transformation (or folding subgroup) of \(q\)-difference Painlev\'e equation} we need translation
\(t_\lambda\) which will act on the quotient $\mathcal{X}_{\vec{a}}/H$. Hence \(t_\lambda\) belongs to the normalizer subgroup \(N\).
We will call folding transformation $(w,\mathcal{A}_w)$ \emph{true folding} if there is translation in \(N\) and the family
$\mathcal{A}_w$ contains points with arbitrary \(q\). Similarly we define \emph{true folding} subgroup. 
    
The true folding condition requires that family $\mathcal{A}_w$ (correspondingly $\mathcal{A}_H$) is at least two
dimensional since one parameter is \(q\), and another parameter is shifted by translation.
The normalizer subgroup \(N\) could contain several linear independent translations, for this case one folding transformation
(or folding subgroup) gives transformation for several $q$-Painlev\'e equations.

Below we give a list of all up to conjugations in $W^{ae}$ possible cyclic true folding subgroups.
To be more precise we classify the images of the folding subgroups in the automorphism group of \(\mathcal{X}_{\vec{a}}\),
this is important since as we saw in Lemma \ref{lemma:reflection} the map from $W^{ae}$ to automorphism group can have a kernel.
Our approach is based on the combinatorics of the affine Weyl group, the only geometric input is Lemma~\ref{lemma:reflection}.
Then, in Section~\ref{sec:geom_answ} we will see that all transformations from the list found by such approach give non-equivalent true
folding transformations.

\subsection{Stabilizer}
\paragraph{Additive case.}
Let \(\Phi\) be a finite simply laced irreducible root system, \(\Delta=\{\alpha_1,\dots,\alpha_r\}\) set of its simple roots, \(W_\Phi\)
be the corresponding Weyl group. Let \(V_\Phi\) be a real vector space with coordinates \(x_1,\dots,x_l\) with action of group
\(W_\Phi\) by the formula \(s_i(x_j) = x_j - C_{ij}x_j\). In more invariant terms \(V_\Phi\) is a space dual to the real vector
space generated by \(\alpha_1,\dots,\alpha_r\).
	
There is a standard action of the extended affine Weyl group \(W^{ae}_\Phi\) on \(V_{\Phi}\), namely for \(\lambda\in P\) translations 
\(t_\lambda\) act as \(x_i\mapsto x_i+(\lambda,\alpha_i)\). Subgroup \(W^{a}_\Phi\subset W^{ae}_\Phi\) is generated by reflections.
The reflection hyperplane for \(s_0\) is given by the equation \(\sum_{i=1}^r \mathfrak{m}_i x_i=1\), where marks \(\mathfrak{m}_i\) are defined by the expansion
of the highest root \(\theta=\sum_{i=1}^r \mathfrak{m}_i \alpha_i\).
	
\begin{Remark}
	In the standard convention \(\alpha_1\dots,\alpha_r\) are roots of finite root system and \(\alpha_0=-\theta+\delta\) is additional
	``affine'' root. Here as usual \(\delta\) is an imaginary root and \(\theta\) is highest root. Hence \(\sum_{i=0}^r \mathfrak{m}_i \alpha_i=\delta\),
	where \(\mathfrak{m}_0=1\) in agreement with formula \eqref{eq:marks def}.
	
	In the paper, we follow conventions of \cite{KNY15} in the numbering of simple roots in \(\Phi^a\). This convention in two cases disagrees
	with one above, namely for \(E_7^{(1)}\) affine root is \(\alpha_7\) and for \(E_8^{(1)}\) affine root is \(\alpha_8\). All statements
	in this section will be given in the invariant form, so this discrepancy will not lead to confusion.
\end{Remark}
	
The (closed) \emph{fundamental alcove} \(\mathrm{Alc}\subset V_\Phi\) is the set bounded by the hyperplanes of simple reflections in
\(W^a\). In coordinates \(\mathrm{Alc}\) is defined by the equations \(x_i\geq 0\), \(1\leq i \leq r\) and \(\sum_{i=1}^r \mathfrak{m}_i x_i\leq 1\).
Below we will identify facets of the fundamental alcove with the set \(\Delta^a\) of simple roots of the affine root system. 
The fundamental alcove is the fundamental domain for the action of \(W^a\) on \(V_\Phi\), see e.g. \cite[Th. 4.8]{Hum90}.
Let \(\widehat{\Omega}\subset W^{ae}\) denotes subgroup which preserves \(\mathrm{Alc}\). Clearly we have decomposition
\(W^{ae} \simeq \widehat{\Omega} \ltimes W^a\). Hence \(\widehat{\Omega}\simeq P/Q\).  Elements of the group \(\widehat{\Omega}\) permutes
facets of \(\mathrm{Alc}\), i.e. perform permutation of the set of \(\alpha_0,\dots,\alpha_r\), this permutation is automorphism
of affine Dynkin diagram. 

The following lemma is standard.
	
\begin{lemma}\label{lemma:stabadd}
	Let \(\vec{x}\in \mathrm{Alc}\), \(I\subset\Delta^a\) denotes the set of facets containing \(\vec{x}\) and \(W^\circ_I\) denotes
	subgroup of \(W^{a}\) generated by \(s_i\), \(i \in I\). Then the stabilizer of \(\vec{x}\) in \(W^{ae}\) has the
	form \(W^{ae}_{\vec{x}}=\widehat{\Omega}_{\vec{x}} \ltimes W^\circ_I\), where \(\widehat{\Omega}_{\vec{x}}\) is stabilizer of \(\vec{x}\) in \(\widehat{\Omega}\). 
\end{lemma}	

\begin{proof} 
	It follows from \cite[proof of Th. 4.8]{Hum90} that stabilizer of \(\vec{x}\) in \(W^a\) is \(W^\circ_I\). If \(\pi w (\vec{x})=\vec{x}\) for \(\pi\in\widehat{\Omega}\), \(w\in W^{a}\) then \(w (\vec{x})\in \mathrm{Alc}\) hence \(w (\vec{x})=\vec{x}\) (see \cite[proof of Th. 4.8]{Hum90}).
\end{proof}

We will need an analog of this statement for the reducible root system \(\Phi=\sqcup_{j=1}^l \Phi_j\). Here and below in such occasions we
assume roots systems \(\Phi_j\) are irreducible, and use notations like \(\Delta^a_j, W_j^{ae}, \mathrm{Alc}_j\) for the corresponding set
of affine simple roots, affine Weyl group, and fundamental alcove. Let \(\mathrm{Aut}_\Phi\) be a group of external automorphisms of \(\Phi\).
Any element \(\sigma \in \mathrm{Aut}_\Phi \) acts on simple roots \(\Delta_j\) as permutation or moves them to roots \(\Delta_j'\) of
isomorphic component. The fundamental alcove \(\mathrm{Alc}_\Phi\) is a product \( \prod \mathrm{Alc}_{j} \), the set of facets of
\(\mathrm{Alc}_\Phi\) is a union \(\sqcup \Delta^a_j\). We can extend finite and affine Weyl groups of \(\Phi\) by \(\mathrm{Aut}_\Phi\)
\begin{equation}
	\widetilde{W}_\Phi =\mathrm{Aut}_\Phi \ltimes \prod_{j=1}^l W_{j},\quad 
	\widetilde{W}^{ae}_\Phi = \mathrm{Aut}_\Phi \ltimes \prod_{j=1}^l W^{ae}_{j}.
\end{equation}

\begin{lemma}\label{lemma:stabadd:reducible} Let \(\vec{x}\in \mathrm{Alc}_\Phi\), and \(I,W_I^\circ\) are as in Lemma \ref{lemma:stabadd}.
Then the stabilizer of \(\vec{x}\) in \(\widetilde{W}^{ae}_\Phi\) has the form
\(\widetilde{W}^{ae}_{\Phi,x}=\big(\mathrm{Aut}_\Phi\ltimes \prod \widehat{\Omega}_j\big)_{\vec{x}} \ltimes W^\circ_I\).
\end{lemma}
The proof of this lemma is the same as proof of Lemma \ref{lemma:stabadd}, with replacement of \(\widehat{\Omega}\)
to \(\mathrm{Aut}_\Phi\ltimes \prod \widehat{\Omega}_j\).


Let \(\Omega \in W\) denotes subgroup which maps fundamental alcove \(\mathrm{Alc}\) into parallel one. It is easy to see that
any element \(\pi \in \widehat{\Omega}\) has unique decomposition \(\pi= t_{-\omega_\pi}\bar{\pi}\), where \(\bar{\pi} \in \Omega\).
Hence the groups \(\widehat{\Omega}\) and  \(\Omega\) are isomorphic. For reducible system \(\Phi=\sqcup_{j=1}^l \Phi_j\) we denote 
\(\widetilde{\Omega}_\Phi=\mathrm{Aut}_\Phi\ltimes \prod \Omega_j\), this is subgroup of \(\widetilde{W}_\Phi\) which maps 
 alcove \(\mathrm{Alc}\) into parallel one.

\begin{corol}\label{corol:Humphereys}
	Let \(\Phi=\sqcup_{j=1}^l \Phi_j\) and \(\vec{x} \in \mathrm{Alc}\), then the group
	\(\{w\in \widetilde{W}_\Phi| w\vec{x}-\vec{x}\in P\}\) has the form
	\(\widetilde{\Omega}_{(\vec{x})}\ltimes W^\circ_I\) where \(I,W_I^\circ\) are as in Lemma \ref{lemma:stabadd} and \(\widetilde{\Omega}_{(\vec{x})}\) is the subgroup of \(\widetilde{\Omega}\).
\end{corol}
\begin{proof}
	 The condition \(w\vec{x}-\vec{x}\in P\) is equivalent to existence of \(\lambda \in P\) such that \(t_{\lambda} w\) belong
	 the the stabilizer of \(\vec{x}\) in \(\widetilde{W}^{ae}\). Hence the statement follows from Lemma \ref{lemma:stabadd:reducible}. 
\end{proof}

	The logic here is similar to \cite[Sec.~2.10]{Hum95}. Note that \(\widetilde{\Omega}_{(\vec{x})}\) is not stabilizer of \(\vec{x}\) in \(\widetilde{\Omega}\) but rather finite part of stabilizer in \(\mathrm{Aut}_\Phi \ltimes \prod_{j=1}^l \widehat{\Omega}_{j}\).

\paragraph{Multiplicative case.} Recall now that our parameters \(\vec{a}=(a_0,\dots,a_r)\) are multiplicative root variables.
For any root \(\alpha=\sum_{i=0}^r k_i \alpha_i\)
we assign root variable \(\prod_{i=0}^r a_i^{k_i}\).
	
\begin{thm}\label{thm:stab}
	Let \(|q|\neq 1\), $\vec{a}\in \mathcal{A}$. Then there exist \(I \subset \Delta^a\) and \(g \in W^{ae}\) such that for stabilizer of
	\(\vec{a}''=g (\vec{a})\) we have 
	\begin{equation}\label{eq:Stab:inclusion}
		W^{ae}_{\vec{a}''} = \Omega_{[\vec{a}]} \ltimes \prod\nolimits_{j=1}^l W^o_{j}.
	\end{equation}
	Here \(\Phi_j\), \(j=1,\dots,l\) are finite root systems corresponding to connected components of the subgraph whose vertex set is \(I\)
	in the graph of affine Dynkin diagram \(\Phi^a\), \(W^o_{j}\) is group generated by reflections by roots of \(\Phi_j\) such that
	corresponding multiplicative root variables equal to 1. The  \(\Omega_{[\vec{a}]}\) is a certain subgroup in
	\(\widehat{\Omega}_I\ltimes \prod\nolimits_{j=1}^l \Omega_{\Phi_j}\), where \(\widehat{\Omega}_I\) is a subgroup of
	\(\widehat{\Omega}\) preserving \(I\).
\end{thm}

Through the paper, we assume that \(q\) is generic so the condition \(|q|\neq 1\) holds. It is easy to see that if Theorem \ref{thm:stab}
holds fo generic \(q\) then it holds for any \(q\) which is not a root of unity.

\begin{proof}
	There is a map from \(\mathcal{A} \rightarrow V_{\Phi}\) given  by the formula
	\(x_i=(\mathrm{Re} \log a_i)/(\mathrm{Re} \log q)\), \(i=1,\dots,r\).
	This map is equivariant with respect to the action of \(W^{ae}\). Indeed, it is easy to check for
	reflections in \(W\) and for translation in \(P\). 
	
	Now fix \(\vec{a}\in \mathcal{A}\), let \(\vec{x} \in V_{\Phi}\) denotes its image. Due to equivariance we have inclusion
	of stabilizers \(W^{ae}_{\vec{a}}\subset W^{ae}_{\vec{x}}\).
	There is \(g'\in W^{a}\) such that \(\vec{x}'=g' (\vec{x})\in \mathrm{Alc}\), let \(\vec{a}'=g' (\vec{a})\).
	As in  Lemma \ref{lemma:stabadd} we introduce \(I\) and we have 
	\begin{equation}
		W^{ae}_{\vec{x}'}=\widehat{\Omega}_{\vec{x}'} \ltimes W^\circ_I =
		\widehat{\Omega}_{\vec{x}'} \ltimes \prod\nolimits_{j=1}^l W_{\Phi_j}.  
	\end{equation}
	Since \(I\) cannot coincide with \(\Delta^a\) the root systems \(\Phi_j\) are of finite type.
	Clearly \( \widehat{\Omega}_{\vec{x}'} \subset \widehat{\Omega}_{I}\).
	
	It follows from definition of \(I\) that \(|a_i'|=1\) iff  \(\alpha_i \in I\). Denote \(a_i'=\exp(2 \pi \ri y_i)\) for real
	\(y_i\). One can consider collection \(y_i\) as an element of \(V_{\sqcup \Phi_j}=\oplus_{j=1}^l V_{\Phi_j}\), note that
	\(\vec{y}\) is defined up to the integer shift \(\vec{y}\mapsto \vec{y}+\lambda\) where \(\lambda \in 	 \oplus_{j=1}^l P_j\). 
	Then there is \(g''\in \prod_{j=1}^l W_{\Phi_j}\) and  \(\lambda \in \oplus_{j=1}^l Q_j\) such that
	\(\vec{y}''=g'' (\vec{y})+\lambda\) belongs to the product of fundamental alcoves
	\(\prod_j \mathrm{Alc}_j\). If some \(w\in W^{ae}_{\vec{x}'}\) preserves all \(a_i''\) for \(\alpha_i\in I\),
	then \(w(\vec{y}'')-\vec{y}''\in P_{\sqcup \Phi_j}\). 
	
	The group \(\widehat{\Omega}_{\vec{x}'}\) acts on \(\sqcup \Phi_j\) as a subgroup of \(\mathrm{Aut}_{\sqcup \Phi_j}\). Hence
	by Corollary \ref{corol:Humphereys}, the stabilizer of set of all \(a_i''\) for \(\alpha_i\in I\) has the form 
	\begin{equation}
		\Big(\widehat{\Omega}_{\vec{x}'} \ltimes \prod\nolimits_{j=1}^l \Omega_{\Phi_j} \Big)_{(\vec{y}'')} \ltimes W^\circ
	\end{equation}
	where \(W^\circ\) is generated by reflections through hyperplanes containing \(\vec{y}''\). 
		
%
%
%
	
	Condition that hyperplane contains \(\vec{y}''\) is equivalent to the condition that corresponding multiplicative root variable
	is equal to 1. Hence the defined above group \(W^\circ\)  preserves the whole \(\vec{a}''\). Therefore we get the
	formula~\eqref{eq:Stab:inclusion}, where \(\Omega_{[\vec{a}]}\) is stabilizer of \(\vec{a}''\) in
	\(\Big(\widehat{\Omega}_{\vec{x}'} \ltimes \prod\nolimits_{j=1}^l \Omega_{\Phi_j} \Big)_{(\vec{y}'')}\).
\end{proof}

Due to Lemma \ref{lemma:reflection} subgroup \(W^\circ\) acts trivially on \(\mathcal{X}_{\vec{a}}\).
Hence, the folding transformations (and folding subgroups) could come only from the group \(\widehat{\Omega}\ltimes \prod\nolimits \Omega_{\Phi_j}\).\footnote{It is not true to say that any folding subgroup can be conjugated to the subgroup of \(\widehat{\Omega}\ltimes \prod\nolimits \Omega_{\Phi_j}\). But this is true up to factors from \(W^\circ\), i.e.
image of any folding subgroup to the automorphism group of \(\mathcal{X}_{\vec{a}}\) can be conjugated to the image of subgroup in \(\widehat{\Omega}\ltimes \prod\nolimits \Omega_{\Phi_j}\).}
	
\subsection{Reduction to type A }
If one of the systems \(\Phi_j\) is of $A$-type, say \(A_n\) then the group \(\Omega\) is cyclic group of order \(n+1\), \(\Omega_{A_n}\simeq C_{n+1}\). The generator of this cyclic group is 		$\bar{\pi}_{n+1}=s_{n} s_{n{-}1}\ldots s_2 s_1$, where we used standard consecutive numbering of the simple roots.

It appears that it is sufficient to consider the case that all root systems \(\Phi_j\) are of type \(A\). 
\begin{lemma}	\label{lemma:detoa}
	For  finite irreducible simply laced root system \(\Phi\) the group \(\Omega_{\Phi}\in W_{\Phi}\) can be conjugated into
	the subgroup of \(\prod_{j=1}^{l'} \Omega_{\Phi'_{j}}\), where \(\Phi'_{j}\) are \(A\)-type roots systems corresponding to connected components
	\(\Delta'_{j}\) in Dynkin diagram \(\Delta\). 
\end{lemma}
	
\begin{proof}
	We will use the following numbering of roots in root diagrams of \(D_n, E_6, E_7\)
	\begin{center}
		\begin{tikzpicture}[elt/.style={circle,draw=black!100,thick, inner sep=0pt,minimum size=2mm},scale=0.75]
			\path 	(-3,0) 	node 	(a1) [elt] {}
			(-2,0) 	node 	(a2) [elt] {}
			(-1,0) 	node 	(a3)  {}
			( 0,0) node  	(a4)  {}
			( 1,0) 	node  	(a5) [elt] {}
			( 1+cos{60},-sin{60}) 	node 	(a6) [elt] {}
			( 1+cos{60},sin{60})	node 	(a7) [elt] {};
			\draw [black,line width=1pt] (a1) -- (a2)  (a5) --(a6) (a5) -- (a7);
			\node at ($(a1.south west) + (-0,-0.3)$) 	{$\alpha_{1}$};
			\node at ($(a2.south west) + (0,-0.3)$)  {$\alpha_{2}$};
			\node at ($(a3) + (0,0)$)  {$\dots$};
			\node at ($(a4) + (0,0)$)  {$\dots$};
			\node at ($(a5.east) + (0.7,0)$)  {$\alpha_{n-2}$};
			\node at ($(a6.east) + (0.7,-0.1)$)  {$\alpha_{n-1}$};	
			\node at ($(a7.east) + (0.5,+0)$)  {$\alpha_{n}$};		
		\end{tikzpicture}
		\quad
		\begin{tikzpicture}[elt/.style={circle,draw=black!100,thick, inner sep=0pt,minimum size=2mm},scale=0.75]
			\path 	(-2,0) 	node 	(a1) [elt] {}
			(-1,0) 	node 	(a2) [elt] {}
			( 0,0) node  	(a3) [elt] {}
			( 1,0) 	node  	(a4) [elt] {}
			( 2,0) 	node 	(a5) [elt] {}
			( 0,1)	node 	(a6) [elt] {};
			\draw [black,line width=1pt ] (a1) -- (a2) -- (a3) -- (a4) -- (a5)   (a3) -- (a6) 		;
			\node at ($(a1.south) + (0,-0.3)$) 	{$\alpha_{1}$};
			\node at ($(a2.south) + (0,-0.3)$)  {$\alpha_{2}$};
			\node at ($(a3.south) + (0,-0.3)$)  {$\alpha_{3}$};
			\node at ($(a4.south) + (0,-0.3)$)  {$\alpha_{4}$};	
			\node at ($(a5.south) + (0,-0.3)$)  {$\alpha_{5}$};		
			\node at ($(a6.west) + (-0.3,0)$) 	{$\alpha_{6}$};	
		\end{tikzpicture}
		\quad
		\begin{tikzpicture}[elt/.style={circle,draw=black!100,thick, inner sep=0pt,minimum size=2mm},scale=0.75]
			\path 	(-3,0) 	node 	(a1) [elt] {}
			(-2,0) 	node 	(a2) [elt] {}
			( -1,0) node  	(a3) [elt] {}
			( 0,0) 	node  	(a4) [elt] {}
			( 1,0) 	node 	(a5) [elt] {}
			( 2,0)	node 	(a6) [elt] {}
			( 0,1)	node 	(a0) [elt] {};
			\draw [black,line width=1pt ] (a1) -- (a2) -- (a3) -- (a4) -- (a5) --  (a6) (a4) -- (a0);
			\node at ($(a1.south) + (0,-0.2)$) 	{$\alpha_{1}$};
			\node at ($(a2.south) + (0,-0.2)$)  {$\alpha_{2}$};
			\node at ($(a3.south) + (0,-0.2)$)  {$\alpha_{3}$};
			\node at ($(a4.south) + (0,-0.2)$)  {$\alpha_{4}$};	
			\node at ($(a5.south) + (0,-0.2)$)  {$\alpha_{5}$};		
			\node at ($(a6.south) + (0,-0.2)$) 	{$\alpha_{6}$};	
			\node at ($(a0.north) + (0,0.2)$) 	{$\alpha_{0}$};		
		\end{tikzpicture}	
	\end{center}

	For \(\Phi=E_8\) the group \(\Omega\) is trivial and there is nothing to prove. For \(E_6\) and \(E_7\) we have the formulas  
	\begin{align}
		&E_6:\Omega=\{e,s_{123543263} s_{1245} s_{123543263}^{-1}, 	s_{123543263} s_{2154} s_{123543263}^{-1}\}\simeq C_3 \label{Omega_E6}\\
		&E_7:\Omega=\{e,s_{123405645342} s_{031} 	s_{123405645342}^{-1}\}\simeq C_2 \label{Omega_E7}
	\end{align}

	For the \(D\)-type recall that Weyl group \(W_{D_n}\) is isomorphic to the group of linear transformations of the vector
	space with basis \(\epsilon_1,\dots,\epsilon_n\) which permutes basis vectors and change signs of even number of them,
	simple roots are expressed in terms of such basis as $\alpha_i=e_i-e_{i+1}, i=1\ldots n-1$, $\alpha_n=e_{n-1}+e_n$.
	Each element of \(W_{D_n}\) can be decomposed into product of cycles \((i_1,\dots,i_k)\), each of them is either positive or
	negative. Here cycle of length \(k\) is said to be positive if its \(k\)-th power is identity and negative if its \(k\)-th power
	is minus identity. See e.g. \cite[Sec.~7]{Car72} for details. Two elements \(w_1,w_2\in W_{D_n}\) are conjugated if they have
	the same cyclic type (including signs) except the case of all cycles are even and positive, in this case there are two conjugacy
	classes, see \cite[Prop 25]{Car72}.

	Let \(n=2k+1\). The group \(\Omega_{D_{2k+1}}\) is isomorphic to \(C_4\) and generated by the element 
	\begin{equation}
		w=\begin{pmatrix}
				\epsilon_1 & \epsilon_2 & \epsilon_3 &\dots& \epsilon_{k+1}  & \dots & \epsilon_{2k} & \epsilon_{2k{+}1}\\ 
				\epsilon_{2k{+}1} & -\epsilon_{2k} & -\epsilon_{2k{-}1} &\dots & -\epsilon_{k+1} &\dots  & -\epsilon_{2} & -\epsilon_{1}
			\end{pmatrix}.
	\end{equation} 
	We have \(w=g^{-1}w' g\), where 
	\begin{align}
		w'&= 
		\begin{pmatrix}
			\epsilon_1 & \epsilon_2 & \ldots & \epsilon_{2k{-}3} & \epsilon_{2k{-}2} & \epsilon_{2k{-}1}  & \epsilon_{2k} & \epsilon_{2k{+}1}\\ 
			\epsilon_{2} & \epsilon_{1} & \ldots & \epsilon_{2k{-}2} & \epsilon_{2k{-}3} & -\epsilon_{2k{+}1}  & -\epsilon_{2k} & \epsilon_{2k{-}1}
		\end{pmatrix},\\
	 g&=\begin{pmatrix}
	 	\epsilon_1 & \epsilon_2 &  \epsilon_3 &  \dots & \epsilon_k & \epsilon_{k{+}1} &  \epsilon_{k+2} & \dots & \epsilon_{2k{-}1} & \epsilon_{2k} & \epsilon_{2k{+}1} \\ 
	 	\epsilon_{2k-1} & \epsilon_1 & \epsilon_3 &  \dots & \epsilon_{2k-3} & (-1)^k \epsilon_{2k} & -\epsilon_{2k-2} & \dots & -\epsilon_4 & -\epsilon_2 & -\epsilon_{2k{+}1},
	 \end{pmatrix}
	\end{align} 
	Note that \(w'=s_{13\dots 2k{-}3 \;2k\; 2k{-}1\; 2k{+}1}
	\in \prod\nolimits \Omega_{\Phi'_{j}}\) for the subset \(I'=\{1,3,\dots,2k{-}3,2k{-}1,2k,2k{+}1\}\).
	The corresponding root system is \((k-1)A_1+A_3\).
	
	Let \(n=2k\). The group \(\Omega_{D_{2k}}\) is isomorphic to \(C_2\times C_2\) and has three nontrivial elements \(w_1, w_2^+, w_2^-\)
	\begin{equation}
		w_1=
		\begin{pmatrix}
			\epsilon_1 & \epsilon_2 & \dots & \epsilon_{2k{-}1} & \epsilon_{2k}\\ 
			-\epsilon_{1} & \epsilon_{2} & \dots & \epsilon_{2k{-}1} & -\epsilon_{2k}
		\end{pmatrix},\quad
		w_2^\pm=
		\begin{pmatrix}
			\epsilon_1 & \epsilon_2 & \dots & \epsilon_{2k{-}1} & \epsilon_{2k}\\ 
			\pm \epsilon_{2k} & -\epsilon_{2k{-}1} & \dots & -\epsilon_{2} & 	
			\pm \epsilon_{1}
		\end{pmatrix}.
	\end{equation} 
	Let us define \(g, w_1',w_2' \) by the formulas 
	\begin{align}
		w_1'= \begin{pmatrix}
			\epsilon_1 &  \dots & \epsilon_{2k{-}2} & \epsilon_{2k{-}1} & \epsilon_{2k}\\ 
			\epsilon_{1} &  \dots & \epsilon_{2k{-}2} & -\epsilon_{2k{-}1} & -\epsilon_{2k}
		\end{pmatrix},\quad
		w_2'= \begin{pmatrix}
			\epsilon_1 & \epsilon_2 &\dots & \epsilon_{2k{-}3} & \epsilon_{2k{-}2} & \epsilon_{2k{-}1} & \epsilon_{2k}\\ 
			\epsilon_{2} & \epsilon_1 & \dots & \epsilon_{2k{-}2} & \epsilon_{2k{-}3} & \epsilon_{2k} & \epsilon_{2k{-}1}
		\end{pmatrix},
		\\
		g=\begin{pmatrix}
			\epsilon_1 & \epsilon_2 &  \dots &  \epsilon_{k-1} &  \epsilon_{k}  &\dots  & \epsilon_{2k{-}2} & \epsilon_{2k{-}1} & \epsilon_{2k} \\ 
			\epsilon_{2k{-}1} & \epsilon_{2k{-}3} &  \dots  & \epsilon_{1} & -\epsilon_2 &\dots  & -\epsilon_{2k{-}4} & -\epsilon_{2k{-}2} & (-1)^{k} \epsilon_{2k}
		\end{pmatrix},
	\end{align} 
 Then we have \(w_1=g^{-1}w_1' g\), and \(w_2^+=g^{-1}w_2' g\) for even \(k\), \(w_2^-=g^{-1}w_2' g\) for odd \(k\). Note that \(w_1'=s_{2k{-}1\, 2k}\) and \(w_2'=s_{13\dots 2k{-1}}\). So \(w_1',w_2' \in \prod\nolimits \Omega_{\Phi'_{j}}\)
 for the subset \(I'=\{1,3,\dots,2k{-}3,2k{-}1,2k\}\), the corresponding root system is \((k+1)A_1\).
\end{proof}

	As was shown in previous subsection all folding subgroups come from subgroups of the groups of the form 
	\(\widehat{\Omega}_I\ltimes \prod\nolimits_{j=1}^l \Omega_{\Phi_j}\) corresponding to certain \(I\subset \Delta^a\). For any \(\Phi_j\) which is of type \(D_n, E_6, E_7\) we can pick \(g_j \in W_j\) as in proof of Lemma~\ref{lemma:detoa}. The following Lemma says that the choice of \(g_j\) is actually canonical 
	
	\begin{lemma}
		Let \(\Phi\) be of type \(D_n, E_6, E_7\), \(g \in W_\Phi\) as in proof of Lemma~\ref{lemma:detoa}. 
		Let \(\sigma \in \mathrm{Aut}_\Phi\). Then \(\sigma g=g \sigma\).
	\end{lemma}
	\begin{proof}
		For \(\Phi=D_n\), \(n>4\) the only automorphism of Dynkin is given by 
	\(
			\sigma= \begin{pmatrix}
	 \epsilon_1 & \epsilon_2 &\dots & \epsilon_{n{-}1} & \epsilon_{n}\\ 
		\epsilon_{1} & \epsilon_2 & \dots & \epsilon_{n-1} &  -\epsilon_{n}
	\end{pmatrix},
	\) 
	and clearly commutes with \(g\). For \(E_7\) there is no automorphisms of \(\Delta\). It remains to check cases \(D_4\) and \(E_6\), this is a direct computation.
	\end{proof}
	
	Due to lemma the element \(g= \prod g_j\) commutes with the group \(\widehat{\Omega}_I\). Then using \(g\)
	we conjugate the group \(\widehat{\Omega}_I\ltimes \prod\nolimits_{j=1}^l \Omega_{\Phi_j}\) to
	\(\widehat{\Omega}_I\ltimes \prod\nolimits_{j=1}^{l'} \Omega_{\Phi'_j}\), where subset \(I'\subset I\) such that
	all corresponding root systems \(\Phi'_{j}\) are of \(A\) type. 	
	
	Therefore all folding subgroups come from from subgroups of the groups of the form \(\widehat{\Omega}\ltimes \prod\nolimits_{j'=1}^{l'} \Omega_{\Phi'_{j}}\)  where all \(\Phi'_{j}\) are of \(A\) type.
%
%
%

\subsection{Selection rules} \label{ssec:selection}
We will see that to find disctinct (true) folding transformations, it is sufficient to consider rather special subsets \(I=\sqcup I_j\)
and special elements of the group \(\widehat{\Omega}\ltimes \prod\nolimits \Omega_{\Phi_j}, \, \Phi_j=A_{n_j}\). 
More precisely, corresponding selection rules will follow from the invariant set condition.
Note that such selection rules are only \emph{necessary}. As we mentioned before in order to show that all folding candidates will give us foldings and these foldings are not equivalent 
we will use geometry in Sections \ref{sec:geom_gs}, \ref{sec:geom_answ}.

It is instructive to consider at first two extreme cases.	

\paragraph{Case 1.} Let  $w=\pi\in \widehat{\Omega}$. In this case, we only have to check that the invariant set
\(\mathcal{A}_w\) is at least two dimensional (true folding condition). This condition means that the number of orbits of \(\pi\) on
the set \(\Delta^a\) (vertices of the affine Dynkin diagram) should be at least 2. 
We list folding candidates satisfying this condition in Section \ref{ssec:answers algebraic}.
We list folding candidates with one orbit in Appendix \ref{sec:ntrf}.

\paragraph{Case 2.} Let \(w\in \prod\nolimits_{j=1}^l \Omega_{\Phi_j}\). 
Recall that any element  of \(\Omega_{A_n}\) has the form \(\bar{\pi}_{n{+}1}^k\),  \(0\leq k \leq n\), where
$\bar{\pi}_{n{+}1}=s_{n} s_{n-1}\ldots s_2 s_1$.

\begin{lemma} \label{lemma:Omega_p1}  
	For $A_n$ let \(\gcd(k,n+1)=d\) and \(m=(n+1)/d-1\). Then \(\bar{\pi}_{n{+}1}^k\) is conjugated to \(\prod_{j=1}^{d} \bar{\pi}_{m+1}^{(j)}\),
	where \(\bar{\pi}_{m+1}^{(j)}=s_{(j-1)(m+1)+m}\cdot \dots \cdot s_{(j-1)(m+1)+1}\).
\end{lemma}

In other words, \(\bar{\pi}_{m+1}^{(j)}\) is a standard generator of \(\Omega_{A_{m}}\), where \(A_m\) has simple roots
\(\alpha_{(j-1)(m+1)+i}, \, i=1\ldots m\). We can also say, that after removing simple roots $\alpha_i$
with $m+1|i$, the Dynkin diagram of $A_n$ breaks up into $d$ subdiagrams $A_m$. 

\begin{proof}
	Under isomorphism $W(A_n)=S_{n+1}$ both \(\bar{\pi}_{n{+}1}^k\) and \(\prod_{j=1}^{d} \bar{\pi}_{m+1}^{(j)}\) are products of \(d\)
	cycles of length \(m+1\). It remains to note that conjugacy classes in \(S_{n+1}\) are labeled by the cyclic types.
\end{proof}
  
It follows from the lemma that it is sufficient to consider elements of the form \(w=\prod_{j=1}^{l} \bar{\pi}^{(j)}_{n_l+1}\) for
standard generators $\bar{\pi}^{(j)}_{n_l+1}$ of $\Phi_j=A_{n_j}$. 
The choice of \(\bar{\pi}_{n_l+1}\) depends on the orientation of Dynkin diagram of $\Phi_{j}$, but up to conjugacy, we can pick any orientation.
 
Below we will visualize the set \(I\) as a coloring of affine Dynkin diagram, namely vertices corresponding to \(\alpha_i\in I\) we color in black and other vertices we color in white. Let us now find condition on the invariant subset \(\mathcal{A}_w\).

For the root system \(A_n\) we have in terms of the root variables
\begin{multline}
	s_{n} s_{n-1}\ldots s_2 s_1 \Big(a_1,a_2,a_3,\dots,a_n\Big)=
	s_{n} s_{n-1}\ldots s_2  \Big(a_1^{-1},a_1a_2,a_3,\dots,a_n\Big)\\
	=s_{n} s_{n-1}\ldots s_3  \Big(a_2,(a_1 a_2)^{-1},a_1a_2 a_3,a_4,\dots,a_n\Big)=\dots=\Big(a_2,a_3,\dots,a_n,(a_1\dots a_n)^{-1}\Big).
\end{multline}
Hence  \(w\)-invariance condition on black vertices $b$ give
\begin{equation}\label{a_b_c2}
	b\in \Phi_j: a_b=\exp \left(2 \pi \ri \frac{m_j}{n_j+1}\right) 
\end{equation}
for some $0\leq m_j \leq n_j$. Moreover if \(\gcd(m_j,n_j+1)=d>1\) then the reflections \(\frac{n_j+1}{d}i,\, i=1\ldots d-1\)
act trivially and corresponding roots can be excluded from \(I\). So we can assume coprimeness of $m_j$ and $n_j+1$.

Action of $w$ will multiply white vertex \(c \in \Delta^{a}\setminus I\) root variable by root variables
of adjacent ($\in N(c)$) and black ($\in I$) vertices. Namely, we have 
\begin{equation}
	w(a_{c})= a_c\prod_{b \in I \cap N(c)} \exp \left( 2 \pi \ri \frac{i_b m_b}{n_{b}+1}\right),
\end{equation}
where \(n_{b}=n_j\), \(m_{b}=m_j\) for \(\Phi_j \ni b\) and \(i_{b}\in 1 \dots n_{b}\) is the number of vertex \(b\) in \(\Phi_j\) in the
numbering used in the definition of \(\bar{\pi}\). The \(w\)-invariance condition (on white vertices)
reads that product in r.h.s. of above formula equals $1$. 

\medskip 
We summarize discussion above as a set of selection rules on the coloring of affine Dynkin diagram:
\begin{itemize}
	\item Each connected component of \(I\) is Dynkin diagram of \(A\) type: \(I=\sqcup \Delta_j\) where \(\Delta_j\) corresponds to root system \(\Phi_j=A_{n_j}\).
	\item Number vertices of each \(\Delta_j\) successively from $1$ to $n_j$, starting from either end.
	Then for all \(\Delta_j\) exist integer number \(m_j, \, 0{<}m_j{\leq}n_j,\, \gcd (m_j,n_j+1)=1\), such that for any white vertex $c$
	\begin{equation}\label{selection_rule_c2}
	\sum_{b\in I\cap N(c)} \frac{i_b m_{j(b)}}{n_{j(b)}+1}\in \mathbb{Z}, 
	\end{equation}
    where the sum runs over the set of black vertices adjacent to \(c\), $j(b)$ is number of $A$-component of $b$ i.e. \(b\in \Delta_{j(b)}\), and $i_b$ is number of $b$ in \(\Delta_{j(b)}\).
	\item There are at least two white vertices: $|\Delta^a\setminus I|\geq 2$.
\end{itemize}
The last item is necessary to have at least two dimensional \(\mathcal{A}_w\), namely to have a true folding.

It is sufficient to classify colorings up to the action of the group \(\widehat{\Omega}\) because such action corresponds to conjugation of $w$.
We give a complete list of colorings that satisfy the above selection rules in Section \ref{ssec:answers algebraic}. 
Additionally, in Appendix \ref{sec:ntrf} we give complete list of colorings with $|\Delta^a\setminus I|=1$.

Group element will be given by $w=\prod_{j=1}^l \bar{\pi}^{(j)}$, where $\bar{\pi}^{(j)}\in \Omega_j$. The order of \(w\) is least common multiple of $n_j+1$.
The values $m_j$ determine connected component $\mathcal{A}_w$ of the invariant set $\mathcal{A}^w$ by
\eqref{a_b_c2}. We will see below that for given $w$ there are
either one such component or several (usually two) components related by the action of normalizer \(N(\langle w \rangle, W^{ae})\).
So it is sufficient to consider only one component. 
          
\begin{Example}\label{ex:colorings}
Let us consider case $\Phi^a=E_7^{(1)}$. In this example we consider four subsets \(I \subset \Delta^a\), they are given in Fig. \ref{fig:excolor}.

First consider $I=3A_1$, it is given on the upper left diagram of Fig. \ref{fig:excolor}. Here there is no choice for \(m_j,i_b\) namely  $m_j=1$ and $i_b=1$. Summands from \eqref{selection_rule_c2} are denoted on the figure as blue fractions.
For all white nodes, we compute sums \eqref{selection_rule_c2}, they are written in green. All of them are integer, so the selection rule is satisfied and we have folding transformation $w=s_0s_1s_3$.

Consider $I=A_1+A_3$ as on upper right diagram of Fig. \ref{fig:excolor}. We choose here $m_j=1$ for both components $A_3$ and again write summands near black nodes and sums near black ones. One of the sums is not
an integer (it is written in red) so the selection rule fails. It is easy to see that if we set $m_j=3$ for \(A_3\) component then the selection rule also fails.

Consider now $I=2A_3$, the corresponding diagrams are in the bottom of Fig. \ref{fig:excolor}. In the left diagram we choose 
$m_j=1$ for left $A_3$ and $m_j=3$ in right $A_3$ and failed. However, in the right diagram, taking both
$m_j=1$ we succeed and obtain folding $w=s_{321}s_{765}$. The corresponding irreducible component in \(\mathcal{A}_w\) is defined by
$a_1=a_2=a_3=a_5=a_6=a_7=\ri$.  Another component could be obtained by taking both $m_j=3$, in this case we obtain $-\ri$ on black vertices.
These components are related by normalizer.

\begin{figure}[h]

\begin{tikzpicture}[elt/.style={circle,draw=black!100,thick, inner sep=0pt,minimum size=2mm},scale=1]
			
			\begin{scope}
			
			\path 	(-3,0) 	node 	(a1) [elt,fill] {}
			(-2,0) 	node 	(a2) [elt] {}
			( -1,0) node  	(a3) [elt,fill] {}
			( 0,0) 	node  	(a4) [elt] {}
			( 1,0) 	node 	(a5) [elt] {}
			( 2,0)	node 	(a6) [elt] {}
			( 3,0)	node 	(a7) [elt] {}
			( 0,1)	node 	(a0) [elt,fill] {};
			\draw [black,line width=1pt ] (a1) -- (a2) -- (a3) -- (a4) -- (a5) --  (a6) -- (a7) (a4) -- (a0);
			
			\node at ($(a1.south) + (0,-0.2)$) 	{\small $-1$};
			\node at ($(a3.south) + (0,-0.2)$) 	{\small $-1$};
			
			\node at ($(a0.west) + (-0.2,0)$) 	{\small $-1$};

			\node at ($(a1.north) + (0,0.2)$) 	{\small \color{blue} $\frac12$};
			\node at ($(a2.north) + (0,0.2)$) 	{\small \color{dgreen} $1$};
			\node at ($(a3.north) + (0,0.2)$) 	{\small \color{blue} $\frac12$};
			
			\node at ($(a5.north) + (0,0.2)$) 	{\small \color{dgreen} $0$};
			\node at ($(a6.north) + (0,0.2)$) 	{\small \color{dgreen} $0$};
			\node at ($(a7.north) + (0,0.2)$) 	{\small \color{dgreen} $0$};
			
			\node at ($(a4.south) + (0,-0.2)$) 	{\small \color{dgreen} $1$};
			\node at ($(a0.north) + (0,0.2)$) 	{\small \color{blue} $\frac12$};
	
	\end{scope}
			
				\begin{scope}[xshift=7.5cm]
			
			\path 	(-3,0) 	node 	(a1) [elt,fill] {}
			(-2,0) 	node 	(a2) [elt,fill] {}
			( -1,0) node  	(a3) [elt,fill] {}
			( 0,0) 	node  	(a4) [elt] {}
			( 1,0) 	node 	(a5) [elt] {}
			( 2,0)	node 	(a6) [elt] {}
			( 3,0)	node 	(a7) [elt] {}
			( 0,1)	node 	(a0) [elt,fill] {};
			\draw [black,line width=1pt ] (a1) -- (a2) -- (a3) -- (a4) -- (a5) --  (a6) -- (a7) (a4) -- (a0);
			
			\node at ($(a1.north) + (0,0.2)$) 	{\small \color{blue} $\frac14$};
			\node at ($(a2.north) + (0,0.2)$) 	{\small \color{blue} $\frac24$};
			\node at ($(a3.north) + (0,0.2)$) 	{\small \color{blue} $\frac34$};
			\node at ($(a0.north) + (0,0.2)$) 	{\small \color{blue} $\frac12$};
			
			\node at ($(a5.north) + (0,0.2)$) 	{\small \color{dgreen} $0$};
			\node at ($(a6.north) + (0,0.2)$) 	{\small \color{dgreen} $0$};
			\node at ($(a7.north) + (0,0.2)$) 	{\small \color{dgreen} $0$};
			
			\node at ($(a4.south) + (0,-0.2)$) 	{\small \color{red} $\frac54$};

	\end{scope}

			\begin{scope}[yshift=-2.5cm]
			
			\path 	(-3,0) 	node 	(a1) [elt,fill] {}
			(-2,0) 	node 	(a2) [elt,fill] {}
			( -1,0) node  	(a3) [elt,fill] {}
			( 0,0) 	node  	(a4) [elt] {}
			( 1,0) 	node 	(a5) [elt,fill] {}
			( 2,0)	node 	(a6) [elt,fill] {}
			( 3,0)	node 	(a7) [elt,fill] {}
			( 0,1)	node 	(a0) [elt] {};
			\draw [black,line width=1pt ] (a1) -- (a2) -- (a3) -- (a4) -- (a5) --  (a6) -- (a7) (a4) -- (a0);

			\node at ($(a1.north) + (0,0.2)$) 	{\small \color{blue} $\frac14$};
			\node at ($(a2.north) + (0,0.2)$) 	{\small \color{blue} $\frac24$};
			\node at ($(a3.north) + (0,0.2)$) 	{\small \color{blue} $\frac34$};
			
			\node at ($(a5.north) + (0,0.2)$) 	{\small \color{blue} $\frac34$};
			\node at ($(a6.north) + (0,0.2)$) 	{\small \color{blue} $\frac64$};
			\node at ($(a7.north) + (0,0.2)$) 	{\small \color{blue} $\frac94$};
			
			\node at ($(a4.south) + (0,-0.2)$) 	{\small \color{red} $\frac32$};
			\node at ($(a0.north) + (0,0.2)$) 	{\small \color{dgreen} $0$};
	
	\end{scope}
		
			\begin{scope}[xshift=7.5cm, yshift=-2.5cm]
			
			\path 	(-3,0) 	node 	(a1) [elt,fill] {}
			(-2,0) 	node 	(a2) [elt,fill] {}
			( -1,0) node  	(a3) [elt,fill] {}
			( 0,0) 	node  	(a4) [elt] {}
			( 1,0) 	node 	(a5) [elt,fill] {}
			( 2,0)	node 	(a6) [elt,fill] {}
			( 3,0)	node 	(a7) [elt,fill] {}
			( 0,1)	node 	(a0) [elt] {};
			\draw [black,line width=1pt ] (a1) -- (a2) -- (a3) -- (a4) -- (a5) --  (a6) -- (a7) (a4) -- (a0);
			
			\node at ($(a1.south) + (0,-0.2)$) 	{\small $\ri$};
			\node at ($(a2.south) + (0,-0.2)$) 	{\small $\ri$};
			\node at ($(a3.south) + (0,-0.2)$) 	{\small $\ri$};
			
			\node at ($(a5.south) + (0,-0.2)$) 	{\small $\ri$};
			\node at ($(a6.south) + (0,-0.2)$) 	{\small $\ri$};
			\node at ($(a7.south) + (0,-0.2)$) 	{\small $\ri$};
			
			\node at ($(a1.north) + (0,0.2)$) 	{\small \color{blue} $\frac14$};
			\node at ($(a2.north) + (0,0.2)$) 	{\small \color{blue} $\frac24$};
			\node at ($(a3.north) + (0,0.2)$) 	{\small \color{blue} $\frac34$};
			
			\node at ($(a5.north) + (0,0.2)$) 	{\small \color{blue} $\frac14$};
			\node at ($(a6.north) + (0,0.2)$) 	{\small \color{blue} $\frac24$};
			\node at ($(a7.north) + (0,0.2)$) 	{\small \color{blue} $\frac34$};
			
			\node at ($(a4.south) + (0,-0.2)$) 	{\small \color{dgreen} $1$};
			\node at ($(a0.north) + (0,0.2)$) 	{\small \color{dgreen} $0$};
	
	\end{scope}

\end{tikzpicture}			
	
\caption{Example of different $E_7^{(1)}$ colorings \label{fig:excolor}}
		
\end{figure}
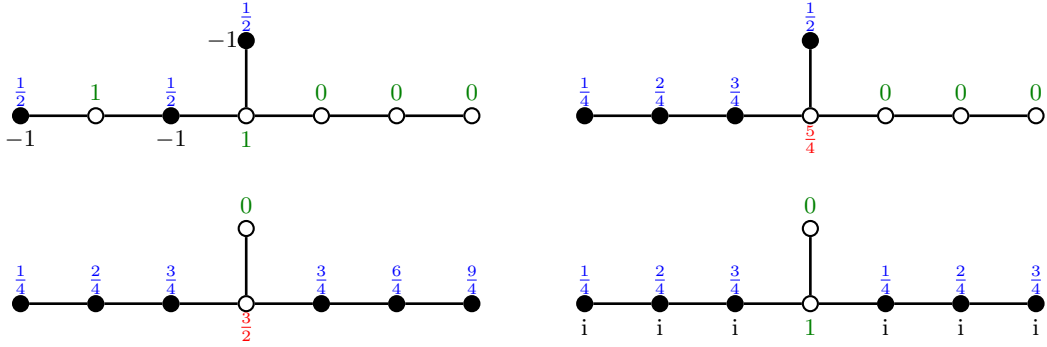

\end{Example}

\paragraph{Case 3.} 
Finally, let us consider intermediate case 
\(w= \pi \cdot \prod\nolimits_{j=1}^l \bar{\pi}_{n_j+1}^{k_j}\), where \(\pi \neq e \in \widehat{\Omega}_I\). Recall the \(\pi\)
acts on the set of connected components of \(I\), i.e. in the set of \(\Delta_j\), as permutation.
	
Let \(\pi\) has order \(d\). Each \(\Delta_j\) is either invariant under the action of \(\pi\) or belongs
to the orbit of the size \(d\). Indeed, since possible values of \(d\) for $\Phi=E_n^{(1)}$ are \(2,3,4,5,6\) 
the only cases to exclude is \(d=4\) and orbit of size 2 and $d=6$ and orbits of sizes $2$ and $3$.
This can possibly happen only for \(\Phi^a=D_5^{(1)},E_3^{(1)}\), but still does not happen. 
	
If component \(\Delta_j\) is invariant then \(\pi\) either acts on it trivially or by \(\sigma\) which reversing order.
In first case we can apply Lemma \ref{lemma:Omega_p1} and (possibly reducing \(\Delta_j\)) we can assume that corresponding
element of \(\Omega_{\Phi_j}\) is \(\bar{\pi}_{n_j+1}\). Otherwise we use 
\begin{lemma}\label{lemma:Omega vs pi 1}  
	Element \(w=\sigma \bar{\pi}_{n+1}^k \in C_2 \ltimes W_{A_n}\) can be conjugated by the elements of
	\(W_{A_n}\) to \(\sigma s_{(n+1)/2}\) for odd $n$ and $k$ or to $\sigma$ otherwise.
\end{lemma}
\begin{proof}
	Using relation \( \bar{\pi}_{n+1}  \sigma=\sigma\bar{\pi}_{n+1}^{-1}\) we can conjugate \(w\) to \(\sigma\) for even \(k\) and to \(\sigma \bar{\pi}_{n+1}\) for odd \(k\). Now using \(s_{n+1-i}\sigma =\sigma s_i\) we conjugate \(\sigma \bar{\pi}_{n+1}\) to \(\sigma\) or to  \(\sigma s_{(n+1)/2}\) depending on parity of \(n\).
\end{proof}
Hence we can consider that \(\pi\) acts trivially on any invariant component \(\Delta_j\).
    
Let us pick on component \(\Delta_j\) for each orbit of \(\pi\). 
We call such component distinguished and denote set of such components by \(D\). By $ I_D=\sqcup_{j\in D} \Delta_j$ we denote the set
of corresponding nodes.
\begin{lemma}\label{lemma:Omega vs pi 2}
	Element $w$ can be conjugated by element from $W_{\sqcup \Phi_j }$ to the form
	\begin{equation}
		w=\pi \prod_{j\in D} \bar{\pi}^{(j)}_{n_j+1} 
	\end{equation}
\end{lemma}
\begin{proof}
	It is sufficient to consider one \(\pi\)-orbit of length \(d\). Denote corresponding components
	by \(\Delta_{0},\Delta_{1},\dots \Delta_{d-1}\), assuming that distinguished component is \(\Delta_{0}\) and
	\(\pi(i)=i+1 \mod d\). Let factor of \(w\) corresponding to this orbit be \(\pi \cdot \prod_i (\bar{\pi}^{(i)})^{k_i}\).
	Conjugating $w$ by the elements  \(\pi^{(i}\) successively for $i=d-1\ldots 1$, we vanish all $k_i$ except probably \(i=0\).
	Now for any \(g \in W_{\Phi_0}\) we can conjugate \(w\) by the product $\prod_{i=0}^{d-1} \pi^i(g)$, this is equivalent to
	conjugation of the \(\Delta_0\) component of \(w\) by \(g\). 
%
%
Then, using Lemma \ref{lemma:Omega_p1}, we reduce to the case \(k_0=1\).
\end{proof}

Let us now find the condition on the invariant subset \(\mathcal{A}_w\). As in previous case, for the black vertices we have
\begin{equation}
	b\in \Delta_j: a_b=\exp \left(2 \pi \ri \frac{m_j}{n_j+1}\right) \label{a_b_c3}
\end{equation}
	with \(\gcd(m_j,n_j+1)=1\) and numbers \(m_j\) are \(\pi\) invariant \(m_j=m_{\pi(j)}\). The white vertex root variables transform under $w$ as follows
\begin{equation}
	w(a_{\pi(c)})=a_{c} \prod_{b \in I_D \cap N(c)}
	\exp \left( 2 \pi \ri \frac{i_b m_{j(b)}}{n_{j(b)}+1}\right),
\end{equation}
where \(i_{b}\) as before is the number of vertex \(b\) in component $j(b)$.
The \(w\)-invariance condition (on white vertices) reads that product in r.h.s. of the above formula equals $1$.

We summarize the discussion above as a set of selection rules on the coloring of the affine Dynkin diagram with automorphism
\(\pi\neq e \in \widehat{\Omega}_\Phi\). As before \(\{\Delta_j\}\) is a set of connected components of \(I\), by \(n_j\) we denote number of vertices in \(\Delta_j\)

\begin{itemize}
	\item Coloring is invariant with respect to $\pi$.
	\item Each \(\Delta_j\) is Dynkin diagram of type \(A\). 
	\item If some \(\Delta_j\) is $\pi$-invariant as a set then \(\pi\) acts trivially on it ($\forall b\in \Delta	_j, \pi j=j: \pi b=b$).
	\item In each $\pi$-orbit choose one distinguished component $\Delta_j$.
	Number vertices of each distinguished component $\Delta_j$ successively from $1$ to $n_j$, starting from either end. Then there is a
	choice of numbers \(m_j, \, 0{<}m_j{\leq}n_j,\, \gcd (m_j,n_j+1)=1\) for each distinguished component, such that for any white vertex $c$
	\begin{equation}
	\sum_{b\in I_D\cap N(c)} \frac{i_b m_{j(b)}}{n_{j(b)}+1}\in \mathbb{Z}, 
	\end{equation}
        where the sum runs over the set of black vertices from distinguished components adjacent to \(c\),
        $j(b)$ is such that \(b\in \Delta_{j(b)}\), and $i_b$ is number of $b$ in component $\Delta_j(b)$.
	\item Number of $\pi$-orbits on the set of vertices $\Delta^a \setminus I$ is at least $2$.
\end{itemize}
The last item is necessary to have at least two-dimensional \(\mathcal{A}_w\), namely to have a true folding.
Clearly these selection rules are generalization of selection rules in case 2. The general formula for the folding transformation is 
\begin{equation}
	w=\pi \prod_{j\in D} \bar{\pi}^{(j)},\quad \pi\in \widehat{\Omega}_{\Phi}, \quad \bar{\pi}^{(j)}\in \Omega_{\Phi_j=A_{n_j}}  \label{wcase3}
\end{equation}

We will list colorings that satisfy selection rules in Section \ref{ssec:answers algebraic}. 
It appears that there are only $2$ true folding transformations of case 3: one for $E_7^{(1)}$ and one for $D_5^{(1)}$.
There exist more case 3 folding transformations, if we drop the true folding selection rule (take one white orbit), they are listed
in Appendix \ref{sec:ntrf}.
We have the same comment about $m_j$ as for case $2$.

\paragraph{Non-cyclic groups.} 

It remains to classify folding groups \(H\) which are not cyclic. As was shown above any such \(H\) is a subgroup of
group of the form \(\widehat{\Omega}\ltimes \prod\nolimits_{j'=1}^{l'} \Omega_{\Phi'_{j}}\)  where all \(\Phi'_{j}\) are of \(A\) type. 
Any nontrivial element of the \(H\) should generate a cyclic folding subgroup, and such elements were classified above up to conjugation. 
For any such element \(w\) we find list of  subsets \(I \in \Delta^a\)  such that there exists element
\(\tilde{w}\in \widehat{\Omega}\ltimes \prod\nolimits_{j'=1}^{l'} \Omega_{\Phi'_{j}}\) conjugated to \(w\). In order to do this we reverse simplification performed in Lemmas \ref{lemma:Omega_p1}, \ref{lemma:Omega vs pi 1} and \ref{lemma:Omega vs pi 2}. Note that condition \eqref{a_b_c2} should hold
for any connected component of \(I\), and we can keep the assumption  \(\gcd(m_j,n_j+1)=1\).

Then we check which of such \(\tilde{w}\) can be combined into one group.  The list of not cyclic folding subgroups is given
in Section \ref{ssec:answers algebraic}. Here we just consider the most nontrivial example.

\begin{Example} 	We have \(6\) folding transformations for \(E_7^{(1)}\). 
	For each of them we list conjugated elements in form \(\widehat{\Omega}\ltimes \prod\nolimits_{j'=1}^{l'} \Omega_{\Phi'_{j}}\).
	We restrict ourselves to the case that $|\Delta^a\setminus I|\geq 2$, since otherwise we do not have dynamics. 
	Similarly in the last case we assumed that
	number of \(\pi\) orbits on \(\Delta^a \setminus I\) is at least~2. 
	By \([i,i+1,\dots, j]\) we denote connected component with nodes \(i,i+1,\dots,j\).
	\begin{itemize}
		\item \(s_{12}^{1,2}s_{40}^{1,2}s_{67}^{1,2}\) \(I=\{[1,2],[4,0],[6,7]\}\);
		\item \(s_{123}^{1,3}s_{567}^{1,3}\),\, \(I=\{[1,2,3],[5,6,7]\}\);
		\item \(\pi s_0 s_1 s_3\),\, \(\pi s_0 s_1 s_5\),\, \(\pi s_0 s_3 s_7\),\,\(\pi s_0 s_5 s_7\),\, \(\pi \ltimes I = \pi \ltimes \{0,1,3,5,7\}\);\; 
		\item \(s_1s_3s_5s_7\),\, \(I=\{1,3,5,7\}\);\; \(s_{123}^2s_5s_7\),\, \(I=\{[1,2,3],5,7\}\);\; \(s_1s_{345}^2s_7\),\, \(I=\{1,[3,4,5],7\}\);\;\\
		\(s_1s_{3}s_{567}^2\),\, \(I=\{1,3,[5,6,7]\}\);\;
		\(s_{12345}^3s_7\),\, \(I=\{[1,2,3,4,5],7\}\);\; 
		\(s_{123}^2s_{567}^2\),\, \(I=\{[1,2,3],[5,6,7]\}\);\;\\
		\(s_{1}s_{34567}^3\),\, \(I=\{1,[3,4,5,6,7]\}\);\;
		\item \(s_0s_1s_3\),\, \(I=\{0,1,3\}\);\; \(s_0s_5s_7\),\, \(I=\{0,5,7\}\);\; \(s_0s_{123}^2s_0\),\, \(I=\{0,[1,2,3]\}\);\; \\ \(s_{1}s_{340}^2\),\, \(I=\{1,[3,4,0]\}\);\; 
		\(s_{045}^2s_{7}\),\, \(I=\{[0,4,5],7\}\);\; \(s_{0}s_{567}^2\),\, \(I=\{0,[5,6,7]\}\);\; \\ \(s_{12340}^3\),\, \(I=\{[1,2,3,4,0]\}\);\;
		\(s_{04567}^3\),\, \(I=\{[0,4,5,6,7]\}\);\;
		\item \(\pi\), \, \(\pi \ltimes I = \pi \);\; \(\pi s_{04}^{1,2}\), \, \(\pi \ltimes I = \pi \ltimes \{[4.0]\}\);\; 
		\(\pi s_{17}\), \, \(\pi \ltimes I = \pi \ltimes\{1,7\}\);\; \\
		\(\pi s_{26}\), \, \(\pi \ltimes I = \pi \ltimes \{2,6\}\);\; 		
		\(\pi s_{35}\), \, \(\pi \ltimes I = \pi \ltimes \{3,5\}\);\; 
		\(\pi s_{345}^2\), \, \(\pi \ltimes I = \pi \ltimes \{[3,4,5]\}\);\; \\
		\(\pi s_{17}s_{35}\), \, \(\pi \ltimes I = \pi \ltimes \{1,3,5,7\}\);\;  \(\pi s_{1267}^{1,2}\), \, \(\pi \ltimes I = \pi \ltimes \{[1,2],[6,7]\}\);\; \\
		\(\pi s_{2356}^{1,2}\), \, \(\pi \ltimes I = \pi \ltimes \{[2,3],[5,6]\}\);\;
		\(\pi s_{04}^{1,2}s_{17}\), \, \(\pi \ltimes I = \pi \ltimes \{1,7,[4,0]\}\);\; \\ \(\pi s_{04}^{1,2}s_{26}\), \, \(\pi \ltimes I = \pi \ltimes \{2,6,[4,0]\}\);\;				 				
		\(\pi s_{23456}^{2,4}\), \, \(\pi \ltimes I = \pi \ltimes \{[2,3,4,5,6]\}\);\; \\ \(\pi s_{345}^{2}s_{17}\), \, \(\pi \ltimes I = \pi \ltimes \{1,[3,4,5],7\}\);\;			 				
		\(\pi s_{123567}^{1,2,3}\), \, \(\pi \ltimes I = \pi \ltimes \{[1,2,3],[5,6,7]\}\);\; 
		\end{itemize}
	We see that possible choices for noncyclic groups are
	
	\medskip 
	\begin{center}
		\begin{tabular}{|c|c|c|c|c|}
			\hline
			\(\pi \ltimes \{0,1,3,5,7\}\) & \(\{0,1,3,5,7\}\) & 		\(\pi \ltimes \{1,3,5,7\}\) & \(\pi \ltimes \{1,[3,4,5],7\}\) &		\(\pi \ltimes \{[1,2,3],[5,6,7]\}\)  	
			\\ \hline
			\(C_2\ltimes C_4\) & \(C_2\times C_2\) & \(C_2\times C_2\) & \(C_2\times C_2\) & 		\(C_2\ltimes C_4\) 
			\\ \hline
		\end{tabular}	
		\medskip 
		
	\end{center}

	Note that the groups in the third and fourth columns are conjugated.

\end{Example}

%
%
       
\subsection{Answers}\label{ssec:answers algebraic}
	
	\medskip
	
	\begin{tabular}{|M{1.5cm}|M{2cm}|M{1.5cm}|M{0.75cm}|M{3.5cm}|M{4.25cm}|M{0.5cm}|}
	    \hline 
	    Sym./surf. & Diagram & Name & Order & $w$ & $\mathcal{A}_w$ & Sec.\\ 
        \hline	
         & 
               \includegraphics[scale=1]{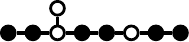}
	      & $3A_2$	& 3 &  $s_{21} s_{54} s_{87}$ & $(\zeta,\zeta,\zeta),\, (\zeta^{{-}1},\zeta^{{-}1},\zeta^{{-}1})$  & 
		\\
		\cline{2-6}
	 $E_8^{(1)}/A_0^{(1)}$	& 
      \includegraphics[scale=1]{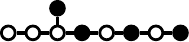}
       & $4A_1$ & 2 &  $s_0 s_4 s_6 s_8$ & $({-}1,{-}1,{-}1,{-}1)$ & \ref{e8_a}
		\\
		\cline{2-6}
		& 
\includegraphics[scale=1]{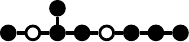}
        & $A_1{+}2A_3$ & 4 & \vspace{1mm}  \parbox{3.5cm}{\centering $s_1 s_{430} s_{876}$\\ {\small $w^2=s_{37} (s_0 s_4 s_6 s_8) s_{73}$}} &
        $({-}1,\ri,\ri), \, ({-}1,-\ri,-\ri)$ & 
		\\
 		\hline
		 &
			\includegraphics[scale=1]{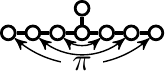}
		& $\pi$ & \vspace{5mm} 2 & $\pi$ & $(a_1{=}a_7,a_2{=}a_6,a_3{=}a_5)$ &  \\
		\cline{2-3} \cline{5-6}	
		& 
			\includegraphics[scale=1]{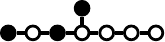}
		     & $3A_1$	&   & $s_0 s_1 s_3$ & $({-}1,{-}1,{-}1)$ & \\
		\cline{2-6} 
		$E_7^{(1)}/A_1^{(1)}$
		&
\includegraphics[scale=1]{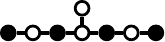}
		& $4A_1$ & 2 &   $s_1 s_3 s_5 s_7$ & $({-}1,{-}1,{-}1,{-}1)$ & \ref{e7_a} \\
		\cline{2-6}
		& 
\includegraphics[scale=1]{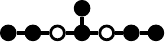}
	  & $3A_2$ & 3   & $s_{12}s_{40}s_{76}$ & $(\zeta,\zeta,\zeta), \, (\zeta^{{-}1},\zeta^{{-}1},\zeta^{{-}1})$  &\\
		\cline{2-6} 
		&
\includegraphics[scale=1]{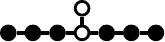}
	    & $2A_3$	&   & \vspace{1mm} \parbox{3.5cm}{\centering $s_{321}s_{765}$\\ {\small $w^2=s_{26} (s_1 s_3 s_5 s_7) s_{62}$}} 
		& $(\ri,\ri), \, (-\ri,-\ri)$ & \\
		\cline{2-3} \cline{5-6}
		 &
%
			\includegraphics[scale=1]{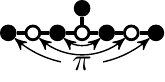}
	    & $\pi{\ltimes}5A_1$	& \vspace{-8mm} 4 & \parbox{3cm}{\centering $\pi s_0 s_1 s_3$\\{\small $w^2=s_1s_3s_5s_7$}}  
		& { $({-}1,{-}1,{-}1,{-}1, {-}1 ;a_2{=}a_6)$} &  \\
		\hline	
	     &
\includegraphics[scale=1]{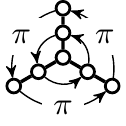}
		& $\pi$ & &   $\pi$ & $(a_0{=}a_1{=}a_5, a_2{=}a_4{=}a_6)$ & 
		\\
		\cline{2-3}  \cline{5-6}
		 
		 $E_6^{(1)}/A_2^{(1)}$ & 
\includegraphics[scale=1]{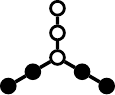}
	      & $2A_2$	& \vspace{-1cm} 3  & $s_{12}s_{45}$  &  $(\zeta,\zeta), \, (\zeta^{{-}1},\zeta^{{-}1})$ &  \ref{e6_a}
		\\	
		\cline{2-6}
		
	    & 
\includegraphics[scale=1]{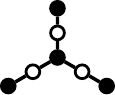}
	&    $4A_1$	& 2 &  $s_0s_1s_3s_5$ & $({-}1,{-}1,{-}1,{-}1)$ & 
		\\
		\hline
		
	 &
%
\includegraphics[scale=0.9]{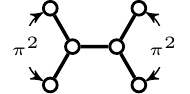}
	& $\pi^2$ & & $\pi^2$ & $(a_0{=}a_1,a_4{=}a_5)$  & \\
	\cline{2-3}  \cline{5-6}
	 
	& 
			\includegraphics[scale=1]{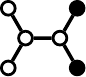}
	& $2A_1$ & \vspace{-8mm} 2 & $s_4s_5$  & $({-}1,{-}1)$ &  \\

	\cline{2-6}
	
	$D_5^{(1)}/A_3^{(1)}$
	
	&	
\includegraphics[scale=0.9]{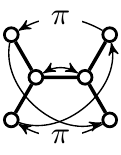}
	& $\pi$	& &  \parbox{3cm}{\centering $\pi$\\{\small $w^2=\pi^2$}} 
		
	    	& $(a_0{=}a_1{=}a_4{=}a_5,a_2{=}a_3)$ &  \ref{d5_a}\\
		\cline{2-3}  \cline{5-6}
		
		&
%
\includegraphics[scale=0.9]{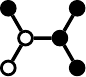}
		&   $A_1{+}A_3$	& 4  & \vspace{1mm} \parbox{3cm}{\centering $s_1 s_{534}$\\{\small $w^2=s_3 (s_4s_5) s_3$}} 
			
			& $({-}1,\ri), \, ({-}1,-\ri)$ & \\
		
		\cline{2-3} \cline{5-6}
		
		&	
%
\includegraphics[scale=0.9]{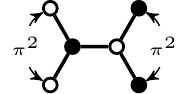}
	    & $\pi^2{\ltimes}3A_1$	&  & \parbox{2cm}{\centering $\pi^2 s_2 s_4$\\{\small $w^2=s_4s_5$}}  &  $({-}1,{-}1,{-}1;a_0{=}-a_1)$ &  \\

			\cline{2-6}
		& 
                        \includegraphics[scale=0.9]{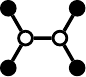}
			& $4A_1$ & 2 &  $s_0s_1s_4s_5$ & $({-}1,{-}1,{-}1,{-}1)$ & \\	
			\hline
	
	 &
%
%
%
%
%
%
%
%
\includegraphics[scale=1]{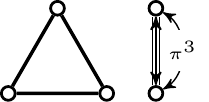}
	& $\pi^3$ & &  $\pi^3$ & $(a_3{=}a_4)$
		 & \\
	
	\cline{2-3} \cline{5-6}
	$E_3^{(1)}/A_5^{(1)}$ &
%
%
%
%
%
%
%
%
\includegraphics[scale=1]{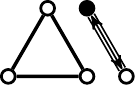}
 	   & $A_1$ & \vspace{-8mm} 2 & $s_4$  & $({-}1)$ & \ref{e3_a}\\
	
	\cline{2-6}

	 &
%
%
%
%
%
%
%
\includegraphics[scale=1]{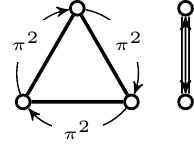}
	& $\pi^2$ & &   $\pi^2$ & $(a_0{=}a_1{=}a_2)$ & \\
	\cline{2-3} \cline{5-6}
	
	& 
%
%
%
%
%
%
\includegraphics[scale=1]{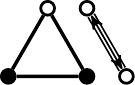}
	   & $A_2$ & \vspace{-8mm} 3 & $s_{12}$ &  $(\zeta), \, (\zeta^{{-}1})$ &  \\

	\hline
	$A_1^{(1)}/A_7^{(1)'}$ &
%
%
%
%
%
%
\includegraphics[scale=1]{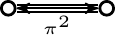}
     & $\pi^2$	& &  $\pi^2$ & - &  \ref{e1_a}\\
	
	\cline{2-3} \cline{5-6}
	
	& 
%
%
%
%
%
%
\includegraphics[scale=1]{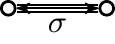}
	   & $\sigma$ & \vspace{-4mm} 2 &  $\sigma$  & - & \\
	    \hline
	\end{tabular}
	
	In the paper we \emph{name} folding transformations by element of $\widehat{\Omega}$ and components of $I$, like $\pi \ltimes A_{n_1}+ \ldots +A_{n_l}$.
	In this naming we will drop $e\in\widehat{\Omega}$ and write $\times$ instead of $\ltimes$, when all components of $I$ are $\pi$-invariant.  
%
%
%
	
	In the next table, we present non-cyclic folding subgroups. For each group, we write formulas for its generators. In the column ``Group'' we also note nontrivial embeddings between different non-cyclic folding subgroups.

%
	\medskip
	
	\begin{tabular}{|M{1.5cm}|M{2.5cm}|M{1.5cm}|M{4.25cm}|M{4.25cm}|}
	    \hline 
	    Sym./surf. & Diagram & Name & Group & $\mathcal{A}_w$ \\ 
        \hline	
         &
%
\includegraphics[scale=0.75]{e7_p11111.pdf}
		& $\pi\ltimes 5A_1$ & $\Dih_4^{(1)}=\underbrace{C_2}_{\pi} \ltimes \underbrace{C_4}_{\pi s_0s_1s_3}$  &
			$({-}1,{-}1,{-}1,{-}1,{-}1;a_2{=}a_6)$\\
        \cline{2-5}
        $E_7^{(1)}/A_1^{(1)}$ & 
        \includegraphics[scale=0.75]{e7_p33.pdf}
        & $\pi\ltimes 2A_3$  & $\Dih_4^{(2)}=\underbrace{C_2}_{\pi}\ltimes \underbrace{C_4}_{s_{321}s_{765}}$ & $(\ri,\ri), \, (-\ri,-\ri)$ \\
        \cline{2-5}
        &
        \includegraphics[scale=0.75]{e7_p1111.pdf}
        & $\pi\ltimes 4A_1$
       & $\underbrace{C_2}_{\pi} \times \underbrace{C_2}_{s_1s_3s_5s_7} \subset \Dih_4^{(1,2)} $ &
       $({-}1,{-}1,{-}1,{-}1; a_2=a_6)$ \\
        \cline{2-5}
         &
       \includegraphics[scale=0.75]{e7_11111.pdf}
			& $5A_1$
        & \vspace{1mm} $\underbrace{C_2}_{s_0s_1s_3}\times \underbrace{C_2}_{s_1s_3s_5s_7}\subset \Dih_4^{(1)}$ & $({-}1,{-}1,{-}1,{-}1,{-}1)$\\
        
        \hline
         &
        \includegraphics[scale=0.75]{d5_pp1111.pdf}
        & $\pi^2\ltimes 4A_1$
        & $\underbrace{C_2}_{\pi^2}\times \underbrace{C_2}_{s_0s_1s_4s_5} \times \underbrace{C_2}_{s_4s_5}$ & $({-}1,{-}1,{-}1,{-}1)$ \\
        \cline{2-5}
        
   $D_5^{(1)}/A_3^{(1)}$        & 
%
%
       \includegraphics[scale=0.75]{d5_1111.pdf}      
      & $4A_1$ & $\underbrace{C_2}_{s_0s_1s_4s_5} \times \underbrace{C_2}_{s_4s_5}  \subset C_2^3$ &  $({-}1,{-}1,{-}1,{-}1)$\\
	\cline{2-4}
	
	&
	\includegraphics[scale=0.75]{d5_pp1111.pdf}
	& $\pi^2\ltimes 4A_1$
	& $\underbrace{C_2}_{\pi^2}\times \underbrace{C_2}_{s_0s_1s_4s_5} \subset C_2^3$  &   $({-}1,{-}1,{-}1,{-}1)$\\
	\cline{2-5}
	
		& \includegraphics[scale=0.75]{d5_pp11.pdf}
	& $\pi^2\ltimes 2A_1$
	& $\underbrace{C_2}_{\pi^2}\times \underbrace{C_2}_{s_4s_5} \subset C_2^3$ & $({-}1,{-}1;a_0{=}a_1)$\\
	
         \hline
          $E_1^{(1)}/A_7^{(1)}$ &  \includegraphics[scale=1]{e1_ppc.pdf}
	    & $\pi^2\times \sigma$  & $\underbrace{C_2}_{\pi^2}\times \underbrace{C_2}_{\sigma}$ & - \\
          \hline
        \end{tabular}

\subsection{Invariant lattices and normalizers}
	\label{ssec:invlat}

\paragraph{Invariant and orthogonal lattices.}
	
For given $w\in W^{ae}$ of finite order $n$ introduce two sublattices of $Q^a$: 
\begin{enumerate}
	\item Invariant sublattice $Q^{a,w}=\mathrm{Ker}_{Q^a}(w-1)$.
	\item Orhtogonal sublattice $Q^{a,\perp w}=\mathrm{Ker}_{Q^a} \sum_{i=0}^{n-1} w^i$.
\end{enumerate}
        
\begin{lemma}\label{lemma:lattices}
	Sublattices $Q^{a,w},Q^{a,\perp w}$ have following properties
	\begin{enumerate}
		\item Lattices $Q^{a,w}$ and $Q^{a,\perp w}$
        are orthogonal, union of their bases is a basis in $Q^a \otimes_\mathbb{Z} \mathbb{Q} $, i.e. $Q^{a,w} \oplus Q^{a,\perp w}$ is a sublattice of full rank in $Q^a$.
        \item Invariant sublattice of $w=t_{\lambda} \bar{w} \in P\rtimes W$ depends only on finite part $\bar{w}\in W$: $Q^{a,w}=Q^{a,\bar{w}}$.
        \item Imaginary root belongs to invariant sublattice: $\delta\in Q^{a,w}, \, \delta\notin Q^{a,\perp w}$.
	\end{enumerate}
\end{lemma}
\begin{proof}
	\begin{enumerate}
     \item  
        Let us consider complexification $Q^a \otimes_\mathbb{Z} \mathbb{C}=\mathbb{C}^{r+1}$. 
        If $w$ has order $n$, eigenvalues of its natural action on this space
        will be $n$-th roots of unity.
		Note that the dimension of space of solutions of the equation $w(\sum_j c_j \alpha_j)=\sum_j c_j \alpha_j$ is the same for the base fields \(\mathbb{Q}\) and \(\mathbb{C}\). Hence the eigenspace of \(w\) with eigenvalue $1$ is complexification $Q^{a,w}\otimes_\mathbb{Z} \mathbb{C}$  of $Q^{a,w}$. Similarly, looking to the equation $\left(\sum_{i=0}^{n-1} w^i\right) \sum_j c_j \alpha_j=0$ we see that span of eigenspaces with eigenvalues $\neq1$ is a complexification $Q^{a,\perp,w}\otimes_\mathbb{Z} \mathbb{C}$  of $Q^{a,\perp,w}$.
		
%
         \item Because $w=t_{\lambda}\bar{w}$ has finite order, then $\lambda\in \mathrm{Ker}_{P} \sum_{i=0}^{n-1} w^i\subset
		 Q^{a,\perp w} \otimes_{\mathbb{Z}} \mathbb{Q}$.
         So $\lambda$ is orthogonal to $Q^{a,w}$, which give us the statement.
         \item This follows from $W^{ae}(\delta)=\delta$.
	\end{enumerate}       
\end{proof}
       
              These sublattices we will use for the calculation of the normalizer of folding transformations.
       Below we present some simplifications for calculation of such sublattices:
       
\begin{enumerate}
        \item Due to the Lemma \ref{lemma:detoa} 
        finitization $\bar{w}$ of folding group elements of case $1$
        is conjugated to folding group element of case $2$.
        A similar statement holds also for the case $3$.
        Hence, their invariant sublattices $Q^{a,w}$ are related by action
        of group elements, which are used for above conjugation.
        \item For case $2$ direct calculation gives $Q^{a,\perp w}=\langle \{\alpha_{i\in I}\} \rangle$.
        \item For case $1$ direct calculation gives $Q^{a,\perp w}=\langle \{\alpha_i-\pi(\alpha_i)\} \rangle$
        \item For case $3$ direct calculation gives $Q^{a,\perp w}=\langle \{\alpha_i-\pi(\alpha_i), \, i\in \Delta^a\setminus I \} \rangle
        \oplus \langle \{\alpha_{i\in I}\} \rangle$.
   \end{enumerate}

\paragraph{Normalizer and components.}
	
Normalizer of the cyclic group $\langle w \rangle$ is a semidirect product of group $N_{flip}$, changing its primitive generator and centralizer $C$:
\begin{equation}
	N(\langle w \rangle,W^{ae})=N_{flip}\ltimes C(w,W^{ae}), 
\end{equation}
$N_{flip}\subset C_2$ for $w$ of order $2,3,4,6$ and $C_4$ for $w$ of order $5$.  We write $\subset$, because we cannot guarantee that all elements of $N_{flip}$ are realized in $W^{ae}$. Indeed, for $w\in \widehat{\Omega}$ we cannot change primitive generator of outer automorphism group.
        
Normalizer acts on $A^w$ and it appears that centralizer preserves folding irreducible components $\mathcal{A}_w$ and $N_{flip}$ permutes them. Note that for true folding, if $t_{\lambda}\in N$, then it immediately belongs to centralizer $C$. 
	
\begin{prop}\label{prop:centr_lattice}
	Assume that $Q^{a,w}$ and $Q^{a,\perp w}$ are generated by root subsystems $\Phi^{a,w}$, $\Phi^{a,\perp w}$ of $\Phi^a$. Then
	\begin{equation}\label{centr_lattice}
		C(w,W^{ae})\subset \Aut\ltimes (W^a_{\Phi^{a,w}}\times W_{\Phi^{a,\perp w}})
	\end{equation}
	where $\Aut\subset \Aut(\Phi^{a,w}) \times \Aut(\Phi^{a,\perp w})$ denotes	subgroup of all automorphisms of both lattices, that is realized in $W^{ae}$.
\end{prop}
	
\begin{proof}
	Sublattices $Q^{a,w}$ and $Q^{a,\perp w}$ are defined by $w$. 
	Hence centralizer $C(w,W^{ae})$ should preserve sublattices $Q^{a,w}$ and $Q^{a,\perp w}$, as
	well as corresponding root systems $\Phi^{a,w}$ and $\Phi^{a,\perp w}$.
\end{proof}
Subgroup $\Aut$ of elements that can be realized in $W^{ae}$
is usually much smaller than $\Aut(\Phi^{a,w}) \times \Aut(\Phi^{a,\perp w})$.
	
By definition of the invariant sublattice $W^a_{\Phi^{a,w}}$ commutes with $w$. 
We can decompose $\Aut=\Aut^{\parallel}\ltimes \Aut^{\perp}$, where $\Aut^{\parallel}$ acts effectively on the $\Phi^{a,w}$,
$\Aut^{\perp}\subset \Aut(\Phi^{a,\perp w})$.
If $\Aut^{\parallel}$ commutes with $w$, then
\begin{equation}\label{centr_lattice_eq}
	C(w,W^{ae})= \Aut^{\parallel}\ltimes (W^a_{\Phi^{a,w}}\times (\Aut^{\perp}\ltimes W_{\Phi^{a,\perp w}})^w), 
\end{equation}
where we take $w$-centralizer of last multiplier.
\begin{Remark} \label{rem:subsublattice}
	Sometimes in cases 1 and 3 sublattice $Q^{a,\perp w}$ is not generated by any root system.
	Namely, this happens for foldings $\pi\ltimes 5A_1\subset E_7^{(1)}$, $\pi \subset E_6^{(1)}$
	and $\pi \subset D_5^{(1)}$.
	But it appears that in such cases we can find $\widetilde{Q^{a,\perp w}}$ of finite index in $Q^{a,\perp w}$ which corresponds to some root system and invariant under the action of normalizer \(N(w,W^{ae})\). So the analog of the formula \eqref{centr_lattice} holds.
\end{Remark}

\begin{Example}\label{ex:sublattices}
	Let us consider folding $2A_3$ for $\Phi^a=E_7^{(1)}$ (second folding from Example \ref{ex:colorings})
	given by formula $w=s_{321}s_{765}$. We have, that $\Phi^{a, \perp w}=2A_3$ generated by $\alpha_1, \alpha_2,
	\alpha_3$ and $\alpha_5, \alpha_6, \alpha_7$. $Q^{a,w}$ obviously contain $\alpha_0$, so together with $\delta$
	we have $\Phi^{a,w}=A_1^{(1)}$.
	
	We choose generator $s_1s_3s_5s_7$ for $N_{flip}$, such generator permutes
	irreducible components $a_{1,2,3,5,6,7}=\ri$ and $a_{1,2,3,5,6,7}=-\ri$.
	
	Group $\Aut$ could be obtained by direct calculations.
	They give, that $\Aut^{\perp}=C_2$, generated by $\pi$.
	Another multiplier $\Aut^{\parallel}=C_2$,
	which generator is defined up to $\Aut^{\perp}$,
	we take $\pi s \pi s^{-1}$, where $s=s_{4354}s_{2132}$. $\Aut^{\perp}$ and
	$\Aut^{\parallel}$ commutes.
	Sublattices and automorphisms are pictured on Fig. \ref{fig:sublattices}
	\footnote{Here for simplicity we pictured action on orthogonal root subsystem of $s\pi s^{-1}$ instead of $\pi s \pi s^{-1}$.}.
	
	\begin{figure}[h]
	
		\begin{tikzpicture}[elt/.style={circle,draw=black!100,thick, inner sep=0pt,minimum size=2mm},scale=2.5]
			\path 	(-1,0) 	node 	(a1) [elt] {}
			(1,0) 	node 	(a2) [elt] {};

		    \draw [black,line width=2.5pt ] (a1) -- (a2);
		    
		     \draw [white,line width=1.5pt ] (a1) -- (a2);
		     
		     \draw [<->,black, line width=0.5pt]
		     (a1) -- (a2);
		     
		     \node at ($(a2.north) + (0,0.1)$) 	{\small$\underline{\delta-\alpha_{0}}$};	
		     	
		     \node at ($(a1.north) + (0,0.1)$) 	{\small $\alpha_{0}$};
		     
		       \draw [<->, dashed]
		     (a1) to[bend right=40] node[fill=white]{\small $\pi s \pi s^{-1}$}  (a2);

		     \node at (0,0.5) 	{$A_1^{(1)}$};
		     
			\begin{scope}[xshift=3cm]     
			 \path (-1,0.5) node (a3) [elt] {}
					(0,0.5) node  (a2) [elt] {}
					(1,0.5) node  (a1) [elt] {}
					(-1,-0.5) node (a5) [elt] {}
					(0,-0.5) node  (a6) [elt] {}
					(1,-0.5) node  (a7) [elt] {};
			
				\draw (a1) -- (a2) -- (a3) (a5) -- (a6) -- (a7);

		    	\node at ($(a1.north) + (0,0.1)$) 	{\small $\alpha_1$};	
		    	\node at ($(a2.north) + (0,0.1)$) 	{\small $\alpha_2$};	
		    	\node at ($(a3.north) + (0,0.1)$) 	{\small $\alpha_3$};	
		    	
		    	\node at ($(a5.south) + (0,-0.1)$) 	{\small $\alpha_5$};	
		    	\node at ($(a6.south) + (0,-0.1)$) 	{\small $\alpha_6$};	
		    	\node at ($(a7.south) + (0,-0.1)$) 	{\small $\alpha_7$};

		    	\draw[<->, dashed] (a3) to[bend left=40] node[fill=white] {$s \pi s^{-1}$} (a1);
		    	\draw[<->, dashed] (a5) to[bend right=40] node[fill=white] {$s \pi s^{-1}$} (a7);
		    	
		    	\draw[<->, dashed] (a5) to[bend right=0] node[fill=white] {$\pi$} (a3);
		    	\draw[<->, dashed] (a6) to[bend right=0] node[fill=white] {$\pi$} (a2);
		    	\draw[<->, dashed] (a7) to[bend right=0] node[fill=white] {$\pi$} (a1);
		    	
		    	\node at (-1.25,1) {$2A_3$};
		\end{scope}
	\end{tikzpicture}
	
		\caption{Root systems $\Phi^{a,w}$, $\Phi^{a,\perp w}$ and their automorphisms	\label{fig:sublattices} }	
	\end{figure}
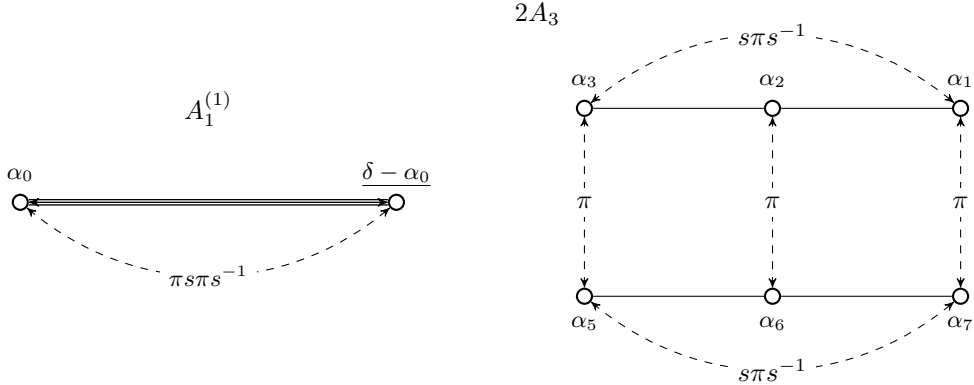
	
	We can check, that $\Aut^{\parallel}$ commutes with $w$,
	so due to \eqref{centr_lattice_eq}, we have
	\begin{equation}\label{eq:example:C}
		C(w,W^{ae}_{E_7})=\underbrace{C_2}_{\pi s \pi s^{-1}}\ltimes (W^a \times \underbrace{C_2}_{\pi s_1s_3s_5s_7}\ltimes
		\underbrace{C_4^2}_{\Omega_{A_3}^2})
		\simeq (\underbrace{C_2}_{\pi s_1s_3s_5s_7} \times W_{A_1}^{ae}) \ltimes C_4^2,
	\end{equation}
	where we already calculated centralizer $(\Aut^{\perp}\ltimes W_{\Phi^{a,\perp w}})^w$.
\end{Example}

\begin{Remark}
	Situation, presented in the Proposition \ref{prop:centr_lattice} occurs only for about half cases of our folding classification. In other cases one of $Q^{a,w}$, $Q^{a,\perp w}$ (or $\widetilde{Q^{a,\perp w}}$) is generated by long root  subsystem with roots of length $|\alpha|^2=2k$ for integer $k>1$. In this case we usually can find elements in \(W^{ae}\) which acts as reflections in this long root system. 
        But such an element could act by automorphisms on the other, $|\alpha|^2=2$ lattice, so direct product in the formula \eqref{centr_lattice} becomes semidirect.
\end{Remark}

Note that symmetry on the image of folding transformation is smaller then $C(w,W^{ae})$, because some of the elements of the centralizer (for example \(w\) itself) act trivially on the image.
In Sec.~\ref{ssec:nodal} we find the symmetry group on the image in each case by similar considerations.
Also we explain, how to identify elements of centralizer with symmetries, 
using explicit formulas for coordinates  \(\f=\f(F,G), \g=\g(F,G)\) found in Sec.~\ref{sec:geom_answ}.  
Now we only state that very often $W_{Q^{a,\perp w}}^w$ acts trivially on the image of the folding transformation. 
        
\paragraph{Projective reduction.} As was already mentioned in Remark \ref{rem:projective reduction} for special values of parameters there could exist group element $t^{1/n}\in W^{ae}$, such that 
\begin{equation}
	(t^{1/n})^n\in P, \quad t^{1/n} a_i=a_i q^{k_i/n},\, k_i\in\mathbb{Z}  
\end{equation}    
	
\begin{Example}
	We had already example of projective reduction in the Introduction (Example \ref{ex:Intr2}). The dynamics of the \(q\)-Painlev\'e equation acts on root variables as \eqref{eq:intro:bar a}, hence it is a translation  $t_{\omega_2-\omega_3}$, recall that \(\omega_i\in P\) denotes \(i\)-th fundamental weight. It is easy to see that 
	\begin{equation}
		t_{\omega_2-\omega_3}=\pi^2 s_{3453} s_{2012} =(\pi s_{2012})^2
	\end{equation}
%
%
	Let us denote  $t^{1/2}_{\omega_2-\omega_3}\equiv \pi s_{2012}$. This element acts as
	\begin{equation}
		t^{1/2}_{\omega_2-\omega_3} (a_0,a_1,a_2,a_3, a_4, a_5)=(a_5, a_4, a_2 a_{0123}, a_3 a_{0123}^{-1},a_1,a_0) 
	\end{equation}
    For example, if $a_0=a_5, a_4=a_1$, this become a projective reduction and $q=a_{0123}^2$.
    In the Example  \ref{ex:Intr2} we have $a_{0,1,4,5}=-1$ so the projective reduction condition is satisfied. This is dashed dashed dynamics as in \eqref{qPVIprr a}, \eqref{qPVIprr}.  
\end{Example}

To obtain a projective reduction for a special value of $\vec{a}$ it is necessary, that some root of $q$ is a Laurent monomial in free root variables. 

For the folding invariant subset projective reduction should act as a translation on the invariant lattice root system $\Phi^{a,w}$. This transformation could act on orthogonal lattice \(Q^{a,\perp w}\) as a nontrivial element of finite order. 
        
\begin{Example}[Continuation of the Example \ref{ex:sublattices}]
	\label{ex:projred_E7}
	
	For the folding transformation $2A_3 \subset E_7^{(1)}$ we have $q=a_0^2a_4^4=(a_0 a_4^2)^2$, so we can expect for projective
	reduction as a square root of elementary translation. The elementary translation in invariant lattice is given by \(\alpha_0=2\omega_0-\omega_4\). In terms of symmetry lattice \(\Phi^{a,w}\) this is translation by simple root, but one can expect existence of the translation by fundamental weight.
	
	According to the formula \eqref{eq:example:C} the elementary translation in \(W^{ae}_{A_1}\) is \(s_0 s \pi s^{-1} \pi\), where $s=s_{4354}s_{2132}$. And we have relation
%
	\begin{equation}
		t_{\omega_4-2\omega_0}=(s_0 s \pi s^{-1} \pi
		)^2=(t^{1/2}_{\omega_4-2\omega_0})^2
	\end{equation}
	This square root \(s_0 s \pi s^{-1} \pi\) is not a translation since it acts on orthogonal root system \(\Phi^{a,w}\) as a nontrivial automorphism of order to 2, see Fig. \ref{fig:sublattices}). Hence 	\(s_0 s \pi s^{-1} \pi\) is projective reduction.
\end{Example}

\begin{Remark} 
	For any \(\Phi^a\) and \(n>1\) we can consider the following coloring. It is defined by the set $\Delta^a\setminus I$ of \emph{white} nodes which is  $\Delta^a\setminus I=\{i\in \Delta^a| \mathfrak{m}_i \mod n=0\}$ where \(\mathfrak{m}_i\) are marks in affine Dynkin diagram, see formula \eqref{eq:marks def}. If we set root variabes corresponding to black nodes to be root of 1 (according to the case 2 condition \eqref{a_b_c2}) then \(q\) becomes the $n$ power of free root variables (possibly with some additional root of unity).
	
	It is interesting, that if all connected components of \(I\) is Dynkin diagram of $A$-type then such coloring correspond to folding transformation. Such constriction gives 7 folding transformations: $3A_2,4A_1,A_1+2A_3 \subset E_8^{(1)}$, $4A_1, 3A_2 \subset E_7^{(1)}$, $4A_1 \subset E_6^{(1)}$, $4A_1 \subset D_5^{(1)}$ (see Section \ref{ssec:answers algebraic}). It also gives some folding transformations without dynamics (see Appendix \ref{sec:ntrf}).
	Let us present the proof of this fact without case-by-case considerations.

	%
	The root marks for simply laced Dynkin diagram satisfy relation
	\begin{equation}\label{markrelation}
		\forall j \in \Delta^a:  \sum_{i\in N(j)} \mathfrak{m}_i=2\mathfrak{m}_j 
	\end{equation}
	which follow from orthogonality of \(\alpha_j\) and \(\delta\). Therefore for any white node \(c\) we have 
	\begin{equation}
		\sum_{b \in I \cap N(c)} \frac{\mathfrak{m}_b }n \in \mathbb{Z}.
	\end{equation}
	
	Let us number vertices on each connected component $\Delta_{j}$ from \(1\) to \(n_j\). We also number corresponding border (white) points by $0$ and \(n_j+1\), if the there is no border (white) point, we consider corresponding mark to be equal zero. Due to relation \eqref{markrelation} 
    for each connected component $\Delta_{j}$ difference $\mathfrak{m}_k-\mathfrak{m}_{k-1} \mod n$ does not depend on $k$. We have 
    \begin{equation} 
    	(\mathfrak{m}_k-\mathfrak{m}_{k-1})(n_j+1) =0 \mod n, \quad (\mathfrak{m}_k-\mathfrak{m}_{k-1}) l \neq 0 \mod n,\;\; \text{ for }l \leq n_j.
    \end{equation} 
	Hence there exist \(0 < m_j \leq n_j\), $\gcd (m_j,n_j+1)=1$ such that \( \frac{m_j}{n_j+1} = \frac{\mathfrak{m}_k-\mathfrak{m}_{k-1}}n \mod 1\). For such choice of \(m_j\) and any vertex \(k\) in \(\Delta_j\) we have \( \frac{k m_j}{n_j+1} = \frac{\mathfrak{m}_k}{n} \mod 1\). Hence the selection rule \eqref{selection_rule_c2} is satisfied.    
	%
\end{Remark} 
\newpage

\section{Geometric realization}	\label{sec:geom_gs}
\subsection{Picard lattice} \label{ssec:Picard}
	
In the previous section, we classified all possible folding transformations and folding subgroups.
Now for each folding, we will find to which Painlev\'e equation it maps. We will also show that folding transformations
found are nontrivial and pairwise non-conjugate. In this section, we explain the methods, illustrate them with certain examples and
give the table with answers. More details are given in Section \ref{sec:geom_answ}. 

Let us fix Painlev\'e family \(\mathcal{X}\). 
Recall that any surface \(\mathcal{X}_{\vec{a}}\) is blowup of $\mathbb{P}^1\times \mathbb{P}^1$ in $8$ points. 
We denote by \(H_F,H_G, E_1,\dots,E_8\) standard generators of the Picard group of \(\mathcal{X}_{\vec{a}}\), here
\(E_i\) is an exceptional divisor corresponding to \(i\)-th blowup, \(H_F\) and \(H_G\) are given by equations \(F=\text{const}\)
and \(G=\text{const}\) correspondingly, for generic constants. 

Let  \(K\) denotes canonical divisor of \(\mathcal{X}_{\vec{a}}\). 
The anti-canonical divisor can be represented by unique curve \(D\), in case of \(q\)-difference equation this curve is a closed
chain of \(\mathbb{P}^1\)-s, namely \((-K)=\delta_0+\delta_1+\dots+\delta_n\). Each \(\delta_i\) has self-intersection
\(\delta_i^2=-2\) and pairwise intersections are \(\delta_i\cdot \delta_j =1\) if \(|i-j|=1\) or \(|i-j|=n\) and zero otherwise.
Hence \(\delta_i\) generate lattice of the type \(A_n^{(1)}\) in the Picard group, this is so called \emph{surface lattice}
\(\Phi^{a}_{\text{surf}}\). The orthogonal complement to the surface lattice is called \emph{symmetry lattice} \(\Phi^{a}_{\text{sym}}=\Phi^{a}\).
The (affine extended) Weyl group of \(\Phi^{a}\) is the symmetry group of the family \(\mathcal{X}\) which we used in the previous section.
The Sakai's families are labeled by their symmetry/surface types.

\begin{Example}\label{Example:geom1}
	Take any numbers \(a_0,\dots, a_7 \in \mathbb{C}^*, \, a_0\neq 1\), let \(\mathcal{X}_{\vec a}\) be a blowup
	of \(\mathbb{P}^1\times \mathbb{P}^1\) at 8 points with coordinates as in the Fig. \ref{Fig:blowup E7}.
	
	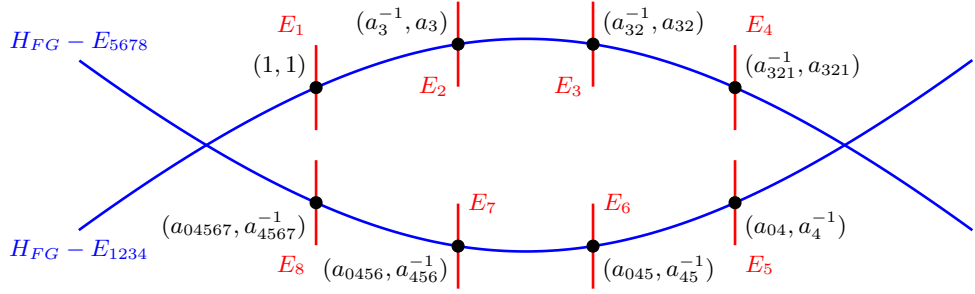
\begin{figure}[h]
	\begin{center}
		\begin{tikzpicture}[scale=4, rotate=45]
			\draw [domain=2.5:0.4,samples=800, blue,line width=1pt] plot (\x,1/\x) node[above] {\small $H_{FG}-E_{5678}$};
			\draw [domain=0.4:2.5,samples=800,blue,line width=1pt] plot 	(\a-\x,\a-1/\x) node[anchor = north] {\small $H_{FG}-E_{1234}$};
		
			\draw[red, line width=1pt]  (\a-1/\s^4.5-\l,\a-\s^4.5-\l) --  (\a-1/\s^4.5+\l,\a-\s^4.5+\l) node[above left] {\small $E_1$};
			\draw[fill] (\a-1/\s^4.5,\a-\s^4.5) circle (0.02cm) 	node[above left] {\small  $(1,1)$};
		
			\draw[red, line width=1pt]  (\a-1/\s^1.5+\l,\a-\s^1.5+\l) --  (\a-1/\s^1.5-\l,\a-\s^1.5-\l) node[left] {\small $E_2$};
			\draw[fill] (\a-1/\s^1.5,\a-\s^1.5) circle (0.02cm) node[above left] {\small $(a_3^{-1},a_3)$};
		
			\draw[red, line width=1pt]  (\a-\s^1.5+\l,\a-1/\s^1.5+\l) --  (\a-\s^1.5-\l,\a-1/\s^1.5-\l) node[left] {\small $E_3$};
			\draw[fill] (\a-\s^1.5,\a-1/\s^1.5) circle (0.02cm) node[above right] {\small  $(a_{32}^{-1},a_{32})$};
		
			\draw[red, line width=1pt]  (\a-\s^4.5-\l,\a-1/\s^4.5-\l) --  (\a-\s^4.5+\l,\a-1/\s^4.5+\l) node[above right] {\small $E_4$};
			\draw[fill] (\a-\s^4.5,\a-1/\s^4.5) circle (0.02cm) node[above right] {\small  $(a_{321}^{-1},a_{321})$};
		
			\draw[red, line width=1pt]  (1/\s^4.5+\l,\s^4.5+\l) --  (1/\s^4.5-\l,\s^4.5-\l) node[below right] {\small $E_5$};
			\draw[fill] (1/\s^4.5,\s^4.5) circle (0.02cm) node[below right] {\small  $(a_{04},a_4^{-1})$};
		
			\draw[red, line width=1pt]  (1/\s^1.5-\l,\s^1.5-\l) --  (1/\s^1.5+\l,\s^1.5+\l) node[right] {\small $E_6$};
			\draw[fill] (1/\s^1.5,\s^1.5) circle (0.02cm) node[below right] {\small  $(a_{045}, a_{45}^{-1})$};
		
			\draw[red, line width=1pt]  (\s^1.5-\l,1/\s^1.5-\l) --  (\s^1.5+\l,1/\s^1.5+\l) node[right] {\small $E_7$};
			\draw[fill] (\s^1.5,1/\s^1.5) circle (0.02cm) node[below left] {\small  $(a_{0456},a_{456}^{-1})$};
		
			\draw[red, line width=1pt]  (\s^4.5+\l,1/\s^4.5+\l) -- (\s^4.5-\l,1/\s^4.5-\l) node[below left] {\small $E_8$};
			\draw[fill] (\s^4.5,1/\s^4.5) circle (0.02cm) node[below left] {\small  $(a_{04567},a_{4567}^{-1})$};
		
		\end{tikzpicture}
	\end{center}
	\caption{Blowup scheme for \(E_7^{(1)}/A_1^{(1)}\) \label{Fig:blowup E7}}
	\end{figure}
	These points belong to two curves of degree \((1,1)\) on  \(\mathbb{P}^1\times \mathbb{P}^1\) given by equations \(F G=1\) and \(FG =a_0\).
	Therefore the meromorphic form \(\dfrac{dF\wedge dG}{(FG-1)(FG-a_0)}\) has no zeroes on \(\mathcal{X}_{\vec a}\) and only poles on
	the preimages of these curves. Hence \((-K)\) consists of two components
	\begin{equation}
		\delta_0=H_{FG}-E_{1234}, \quad 	\delta_1=H_{FG}-E_{5678},
	\end{equation}
	which generate surface lattice. 
	Its type is \(A_1^{(1)}\). We see that if $a_0=1$, then $-K$ degenerate, that's why this value is removed.
	The symmetry lattice is generated by classes 
	\begin{align}
		\alpha_0=H_F-H_G, \; \alpha_1=E_3-E_4,\; \alpha_2=E_2-E_3,\; \alpha_4=E_1-E_2,\\
		\alpha_5=H_G-E_{15}, \; \alpha_6=E_5-E_6,\; \alpha_7=E_6-E_7,\; \alpha_7=E_7-E_8.
	\end{align}
	The intersection pairing between these roots is (opposite to) intersection pairing in root lattice \(E_7^{(1)}\).
    \begin{center}
		\begin{tikzpicture}[elt/.style={circle,draw=black!100,thick, inner sep=2pt,minimum size=2mm},scale=1.25]
		\path 	(-3,0) 	node 	(a1) [elt] {\scriptsize $1$}
		(-2,0) 	node 	(a2) [elt] {\scriptsize $2$}
		( -1,0) node  	(a3) [elt] {\scriptsize $3$}
		( 0,0) 	node  	(a4) [elt] {\scriptsize $4$}
		( 1,0) 	node 	(a5) [elt] {\scriptsize $5$}
		( 2,0)	node 	(a6) [elt] {\scriptsize $6$}
		( 3,0)	node 	(a7) [elt] {\scriptsize $7$}
		( 0,1)	node 	(a0) [elt] {\scriptsize $0$};
		\draw [black,line width=1pt] (a1) -- (a2) -- (a3) -- (a4) -- (a5) --  (a6) -- (a7) (a4) -- (a0);
		\node at ($(a1.north) + (0,0.2)$) 	{\small $E_3-E_4$};
		\node at ($(a2.south) + (0,-0.2)$)  {\small  $E_2-E_3$};
		\node at ($(a3.north) + (0,0.2)$)  {\small  $E_1-E_2$};
		\node at ($(a4.south) + (0,-0.2)$)  {\small $H_G-E_{15}$};	
		\node at ($(a5.north) + (0,0.2)$)  {\small  $E_5-E_6$};		
		\node at ($(a6.south) + (0,-0.2)$) 	{\small $E_6-E_7$};	
		\node at ($(a7.north) + (0,0.2)$) 	{\small  $E_7-E_8$};	
		\node at ($(a0.north) + (0,0.2)$) 	{\small  $H_F-H_G$};		
	\end{tikzpicture}
	
	\end{center} 
	Hence we get symmetry/surface type \(E_7^{(1)}/A_1^{(1)}\).
\end{Example}

\begin{wraptable}{r}{5.5cm}	
	\begin{tabular}{|c|c|}
		\hline
		Self-intersection & Line \\
		\hline
		$\geq 0$ &  \begin{tikzpicture}
			\draw [black]  (0,0) -- 
			(2,0);
			\draw [black, line width=1pt]  (0,-\l) -- 
			(2,-\l);
		\end{tikzpicture} \\
		\hline
		-1 & \begin{tikzpicture}
			\draw [red, line width=1pt]  (0,0) -- 
			(2,0);
		\end{tikzpicture} \\
		\hline
		-2 & \begin{tikzpicture}
			\draw [blue, line width=1pt]  (0,0) -- 
			(2,0);
		\end{tikzpicture} \\
		\hline
		-3 & \begin{tikzpicture}
			\draw [cyan, line width=1pt]  (0,0) -- 
			(2,0);
		\end{tikzpicture} \\
		\hline
		-4 & \begin{tikzpicture}
			\draw [brown, line width=1pt]  (0,0) -- 
			(2,0);
		\end{tikzpicture} \\
		\hline
		-6 & \begin{tikzpicture}
			\draw [violet, line width=1pt]  (0,0) -- 
			(2,0);
		\end{tikzpicture} \\
		\hline
		-8 & \begin{tikzpicture}
			\draw [magenta, line width=1pt]  (0,0) -- 
			(2,0);
		\end{tikzpicture} \\
		\hline
		non-specified & \begin{tikzpicture}
			\draw [gray, line width=1.5pt]  (0,0) --
			(2,0);
		\end{tikzpicture} \\
		\hline
	\end{tabular}
\caption{Self-intersections in Figures   \label{Tab:selfintersections}}
\end{wraptable}

\begin{Remark}
	Here and below in figures we draw rational curves with self-intersection \((-1)\) in red and rational
	curves with self-intersection \((-2)\) in blue. Generally we will  distinguish curves with different self-intersections by different
	colors according to the Table \ref{Tab:selfintersections}		
\end{Remark}

	The parameters \(a_i\) used there are multiplicative root variables. They can be defined using period map
	\(\chi\colon Pic(\mathcal{X}_{\vec a})\rightarrow \mathbb{C}^*\), see \cite[Sec.~5]{Sakai01}
	\footnote{Note that in multiplicative case (on which we restrict in this paper) the image of the period map \(\chi\)
	is \(\mathbb{C} \mod \mathbb{Z}\). So we applied \(\exp(2 \pi \ri \cdot)\) to \(\chi\) in loc. cit}.
	Then \(a_i=\chi(\alpha_i)\), where \(\alpha_0,\dots,\alpha_r\) are simple roots in the symmetry system \(\Phi^a\).
	
 	Usually the classes of roots in \(\Phi^a\) are not effective, there are no curves representing such classes.
 	But for some special values of root variables some of these classes might become effective. Recall \cite[Sec.~3]{Sakai01}
 	that a smooth rational curve with self-intersection $(-2)$ is called a \emph{nodal curve}. According to this definition all
 	components \(\delta_0,\dots, \delta_n\) of \((-K)\) are nodal curves, but we will mostly interest in nodal curves in symmetry
 	lattice \(\Phi^a\), the set of such curves will be denoted as \(\Delta^{\nod}\). 
 	
 	\begin{prop}[{\cite[Prop. 22]{Sakai01}}]\label{prop:nodal}
 		Denote by \(W^{\nod}\) the subgroup of \(W_{\Phi^a}\)
 		generated by the reflections with respect to \(\alpha\in \Delta^{\nod}\). Then \(\chi(\alpha)=1\) if and
 		only if \(\alpha\in W^{\nod}\cdot \Delta^{\nod}\).
 	\end{prop}
 	
	Let \(\alpha \in Q^a\) be a root (i.e. \(\alpha^2=-2\)), \(\chi(\alpha)\) be a corresponding root variable and
	\(s_\alpha\colon X_{\vec{a}} \rightarrow X_{s_\alpha(\vec{a})}\) corresponding reflection. For generic values of parameters
	we have \(\chi(s_\alpha(\alpha))=\chi(\alpha)^{-1}\). Now assume that \(\alpha\) is a nodal curve \(\alpha\in \Delta^{\nod}\), hence
	\(\chi(\alpha)=1\) and assume that other root variables are generic. We have \(\chi(s_\alpha(\alpha))=1\),
	hence \(s_\alpha(\alpha)=n\alpha\) for \(n \in \mathbb{Z}\), hence \(s_\alpha(\alpha)=\alpha\) since \((-\alpha)\) could not be effective.
	Moreover, for any \(\beta \in Q^a\)	we have \(\beta \cdot \alpha = s_\alpha(\beta) \cdot \alpha\) and \(\chi(s_\alpha(\beta))=\chi(\beta)\), hence \(s_\alpha(\beta)=\beta\). Since \(s_\alpha\) acts trivially on Picard lattice it should be a trivial transformation of \(\mathcal{X}_{\vec{a}}\). Clearly, the same property holds for divisors in \(\alpha \in W^{\nod}\cdot \Delta^{\nod}\), in other words for any divisor with \(\chi(\alpha)=1\) the corresponding
 	reflection should be trivial transformation of \(\mathcal{X}_{\vec{a}}\). This can be considered as a geometric proof of
 	Lemma~\ref{lemma:reflection}.

	For any two nodal curves \(\alpha,\beta \in \Delta_{\nod}\) we have either \(\alpha \cdot \beta =0\) or \(\alpha \cdot \beta =1\)
	since otherwise \((\alpha+\beta)^2 \geq 0\) and non proportional to \(\delta\) and this is forbidden in \(Q^a\).
	Therefore elements of \(\Delta_{\nod}\) are simple roots for some finite root system which we denote by \(\Phi_{\nod}\).
	We denote by \(Q_{\nod}\) the sublattice in \(Q^a\) generated by \(\Phi_{\nod}\) and
	\(\overline{Q}_{\nod}=(Q_{\nod} \otimes_\mathbb{Z} \mathbb{Q} ) \cap Q^a\). In general the lattices \(Q_{\nod}\) and
	\(\overline{Q}_{\nod}\) are different, but \(Q_{\nod}\) is a sublattice of finite index in \(\overline{Q}_{\nod}\).
	Usually the lattice \(\overline{Q}_{\nod}\) is generated by some root system, we denote this root system by 
	\(\overline{\Phi}_{\nod}\).

\subsection{Quotient}

\paragraph{Surface.} Let \(w \in W^{ae}\) be a true folding transformation. Recall that it means there exists \(\vec{a}\) such that
\(w (\vec{a})=\vec{a}\), but \(w\) acts on \(X=\mathcal{X}_{\vec{a}}\) nontrivially and there exists translation \(t_\lambda \in W^{ae}\)
which commutes with \(w\). Hence \(t_\lambda\) acts on the quotient \(Y'=X/\langle w \rangle\), where \(\langle w \rangle\) is
a group generated by \(w\). This action of \(t_\lambda\) is a new Painlev\'e dynamics, dynamics after folding transformation.
But in general \(Y'\) is not a rational surface corresponding to the Painlev\'e equation, one have to first resolve singularities
\(\widetilde{Y}\rightarrow Y'\) and then (possibly) perform additional blowdown \(\widetilde{Y}\rightarrow Y\).
 \begin{center}
 \includegraphics{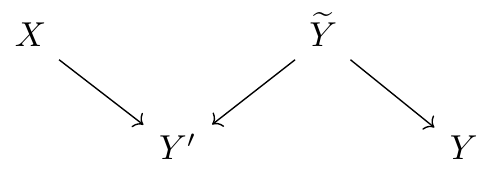}
 \end{center}

Let us discuss these steps in more detail. The singularities of \(Y'\) are images of points of \(p\in X\) which are
stabilized by nontrivial subgroups in \(\langle w \rangle\), denote generator of this subgroup by \(w^\sharp\).
The order of \(w\) can be \(2,3,4\) (this follows from classification performed in Section \ref{sec:class_gp}),
then \(w^\sharp=w\) or \(w^\sharp=w^2\) if order of \(w\) is 4. There exist some coordinates in the neighborhood of
\(p\) such that \(w^\sharp\) acts as \(F \mapsto F \zeta_m, G \mapsto G \zeta_{m}^d\), where \(\zeta_m\) is some primitive
root of unity of order \(m\), \(m\) is the order of \(w^\sharp\), and \(0<  d < m\). 

The minimal resolution of such singularities can be constructed using toric geometry as in \cite[Sec 2.6]{Fulton93}. After the resolution, we get an exceptional
divisor which is the sequence of \(\mathbb{P}^1\) with negative self-intersection.  In the table below we
give the structure of the exceptional divisor and intersection pairing on the components of the exceptional divisor.
We restrict to the case \(m\leq 4\) needed in the paper.  We also give formulas for the coordinates in the neighborhood
of the exceptional divisor. The last column stands for the contribution in self-intersection and will be explained below.
\begin{table}[h]
	\begin{center}
		\begin{tabular}{|c|c|c|c|c|c|}
			\hline
			Order  & Action 
			& Exceptional divisor & Intersection  & \(-d/m\)\\
			\hline 
			$2$ & $F\mapsto \zeta_2F,\, G\mapsto \zeta_2G$ 
			& \begin{tikzpicture}
				\draw [blue, line width=1pt]  (0,0) --
				(2,0);
				\draw (0,0) node[left] {\tiny $(F^2, G/F)$};
				\draw (2,0) node[right] {\tiny $(F/G, G^2)$};
			\end{tikzpicture} & $A_1$ & -1/2 \\
			\hline
			$3$ & $F\mapsto \zeta_3 F,\, G\mapsto \zeta_3 G$ 
			& \begin{tikzpicture}
				\draw [cyan, line width=1pt]  (0,0) -- 
				(2,0);
				\draw (0,0) node[left] {\tiny $(F^3, G/F)$};
				\draw (2,0) node[right] {\tiny $(F/G, G^3)$};
			\end{tikzpicture} & $(-3)$ & -1/3 \\
			\hline
			$3$ & $F\mapsto \zeta_3 F,\, G\mapsto \zeta_3^{-1}G$ 
			&
			\begin{tikzpicture}
				\draw [blue, line width=1pt]  (2*\l,0) --
				(-1.5,0.5);
				\draw [blue, line width=1pt]  (-2*\l,0) --
				(1.5,0.5);
				\draw (-1.5,0.5) node[left] {\tiny $(F^3, G/F^2)$};
				\draw (1.5,0.5) node[right] {\tiny $(F/G^2, G^3)$};
				\draw (0,0) node[below] {\tiny $(F^2/G, G^2/F)$};
			\end{tikzpicture}		
			& $A_2$ & -2/3 \\
			\hline
			$4$& $F\mapsto \zeta_4F,\, G\mapsto \zeta_4G$ 
			& \begin{tikzpicture}
				\draw [brown, line width=1pt]  (0,0) -- 
				(2,0);
				\draw (0,0) node[left] {\tiny $(F^4, G/F)$};
				\draw (2,0) node[right] {\tiny $(F/G, G^4)$};
			\end{tikzpicture} & $(-4)$ & -1/4\\
			\hline
			$4$& $F\mapsto \zeta_4F,\, G\mapsto \zeta_4^{-1}G$ 
			& 
			\begin{tikzpicture}
				\draw [blue, line width=1pt] (-0.5+\l,-\l) -- (-1.5,1);
				\draw [blue, line width=1pt]  (-0.5-2*\l,0) -- (0.5+2*\l,0);
				\draw [blue, line width=1pt]  (0.5-\l,-\l) -- (1.5,1);
				\draw (-0.5-\l,0) node[left] {\tiny $(F^3/G, F^2/G^2)$};
				\draw (0.5+\l,0) node[right] {\tiny $(G^2/F^2, G^3/F)$};
				\draw (-1.5,1) node[left]  {\tiny $(F^4, G/F^3)$};
				\draw (1.5,1) node[right] {\tiny $(F/G^3,G^4)$};
			\end{tikzpicture}
			& $A_3$ & -3/4 \\
			\hline
			$4$ & $F\mapsto -F, \, G\mapsto \zeta_4 G$ 
			&
			\begin{tikzpicture}
				\draw [blue, line width=1pt]  (0,0) -- (2,0);
				\draw (0,0) node[left] {\scriptsize $(F^2, G^2/F)$};
				\draw (2,0) node[right] {\scriptsize $(F/G^2, G^4)$};
			\end{tikzpicture}
			& $A_1$ & -1/2 \\
			\hline
		\end{tabular}
	\end{center}
	\caption{Singularities \label{Tab:singluar}}
\end{table}

It appears that in some cases the resolution \(\widetilde{Y}\) is not minimal, i.e. differs from representative of Painlev\'e family by 
additional blowup. In such case effective divisor \((-K)\) contains \((-1)\) curves. After blowdown of these curves, we get the surface \(Y\)
on which new Painlev\'e dynamics acts.

\begin{Example}[Continuation of Example \ref{Example:geom1}] \label{Example:geom2}
Consider folding transformation \(2 A_3 \subset E_7\) from the table in Section \ref{ssec:answers algebraic}. According to this
table the corresponding root variables should satisfy (for one of two irreducible components, see Remark \ref{rem:concomp})
\begin{equation}\label{eq:2A3 in E7 root variables}
	a_{1}=a_2=a_3=a_5=a_6=a_7=\ri.
\end{equation}

The action of the \(w\) is given by 
\begin{equation}
	w=s_{321765}: \quad F\mapsto -{\ri} F, \quad G\mapsto {\ri} G   
\end{equation}
where we used realization of the Weyl group given in Section \ref{ssec:E7/A1}. The order of \(w\) is 4. 
There are 4 fixed points in \(X\), we list them in the Table \ref{Tab:2A3 points}.
The meaning of the last column will be explained later, after identification with standard geometry.

\begin{table}[h]
	\begin{center}
		\begin{tabular}{|c|c|c|c|}
			\hline 
			sing. pt. & action & except. div. & Image\\
			\hline
			$(0,0)$ &  $F\mapsto -{\ri} F, \, G\mapsto {\ri} G$   & $A_3$ & $\e_5-\e_6, \, \e_6-\e_7, \, \e_7-\e_8$ \\
			\hline
			$(0,\infty)$ &  $F\mapsto -{\ri} F, \, G^{-1}\mapsto -{\ri} G^{-1} $
			& $(-4)$ & 
			$\h_{\g}-\e_{1234}$  \\
			\hline
			$(\infty,0)$ &  $F^{-1}\mapsto {\ri} F^{-1}, \, G\mapsto {\ri} G$ & $(-4)$ & 
			$\h_{\g}-\e_{5678}$ \\
			\hline
			$(\infty,\infty)$ & $F^{-1}\mapsto {\ri} F^{-1}, \, G^{-1}\mapsto -{\ri} G^{-1}$ & $A_3$ & $\e_1-\e_2, \, \e_2-\e_3, \, \e_3-\e_4$ \\
			\hline
		\end{tabular}
	\end{center}
\caption{Singular points for the folding \(2A_3 \subset E_7\)  \label{Tab:2A3 points}}
\end{table}

Now we proceed to geometry. It is easy to see that \(w\) permutes exceptional divisors \(E_1,E_2,E_3,E_4\) and also \(E_5,E_6,E_7,E_8\)
hence in the \(Y'\) they will go to two \((-1)\) curves.

Consider curves \(F G=1\), \(F G=a_0\) which are components of the \((-K)\).
Curves of the form $FG=const$ form a pencil with intersection \(2\) between each two curves. In the image of this pencil in \(\widetilde{Y}\) such intersections become 0, since we made resolution at points \((0,\infty)\) and \((\infty,0)\). Hence, the images of curves $FG=const$ have self-intersection 0, for generic \(const\). Therefore, the images of curves  \(F G=1\), \(F G=a_0\)  have  self-intersection \(-1\) due to one additional blowup mentioned above.

It appears that images in \(\widetilde{Y}\) of coordinate lines
\(F=0\), \(F=\infty\), \(G=0\), \(G=\infty\) are also \(-1\) divisors.
This (and also previous) self-intersection could be calculated, using standard intersection theory results, which we review below, see formula \eqref{eq:self-intersection}.


The first two pictures in  Fig. \ref{Fig:geom examp} represent \(X\) and \(\widetilde{Y}\).
Here and below we draw fixed point \(p \in X\) by black circles.
\begin{figure}[ht]
\begin{center}
\begin{tikzpicture}[scale=3]
	\draw [domain=0.4:2.5,samples=800, line width=1pt,blue] plot (\x,1/\x) node[below] {\small $H_{FG}-E_{5678}$};
	\draw [domain=2.5:0.4,samples=800,line width=1pt,blue] plot (\a-\x,\a-1/\x) node[left] {\small $H_{FG}-E_{1234}$};
	
	\draw [line width=1pt,red] (\s^3-\l,1/\s^3-\l) -- (\s^3+\l,1/\s^3+\l) node[above] {\small $E_8$};
	\draw [line width=1pt,red] (\s-\l,1/\s-\l) -- (\s+\l,1/\s+\l) node[pos=0, below left] {\small $E_7$};
	\draw [line width=1pt,red] (1/\s-\l,\s-\l) -- (1/\s+\l,\s+\l) node[pos=0, below left] {\small $E_6$};
	\draw [line width=1pt,red] (1/\s^3-\l,\s^3-\l) -- (1/\s^3+\l,\s^3+\l) node[right] {\small $E_5$};

	\draw[->, dashed] 
	(\s^3+\l,1/\s^3+\l) to
	[bend left=20] (\s+\l,1/\s+\l);
	
	\draw[->, dashed] 
	(\s+\l,1/\s+\l) to
	[bend left=20] (1/\s+\l,\s+\l);
	
	\draw[->, dashed] 
	(1/\s+\l,\s+\l) to
	[bend left=20] (1/\s^3+\l,\s^3+\l);
	
	\draw[<-, dashed] 
	(\s^3-\l,1/\s^3-\l) to
	[bend right=20] (1/\s^3-\l,\s^3-\l);
	
	\draw [line width=1pt,red] (\a-\s^3-\l,\a-1/\s^3-\l) -- (\a-\s^3+\l,\a-1/\s^3+\l) node[pos=0, below] {\small $E_1$};
	\draw [line width=1pt,red] (\a-\s-\l,\a-1/\s-\l) -- (\a-\s+\l,\a-1/\s+\l) node[above right] {\small $E_2$};
	\draw [line width=1pt,red] (\a-1/\s-\l,\a-\s-\l) -- (\a-1/\s+\l,\a-\s+\l) node[above right] {\small $E_3$};
	\draw [line width=1pt,red] (\a-1/\s^3-\l,\a-\s^3-\l) -- (\a-1/\s^3+\l,\a-\s^3+\l) node[pos=0, left] {\small $E_4$};
	
	\draw[->, dashed] 
	(\a-\s^3+\l,\a-1/\s^3+\l) to
	[bend right=20] (\a-\s+\l,\a-1/\s+\l);
	
	\draw[->, dashed] 
	(\a-\s+\l,\a-1/\s+\l) to
	[bend right=20] (\a-1/\s+\l,\a-\s+\l);
	
	\draw[->, dashed] 
	(\a-1/\s+\l,\a-\s+\l) to
	[bend right=20] (\a-1/\s^3+\l,\a-\s^3+\l);
	
	\draw[<-, dashed] 
	(\a-\s^3-\l,\a-1/\s^3-\l) to
	[bend left=20] (\a-1/\s^3-\l,\a-\s^3-\l);
	
	\draw (0.4,2) -- (2.5,2); 
	\draw (0.4,0.5) -- (2.5,0.5); 
	\draw (0.5,0.4) -- (0.5,2.5); 
	\draw (2,0.4) -- (2,2.5);

	
	\draw[fill] (0.5,0.5) circle (0.02cm);
	\draw[fill] (0.5,2) circle (0.02cm);
	\draw[fill] (2,0.5) circle (0.02cm);
	\draw[fill] (2,2) circle (0.02cm);
	
	\draw[->] (2.2,1.5) -- (2.8,1.5) node[pos=0.5, above] {$/w$};
	
	\begin{scope}[xshift=3cm,scale=1]
		\draw [line width=1pt,red] (0.5-2*\l,2-2*\l)  -- (2-2*\l,0.5-2*\l)
		node [right] {\tiny $\e^a$};
		\draw [line width=1pt,red]   (2+2*\l,0.5+2*\l) -- (0.5+2*\l,2+2*\l)
		node [left] {\tiny $\e^b$};
		
		\draw [line width=1pt,red]  (1.25+2*\l-\l,1.25+2*\l-\l) -- (1.25+2*\l+\l,1.25+2*\l+\l)
		;
		
		\draw [line width=1pt,red]  (1.25-2*\l-\l,1.25-2*\l-\l) -- (1.25-2*\l+\l,1.25-2*\l+\l)
		;
		
		\draw [line width=1pt,brown]  (0.5+\l-5*\l,2-\l-5*\l) -- (0.5+\l+5*\l,2-\l+5*\l)
				node [above] {\tiny $\h_{\f\g}-\e_{1234}-\e^{ab}$}
		;
		
		\draw [line width=1pt,brown]   (2-\l+5*\l,0.5+\l+5*\l) -- (2-\l-5*\l,0.5+\l-5*\l)
				node [below] {\tiny $\h_{\f\g}-\e_{5678}-\e^{ab}$}
		;
		
		\draw [line width=1pt,red]  (0.5-3*\l,0.5) -- 
		(0.5-3*\l,2)
		;
		
		\draw [line width=1pt,red]    
		(2+3*\l,2) -- (2+3*\l,0.5)
		;
		
		\draw [line width=1pt,red]  (0.5,0.5-3*\l) -- 
		(2,0.5-3*\l)
		;
		
		\draw [line width=1pt, red]   
		(2,2+3*\l) -- (0.5,2+3*\l)
		;
		
		\draw [line width=1pt,blue]  (5*\l,0.5+\l) -- (\l,0.5+\l)
		;
		
		\draw [line width=1pt,blue]  (0.5+\l,5*\l) -- (0.5+\l,\l)
		;
		
		\draw [line width=1pt,blue]  (5*\l-0.5*\l-\l,0.5+\l-0.5*\l+\l) -- (+0.5+\l-0.5*\l+\l,5*\l-0.5*\l-\l)
		;
		
		\draw [line width=1pt,blue]  (2,2-\l) -- (2+4*\l,2-\l)
		;
		
		\draw [line width=1pt,blue]  (2-\l,2) -- (2-\l,2+4*\l) 
		;
		
		\draw [line width=1pt,blue]  (2+0.5*\l+\l,2-\l+0.5*\l-\l) -- (2-\l+0.5*\l-\l,2+0.5*\l+\l)
		;
		
		\draw[->] (1.25,-0.25) -- (0.75,-0.75) node[pos=0.5, right] {bl. d. $\e^{ab}$};
		
	\end{scope}
 	\begin{scope}[rotate=-45, xshift=0.75cm, yshift=-1cm, scale=1.5]
		\draw [domain=0.4:2.5,samples=800, line width=1pt,blue] plot (2.5-\x,1/\x)
		node[above] {\tiny $\h_{\f\g}-\e_{5678}$};
		\draw [domain=0.4:2.5,samples=800,line width=1pt,blue] plot (2.5-\a+\x,\a-1/\x)
		node[below] {\tiny $\h_{\f\g}-\e_{1234}$};

		\draw [line width=1pt,red]
		(2.5-\s^2-1.5*\l,1/\s^2+1.5*\l) -- (2.5-\s^2+2*\l,1/\s^2-2*\l)
		node[right] {\tiny $\h_{\g}-\e_1$};
		
		\draw [line width=1pt,red]
		(2.5-1/\s^2-1.5*\l,\s^2+1.5*\l) -- (2.5-1/\s^2+2*\l,\s^2-2*\l)
		node[right] {\tiny $\e_8$};
		
		\draw [line width=1pt,red]
		(2.5-\a+\s^2-2*\l,\a-1/\s^2+2*\l) -- (2.5-\a+\s^2+1.5*\l,\a-1/\s^2-1.5*\l)
		node[pos=0,left] {\tiny $\h_{\g}-\e_5$};
		
		\draw [line width=1pt,red]
		(2.5-\a+1/\s^2-2*\l,\a-\s^2+2*\l) -- (2.5-\a+1/\s^2+1.5*\l,\a-\s^2-1.5*\l)
		node[pos=0,left] {\tiny $\e_4$};
		
		\draw [line width=1pt,blue]
		(2.5-\s^2-1.5*\l+\l+\l,1/\s^2+1.5*\l+\l-\l) -- (2.5-\a+1/\s^2+1.5*\l+\l-\l,\a-\s^2-1.5*\l+\l+\l)
		node[pos=0.5, right] {\tiny $\e_2-\e_3$};
		
		\draw [line width=1pt,blue]
		(2.5-1/\s^2-1.5*\l-\l+\l,\s^2+1.5*\l-\l-\l) -- (2.5-\a+\s^2+1.5*\l-\l-\l,\a-1/\s^2-1.5*\l-\l+\l)
		node[pos=0.5, left] {\tiny $\e_6-\e_7$};
		
		\draw [line width=1pt,blue]
		(2.5-\s^2-1.5*\l+\l+\l,1/\s^2+1.5*\l+\l-\l+\l) -- (2.5-\s^2-1.5*\l+\l+\l-2*\l,1/\s^2+1.5*\l+\l-\l+\l-2*\l)
		node[below left] {\tiny $\e_1-\e_2$};
		
		\draw [line width=1pt,blue]
		(2.5-\a+1/\s^2+1.5*\l+\l-\l+\l,\a-\s^2-1.5*\l+\l+\l) -- (2.5-\a+1/\s^2+1.5*\l+\l-\l+\l-2*\l,\a-\s^2-1.5*\l+\l+\l-2*\l)
		node[below left] {\tiny $\e_3-\e_4$};
		
		\draw [line width=1pt,blue]
		(2.5-1/\s^2-1.5*\l-\l+\l-\l,\s^2+1.5*\l-\l-\l) -- (2.5-1/\s^2-1.5*\l-\l+\l-\l+2*\l,\s^2+1.5*\l-\l-\l+2*\l)
		node[above right] {\tiny $\e_7-\e_8$};
		
		\draw [line width=1pt,blue]
		(2.5-\a+\s^2+1.5*\l-\l-\l,\a-1/\s^2-1.5*\l-\l+\l-\l) -- 
		(2.5-\a+\s^2+1.5*\l-\l-\l+2*\l,\a-1/\s^2-1.5*\l-\l+\l-\l+2*\l)
		node[above right] {\tiny $\e_5-\e_6$};
		
	\end{scope}
\end{tikzpicture}
\end{center}
\caption{Geometry for the folding \(2A_3 \subset E_7^{(1)}\) \label{Fig:geom examp}}
\end{figure}
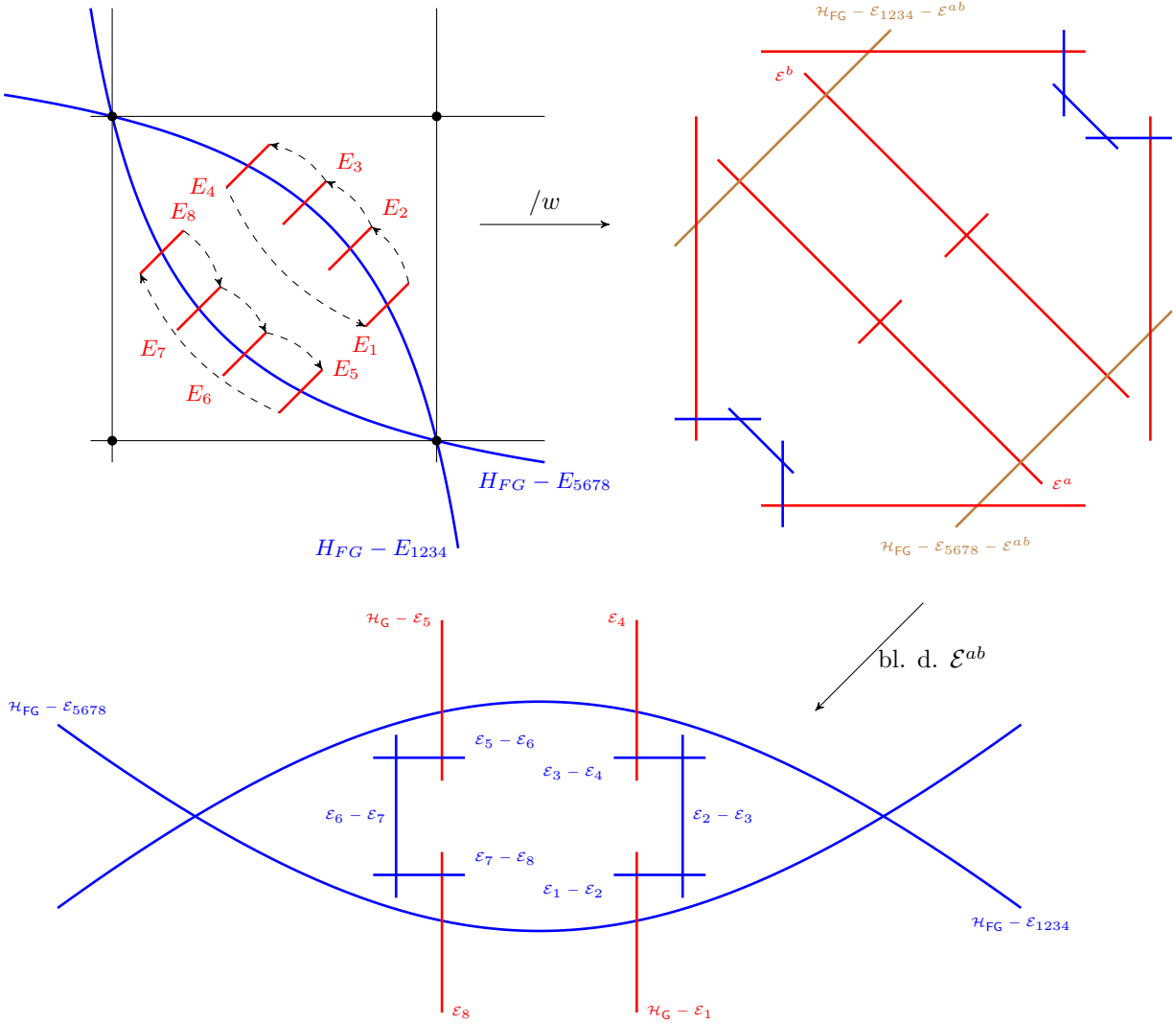

We see that anticanonical divisor on \(\widetilde{Y}\) consists of 4 components which are \((-4),(-1),(-4),(-1)\) curves.
The map \(\widetilde{Y}\rightarrow Y\) is blow down of these two \((-1)\) curves. The picture for \(Y\) is the third picture on
Fig. \ref{Fig:geom examp}. On the surface \(Y\) the anticanonical divisor consists of two \((-2)\) curves
and has zero self-intersection, this is \(A_1\) surface type.
Hence we have folding transformation from \(E_7^{(1)}/A_1^{(1)}\) \(q\)-Painlev\'e equation to itself.
\end{Example}

\paragraph{Intersections and Standard geometry.}
	
We have constructed above the birational map \(\psi \colon X \dashrightarrow Y\), such that \(\psi(w(p))=\psi(p)\) for any \(p \in X\). 
Now we have to identify \(Y\) with \(\mathcal{Y}_{\vec{\A}}\) which is the member of some Painlev\'e family \(\mathcal{Y}\). 
	
The surface \(\mathcal{Y}_{\vec{\A}}\) is defined through its blowup structure, namely choice of points
\(p_1,\dots, p_8 \in \mathbb{P}^1\times \mathbb{P}^1\). The configuration of these points could be quite degenerate,
including many infinitely near points. We follow \cite{KNY15} in the description of blowup structure on Painlev\'e family
and call the structure from loc. cit. a standard geometry. We recall necessary details in Section \ref{sec:geom_answ}. 

The choice of standard geometry leads to choice of the basis \(\e_1,\dots \e_8, \h_{\f}, \h_\g \), here \(\f,\g\)  are some coordinates on \(Y\),
choice of them is part of the choice of standard geometry. The standard geometry is not unique. In order to have formulas for \(\f,\g\)
simple it is natural to have \(\e_1,\dots \e_8, \h_{\f}, \h_\g \) as images of some natural divisors in \(X\).

Let us recall basic statements of intersection theory which we will use below. Let \(\psi' \colon X \rightarrow Y'\) denotes the quotient map.
This is a finite map of degree \(m\), where \(m\) is order of \(w\). Let \(C\) be a curve on \(X\) and \(D'\) be a curve on \(Y'\) such that
we have relation in Picard group \([C]=(\psi')^* ([D']) \Rightarrow \psi'_*([C])=m[D']\). Due to push-pull formula
(see e.g. \cite{Eisenbud:Harris:2016}) we have relation on self-intersection \((D',D')=(C,C)/m\).
Note that curve \(C\) should be \(w\)-invariant but not necessarily irreducible or reduced.

Let \(\widetilde{D}\) denotes the proper transform of \(D'\) in \(\widetilde{Y}\). The self-intersection of \(\widetilde{D}\) differs
from self-intersection \(D'\) by contribution of exceptional divisor in \(\widetilde{Y}\). The last contributions are local and can
be computed for example using toric description of resolution of singularities in  \cite[Sec 2.6]{Fulton93}. Here we present an answer,
in order to do this we need some notations.

Let \(x_i \in C\) be point which goes to singular point of \(Y'\). Assume first for simplicity that \(x_i\) is \(w\)-invariant and
action of \(w\) on \(C\) has order \(m\).  Then for certain primitive \(m\)-th root of unity \(\zeta_m\) and integer \(d_i\), \(0<d_i<m\)
the \(w\) acts in vicinity of the point \(p\) with eigenvalue \(\zeta_m\) in tangent to \(C\) direction and \(\zeta_m^{d_i}\)
in normal direction. Using these notations the contribution of the point \(x_i\) is \(-d_i/m\). 

In general \(x_i\) could be not invariant under the action. By \(w_i\) denote a power of \(w\) which preserves \(x_i\).
As was explained above \(w_i\) is either \(w\) or \(w^2\). Denote by \(m_i\) the order of \(w_i\) on \(C\) in vicinity of the point \(x_i\).
The number \(m_i\) could be smaller than the order of \(w_i\), we denote \(\mathrm{ord} (w_i)=m_i m_i^\perp\).
Then for certain primitive \(m_i\)-th root of unity \(\zeta_{m_i}\) and integer \(d_i\), \(0<d_i<m_i\) the \(w\) acts in vicinity
of the point \(x_i\) with eigenvalue \(\zeta_{m_i}\) in tangent to \(C\) direction and \((\zeta_{m_i}^{d_i/m_i^\perp})\) in normal direction.
Using these notations the contribution of the point \(x_i\) is \(-d_i/m_i\). Putting all this together we get the formula for self-intersection. 
\begin{equation} \label{eq:self-intersection}
	\widetilde{D}\cdot \widetilde{D} = \frac{1}{m} C\cdot C -\sum_{x_i} \frac{d_i}{m_i}.
\end{equation}
We present these local contributions in the last column of Table \ref{Tab:singluar}.

\begin{Example}[Continuation of Examples \ref{Example:geom1}, \ref{Example:geom2}]\label{Example:geom3}
Let \(C\) be one of the coordinate lines \(F=0\), It has self-intersection \(0\) and goes through two singular points, one of the is \(A_3\)
singularity and another one correponds to \((-4)\) curve. Hence the self-intersection of  \(\widetilde{D}\) is \((-1)\). 

Let \(C\) be one of the curves $F G=1$ or $FG=a_0$. The corresponding divisors on \(X\) are \(H_{FG}-E_{1234}\) and \(H_{FG}-E_{5678}\),
so \(C\) has self-intersection \((-2)\). It goes through two singular points, each of them is  \(A_3\) singularity.  Hence the self-intersection of
\(\widetilde{D}\) is \((-1)\).   

In Table \ref{Tab:2A3 curves} we list \(w\)-invariant curves which are drawn in the figure. For each curve, we specify fixed points
through which it goes and the corresponding tangent direction. As we computed above the images of these curves are \((-1)\) curves,
hence they are drawn red in Fig. \ref{Fig:geom examp}.

\begin{table}[h]
	\begin{center}
	\begin{tabular}{|c|c|c|c|c|c|}
		\hline
		Curve & $(0,0)$ & $(0,\infty)$ & $(\infty,0)$ & $(\infty,\infty)$ & Image \\
		\hline
		$F=0$ & $dF=0$ & $dF=0$ & - & - & $\h_{\g}-\e_5$\\
		\hline
		$F=\infty$ & - & - & $dF^{-1}=0$ & $dF^{-1}=0$ & $\h_{\g}-\e_1$ \\
		\hline
		$G=0$ & $dG=0$ & - & $dG=0$ & - & $\e_8$\\
		\hline
		$G=\infty$ & - & $dG^{-1}=0$ & - & $dG^{-1}=0$ & $\e_4$\\
		\hline
		$F G=1$ & - & $dF=dG^{-1}$  & $dF^{-1}=dG$ & - &  $(\infty,0)$ \\
		\hline
		$FG=a_0$ & - & $dF=a_0 dG^{-1}$ & $dF^{-1}=a_0^{-1}dG$ & - &  $(0,\infty)$ \\
		\hline 
	\end{tabular}	
	\end{center}
	\caption{Invariant curves for the folding \(2A_3 \subset E_7\) \label{Tab:2A3 curves}}
\end{table}

Now we can identify the standard geometry on \(Y\). According to standard geometry components of the anticanonical divisor
should be \(\h_{\f\g}-\e_{1234}\) and \(\h_{\f\g}-\e_{5678}\). We pick image of \(G=0\) to be \(\e_4\), image \(G=\infty\) to be \(\e_8\).
Other exceptional divisors are reducible, this is because of many nodal curves in the \(Y\), in the Fig. \ref{Fig:geom examp} we add names
of these curves. We also add names of the components of an anticanonical divisor on \(\widetilde{Y}\). In the last columns of the
Tables \ref{Tab:2A3 points}, \ref{Tab:2A3 curves} we write image in terms of the standard geometry on \(Y\). 
\end{Example}

\begin{Remark}
	In order to describe standard geometry on \(Y\) we use certain simple \(w\)-invariant curves on \(X\). The most typical example
	is the following. Let \(C\) be a rational curve with self-intersection by  \(k\). Assume that \(w\) acts on \(C\) as
	\(z \mapsto \zeta_n z\), where \(z\) is some coordinate on \(C\). In this case the formula \eqref{eq:self-intersection}
	can be reproduced by the following elementary calculation. 
	
	Let us choose local coordinates \((z,v_0)\) near \(0 \in C\) and \((z^{-1},v_\infty)\) near \(\infty \in C\).
	We can assume that \(v_\infty =v_0 z^{-k}\) since \(C^2=k\). Assume now that \(w\) acts on these coordinates as 
	\begin{equation}
		w \colon\; z\mapsto \zeta_n z,\; v_0\mapsto \zeta_n^{d_0} v_0,\; v_\infty\mapsto \zeta_n^{-d_\infty} v_\infty.
	\end{equation}
	After the quotient \(C\) becomes rational curve with coordinate \(z^n\). And after the resolution local coordinates near the
	images of fixed points are \((z^n,v_0 z^{-d_0})\) and \((z^{-n},v_\infty z^{d_\infty})\). Since 
	\begin{equation}
		v_\infty z^{d_\infty} = v_0 z^{-d_0} (z^n)^{-\frac{k-d_0-d_\infty}{n}}
	\end{equation}
	the self intersection of \(\widetilde{D}\) is $\frac{k-d_0-d_\infty}{n}$. This is consistent with formula \eqref{eq:self-intersection}.	
	\begin{center}
		\begin{tikzpicture}[scale=2]
			
			\begin{scope}
				\draw [gray, line width=1.5pt]  (0,0) --
				(2,0);
				\draw[->,line width=1.5pt] (0,0) -- (0.5,0);
				\draw[->,line width=1.5pt] (2,0) -- (1.5,0);
				\draw (0.5,0) node[below] {$z$};
				\draw (1.5,0) node[below] {$z^{-1}$};
				\draw (1,0) node[above] {$k$};
				\draw[fill] (0,0) circle (0.02cm);
				\draw[fill] (2,0) circle (0.02cm);
				
				\draw[->] (\l,0) -- (\l,3*\l);
				\draw[->] (2-\l,0) -- (2-\l,3*\l);
				\draw (\l,3*\l) node[right] {$v_0$};
				\draw (2-\l,3*\l) node[left] {$v_{\infty}$};
				
			\end{scope}
			
			\draw[->] (2.5,\l) -- (3.2,\l);
			\draw (2.8,\l) node[above] {$/w$};
			
			\begin{scope}[xshift=4cm]
				\draw [gray, line width=1.5pt]  (-\l,0) --
				(2+\l,0);
				
				
				\draw[->] (\l,0) -- (\l,3*\l);
				\draw[->] (2-\l,0) -- (2-\l,3*\l);
				\draw ( \l,3*\l) node[above] {$v_0 z^{-d_0}$};
				\draw (2-\l,3*\l) node[above] {$v_{\infty}z^{d_\infty}$};
				
				\draw[->,line width=1.5pt] (0,0) -- (0.5,0);
				\draw[->,line width=1.5pt] (2,0) -- (1.5,0);
				\draw (0.5,0) node[below] {$z^n$};
				\draw (1.5,0) node[below] {$z^{-n}$};
			\end{scope}
		\end{tikzpicture}
	\end{center}
	
	Let us mention another typical case. Let \(C^\sharp\) be rational curve with self-intersection \(k\) and assume that \(w\) acts
	on \(C^\sharp\) trivially. Then there is no singular points on \(D'\), hence \(\widetilde{D}=D'\). But \((\psi')^*([D'])=n [C^\sharp]=[C]\).
	Hence \((\widetilde{D})^2=C^2/n=nk\).
\end{Remark}

\begin{Remark} \label{rem:concomp}
	One can take another irreducible component, namely instead of \eqref{eq:2A3 in E7 root variables} we can set 
	\begin{equation}
		a'_{1}=a'_2=a'_3=a'_5=a'_6=a'_7=-\ri.
	\end{equation}
	But looking to the blowup scheme in Fig. \ref{Fig:blowup E7} we see that surfaces \(\mathcal{X}_{\vec{a}}\) and \(\mathcal{X}_{\vec{a'}}\) are naturally isomorphic if \(a'_0=a_0\) and \(a'_4=a_4\). Moreover, this isomorphism intertwines orbits of the group \(\langle w \rangle\) generated by \(w\). So the geometry of the quotients \(\mathcal{X}_{\vec{a}}/\langle w \rangle\) and \(\mathcal{X}_{\vec{a'}}/ \langle w \rangle\) coincide. 
	
	This can be stated more algebraically. Namely there is an element \(n_{flip}=s_{1357}\) which belong to normalizer of \(\langle w \rangle\), namely \( n_{flip} w n_{flip}^{-1}=w^3\). The isomorphism between \(\mathcal{X}_{\vec{a}}\) and \(\mathcal{X}_{\vec{a'}}\) is given by \( w n_{flip} \in W^{ae}\). 
	
	And this is always the case, if we have two different irreducible folding components in \(\mathcal{A}^w\) then there exists element \(n_{flip}\)
	in the normalizer of \(H\) which permutes these components. See the last column in tables in Sec. \ref{sec:class_answ}. So the folding transformation does
	not depend on an irreducible component in \(\mathcal{A}^w\).

%
%
\end{Remark}

\paragraph{Coordinates and parameters.}

The next goal is to write down the map \(\psi \colon X \dashrightarrow Y\) in coordinates. In other words we want to write down formulas 
\begin{equation}
	\f=\f(F,G),\quad  \g=\g(F,G).
\end{equation}
This is necessary for the study of explicit relation between solutions of the $q$-Painlev\'e equations. And also this provides some double
check of the geometric answer obtained above. 

It is convenient to find first the formulas for the pencils of the curves \(\f=\text{const}\) and \(\g=\text{const}\). Curves from such
pencils should be \(w\)-invariant and from the geometry, we can compute their degrees and properties which determines the pencils.
The structure of the pencil determines formulas for \(\f,\g\) up to M\"obius transformation,
as automorphisms of $\mathbb{P}^1\times \mathbb{P}^1$ on the image.

We find multiplicative root variables parameters \(\A_i\) on \(Y\) using formulas for coordinates \(\f,\g\). They always are 
products of some powers of the initial multiplicative root variables \(a_i\).

\begin{Example}[Continuation of Examples \ref{Example:geom1}, \ref{Example:geom2}, \ref{Example:geom3}]\label{Example:geom4}

	Consider first pencils of curves \(F=\text{const}\) and \(G=\text{const}\) on \(X\). We claim that
	\begin{equation}\label{eq:psi* HF HG}
		\psi_*(H_F)=\h_\f+5\h_\g-\e_{12345678},\quad 	\psi_*(H_G)=\h_\f+\h_\g.
	\end{equation}
	Indeed, it follows from the geometry (see image columns in Tables \ref{Tab:2A3 points}, \ref{Tab:2A3 curves}) that 
	\begin{align}
		&\psi_*(H_F)\cdot \e_4=	\psi_*(H_F)\cdot \e_8=1,\quad \psi_*(H_F)\cdot (\h_\g-\e_1)=\psi_*(H_F)\cdot (\h_\g-\e_5)=0\\ 
		&\psi_*(H_F)\cdot (\e_1-\e_2)=	\psi_*(H_F)\cdot (\e_2-\e_3)=\psi_*(H_F)\cdot (\e_3-\e_4)=0\\
		&\psi_*(H_F)\cdot (\e_5-\e_6)=\psi_*(H_F)\cdot (\e_6-\e_7)=	\psi_*(H_F)\cdot (\e_7-\e_8)=0,\\
		& \psi_*(H_F)\cdot (\h_\f+\h_\g-\e_{1234})=\psi_*(H_F)\cdot (\h_\f+\h_\g-\e_{5678})=2.
	\end{align}
	From this intersections we get the formula \eqref{eq:psi* HF HG}  for \(\psi_*(H_F)\). As a double-check note that by push-pull
	formula the self-intersection of \(\psi'_*(H_F)\) is zero and after blowdown of \(\e^{ab}\) (see Fig. \ref{Fig:geom examp}) the
	self-intersection of \(\psi_*(H_F)\) will be 2, this is consistent with formula \eqref{eq:psi* HF HG}.
	
	Similarly we have
	\begin{align}
		&\psi_*(H_G)\cdot (\h_\g-\e_1)=	\psi_*(H_G)\cdot (\h_\g-\e_5)=1,\quad \psi_*(H_G)\cdot \e_4=	\psi_*(H_G)\cdot \e_8=0,\\
		&\psi_*(H_G)\cdot (\e_1-\e_2)=	\psi_*(H_G)\cdot (\e_2-\e_3)=\psi_*(H_G)\cdot (\e_3-\e_4)=0\\
		&\psi_*(H_G)\cdot (\e_5-\e_6)=
		\psi_*(H_G)\cdot (\e_6-\e_7)=	\psi_*(H_G)\cdot (\e_7-\e_8)=0,\\
		&\psi_*(H_G)\cdot (\h_\f+\h_\g-\e_{1234})=\psi_*(H_G)\cdot (\h_\f+\h_\g-\e_{5678})=2.
	\end{align}
	This leads to the formula \eqref{eq:psi* HF HG} for \(\psi_*(H_G)\). As a double-check, we can compute that self-intersection of
	\(\psi_*(H_G)\) is 2, this is consistent with  formula \eqref{eq:psi* HF HG}.
	
	Now return to the curves given by \(\f=\text{const}\) and \(\g=\text{const}\) in \(X\). From the formula \eqref{eq:psi* HF HG}
	we know intersections of \(\psi_*(H_F),\psi_*(H_G)\) with \(\h_\f\) and \(\h_\g\). Using this we conclude that  \(\f=\text{const}\)
	should be \((1,5)\) curve and \(\g=\text{const}\) should be \((1,1)\) curve. We claim that  
	\begin{equation}\label{eq:psi* hf hg}
		\psi^*(\h_\f)=H_F+5H_G-E_{12345678},\quad \psi^*(\h_\g)=H_F+H_G
	\end{equation}
	Indeed, on the surface \(Y\) curves \(\g=\text{const}\) intersect each component of \(-K_Y\) in one point.
	Therefore curves \(\g=\text{const}\) on \(X\) intersect hyperbolas \( FG=1\) and \(FG=a_0\) in two points \((0,\infty)\)
	and \((\infty,0)\), hence do not contain exceptional divisors, so we get formula \eqref{eq:psi* hf hg} for \(\psi^*(\h_\g)\).
	Similarly, the curves \(\f=\text{const}\) on \(X\) intersect hyperbolas \( FG=1\) and \(FG=a_0\) in two points \((0,\infty)\)
	and \((\infty,0)\), but since they are \((1,5)\) curves they should contain exceptional divisors. Due to \(w\)-invariance the
	only choice is given by the formula~\eqref{eq:psi* hf hg}. 
	
	The formulas \eqref{eq:psi* hf hg} can be also deduced (or checked) computing self-intersection. Namely, note that the curves
	\(\f=\text{const}\) and \(\g=\text{const}\) on \(X\) goes through two singular points with \(d_i=1\) and in the image have
	self-intersection 0. Therefore on \(X\) such curves have self-intersection 2 (using \eqref{eq:self-intersection}).
	This and \(w\)-invariance give formulas~\eqref{eq:psi* hf hg}.
	
	Now we can write formula for the pencils \(\f=\text{const}\) and \(\g=\text{const}\). The first one consist of \((1,1)\) curves
	which are \(w\)-invariant and go through \((0,\infty)\) and \((\infty,0)\). This determines the 1-parameter family 
	\begin{equation}
		c_1 FG - c_2=0, 
	\end{equation}
	where \(c_1,c_2\) are complex constants. The second pencil consist of \((1,5)\) curves which go through  \((0,\infty)\)
	and \((\infty,0)\) and intersect all exceptional divisors \(E_1,\dots,E_8\). This determines the 1-parameter family 
	\begin{equation}
		c_1 (FG-a_0)(G^4-1) - c_2 (FG-1)(G^4-a_4^{-4})=0,
	\end{equation}
	where \(c_1,c_2\) are complex constants (recall that root variables are subject of \eqref{eq:2A3 in E7 root variables}).

	The final formulas for \(\f\) and \(\g\) are ratios of two equations from the families above. They are determined by the
	properties that curves \(FG-1\), \(FG-a_0\), \(G=\infty\) go to \((\infty,0)\), \((0,\infty)\) and \((1,1)\) correspondingly.
	We get formulas \eqref{eq:2A3 in E7 fg}.

	It remain to find root variables on the image. Since on the image we have effective divisors \(\e_1-\e_2\), \(\e_2-\e_3\), \(\e_3-\e_4\)
	the blowup points \(p_1,p_2,p_3,p_4\) on the image are infinitely near points. Hence \(\A_1=\A_2=\A_3=1\). Similarly \(\A_5=\A_6=\A_7=1\).
	Since curve \(G=0\) goes to \(\e_8\) we see from formulas \eqref{eq:2A3 in E7 fg} that \(p_8=(a_0 a_4^4,a_0^{-1})\).
	Hence \(\A_0=a_4^4\) and \(\A_4=a_0\). We put these root variables in the diagram below.
			    
	\begin{tabular}{m{7.5cm}m{7.5cm}}
		\begin{center}
			 \begin{tikzpicture}[elt/.style={circle,draw=black!100,thick, inner sep=0pt,minimum size=2mm},scale=0.75]
				\path 	(-3,0) 	node 	(a1) [elt] {}
				(-2,0) 	node 	(a2) [elt] {}
				( -1,0) node  	(a3) [elt] {}
				( 0,0) 	node  	(a4) [elt] {}
				( 1,0) 	node 	(a5) [elt] {}
				( 2,0)	node 	(a6) [elt] {}
				( 3,0)	node 	(a7) [elt] {}
				( 0,1)	node 	(a0) [elt] {};
				\draw [black,line width=1pt ] (a1) -- (a2) -- (a3) -- (a4) -- (a5) --  (a6) -- (a7) (a4) -- (a0);
				\node at ($(a1.south) + (0,-0.3)$) 	{\small $1$};
				\node at ($(a2.south) + (0,-0.3)$)  {\small $1$};
				\node at ($(a3.south) + (0,-0.3)$)  {\small $1$};
				\node at ($(a4.south) + (0,-0.3)$)  {\small $a_0$};	
				\node at ($(a5.south) + (0,-0.3)$)  {\small $1$};	
				\node at ($(a6.south) + (0,-0.3)$) 	{\small $1$};
				\node at ($(a7.south) + (0,-0.3)$) 	{\small $1$};
				\node at ($(a0.west) + (-0.3,0)$) 	{\small $a_4^4$};
			\end{tikzpicture}
		\end{center}	
		&
		\begin{center}
			\begin{equation}\label{eq:2A3 in E7 fg} 
				\begin{aligned}
			             &\f=\frac{FG-a_0}{FG-1}\, \frac{G^4-1}{G^4-a_4^{-4}}, \\
			             &\g=\frac{FG-1}{FG-a_0}.
			    \end{aligned}
		    \end{equation}
		\end{center}
	\end{tabular}   
\end{Example}			    

Finally, after defining geometry on the image for all foldings transformations, we should
check, that foldings with the same order and the same symmetry/surface type both
in the image and in the preimage are not equivalent.
Such groups of foldings appear iff their finitizations $\bar{w}$ are conjugated (see Remark \ref{rem:diffaff_sameim}).
It appears that these groups have also the same nodal root subsystem $\Phi_{\nod}$.
But $\overline{\Phi}_{\nod}$ appear to be different for them, so we take $\overline{\Phi}_{\nod}$ as distinguishing invariant.

\subsection{Nodal curves and symmetries on image}
\label{ssec:nodal}

\paragraph{Symmetries on the image from nodal root systems.}  
Let us define \emph{special symmetry lattice} --- root sublattice in root lattice $\Phi^a$,
orthogonal to the nodal curves $\Delta_{\nod}$: 
$Q^{\symm}=(\alpha \in Q^a| (\alpha, \Delta_{\nod})=0)$.

Assume that $Q^{\symm}$ is generated by certain (affine) root subsystem $\Phi^{\symm}$.
Then we can calculate symmetries of our special subfamily $Y_{\vec{\A}}$ analogously to Prop. \ref{prop:centr_lattice}.
Indeed, our special configuration is defined by $\Delta_{\nod}$, so we should
take only those elements from $W^{ae}$ in the image $Y$, that preserve $\Delta_{\nod}$, as well as its
orthogonal special symmetry lattice $\Phi^{\symm}$. 
So we have
\begin{equation}\label{symm_lattice_eq}
\mathrm{Sym}(Y_{\vec{a}})=\Auti\ltimes W_{\Phi^{\symm}}=
(\mathrm{Aut}^{\symm} \ltimes W_{\Phi^{\symm}})\ltimes \Auti^{\nod},
\end{equation}
where $\Auti=\Auti^{\symm}\ltimes \Auti^{\nod}$
for automorphism subgroups $\Auti^{\nod}\subset \Auti(\Delta_{\nod})$
and $\Auti^{\symm}\subset \Auti(\Phi^{\perp, \nod}) \times \Auti(\Delta_{\nod})$,
the latter acts effectively on $\Phi^{\symm}$. Note that automorphisms $\Auti(\Delta_{\nod})$ is only permutations.
As in the algebraic case, this is subgroups, that are realized in symmetry group $W^{ae}$ of subfamily $Y_{\vec{\A}}$. 

\begin{Remark}
There are some peculiarities about sublattice $Q^{\symm}$.
Namely, for $E_3^{(1)}$ and $E_1^{(1)}$ cases there sublattices could be realized only
by long root system $\Phi^{\perp, \nod}|_{|\alpha|^2=2n}$, where $n$ is the order of folding.
Then we should find reflections by hands and they could act on $\Delta_{\nod}$.
Even more problems we occur for
folding transformations $\pi^3$ and $\pi^2$ for $E_3^{(1)}$.
In this case elementary long root system simple reflections cannot be realized in $W^{ae}$,
that's why it is convenient to take instead of $Q^{\symm}$ its sublattice
with imaginary root $2\delta$ and $3\delta$ respectively.
Cf. Remark \ref{rem:subsublattice} from the preimage side.
\end{Remark}

\begin{Remark}
In the special configuration $\vec{\A}$ in the image we can also have projective reduction
and it could be found in the same way, as in the preimage. I.e., $\mathsf{q}$ should be
some full power of root variables $\A$ and projective reduction should be translation
in the special symmetry root subsystem $\Phi^{\symm}$ and non-trivial permutation of $\Delta^{\nod}$.
\end{Remark}

\begin{Example}[Continuation of Examples \ref{Example:geom1} - \ref{Example:geom4}]
\label{Example:geom5}

On the Fig. \ref{Fig:geom examp} we see 6 nodal curves, which form two $A_3$ root systems, consisting from
$\alpha_3=\e_1-\e_2, \, \alpha_2=\e_2-\e_3, \, \alpha_1=\e_3-\e_4$ and $\alpha_5=\e_5-\e_6, \, \alpha_6=\e_6-\e_7, \, \alpha_7=\e_7-\e_8$.
Root variables $\A_0$ and $\A_4$ are free, so, according to Prop. \ref{prop:nodal} there are no other irreducible nodal curves.
Extended nodal root subsystem equals to the nodal root subsystem \(\overline{\Phi}_{\nod}=\Phi_{\nod}=2A_3\).

It is easy to see that $\Phi^{\symm}=A_1^{(1)}$, generated by $\alpha_0$ and $\delta$.
Special symmetry and nodal root subsystems as well as its automorphisms are represented by left and right pictures
on Fig. \ref{fig:sublattices_im}. Calculation gives, that $\Auti^{\nod}$, $\Auti^{\symm}$ are commuting $C_2$, generated by $\pi$ and 
$\pi s \pi s^{-1}$ for $s=s_{4354} s_{2132}$ respectively\footnote{On the Fig. \ref{fig:sublattices_im} we
pictured action on nodal root subsystem of $s \pi s^{-1}$ instead of $s$ for simplicity.}. 
So symmetry on the image will be $C_2 \times W_{A_1}^{ae}$, this is written at the right side of Fig. \ref{fig:sublattices_im}.
	   
	   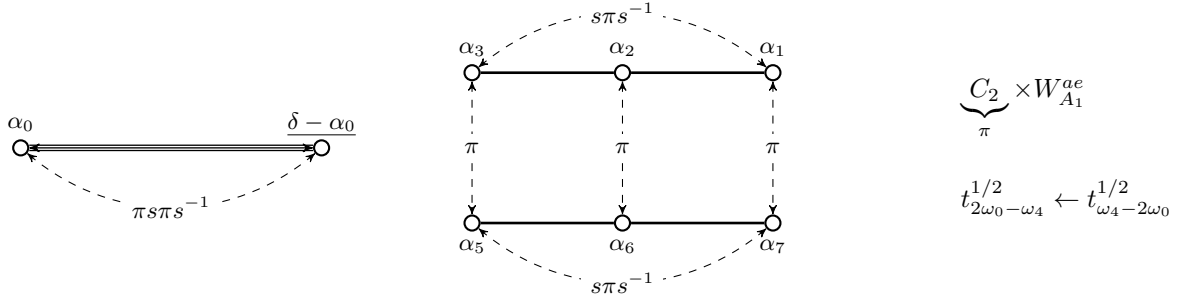
\begin{figure}[h]
	   \centering
	    \begin{tikzpicture}[elt/.style={circle,draw=black!100,thick, inner sep=0pt,minimum size=2mm},scale=2]
			\path 	(-1,0) 	node 	(a1) [elt] {}
			(1,0) 	node 	(a2) [elt] {};

		    \draw [black,line width=2.5pt ] (a1) -- (a2);
		    
		     \draw [white,line width=1.5pt ] (a1) -- (a2);
		     
		     \draw [<->,black, line width=0.5pt]
		     (a1) -- (a2);
		     
		     \node at ($(a2.north) + (0,0.1)$) 	{\small$\underline{\delta-\alpha_0}$};	
		     	
		     \node at ($(a1.north) + (0,0.1)$) 	{\small $\alpha_0$};
		     
		     \draw [<->, dashed]
		     (a1) to[bend right=40] node[fill=white]{\small $\pi s \pi s^{-1}$}  (a2);

		 \begin{scope}[xshift=3cm,scale=1]
			\path 	(1,0.5) 	node 	(a1) [elt] {}
			(0,0.5) 	node 	(a2) [elt] {}
			(-1,0.5) node  	(a3) [elt] {}
			
			( 1,-0.5) 	node 	(a7) [elt] {}
			( 0,-0.5)	node 	(a6) [elt] {}
			( -1,-0.5)	node 	(a5) [elt] {};
			
			\draw [black,line width=1pt ] (a1) -- (a2) -- (a3) (a5) --  (a6) -- (a7);
		\node at ($(a1.north) + (0,0.1)$)  {\small $\alpha_1$};
		\node at ($(a2.north) + (0,0.1)$)  {\small $\alpha_2$};
		\node at ($(a3.north) + (0,0.1)$)  {\small $\alpha_3$};
		
		\node at ($(a5.south) + (0,-0.1)$) {\small $\alpha_5$};	
		\node at ($(a6.south) + (0,-0.1)$) {\small $\alpha_6$};
		\node at ($(a7.south) + (0,-0.1)$) {\small $\alpha_7$};
		
		\draw [<->, dashed]
		     (a1) to[bend right=40] node[fill=white]{\small $s\pi s^{-1}$}  (a3);
		 \draw [<->, dashed]
		     (a5) to[bend right=40] node[fill=white]{\small $s\pi s^{-1}$}  (a7);
		 
		 \draw[<->,dashed] (a1) edge[bend right=0] node[fill=white]{\small $\pi$} (a7);
			\draw[<->,dashed] (a2) edge[bend right=0] node[fill=white]{\small $\pi$} (a6);
			\draw[<->,dashed] (a3) edge[bend right=0] node[fill=white]{\small $\pi$} (a5);
		
		\end{scope}
		
		 \node at (6,0) {\parbox{3cm}
		{$\underbrace{C_2}_{\pi} \times W_{A_1}^{ae}$\\[0.5cm]
                 $t^{1/2}_{2\omega_0-\omega_4} \leftarrow t^{1/2}_{\omega_4-2\omega_0}$ 
                 }};

			\end{tikzpicture}
			
\caption{Special symmetry and nodal root subsystems for the image of $2A_3\subset E_7^{(1)}$}
		
		\label{fig:sublattices_im}
		
	    \end{figure}

	We have $\mathrm{q}=\a_0^2\a_4^4=(\a_0 \a_4^2)^2$ and the elementary translation in special symmetry lattice is given by
	\(\alpha_0=2\omega_0-\omega_4\) so one can expect existence of the translation by fundamental weight.
	We have projective reduction $t^{1/2}_{2\omega_0-\omega_4}$:
	\begin{equation}
	t^{1/2}_{2\omega_0-\omega_4}=s^{-1} \pi s \pi s_0,
	\end{equation}
	which non-trivially permute nodal curves.
	We see that it corresponds to the projective reduction $t^{1/2}_{\omega_4-2\omega_0}$ in the preimage due
	to $\A_0=a_4^4, \A_4=a_0$ and $\mathsf{q}=q^2$. We will talk systematically about correspondence in the image 
	and preimage in the coming paragraph.
	Such correspondence is also written on the right side of Fig. \ref{fig:sublattices_im}.
	
	Note that configuration of root variables in the preimage and in the image are similar:
	$a_{1,2,3,5,6,7}=\ri$ vs. $\A_{1,2,3,5,6,7}=1$, that is why Fig. \ref{fig:sublattices} and \ref{fig:sublattices_im} coincide
	(cf. Remark \ref{rem:pairs}).
	\end{Example}

	\paragraph{Relation to centralizer in the preimage.}
	
	Finally, let us connect symmetries $\mathrm{Sym}(Y_{\vec{\A}})$ in the image and centralizer $C(w,W^{ae})$ in the preimage
	for folding transformation $w$ of order $n$. We have map $\psi_*=C(w,W^{ae}) \mapsto \mathrm{Sym}(Y_{\vec{\A}})$,
	which is not injective, because at least folding sungroup $\langle w \rangle$ becomes trivial.  
	It appears, that in several cases it is even not surjective, see below.
	Connection between centralizer and symmetry in the image is done by such steps:
	\begin{enumerate}
	 \item Let $\alpha\in \Phi^{\symm}$ be a root. Then we have $|\psi_*^{-1}(\alpha)|^2=n|\alpha|^2$,
	due to push-pull formula. So we take in the preimage ``long'' root subsystem
	$\Phi^{a,w}_{long}\subset Q^{a,w}$ with $|\alpha|^2=2n$ (or $|\alpha|^2=2ln$, where $2l$ is a length of special symmetry root subsystem
	in the image) instead of $\Phi^{a,w}$. Type of this long root subsystem should
	coincide with special symmetry root subsystem in the image. We choose simple roots of this long root subsystem in a way
	that root variables of $\Phi^{a,w}_{long}$ equal to the corresponding root variables of $\Phi^{\symm}$ in the image
	\footnote{For a few cases, up to exponentiation to a common power.}.
	 
	 \item Then we isomorphically rewrite $C(w,W)$ in the form of \eqref{centr_lattice_eq}, but for long $\Phi^{a,w}_{long}$.
	 A common situation is that reflections in short invariant root subsystem become outer automorphisms of long invariant root subsystem
	 and vice versa. Groups of automorphisms $\Aut^{\parallel}$ and $\Auti^{\symm}$ also should become the same.
	 
	 \item Check, what elements of $(\Aut^{\perp}\ltimes W_{\Phi^{a,\perp w}})^w$
	 act non-trivially on the image (their image belongs to $\Auti^{\nod}$). That could be done using explicit coordinates
	 $(\f(F,G),\g(F,G))$.
	 
	 \item If there are additional symmetries $\mathrm{Sym}(Y_{\vec{\A}})$, which do not come from $C(w,W)$, then
	 they act trivially on initial coordinates and root variables. From the point of view of coordinates, this means that
	 formulas for $\f(F,G), \g(F,G)$ contain fractional powers of root variables (or(and) roots of unity) and such symmetries
	 occur from their branching (or Galois group). Such symmetries belong to $\Auti^{\nod}$, so they permute nodal curves
	 and corresponding fixed points in the preimage.
	 
	 \item Finally, we can check the correspondence $\psi_*$ explicitly.
	 Namely, for $w\in \Aut^{\parallel}\ltimes W^a_{\Phi^{a,w}}$
	 and $\mathsf{w}=\psi_*(w)\in\Auti^{\symm} \ltimes W_{\Phi^{\perp, \nod}}$
	 we should have
	 \begin{equation}
	 \mathsf{w}(\vec{\A})(\vec{a})=\vec{\A}(w(\vec{a})), \qquad  \mathsf{w}(\f,\g)(\vec{a},F,G)=
	 (\f,\g)(w(\vec{a}),w(F),w(G)).\label{docking}
	 \end{equation}
	 Automorphism group $\Auti^{\symm}$ is defined up to $\Auti^{\nod}$,
	 the latter formula fix this freedom.
	 We also should relate nontrivial $(\Aut^{\perp}\ltimes W_{\Phi^{a,\perp w}})^w$
         to the subgroup of $\Auti^{\nod}$ in similar manner.
         
         In fact this is non-trivial double-check both for calculation of $C(w,W)$ and $\mathrm{Sym}(Y_{\vec{\a}})$.
	\end{enumerate}

\begin{Example}[Continuation of Examples \ref{Example:geom1} - \ref{Example:geom5}]
 
 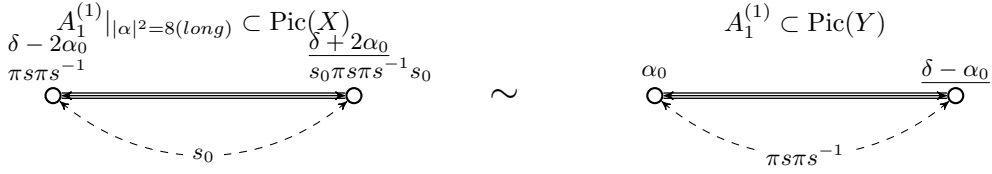
\begin{figure}[h]

\begin{tikzpicture}[elt/.style={circle,draw=black!100,thick, inner sep=0pt,minimum size=2mm},scale=2]
  
  \node at (0,0.5) 	{$A_1^{(1)}|_{|\alpha|^2=8 (long)} \subset \mathrm{Pic}(X)$};
			\path 	(-1,0) 	node 	(a1) [elt] {}
			(1,0) 	node 	(a2) [elt] {};

		    \draw [black,line width=2.5pt ] (a1) -- (a2);
		    
		     \draw [white,line width=1.5pt ] (a1) -- (a2);
		     
		     \draw [<->,black, line width=0.5pt]
		     (a1) -- (a2);
		     
		     \node at ($(a2.north) + (0.2,0.2)$) {\small \parbox{2cm}{$\underline{\delta+2\alpha_{0}}$\\$s_0 \pi s \pi s^{-1}s_0$}};	
		     	
		     \node at ($(a1.north) + (0.2,0.2)$) {\small \parbox{2cm}{$\delta-2\alpha_0$\\$\pi s \pi s^{-1}$}};
		     
		     \draw [<->, dashed]
		     (a1) to[bend right=40] node[fill=white]{\small $s_0$}  (a2);
		     
		     \node at (2,0) {\Large $\sim$};
		     
		     \begin{scope}[xshift=4cm]
		     
		      \node at (0,0.5) 	{$A_1^{(1)}\subset \mathrm{Pic}(Y)$};
		     
		     \path 	(-1,0) 	node 	(a1) [elt] {}
			(1,0) 	node 	(a2) [elt] {};

		    \draw [black,line width=2.5pt ] (a1) -- (a2);
		    
		     \draw [white,line width=1.5pt ] (a1) -- (a2);
		     
		     \draw [<->,black, line width=0.5pt]
		     (a1) -- (a2);
		     
		     \node at ($(a2.north) + (0,0.1)$) 	{\small$\underline{\delta-\alpha_0}$};	
		     	
		     \node at ($(a1.north) + (0,0.1)$) 	{\small $\alpha_0$};
		     
		     \draw [<->, dashed]
		     (a1) to[bend right=40] node[fill=white]{\small $\pi s \pi s^{-1}$}  (a2); 
		     \end{scope}

		     \end{tikzpicture} 

\caption{Relation between symmetries in image and preimage}

\label{fig:docking}
		     
\end{figure}
 
Here we compare symmetry from Fig. \ref{fig:sublattices_im} with centralizer \eqref{eq:example:C}:
 \begin{enumerate}
 \item We should take root subsystem in $Q^{a,w}$ with roots $|\alpha|^2=8$, so  $Q^{a,w}_{long} \subset 2\alpha_0, \delta$.
 To match root variables $\A_0=a_4^4$ and $\chi(\delta-\alpha_0)=a_{04}^4$ in the image,
 we should take simple roots $\delta-2\alpha_0$, $\delta+2\alpha_0$, see Fig. \ref{fig:docking}.
 Note that imaginary root become $2\delta$ instead of $\delta$.
 
 \item $W^{ae}_{\Phi^{a,w}}\simeq W^{ae}_{\Phi^{a,w}_{long}}$, $s_0$ becomes outer automorphism,
 $\pi s \pi s^{-1}$ becomes reflection.
 \item $\Omega_{A_3}^2$ act trivially on the image, $\pi s_1s_3s_5s_7$ --- non-trivially.
 
 \item We have no such additional symmetries.
 
 \item Changing root variables symmetry generators  in image and preimage correspond
 to each other by superimposition of left and right pictures of Fig. \ref{fig:docking}:
 \begin{equation}
 \psi_*(\pi s \pi s^{-1})=s_0, \quad \psi_*(s_0)=\pi s \pi s^{-1}.
 \end{equation}
 Also we have $\psi_*(\pi s_1s_3s_5s_7)=\pi$.
\end{enumerate}
 
 Generator $\pi$ is folding in the image and provide further extension of
 folding group $C_4$ to $C_2\ltimes C_4=\Dih_4^{(2)}$ 
 
\end{Example}

	\newpage

	\subsection{Answers} \label{ssec:geom answers}
	
	\begin{tabular}{|M{1.5cm}|M{2cm}|M{1.5cm}|M{1cm}|M{1.5cm}|M{1.5cm}|M{3cm}|M{1.5cm}|}
	    \hline 
	    Sym./surf. & Diagram & Name & Order & Goes to & $\overline{\Phi}_{\nod}$ & Symmetry & Sect.\\ 
        \hline	
          & 
\includegraphics[scale=1]{e8_222_s.pdf}
                &   
                $3A_2$
		& 3 &  $E_6^{(1)}/A_2^{(1)}$ & \(2A_2\) & $C_2 \ltimes W_{A_2}^a$ & \ref{e8_222}
		\\
		\cline{2-8}
		$E_8^{(1)}/A_0^{(1)}$ & 
                 \includegraphics[scale=1]{e8_1111_s.pdf}
  &  $4A_1$
        & 2 &  $E_7^{(1)}/A_1^{(1)}$ & \(3A_1\) & $S_3\ltimes W_{D_4}^a$ & \ref{e8_1111}
		\\
		\cline{2-8}
		
		& 
\includegraphics[scale=1]{e8_133_s.pdf}
        & 
        $A_1+ 2A_3$
        & 4 &  $D_5^{(1)}/A_3^{(1)}$ & \(A_3+A_1\) & $W_{A_1}^a$  & \ref{e8_133}
		\\
 		\hline
		
		 &
%
          \includegraphics[scale=1]{e7_p_s.pdf}
		&
		$\pi$
		&  &  & \(D_4\) & $ S_3 \ltimes (C_2^2 \times W_{D_4}^a)$ & \ref{e7_p} \\
		\cline{2-3} \cline{6-8}
		
		& 
			\includegraphics[scale=1]{e7_111_s.pdf}
			&
			$3A_1$
			& \vspace{-8mm} 2  & \vspace{-8mm} $E_8^{(1)}/A_0^{(1)}$ & \(4A_1\) & $S_4\ltimes W_{D_4}^a$ & \ref{e7_111} \\
		
		\cline{2-8}

$E_7^{(1)}/A_1^{(1)}$	&
\includegraphics[scale=1]{e7_1111_s.pdf}
                &
                $4A_1$
		& 2 &   $D_5^{(1)}/A_3^{(1)}$ & \(2A_1\) & $C_2^2 \ltimes W_{A_3}^a$ & \ref{e7_1111} \\
		\cline{2-8}
		& 
\includegraphics[scale=1]{e7_222_s.pdf}
                &
                $3A_2$ 
		& 3   & $E_3^{(1)}/A_5^{(1)}$ & \(A_2\) & $W_{A_1}^{ae}$ &\ref{e7_222}\\
		
		\cline{2-8}

		&
\includegraphics[scale=1]{e7_33_s.pdf}
                &
                $2A_3$
		& & & \(2A_3\) &  $ C_2 \times W_{A_1}^{ae}$ & \ref{e7_33}\\
		
		\cline{2-3} \cline{6-8}
		
		 &
%
\includegraphics[scale=1]{e7_p11111_s.pdf}
		&
		$\pi \ltimes 5A_1$
		& \vspace{-6mm} 4   &  \vspace{-6mm} $E_7^{(1)}/A_1^{(1)}$   & \(D_6\)& $C_2^2\times W_{A_1}^a$  & \ref{e7_p11111} \\

		\hline
			
	      &
\includegraphics[scale=1]{e6_p_s.pdf}
                &
                $\pi$
		& &   & \(E_6\) & $C_2\ltimes (S_3\times W_{A_2}^a)$ & \ref{e6_p}  
		\\
		\cline{2-3}  \cline{6-8}
		 
		$E_6^{(1)}/A_2^{(1)}$ &
%
\includegraphics[scale=1]{e6_22_s.pdf}
                &
                $2A_2$
		& \vspace{-1cm} 3  &  \vspace{-1cm}  $E_8^{(1)}/A_0^{(1)}$ & \(3A_2\)& $S_3\ltimes W_{A_2}^a $ & \ref{e6_22}
		\\
		
		\cline{2-8}
		&
\includegraphics[scale=1]{e6_1111_s.pdf}
                & $4A_1$
		& 2 &  $E_3^{(1)}/A_5^{(1)}$ & \(A_1\)& $W_{A_2}^{ae}$  & \ref{e6_1111}
		\\
		\hline
		
	 &
%
\includegraphics[scale=1]{d5_pp_s.pdf}
        & $\pi^2$ & &  &  \(D_4\)  & $\!\!S_3 {\ltimes} (C_2^2 {\times} (C_2 {\ltimes} W_{3A_1}^a))$  &\ref{d5_pp}\\
	 
	\cline{2-3}  \cline{6-8}
	& 
\includegraphics[scale=1]{d5_11_s.pdf}
       &
	$2A_1$		
	& \vspace{-8mm} 2 & \vspace{-8mm} $E_7^{(1)}/A_1^{(1)}$ &  \(4A_1\) & $(S_3\ltimes C_2^3) \ltimes W_{3A_1}^a$ & \ref{d5_11}	\\	
	\cline{2-7}
	
      $D_5^{(1)}/A_3^{(1)}$	&
               \includegraphics[scale=1]{d5_p_s.pdf}
               & $\pi$ 
		&  &    & \(E_7\)& $C_2^2 \times W_{A_1}^{a}$ & \ref{d5_p}\\
		\cline{2-3}  \cline{6-8}
		
		&
\includegraphics[scale=1]{d5_13_s.pdf}
                        & $A_1+A_3$
			& 4  & $E_8^{(1)}/A_0^{(1)}$ & \(2A_3+A_1\) & $W_{A_1}^{ae}$ & \ref{d5_13} \\
		
		\cline{2-3} \cline{6-8}
		
		&	
%
\includegraphics[scale=1]{d5_pp111_s.pdf}
                & $\pi^2\ltimes 3A_1$ 
		& &  &   \(D_6+A_1\) &  $ C_2^2\times W_{A_1}^a$ & \ref{d5_pp111}\\

			\cline{2-8}
		& 
\includegraphics[scale=1]{d5_1111_s.pdf}
                        & $4A_1$
			& 2 &  $A_1^{(1)}/A_7^{(1)'}$ & \(0\)& $C_2 \times (C_4\ltimes W_{A_1}^{a})$  & \ref{d5_1111}\\	
			\hline
	
	 &
%
%
%
%
%
\includegraphics[scale=1]{e3_ppp_s.pdf}
       &
     $\pi^3$
	& &   & \(D_4\)
		 & $C_2^2 \times W^{ae}_{A_2}\, (2\delta)$ & \ref{e3_ppp}\\
	
	\cline{2-3} \cline{6-8}
	$E_3^{(1)}/A_5^{(1)}$ & 
%
%
%
%
\includegraphics[scale=1]{e3_1_s.pdf}
            & $A_1$ 
	    & \vspace{-1cm} 2 & \vspace{-1cm} $E_6^{(1)}/A_2^{(1)}$  & \(4A_1\)& $W^{ae}_{A_2}$ & \ref{e3_1}\\
	
	\cline{2-8}

	 &
%
%
%
    \includegraphics[scale=1]{e3_pp_s.pdf}
    & $\pi^2$
	&  &    & \(E_6\)& $S_3\times W_{A_1}^{ae}, \, (3\delta)$ & \ref{e3_pp} \\
	\cline{2-3} \cline{6-8}
	
	& 
%
%
%
%
%
\includegraphics[scale=1]{e3_2_s.pdf}
& $A_2$
	    & \vspace{-1cm} 3 & \vspace{-1cm} $E_7^{(1)}/A_1^{(1)}$ &   \(3A_2\)& $W_{A_1}^{ae}$ & \ref{e3_2} \\

	\hline
	$A_1^{(1)}/A_7^{(1)'}$ &
%
%
%
%
%
%
\includegraphics[scale=1]{e1_pp_s.pdf}
         & $\pi^2$
	& &   & \(D_4\)& $C_2^2 \times W_{A_1}^a$ & \ref{e1_pp}\\
	\cline{2-3} \cline{6-8}
	& 
%
%
%
%
%
%
      \includegraphics[scale=1]{e1_c_s.pdf}
          & $\sigma$
	    & \vspace{-6mm} 2 & \vspace{-6mm} $D_5^{(1)}/A_3^{(1)}$   & \(4 A_1\) & $C_4 \ltimes W_{A_1}^a$ & \ref{e1_c}\\
	    \hline
	\end{tabular}

%
%
%
%
%
%
%
%

\bigskip
	
	Detailed calculations are given in Section \ref{sec:geom_answ}.
	Using calculations for cyclic foldings, we can calculate data for non-cyclic groups
	
	\medskip
	
	\begin{tabular}{|M{1.5cm}|M{2.5cm}|M{1.5cm}|M{3cm}|M{3cm}|M{2.5cm}|}
	    \hline 
	    Sym./surf. & Diagram & Name & Group & Goes to & $\overline{\Phi}_{\nod}$ \\ 
        \hline	
         &
%
 \includegraphics[scale=0.75]{e7_p11111.pdf}
			& $\pi \ltimes 5A_1$ & $\Dih_4^{(1)}$  & $E_8^{(1)}/A_0^{(1)}$ & $D_5+2A_1$
			\\
        \cline{2-6}
       $E_7^{(1)}/A_1^{(1)}$ & \includegraphics[scale=0.75]{e7_p33.pdf} &
       $\pi\ltimes 2A_3$ & $\Dih_4^{(2)}$ & $E_8^{(1)}/A_0^{(1)}$ & $D_4+A_3$\\
        \cline{2-6}
        &
        \includegraphics[scale=0.75]{e7_p1111.pdf}
       & $\pi \ltimes 4A_1$ & $C_2^2\subset \Dih_4^{(1,2)}  $ & $E_7^{(1)}/A_1^{(1)}$ & $D_4+A_1$ \\
        \cline{2-6}
         &\includegraphics[scale=0.75]{e7_11111.pdf}
        & $5A_1$ & $C_2^2\subset \Dih_4^{(1)}$ & $E_7^{(1)}/A_1^{(1)}$ & $A_3+2A_1$ \\
        
        \hline
         &
         \includegraphics[scale=0.75]{d5_pp1111.pdf}
         & $\pi^2 \ltimes 4A_1$ & $C_2^3$ &  $E_7^{(1)}/A_1^{(1)}$ & $D_4+2A_1$\\
         \cline{2-6}
         
        $D_5^{(1)}/A_3^{(1)}$  & \includegraphics[scale=0.75]{d5_1111.pdf}
	& $4A_1$ & $C_2^2 \subset C_2^3$ &  $D_5^{(1)}/A_3^{(1)}$ & $4A_1$\\
	\cline{2-6}
	
	&\includegraphics[scale=0.75]{d5_pp1111.pdf}
	
	& $\pi^2\ltimes 4A_1$ & $C_2^2 \subset C_2^3$  & $D_5^{(1)}/A_3^{(1)}$ & $D_4$\\
	\cline{2-6}
	
		&\includegraphics[scale=0.75]{d5_pp11.pdf}
	& $\pi^2\ltimes 2A_1$ & $C_2^2 \subset C_2^3$ & $E_8^{(1)}/A_0^{(1)}$ & $D_4+2A_1$\\
	
         \hline
          $E_1^{(1)}/A_7^{(1)}$ & \includegraphics[scale=1]{e1_ppc.pdf} & $\pi^2 \times \sigma$  & $C_2^2$ &
	    $E_7^{(1)}$ & $D_4+2A_1$\\ 
          \hline
        \end{tabular}

%
%

\medskip

\begin{Remark}\label{rem:diffaff_sameim}
It is remarkable, that folding transformations with conjugated finitizations $\bar{w}$ always
go to the same symmetry/surface type. We spotted that from answers above case by case, however, we have no
explanation of this phenomenon.
\end{Remark}

\begin{Remark} \label{rem:pairs}
	We have a lots of pairs 
	foldings transformations $w_1, w_2$ going in ``opposite direction'':
	\begin{equation}
		\begin{aligned}
			w_1:E_{m_1}^{(1)}\xrightarrow{n} E_{m_2}^{(1)}: \mathcal{A}_{w_1} (\zeta_n) \mapsto  \mathcal{A}_{w_2} (1)
			\\
			w_2:E_{m_2}^{(1)}\xrightarrow{n} E_{m_1}^{(1)}: \mathcal{A}_{w_2} (\zeta_n) \mapsto  \mathcal{A}_{w_1} (1).
		\end{aligned}
	\end{equation}
	Here $\mathcal{A}_{w_{1,2}}(\zeta_n)$ denotes irreducible component of the folding invariant subset. In particular $\mathcal{A}_{w_{1,2}}(1)$ is component on which corresponding folding acts trivially. 
	
	In particular in folding transformation in Example \ref{Example:geom5} is opposite to itself.
\end{Remark}

	\newpage

	\section{Folding transformation: lattices, normalizers and components}
	
	\label{sec:class_answ}
	
	\subsection{$E_8^{(1)}/A_0^{(1)}$}
	
	\label{e8_a}

	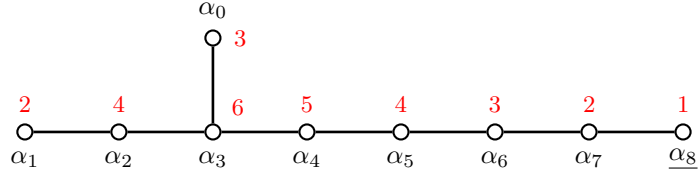
\begin{figure}[h]
	
	\begin{center}
	
 \begin{tikzpicture}[elt/.style={circle,draw=black!100,thick, inner sep=0pt,minimum size=2mm},scale=1.25]
			\path 	(-2,0) 	node 	(a1) [elt] {}
				(-1,0) 	node 	(a2) [elt] {}
				( 0,0) node  	(a3) [elt] {}
				( 1,0) 	node  	(a4) [elt] {}
				( 2,0) 	node 	(a5) [elt] {}
				( 3,0)	node 	(a6) [elt] {}
				( 4,0)	node 	(a7) [elt] {}
				( 5,0)	node 	(a8) [elt] {}
				( 0,1)	node 	(a0) [elt] {};
			\draw [black,line width=1pt ] (a1) -- (a2) -- (a3) -- (a4) -- (a5) --  (a6) -- (a7) --(a8) (a3) -- (a0);
			\node at ($(a1.south) + (0,-0.2)$) 	{$\alpha_{1}$};
			\node at ($(a2.south) + (0,-0.2)$)  {$\alpha_{2}$};
			\node at ($(a3.south) + (0,-0.2)$)  {$\alpha_{3}$};
			\node at ($(a4.south) + (0,-0.2)$)  {$\alpha_{4}$};	
			\node at ($(a5.south) + (0,-0.2)$)  {$\alpha_{5}$};		
			\node at ($(a6.south) + (0,-0.2)$) 	{$\alpha_{6}$};	
			\node at ($(a7.south) + (0,-0.2)$) 	{$\alpha_{7}$};	
			\node at ($(a8.south) + (0,-0.2)$) 	{$\underline{\alpha_{8}}$};	
			\node at ($(a0.north) + (0,0.2)$) 	{$\alpha_{0}$};		
			\node at ($(a0.east) + (0.2,0)$) 	{\color{red}\small$3$};
			\node at ($(a1.north) + (0,0.2)$) 	{\color{red}\small$2$};
			\node at ($(a2.north) + (0,0.2)$) 	{\color{red}\small$4$};
			\node at ($(a3.north  east) + (0.2,0.2)$) 	{\color{red}\small$6$};
			\node at ($(a4.north) + (0,0.2)$) {\color{red}\small$5$};
			\node at ($(a5.north) + (0,0.2)$) 	{\color{red}\small$4$};
			\node at ($(a6.north) + (0,0.2)$) 	{\color{red}\small$3$};
			\node at ($(a7.north) + (0,0.2)$) 	{\color{red}\small$2$};
			\node at ($(a8.north) + (0,0.2)$) 	{\color{red}\small$1$};
		\end{tikzpicture}
		
		\end{center}
		
		\caption{Root numbering and marks for $E_8^{(1)}$ Dynkin diagram}
		
		\label{fig:E8_rn}
		
		\end{figure}

		For $E_8^{(1)}$ we have only case $2$ foldings: of order $3$, of order $4$, and (conjugated to) its square of order $2$.
		We start from folding of order $3$ on which we explain our notations.

	\begin{center}
			\begin{tabular}{|m{3cm}|m{4.25cm}|m{2cm}|m{3cm}|m{3.25cm}|}
			\hline
			  \begin{center}
			  Folding  \end{center}
			      & \begin{center}
			     $Q^{a,w}$  
			   \end{center}  & \begin{center} $Q^{a, \perp w}$ \end{center}	                                                 
			     & \begin{center}  $C(w)$ \end{center} & \begin{center} $\mathcal{A}^w, \, N_{flip}$ \end{center}
			                                                  \\
			     \hline
			     
			     \begin{center}

%
                        \includegraphics[scale=0.8]{e8_222.pdf}

			\medskip
			
			$w=s_{215487}$

			\end{center}
			     
			    &
			    
			     \begin{center}

			      \begin{tikzpicture}[elt/.style={circle,draw=black!100,thick, inner sep=0pt,minimum size=1.4mm},scale=1]
			      \path 	
			(-1,0) 	node 	(a1) [elt] {}
			(1,0) 	node 	(a2) [elt] {}
			(0,1.73) 	node 	(a0) [elt] {};
			\draw [black,line width=1pt] (a0) -- (a1) -- (a2) -- (a0);
			
			 \node at ($(a0.north) + (0,0.2)$) 	{\small $\underline{\alpha_{01 2^2 3^3 4^3 5^3 6^3 7^2 8}}$};	
		     	
		     \node at ($(a1.west) + (-0.2,0)$) 	{\small $\alpha_0$};
		     
		      \node at ($(a2.east) + (0.7,-0.1)$) 	{\small $\alpha_{012^2 3^3 4^2 5}$};
			
			\draw[<->,dashed] (a1) to[bend right=40]
			node[fill=white]{\tiny $s_{324}s_{135}s_{423}$} (a2);
			
			\draw[<->,dashed] (a2) to[bend right=40]
			node[fill=white]{\tiny $s_{675} s_{468} s_{576}$} (a0);
			
			\node at (0,3) {$A_2^{(1)}$};
			
			\node at (0,-1.5) {$(A_2^{(1)})_{|\alpha|^2=6}$};

			     \begin{scope}[yshift=-4.5cm]
			      \path 	
			(0,1.73) node 	(a1) [elt] {}
			(1,0) node 	(a2) [elt] {}
			(-1,0) node (a0) [elt] {};
			\draw [black,line width=1pt] (a0) -- (a1) -- (a2) -- (a0);
			
			 \node at ($(a0.south) + (0,-0.4)$) 	{\small \parbox{1cm}{$\underline{\alpha_{45^2 6^3 7^2 8}}$\\ $s_{675} s_{468} s_{576}$}};	
		     	
		     \node at ($(a1.north) + (0,0.4)$) 	{\small \parbox{1cm}{$\alpha_{12^2 3^3 4^2 5}$\\ $s_{324} s_{135} s_{423}$}};
		     
		      \node at ($(a2.south) + (0.1,-0.5)$) 	{\small \parbox{1cm}{$\alpha_{0^3 1 2^2 3^3 4^2 5}$\\ $s_{0324} s_{135} s_{4230}$}};
			
			\draw[<->,dashed] (a1) edge[bend left=40] node[fill=white]{\small $s_0$} (a2);
			
			\end{scope}
			
			      \end{tikzpicture}
			      
			      \medskip
			      
			      $t_{\omega_1^{A_2}+\omega_{2}^{A_2}}\sim t^{1/3}_{\omega_3-2\omega_6}$
			      
			      \end{center} 
			    
			    & 
			    \begin{center}
			    
			    $3A_2$
			    
			     \medskip
			      
			      \begin{tikzpicture}[elt/.style={circle,draw=black!100,thick, inner sep=0pt,minimum size=2mm},scale=0.75]
			\path (0,0) node (a1) [elt] {}
			(0,1) node (a2) [elt] {}
			(0,3) node  (a4) [elt] {}
			(0,4) node  (a5) [elt] {}
			(0,6) node  (a7) [elt] {}
			(0,7) node  (a8) [elt] {};
			\draw [black,line width=1pt] (a1) -- (a2) (a4) -- (a5) (a7) -- (a8);
		    
		    	\node at ($(a1.south) + (0,-0.2)$) 	{\small $\alpha_1$};	
		    	
		    		\node at ($(a2.north) + (0,0.2)$) 	{\small $\alpha_2$};
		    		
		    			\node at ($(a4.south) + (0,-0.2)$) 	{\small $\alpha_4$};	
		    			
		    			\node at ($(a5.north) + (0,0.2)$) 	{\small $\alpha_5$};	
		    			
		    			\node at ($(a7.south) + (0,-0.2)$) 	{\small $\alpha_7$};	
		    			
		    			\node at ($(a8.north) + (0,0.2)$) 	{\small $\alpha_8$};

		    	\draw[<->, dashed] (0,0.5) to[bend left=40]
			node[fill=white]{\small $s_{324} s_{135} s_{423}$} (0,3.5);
			
			\draw[<->, dashed] (0,3.5) to[bend left=40]
			node[fill=white]{\small $s_{675} s_{468} s_{576}$} (0, 6.5);
			
			\end{tikzpicture} 
			    
			\end{center}

			    & 
			    
			    \begin{center}
			    
			    \begin{gather*}
			      S_3\ltimes (W_{A_2}^a \times \underbrace{C_3^3}_{triv})\\
			      \simeq C_2\ltimes (W_{A_2^l}^a \ltimes \underbrace{C_3^3}_{triv})\\
			      \downarrow\\
			      C_2\ltimes W_{A_2}^a\\
			      t^{1/3}_{\omega_3-2\omega_6}\rightarrow t_{\omega_4-2\omega_5}
			      \end{gather*}
			      
			      \medskip
			      
			      \eqref{symm_E8_E6}
			      
			      \end{center}
			    
			    &
			    \begin{tikzpicture}[elt/.style={circle,draw=black!100,thick, inner sep=0pt,minimum size=1.4mm},scale=0.4]
				\path 	(-2,0) 	node 	(a1) [elt,fill] {}
					(-1,0) 	node 	(a2) [elt,fill] {}
					( 0,0) node  	(a3) [elt] {}
					( 1,0) 	node  	(a4) [elt,fill] {}
					( 2,0) 	node 	(a5) [elt,fill] {}
					( 3,0)	node 	(a6) [elt] {}
					( 4,0)	node 	(a7) [elt,fill] {}
					( 5,0)	node 	(a8) [elt,fill] {}
					( 0,1)	node 	(a0) [elt] {};
				\draw [black,line width=1pt ] (a1) -- (a2) -- (a3) -- (a4) -- (a5) --  (a6) -- (a7) --(a8) (a3) -- (a0);
				\node at ($(a1.south) + (0,-0.3)$) 	{\tiny $\zeta$};
			\node at ($(a2.south) + (0,-0.3)$) 	{\tiny $\zeta$};
			\node at ($(a4.south) + (0,-0.3)$) 	{\tiny $\zeta$};
			\node at ($(a5.south) + (0,-0.3)$) 	{\tiny $\zeta$};
			\node at ($(a7.south) + (0,-0.3)$) 	{\tiny $\zeta$};
			\node at ($(a8.south) + (0,-0.3)$) 	{\tiny $\zeta$};
			
			\draw [<->, dashed] (0,-1) to node[fill=white] {\small $s_{147}$} (0,-4);

			\begin{scope}[yshift=-6cm]
			 \path 	(-2,0) 	node 	(a1) [elt,fill] {}
					(-1,0) 	node 	(a2) [elt,fill] {}
					( 0,0) node  	(a3) [elt] {}
					( 1,0) 	node  	(a4) [elt,fill] {}
					( 2,0) 	node 	(a5) [elt,fill] {}
					( 3,0)	node 	(a6) [elt] {}
					( 4,0)	node 	(a7) [elt,fill] {}
					( 5,0)	node 	(a8) [elt,fill] {}
					( 0,1)	node 	(a0) [elt] {};
				\draw [black,line width=1pt ] (a1) -- (a2) -- (a3) -- (a4) -- (a5) --  (a6) -- (a7) --(a8) (a3) -- (a0);
				\node at ($(a1.south) + (0,-0.3)$) 	{\tiny $\zeta^{-1}$};
			\node at ($(a2.south) + (0,-0.3)$) 	{\tiny $\zeta^{-1}$};
			\node at ($(a4.south) + (0,-0.3)$) 	{\tiny $\zeta^{-1}$};
			\node at ($(a5.south) + (0,-0.3)$) 	{\tiny $\zeta^{-1}$};
			\node at ($(a7.south) + (0,-0.3)$) 	{\tiny $\zeta^{-1}$};
			\node at ($(a8.south) + (0,-0.3)$) 	{\tiny $\zeta^{-1}$};
			 
			\end{scope}

			\end{tikzpicture}
			     \\
			     \hline
			     
			     \end{tabular}
			     \end{center}
	
	\medskip
	
	\begin{itemize}
	
	\item By dashed arrows we denote the action of $\Aut$. By underline we denote our choice of the affine root.
	
	\item In the second column for roots of length $|\alpha|^2=6$, under the
	decomposition of given root in terms of simple ones we write corresponding element from $W^{ae}$, which realizes reflection in such root.
        At the bottom of this column, we write projective reduction occurs and corresponding translation in Weyl group of 
        $\Phi^{a,w}|_{|\alpha|^2=6}$.
        
        \item In the fourth column we write two centralizers with respect
        to short and long (l) $\Phi^{a,w}$, write image of centralizer after folding matching elements, which go to trivial in the image.
        We also write the image of projective reduction after folding. Note that the image of the centralizer is found from the coordinates,
        found in Section \ref{sec:geom_answ}.
		
	\end{itemize}
		
			     \begin{center}
			\begin{tabular}{|m{3cm}|m{4cm}|m{2.5cm}|m{3cm}|m{3cm}|}
			\hline
			  \begin{center}
			  Folding  \end{center}
			      & \begin{center}
			     $Q^{a,w}$  
			   \end{center}  & \begin{center} $Q^{a, \perp w}$ \end{center}	                                                 
			     & \begin{center}  $C(w)$ \end{center} & \begin{center} $\mathcal{A}^w, \, N_{flip}$ \end{center}\\		                                                
			     \hline

			     \begin{center}


                        \includegraphics[scale=0.8]{e8_1111.pdf} 
			
			\medskip
			
			$w=s_{0468}$

			\end{center}

			     &
			   \begin{center}

			        \begin{tikzpicture}[elt/.style={circle,draw=black!100,thick, inner sep=0pt,minimum size=2mm},scale=1.25]
			\path 	(-1,0) 	node 	(a1) [elt] {}
			(0,1) 	node 	(a0) [elt] {}
			(0,0) node  	(a2) [elt] {}
			(0,-1) node  	(a3) [elt] {}
			(1,0) node  	(a4) [elt] {};

		    \draw [black,line width=1pt] (a0) -- (a2) -- (a4) (a1) -- (a2) -- (a3);

		     	\node at ($(a1.north) + (0,0.2)$) 	{\small$\alpha_2$};
		     	\node at ($(a0.north) + (0,0.2)$) 	{\small$\underline{\alpha_{0233445566778}}$};
		     	\node at ($(a2.west) + (-0.2,0.2)$) 	{\small$\alpha_1$};
		     	\node at ($(a3.south) + (0,-0.2)$) 	{\small$\alpha_{02334}$};
		     	\node at ($(a4.north) + (0,0.2)$) 	{\small$\alpha_{023344556}$};
		     	
		     	\draw[<->, dashed] (a1) to [bend right=40] node[fill=white]{\small $s_{3043}$} (a3); 
		     	\draw[<->, dashed] (a3) to [bend right=40] node[fill=white]{\small $s_{5465}$} (a4); 
		     	\draw[<->, dashed] (a4) to [bend right=40] node[fill=white]{\small $s_{7687}$} (a0); 
			     
			\node at (0,2) 	{$D_4^{(1)}$};
			
			\node at (0,-2) {$(D_4^{(1)})|_{|\alpha|^2=4}$};

			        \begin{scope}[yshift=-4cm]
			\path 	(1,0) 	node 	(a1) [elt] {}
			(-1,0) 	node 	(a0) [elt] {}
			(0,0) node  	(a2) [elt] {}
			(0,1) node  	(a3) [elt] {}
			(0,-1) node  	(a4) [elt] {};

		    \draw [black,line width=1pt] (a0) -- (a2) -- (a4) (a1) -- (a2) -- (a3);

		     	\node at ($(a1.west) + (-0.2,0)$) 	{\small \parbox{1cm}{$\alpha_{0334}$ \\ $s_{3043}$}};
		     	\node at ($(a0.north) + (0,0.3)$) 	{\small \parbox{1cm}{$\underline{\alpha_{6778}}$\\$s_{7687}$}};
		     	\node at ($(a2.west) + (-0.2,0.3)$) 	{\small \parbox{1cm}{$\alpha_{4556}$\\$s_{5465}$}};
		     	\node at ($(a3.north) + (0,0.3)$) 	{\small \parbox{1cm}{$\alpha_{022334}$\\$s_{230432}$}};
		     	\node at ($(a4.south) + (0,-0.3)$) 	{\small \parbox{1cm}{$\alpha_{01122334}$\\$s_{12304321}$}};
		     	
		     	\draw[<->, dashed] (a1) to [bend right=40] node[fill=white]{\small $s_2$} (a3); 
		     	\draw[<->, dashed] (a1) to [bend left=40] node[fill=white]{\small $s_{121}$} (a4); 
			     
			 \end{scope}    
			     
			 \end{tikzpicture}
			 
			 \medskip
			 
			{\small $t_{\omega_{\alpha_{02^23^24}}+\omega_{\alpha_{03^24}}+\omega_{\alpha_{01^22^23^24}}-\omega_{\alpha_{45^26}}}$}\\
			$\sim t^{1/2}_{\omega_3-\omega_{57}}$\\
			 $t_{\omega_{\alpha_{4556}}}\sim t^{1/2}_{\omega_5-2\omega_7}$

			 \end{center}

			     &
			     
			     \begin{center}
			     
			      $4A_1$
			      			      
			      \medskip
			      
			      \begin{tikzpicture}[elt/.style={circle,draw=black!100,thick, inner sep=0pt,minimum size=2mm},scale=2]
			\path (0,0) node (a0) [elt] {}
			(0,1) node (a4) [elt] {}
			(0,2) node  (a6) [elt] {}
			(0,3) node  (a8) [elt] {};
		    
		    	\node at ($(a0.west) + (-0.2,0)$) 	{\small $\alpha_0$};	
		    	
		    		\node at ($(a4.west) + (-0.2,0)$) 	{\small $\alpha_4$};
		    		
		    			\node at ($(a6.west) + (-0.2,0)$) 	{\small $\alpha_6$};	
		    			
		    			\node at ($(a8.west) + (-0.2,0)$) 	{\small $\alpha_8$};

		    	\draw[<->,dashed] (a0) to[bend right=0]
			node[fill=white]{\small $s_{3043}$} (a4);
			
			\draw[<->,dashed] (a4) to[bend right=0]
			node[fill=white]{\small $s_{5465}$} (a6);
			
			\draw[<->,dashed] (a6) to[bend right=0]
			node[fill=white]{\small $s_{7687}$} (a8);
		    
			\end{tikzpicture} 
			    
			    \end{center}
			     
			      &
			      
			      \begin{center}
			      
			      \begin{gather*}
			      S_4\ltimes (W_{D_4}^a \times \underbrace{C_2^4}_{triv})\\
			      \simeq S_3\ltimes (W_{D_4^l}^a \ltimes \underbrace{C_2^4}_{triv})\\
			      \downarrow\\
			      S_3\ltimes W_{D_4}^a\\
			      t^{1/2}_{\omega_3-\omega_{57}}\rightarrow t_{\omega_3-\omega_{12}}\\
			      t^{1/2}_{\omega_5-2\omega_7}\rightarrow t_{\omega_2-2\omega_1}
			      \end{gather*}
			      
			      \medskip
			      
			      \eqref{symm_E8_E7}
			      
			      \end{center}

			      &
			       \begin{tikzpicture}[elt/.style={circle,draw=black!100,thick, inner sep=0pt,minimum size=1.4mm},scale=0.4]
				\path 	(-2,0) 	node 	(a1) [elt] {}
					(-1,0) 	node 	(a2) [elt] {}
					( 0,0) node  	(a3) [elt] {}
					( 1,0) 	node  	(a4) [elt,fill] {}
					( 2,0) 	node 	(a5) [elt] {}
					( 3,0)	node 	(a6) [elt,fill] {}
					( 4,0)	node 	(a7) [elt] {}
					( 5,0)	node 	(a8) [elt,fill] {}
					( 0,1)	node 	(a0) [elt,fill] {};
				\draw [black,line width=1pt ] (a1) -- (a2) -- (a3) -- (a4) -- (a5) --  (a6) -- (a7) --(a8) (a3) -- (a0);
				\node at ($(a0.north) + (0,0.3)$) 	{\tiny $-1$};
			\node at ($(a4.south) + (0,-0.3)$) 	{\tiny $-1$};
			\node at ($(a6.south) + (0,-0.3)$) 	{\tiny $-1$};
			\node at ($(a8.south) + (0,-0.3)$) 	{\tiny $-1$};

			\end{tikzpicture}
			      
			      \\
			     \hline
			     
			     \begin{center}

%

                       \includegraphics[scale=0.8]{e8_133.pdf}
			
			\medskip
			
			$w=s_{1430876}$
			
			\end{center}
			     
			     &
			     
			     \begin{center}
			
			\begin{tikzpicture}[elt/.style={circle,draw=black!100,thick, inner sep=0pt,minimum size=2mm},scale=1]
			\path 	(-1,0) 	node 	(a1) [elt] {}
			(1,0) 	node 	(a2) [elt] {};

		    \draw [black,line width=2.5pt ] (a1) -- (a2);
		    
		     \draw [white,line width=1.5pt ] (a1) -- (a2);
		     
		     \draw [<->,black, line width=0.5pt]
		     (a1) -- (a2);
		     
		     \node at ($(a2.north) + (0,0.2)$) 	{\small$\underline{\delta-\alpha_{0122334}}$};	
		     	
		     \node at ($(a1.north) + (0,0.2)$) 	{\small $\alpha_{0122334}$};
		     
		      \draw [<->, dashed]
		     (a1) to[bend right=40] node[fill=white]{\small $s_{567435}s_{0468}s_{534765}$} (a2);
		     
		     \node at (0,1) 	{$A_1^{(1)}$};
		     
		      \node at (0,-1.5) { $(A_1^{(1)})|_{|\alpha|^2=8}$};

			\begin{scope}[yshift=-2.5cm]
			\path 	(1,0) 	node 	(a1) [elt] {}
			(-1,0) 	node 	(a2) [elt] {};

		    \draw [black,line width=2.5pt ] (a1) -- (a2);
		    
		     \draw [white,line width=1.5pt ] (a1) -- (a2);
		     
		     \draw [<->,black, line width=0.5pt]
		     (a1) -- (a2);
		     
		     \node at ($(a2.south) + (0.2, -0.2)$) 	{\tiny \parbox{2cm}{$\underline{\delta-2\alpha_{0122334}}$\\$s_{567435}s_{0468}s_{534765}$ }};	
		     	
		     \node at ($(a1.north) + (0.3, 0.3)$) 	{\tiny \parbox{2cm}{$2\alpha_{0122334}$ \\ $s_{\alpha_{0122334}}$}};
		     
		     \end{scope}
		     
		     \end{tikzpicture}
		     
		     \medskip
		     
		     $t_{\omega_{2\alpha_{0122334}}}\sim t^{1/2}_{\omega_2-\omega_5}$
		     
		     \end{center}
			     
			     & 
			     
			     \begin{center}
			     
			     $A_1+ 2A_3$

			      \medskip
			      
			      \begin{tikzpicture}[elt/.style={circle,draw=black!100,thick, inner sep=0pt,minimum size=2mm},scale=1]
			\path (0,2) node (a1) [elt] {}
			(-1,1) node (a0) [elt] {}
			(0,1) node  (a3) [elt] {}
			(1,1) node  (a4) [elt] {}
			(-1,-1) node (a6) [elt] {}
			(0,-1) node  (a7) [elt] {}
			(1,-1) node  (a8) [elt] {};
			
			\draw (a0) -- (a3) -- (a4) (a6) -- (a7) -- (a8);
		    
		    	\node at ($(a1.north) + (0,0.2)$) 	{\small $\alpha_1$};	
		    	
		    	\node at ($(a0.north) + (0,0.2)$) 	{\small $\alpha_0$};	
		    	\node at ($(a3.north) + (0,0.2)$) 	{\small $\alpha_3$};	
		    	\node at ($(a4.north) + (0,0.2)$) 	{\small $\alpha_4$};	
		    	
		    	\node at ($(a6.south) + (0,-0.2)$) 	{\small $\alpha_6$};	
		    	\node at ($(a7.south) + (0,-0.2)$) 	{\small $\alpha_7$};	
		    	\node at ($(a8.south) + (0,-0.2)$) 	{\small $\alpha_8$};
		    	
		    	\draw[<->, dashed] (a0) -- (a6);
		    	\draw[<->, dashed] (a3) to[] node[fill=white] {\small $s_{567435}s_{0468}s_{534765}$} (a7);
		    	\draw[<->, dashed] (a4) -- (a8);
		  
			\end{tikzpicture} 
			     
			     \end{center}
   
			     &
			     
			     \begin{center}
			     
			     \begin{gather*}
			      C_2{\ltimes}(W_{A_1}^a{\times}(\underbrace{C_2{\times}C_4^2}_{triv}))\\
			      \simeq  W_{A_1}^a \ltimes (\underbrace{C_2\times C_4^2}_{triv})\\
			      \downarrow\\
			      W_{A_1}^a\\
			      t^{1/2}_{\omega_2-\omega_5}\rightarrow t_{\omega_2-\omega_0}
			      \end{gather*}
			      
			      \medskip
			      
			      \eqref{symm_E8_D5}
			      
			      \end{center}
			     
			     &
			     \begin{tikzpicture}[elt/.style={circle,draw=black!100,thick, inner sep=0pt,minimum size=1.4mm},scale=0.4]
				\path 	(-2,0) 	node 	(a1) [elt,fill] {}
					(-1,0) 	node 	(a2) [elt] {}
					( 0,0) node  	(a3) [elt,fill] {}
					( 1,0) 	node  	(a4) [elt,fill] {}
					( 2,0) 	node 	(a5) [elt] {}
					( 3,0)	node 	(a6) [elt,fill] {}
					( 4,0)	node 	(a7) [elt,fill] {}
					( 5,0)	node 	(a8) [elt,fill] {}
					( 0,1)	node 	(a0) [elt,fill] {};
				\draw [black,line width=1pt ] (a1) -- (a2) -- (a3) -- (a4) -- (a5) --  (a6) -- (a7) --(a8) (a3) -- (a0);
				\node at ($(a1.south) + (0,-0.3)$) 	{\tiny $-1$};
			\node at ($(a0.north) + (0,0.3)$) 	{\tiny $\ri$};
			\node at ($(a3.south) + (0,-0.3)$) 	{\tiny $\ri$};
			\node at ($(a4.south) + (0,-0.3)$) 	{\tiny $\ri$};
			\node at ($(a6.south) + (0,-0.3)$) 	{\tiny $\ri$};
			\node at ($(a7.south) + (0,-0.3)$) 	{\tiny $\ri$};
			\node at ($(a8.south) + (0,-0.3)$) 	{\tiny $\ri$};
			
			\draw [<->,dashed] (0,-1) to node[fill=white] {\small $s_{0468}$} (0,-4); 
			
			\begin{scope}[yshift=-6cm]
			 \path 	(-2,0) 	node 	(a1) [elt,fill] {}
					(-1,0) 	node 	(a2) [elt] {}
					( 0,0) node  	(a3) [elt,fill] {}
					( 1,0) 	node  	(a4) [elt,fill] {}
					( 2,0) 	node 	(a5) [elt] {}
					( 3,0)	node 	(a6) [elt,fill] {}
					( 4,0)	node 	(a7) [elt,fill] {}
					( 5,0)	node 	(a8) [elt,fill] {}
					( 0,1)	node 	(a0) [elt,fill] {};
				\draw [black,line width=1pt ] (a1) -- (a2) -- (a3) -- (a4) -- (a5) --  (a6) -- (a7) --(a8) (a3) -- (a0);
				\node at ($(a1.south) + (0,-0.3)$) 	{\tiny $-1$};
			\node at ($(a0.north) + (0,0.3)$) 	{\tiny $-\ri$};
			\node at ($(a3.south) + (0,-0.3)$) 	{\tiny $-\ri$};
			\node at ($(a4.south) + (0,-0.3)$) 	{\tiny $-\ri$};
			\node at ($(a6.south) + (0,-0.3)$) 	{\tiny $-\ri$};
			\node at ($(a7.south) + (0,-0.3)$) 	{\tiny $-\ri$};
			\node at ($(a8.south) + (0,-0.3)$) 	{\tiny $-\ri$};
			\end{scope}

			\end{tikzpicture}
			     
			     \\
			     \hline
	\end{tabular}
	\end{center}

	\newpage

	\subsection{$E_7^{(1)}/A_1^{(1)}$}
	
	\label{e7_a}
	
	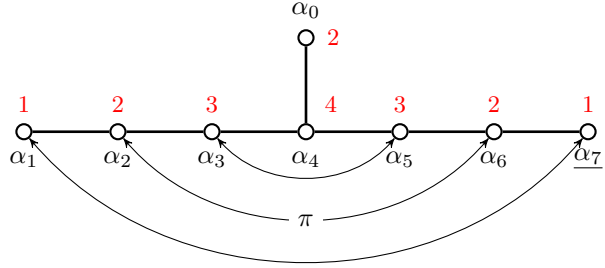
\begin{figure}[h]
	
	\begin{center}
	\begin{tikzpicture}[elt/.style={circle,draw=black!100,thick, inner sep=0pt,minimum size=2mm},scale=1.25]
			\path 	(-3,0) 	node 	(a1) [elt] {}
				(-2,0) 	node 	(a2) [elt] {}
				( -1,0) node  	(a3) [elt] {}
				( 0,0) 	node  	(a4) [elt] {}
				( 1,0) 	node 	(a5) [elt] {}
				( 2,0)	node 	(a6) [elt] {}
				( 3,0)	node 	(a7) [elt] {}
				( 0,1)	node 	(a0) [elt] {};
			\draw [black,line width=1pt ] (a1) -- (a2) -- (a3) -- (a4) -- (a5) --  (a6) -- (a7) (a4) -- (a0);
			\node at ($(a1.south) + (0,-0.2)$) 	{$\alpha_{1}$};
			\node at ($(a2.south) + (0,-0.2)$)  {$\alpha_{2}$};
			\node at ($(a3.south) + (0,-0.2)$)  {$\alpha_{3}$};
			\node at ($(a4.south) + (0,-0.2)$)  {$\alpha_{4}$};	
			\node at ($(a5.south) + (0,-0.2)$)  {$\alpha_{5}$};		
			\node at ($(a6.south) + (0,-0.2)$) 	{$\alpha_{6}$};	
			\node at ($(a7.south) + (0,-0.2)$) 	{$ \underline{\alpha_{7}}$};	
			\node at ($(a0.north) + (0,0.2)$) 	{$\alpha_{0}$};		
			\node at ($(a0.east) + (0.2,0)$) 	{\color{red}\small$2$};
			\node at ($(a1.north) + (0,0.2)$) 	{\color{red}\small$1$};
			\node at ($(a2.north) + (0,0.2)$) 	{\color{red}\small$2$};
			\node at ($(a3.north) + (0,0.2)$) 	{\color{red}\small$3$};
			\node at ($(a4.north east) + (0.2,0.23)$) {\color{red}\small$4$};
			\node at ($(a5.north) + (0,0.2)$) 	{\color{red}\small$3$};
			\node at ($(a6.north) + (0,0.2)$) 	{\color{red}\small$2$};
			\node at ($(a7.north) + (0,0.2)$) 	{\color{red}\small$1$};
			
			\draw[<->] (a1) edge[bend right=50] node[]{} (a7);
			\draw[<->] (a2) edge[bend right=50] node[fill=white]{$\pi$} (a6);
			\draw[<->] (a3) edge[bend right=50] node[]{} (a5);
			
		\end{tikzpicture}
	        \end{center} 
	      
	        \caption{Root numbering and marks for $E_7^{(1)}$ Dynkin diagram}
	        
	        \label{fig:E7_rn}
                
		\end{figure}
			
		Affine Weyl group $W(E_7^{(1)})$ (and corresponding Dynkin diagram) has outer automorphism of order $2$\footnote{It is absent in \cite{KNY15}.}.
			
		For $E_7^{(1)}$ we have $6$ foldings: 
		\begin{itemize}
		\item Two foldings of order $2$, of case $1$ and $2$, which have common (up to conjugation)
		$\bar{w}$.
		\item Two foldings of order $4$, of case $2$ and $3$, which have common (up to conjugation)
		$\bar{w}$.
		\item Folding of order $2$, of case $2$, which is (up to conjugation) square of both foldings of order $4$.
		\item Folding of order $3$, of case $2$.
		\end{itemize}
			
	        Several new notations and peculiarities appear here in comparison with $E_8^{(1)}$ case.
	        
	        \begin{itemize}
	        \item Here appear foldings which have conjugated finitization $\bar{w}$.
	        This is reflected in the first column, where we wrote decomposition of $w$
	        into translational and finite parts and write conjugation of $\bar{w}$
	        explicitly. Due to Lemma \ref{lemma:lattices} invariant lattices $Q^{a,w}$ are of the same type.
	     
		\item We denote by $w_0(S_n)$ longest element of Weyl group $S_n$
		
		\item In the $\mathcal{A}_w$ arrow with ``='' means, that corresponding root variables equal (we write one ``='' for
		several arrows for simplicity)
		
		\item In case $3A_2$ we have no short invariant lattice, we immediately obtain long lattice. 
		For the next cases, one finds more such examples.
		
		\item In case $\pi\ltimes 5A_1$ we have $Q^{a,\perp w}$ generated
		by $\alpha_0, \alpha_1, \alpha_3, \alpha_5, \alpha_7$ and $\alpha_2-\alpha_6$.
		This sublattice is not generated by any root subsystem, this
		is inconvenient, so instead we take its sublattice of index $2$, taking doubled generator $2(\alpha_2-\alpha_6)$.  
		Such sublattice already have root subsystem description, we denote it $\widetilde{Q^{a,\perp w}}$.
		\end{itemize}
			
			\begin{center}
			\begin{tabular}{|m{2.5cm}|m{4cm}|m{2.5cm}|m{3cm}|m{2.5cm}|}
			\hline
			  \begin{center} Folding \end{center}    & \begin{center}
			     $Q^{a,w}$  
			   \end{center}  & \begin{center}  $Q^{a,\perp w}$ \end{center}		                                                  
			     & \begin{center}  $C(w)$ \end{center} & \begin{center} $\mathcal{A}^w$ \end{center}
			                                                           \\
			    \hline
			    
			    \begin{center}
 
%

                         \includegraphics[scale=0.8]{e7_p.pdf}
			
			\medskip
			
			 $w=\pi$
			{\small $=t_{\omega_1}\bar{\pi}$
			 $\bar{\pi}=s s_{013} s^{-1}$,
			    $s=s_{123405645342}$}

			\end{center}

			    &

			\begin{center}
			
			   \begin{tikzpicture}[elt/.style={circle,draw=black!100,thick, inner sep=0pt,minimum size=2mm},scale=1.25]
			\path 	(-1,0) 	node 	(a0) [elt] {}
			(0,1) 	node 	(a1) [elt] {}
			(0,0) node  	(a2) [elt] {}
			(0,-1) node  	(a3) [elt] {}
			(1,0) node  	(a4) [elt] {};

		    \draw [black,line width=1pt] (a0) -- (a2) -- (a4) (a1) -- (a2) -- (a3);

		     	\node at ($(a0.north) + (0,0.2)$) 	{\small$\alpha_4$};
		     	\node at ($(a1.north) + (0,0.2)$) 	{\small$\underline{\alpha_{1234567}}$};
		     	\node at ($(a2.west) + (-0.2,0.2)$) 	{\small$\alpha_0$};
		     	\node at ($(a3.south) + (0,-0.2)$) 	{\small$\alpha_{345}$};
		     	\node at ($(a4.east) + (0.4,0)$) 	{\small$\alpha_{23456}$};
		     	
		     \draw[<->, dashed] (a4) to [bend right=40] node[fill=white]{\small $s_{17}$} (a1);  	
		     
 		    \draw[<->, dashed] (a3) to [bend right=40] node[fill=white]{\small $s_{26}$} (a4); 
 		    
		    \draw[<->, dashed] (a3) to [bend left=40] node[fill=white]{\small $s_{35}$} (a0); 
 		  
			\node at (0,1.75) 	{$D_4^{(1)}$};
			
			\node at (0,-2) {$(D_4^{(1)})|_{|\alpha|^2=4}$};

                       \begin{scope}[yshift=-4cm]
			\path 	(-1,0) 	node 	(a0) [elt] {}
			(0,1) 	node 	(a1) [elt] {}
			(0,0) node  	(a2) [elt] {}
			(0,-1) node  	(a3) [elt] {}
			(1,0) node  	(a4) [elt] {};

		    \draw [black,line width=1pt] (a0) -- (a2) -- (a4) (a1) -- (a2) -- (a3);

		     	\node at ($(a0.north) + (-0.1,0)$) 	{\small \parbox{1cm}{$\alpha_{35}$\\$s_{35}$}};
		     	\node at ($(a1.north) + (0,0.3)$) 	{\small \parbox{1cm}{$\underline{\alpha_{17}}$\\$s_{17}$}};
		     	\node at ($(a2.west) + (-0.1,0.3)$) 	{\small \parbox{1cm}{$\alpha_{26}$\\$s_{26}$}};
		     	\node at ($(a3.south) + (0,-0.2)$) 	{\small$\alpha_{003445}$};
		     	\node at ($(a4.north) + (0,0.2)$) 	{\small$\alpha_{3445}$};
		     	
 		    \draw[<->, dashed] (a3) to [bend right=40] node[fill=white]{\small $s_0$} (a4); 
 		    
		    \draw[<->, dashed] (a3) to [bend left=40] node[fill=white]{\small $s_{040}$} (a0); 
 		     
 		     \end{scope}
 		     
			\end{tikzpicture}  
			
			\medskip
			
			 $t_{\omega_{\alpha_{35}}+\omega_{\alpha_{3445}}+\omega_{\alpha_{003445}}-\omega_{\alpha_{26}}}$\\
			$\sim t_{\omega_{35}-\omega_{126}}^{1/2}$\\
			$t_{\omega_{\alpha_{26}}} \sim t_{\omega_{26}-2\omega_1}^{1/2}$\\

			\end{center}

			    &

		\begin{center}
		
		$(A_3)_{|\alpha|^2=4}$
		
		\medskip
			
                   \begin{tikzpicture}[elt/.style={circle,draw=black!100,thick, inner sep=0pt,minimum size=2mm},scale=2]
			\path (0,-1) node (a1) [elt] {}
			(0,0) node (a2) [elt] {}
			(0,1) node  (a3) [elt] {};
		    
		        \draw (a1) -- (a2) -- (a3); 
		    
		    	\node at ($(a1.east) + (0.3,0)$) 	{\small \parbox{1cm}{$\alpha_1{-}\alpha_7$\\$s_{17}$}};	
		    	
		    		\node at ($(a2.east) + (0.3,0)$) 	{\small \parbox{1cm}{$\alpha_2{-}\alpha_6$\\$s_{26}$}};
		    		
		    			\node at ($(a3.east) + (0.3,0)$) 	{\small \parbox{1cm}{$\alpha_3{-}\alpha_5$\\ $s_{35}$}};

 		    	\draw[<->, dashed] (a1) to[bend left=45]
 			node[fill=white]{\small $\pi w_0(S_4)$} (a3);

			\end{tikzpicture} 	
			    
			\end{center}

			    &
			    \begin{center}
			    
			    {\small
			    \begin{gather*}
			    (S_4{\ltimes} W_{D_4}^a) \times \underbrace{C_2}_{triv}\\
			    \simeq(S_3{\ltimes} W_{D_4}^a) \times \underbrace{C_2}_{triv}\\
			    \downarrow\\
			    S_3{\ltimes} (C_2^2 \times W_{D_4}^a)\\
			    t_{\omega_{35}-\omega_{126}}^{1/2}\rightarrow t_{\omega_6-\omega_7}\\
			    t_{\omega_{26}-2\omega_1}^{1/2}\rightarrow t_{\omega_7}
			    \end{gather*}
                            }  
			    \eqref{symm_E7_E8_c1}  
			    
			    \end{center}

			    &

			    \begin{center}
			     \begin{tikzpicture}[elt/.style={circle,draw=black!100,thick, inner sep=0pt,minimum size=1.4mm},scale=0.4]
			\path 	(-3,0) 	node 	(a1) [elt] {}
			(-2,0) 	node 	(a2) [elt] {}
			( -1,0) node  	(a3) [elt] {}
			( 0,0) 	node  	(a4) [elt] {}
			( 1,0) 	node 	(a5) [elt] {}
			( 2,0)	node 	(a6) [elt] {}
			( 3,0)	node 	(a7) [elt] {}
			( 0,1)	node 	(a0) [elt] {};
			\draw [black,line width=1pt ] (a1) -- (a2) -- (a3) -- (a4) -- (a5) --  (a6) -- (a7) (a4) -- (a0);
			\draw[<->] (a1) edge[bend right=50] node[]{} (a7);
			\draw[<->] (a2) edge[bend right=50] node[fill=white]{$=$} (a6);
			\draw[<->] (a3) edge[bend right=50] node[]{} (a5);
			
			\end{tikzpicture}
			\end{center}

			    \\
			     \cline{1-1} \cline{3-5}

			    \begin{center}
			
%

                        \includegraphics[scale=0.8]{e7_111.pdf}
			
			\medskip
			
			$w=s_{013}=\bar{w}$
			
			\end{center}
			
			& 
			
			\begin{center}
			
			   \begin{tikzpicture}[elt/.style={circle,draw=black!100,thick, inner sep=0pt,minimum size=2mm},scale=1.25]
			\path 	(-1,0) 	node 	(a0) [elt] {}
			(0,1) 	node 	(a1) [elt] {}
			(0,0) node  	(a2) [elt] {}
			(0,-1) node  	(a3) [elt] {}
			(1,0) node  	(a4) [elt] {};

		    \draw [black,line width=1pt] (a0) -- (a2) -- (a4) (a1) -- (a2) -- (a3);

		     	\node at ($(a0.north) + (0,0.2)$) 	{\small$\alpha_{5}$};
		     	\node at ($(a1.north) + (0,0.2)$) 	{\small$\underline{\alpha_{7}}$};
		     	\node at ($(a2.west) + (-0.2,0.2)$) 	{\small$\alpha_{6}$};
		     	\node at ($(a3.south) + (0,-0.2)$) 	{\small$\alpha_{03445}$};
		     	\node at ($(a4.north) + (0,0.2)$) 	{\small$\alpha_{012233445}$};

 		    \draw[<->, dashed] (a3) to [bend right=40] node[fill=white]{\small $s_{2132}$} (a4); 
 		    
		    \draw[<->, dashed] (a3) to [bend left=40] node[fill=white]{\small $s_{4304}$} (a0); 
 		  
			\node at (0,1.75) 	{$D_4^{(1)}$};
			
			\node at (0,-2) {$(D_4^{(1)})|_{|\alpha|^2=4}$};
                     
                       \begin{scope}[yshift=-4cm]
			\path 	(-1,0) 	node 	(a0) [elt] {}
			(0,1) 	node 	(a1) [elt] {}
			(0,0) node  	(a2) [elt] {}
			(0,-1) node  	(a3) [elt] {}
			(1,0) node  	(a4) [elt] {};

		    \draw [black,line width=1pt] (a0) -- (a2) -- (a4) (a1) -- (a2) -- (a3);

		     	\node at ($(a0.north) + (0,0.3)$) 	{\small \parbox{1cm}{$\alpha_{0344}$\\$s_{4034}$}};
		     	\node at ($(a1.north) + (0,0.3)$) 	{\small \parbox{1cm}{$\underline{\alpha_{0344556677}}$\\$s_{7654034567}$}};
		     	\node at ($(a2.west) + (-0.2,0.3)$) 	{\small \parbox{1cm}{$\alpha_{1223}$\\$s_{2132}$}};
		     	\node at ($(a3.south) + (0,-0.3)$) 	{\small \parbox{1cm}{$\alpha_{034455}$\\$s_{540345}$}};
		     	\node at ($(a4.north) + (0,-0.05)$) 	{\small \parbox{1cm}{$\alpha_{03445566}$\\$s_{65403456}$}};
		     	
 		    \draw[<->, dashed] (a3) to [bend right=40] node[fill=white]{\small $s_6$} (a4); 
 		    
		    \draw[<->, dashed] (a3) to [bend left=40] node[fill=white]{\small $s_5$} (a0); 
 		    
 		    \draw[<->, dashed] (a4) to [bend right=40] node[fill=white]{\small $s_7$} (a1); 
 		     
 		     \end{scope}
 		     
			\end{tikzpicture}  
			
			\end{center}

			&
			
			\begin{center}

			$3A_1$
			
			\medskip
			
			\begin{tikzpicture}[elt/.style={circle,draw=black!100,thick, inner sep=0pt,minimum size=2mm},scale=2]
			\path (0,-1) node (a1) [elt] {}
			(0,0) node (a2) [elt] {}
			(0,1) node  (a3) [elt] {};
		    
		    	\node at ($(a1.west) + (-0.2,0)$) 	{\small $\alpha_{1}$};	
		    	
		    		\node at ($(a2.west) + (-0.2,0)$) 	{\small $\alpha_{3}$};
		    		
		    			\node at ($(a3.west) + (-0.2,0)$) 	{\small $\alpha_{0}$};

		    	\draw[<->, dashed] (a1) to[bend right=0]
			node[fill=white]{\small $s_{2132}$} (a2);
			
			\draw[<->, dashed] (a2) to[bend right=0]
			node[fill=white]{\small $s_{4304}$} (a3);
		    
			\end{tikzpicture} 
			
			\end{center} 
			
			&
			
			\begin{center}
			
			 \begin{gather*}
			      S_3 \ltimes (W_{D_4}^a \times \underbrace{C_2^3}_{triv})\\
			      \simeq S_4 \ltimes (W_{D_4^l}^a \ltimes \underbrace{C_2^3}_{triv})\\
			      \downarrow\\
			      S_4 \ltimes W_{D_4}^a\\
			       \textrm{no proj. red.}\\
			      \textrm{in the preimage}
			      \end{gather*}
			      
			      \medskip
			      
			      \eqref{symm_E7_E8_c2}
			      
			      \end{center}

			& 
			
			\begin{center}

			  \hspace*{-0.3cm}\begin{tikzpicture}[elt/.style={circle,draw=black!100,thick, inner sep=0pt,minimum size=1.4mm},scale=0.4]
			\path 	(-3,0) 	node 	(a1) [elt,fill] {}
			(-2,0) 	node 	(a2) [elt] {}
			( -1,0) node  	(a3) [elt,fill] {}
			( 0,0) 	node  	(a4) [elt] {}
			( 1,0) 	node 	(a5) [elt] {}
			( 2,0)	node 	(a6) [elt] {}
			( 3,0)	node 	(a7) [elt] {}
			( 0,1)	node 	(a0) [elt,fill] {};
			\draw [black,line width=1pt ] (a1) -- (a2) -- (a3) -- (a4) -- (a5) --  (a6) -- (a7) (a4) -- (a0);
			
			\node at ($(a0.north) + (0,0.3)$) 	{\tiny $-1$};
			\node at ($(a1.south) + (0,-0.3)$) 	{\tiny $-1$};
			\node at ($(a3.south) + (0,-0.3)$) 	{\tiny $-1$};
			
			\end{tikzpicture} 
			\end{center}

			\\
			    
			    \hline
			    \end{tabular}
			    \end{center}

			    \begin{center}
			    \begin{tabular}{|m{2.5cm}|m{4cm}|m{2.75cm}|m{3.25cm}|m{2.5cm}|}
			    \hline
			   \begin{center} Folding \end{center}  & \begin{center}
			     $Q^{a,w}$  
			   \end{center}  & \begin{center} $Q^{a,\perp w}$ \end{center}
			     & \begin{center} $C(w)$  \end{center} & \begin{center} $\mathcal{A}^w, \, N_{flip}$ \end{center}  \\
			    \hline
			    
			    \begin{center}


                    \includegraphics[scale=0.8]{e7_1111.pdf}
			
		   \medskip
		
			{$w{=}s_1 s_3 s_5 s_7$}

         \end{center}		
		
		    &

		    \begin{center}

		    \begin{tikzpicture}[elt/.style={circle,draw=black!100,thick, inner sep=0pt,minimum size=2mm},scale=1]
			\path 	(-1,0) 	node 	(a1) [elt] {}
			(1,0) 	node 	(a2) [elt] {}
			(-1,2) node  	(a3) [elt] {}
			(1,2) 	node  	(a4) [elt] {}
			(-1,4) 	node  	(a5) [elt] {}
			(1,4) 	node  	(a6) [elt] {};
			
		    \draw [black,line width=2.5pt ] (a1) -- (a2);
		    
		     \draw [white,line width=1.5pt ] (a1) -- (a2);
		     
		     \draw [<->,black, line width=0.5pt]
		     (a1) -- (a2);
		     
		      \draw [black,line width=2.5pt ] (a3) -- (a4);
		    
		     \draw [white,line width=1.5pt ] (a3) -- (a4);
		     
		     \draw [<->,black, line width=0.5pt]
		     (a3) -- (a4);
		     
		      \draw [black,line width=2.5pt ] (a5) -- (a6);
		    
		     \draw [white,line width=1.5pt ] (a5) -- (a6);
		     
		     \draw [<->,black, line width=0.5pt]
		     (a5) -- (a6);
		     
		     	\node at ($(a1.west) + (-0.2,0)$) 	{\small $\alpha_0$};	
		     	
		     		\node at ($(a2.east) + (0.5,0)$) 	{\small $\underline{\delta-\alpha_0}$};	
		     		
		     	\node at ($(a3.south) + (0,-0.3)$) 	{\small $\alpha_{03445}$};	
		     	
		     		\node at ($(a4.south) + (0,-0.3)$) 	{\small $\underline{\delta-\alpha_{03445}}$};	
		     		
		     		\node at ($(a5.north) + (0,0.2)$) 	{\small$\alpha_{012233445}$};	
		     		
		     		\node at ($(a6.north) + (0,0.2)$) 	{\small$\underline{\alpha_{034455667}}$};

		    \draw[<->, dashed] (a1) to [bend left=0] node[fill=white]{\small $s_{4354}$} (a3); 
		    
		     \draw[<->, dashed] (a2) to [bend left=0] node[fill=white]{\small $s_{4354}$} (a4); 
		    
		     \draw[<->, dashed] (a3) to [bend left=0] node[fill=white]{\small $s_{2132}$} (a5); 
		     
		      \draw[<->, dashed] (a4) to [bend left=0] node[fill=white]{\small $s_{2132}$} (a6);

		        \draw[<->, dashed] (a5) to [bend right=40] node[fill=white]{\small $\pi$} (a6); 
		      
		        \draw[<->, dashed] (a3) to [bend left=40] node[fill=white]{\small $s_{2132}\pi s_{2132}$} (a4); 
		     
		      \draw[<->, dashed] (a1) to [bend right=40] node[fill=white]{\small $s_{4354}s_{2132}\pi s_{2132}s_{4354}$} (a2); 
		    
			\node at (0,5) 	{$(3A_1)^{(1)}$};
			
			\node at (0,-1.5) {$(A_3^{(1)})|_{|\alpha|^2=4}$};

			 \begin{scope}[yshift=-3.75cm]
			      \path 	
			 (-1,1) 	node 	(a0) [elt] {}    
			(-1,-1) 	node 	(a1) [elt] {}
			(1,-1) 	node 	(a2) [elt] {}
			(1,1) 	node 	(a3) [elt] {}; 
			
			\draw [black,line width=1pt] (a0) -- (a1) -- (a2) -- (a3) -- (a0);
			
			 \node at ($(a0.north) + (0,0.3)$) 	{\small \parbox{1cm}{$\alpha_{003445}$\\$s_{043540}$}};	
		     	
		     \node at ($(a1.south) + (0,-0.3)$) 	{\small \parbox{1cm}{$\alpha_{1223}$\\$s_{2132}$}};
		     
		      \node at ($(a2.south) + (0.1,-0.3)$) 	{\small \parbox{1cm}{$\alpha_{3445}$\\$s_{4354}$}};
		      
		      \node at ($(a3.north) + (0.1,0.3)$) 	{\small \parbox{1cm}{$\alpha_{5667}$\\$s_{6576}$}};
		      
		      \draw[<->, dashed] (a1) to[bend right=20] node[fill=white]{\small $\pi$}  (a3);
		      
		       \draw[<->, dashed] (a2) to[bend right=20] node[fill=white]{\small $s_0$}  (a0);
			
			\end{scope}
			
			      \end{tikzpicture}
			      
			      \medskip
			      
			      $t_{\omega_{\alpha_{1223}}-\omega_{\alpha_{5667}}}\sim t_{\omega_2-\omega_6}^{1/2}$\\
			      $t_{\omega_{\alpha_{3445}}-\omega_{\alpha_{003445}}}\sim t_{\omega_4-2\omega_0}^{1/2}$

			\end{center}

			&

		 \begin{center}
			     
			      $4A_1$
			      			      
			      \medskip
			      
			      \begin{tikzpicture}[elt/.style={circle,draw=black!100,thick, inner sep=0pt,minimum size=2mm},scale=2]
			\path (0,0) node (a1) [elt] {}
			(0,1) node (a3) [elt] {}
			(0,2) node  (a5) [elt] {}
			(0,3) node  (a7) [elt] {};
		    
		    	\node at ($(a1.west) + (-0.2,0)$) 	{\small $\alpha_1$};	
		    	
		    		\node at ($(a3.west) + (-0.2,0)$) 	{\small $\alpha_3$};
		    		
		    			\node at ($(a5.west) + (-0.2,0)$) 	{\small $\alpha_5$};	
		    			
		    			\node at ($(a7.west) + (-0.2,0)$) 	{\small $\alpha_7$};

		    	\draw[<->,dashed] (a1) to[bend right=0]
			node[fill=white]{\small $s_{2132}$} (a3);
			
			\draw[<->,dashed] (a3) to[bend right=0]
			node[fill=white]{\small $s_{4354}$} (a5);
			
			\draw[<->,dashed] (a3) to[bend left=60]
			node[fill=white]{\small $\pi$} (a5);
			
			\draw[<->,dashed] (a1) to[bend right=20]
			node[fill=white]{\small $\pi$} (a7);
		    
			\end{tikzpicture} 
			    
			    \end{center}

		& 
		
		\begin{center}
		{\small
		\begin{gather*}
		 (S_3{\ltimes}C_2^3){\ltimes}(W_{3A_1}^a{\times}\underbrace{C_2^4}_{triv})\\
		 \simeq\\
		 (C_2^2{\ltimes}W_{A_3}^a){\ltimes}\underbrace{C_2^4}_{triv})\\
		 \downarrow\\
		 C_2^2{\ltimes}W_{A_3}^a\\
		 t_{\omega_2-\omega_6}^{1/2} \rightarrow t_{\omega_3-\omega_2}\\
		 t_{\omega_4-2\omega_0}^{1/2} \rightarrow t_{\omega_4-\omega_5}
		\end{gather*}
		}
		\eqref{symm_E7_D5}
		\end{center}
			
			&

			\begin{center}
			
			\hspace*{-0.3cm}\begin{tikzpicture}[elt/.style={circle,draw=black!100,thick, inner sep=0pt,minimum size=1.4mm},scale=0.4]
			\path 	(-3,0) 	node 	(a1) [elt,fill] {}
			(-2,0) 	node 	(a2) [elt] {}
			( -1,0) node  	(a3) [elt,fill] {}
			( 0,0) 	node  	(a4) [elt] {}
			( 1,0) 	node 	(a5) [elt,fill] {}
			( 2,0)	node 	(a6) [elt] {}
			( 3,0)	node 	(a7) [elt,fill] {}
			( 0,1)	node 	(a0) [elt] {};
			\draw [black,line width=1pt ] (a1) -- (a2) -- (a3) -- (a4) -- (a5) --  (a6) -- (a7) (a4) -- (a0);
			
 			\node at ($(a1.south) + (0.2,-0.3)$) 	{\tiny $-1$};	
 			\node at ($(a3.south) + (0.2,-0.3)$) 	{\tiny $-1$};	
 			\node at ($(a5.south) + (0.2,-0.3)$) 	{\tiny $-1$};	
			\node at ($(a7.south) + (0.2,-0.3)$) 	{\tiny $-1$};	
		
			\end{tikzpicture}  
			\end{center}

			\\

			     \hline
			     
			   \begin{center}


                       \includegraphics[scale=0.8]{e7_222.pdf}
			
			\medskip
			
			 $w{=}s_{12}s_{40}s_{76}$

			 \end{center}

			&
			
			\begin{center}
			
			\begin{tikzpicture}[elt/.style={circle,draw=black!100,thick, inner sep=0pt,minimum size=2mm},scale=1.25]
			\path 	(1,0) 	node 	(a1) [elt] {}
			(-1,0) 	node 	(a2) [elt] {};

		    \draw [black,line width=2.5pt ] (a1) -- (a2);
		    
		     \draw [white,line width=1.5pt ] (a1) -- (a2);
		     
		     \draw [<->,black, line width=0.5pt]
		     (a1) -- (a2);
		     
		     \node at ($(a2.south) + (-0.2,-0.3)$) 	{\small 
		     \parbox{1cm}{$\underline{\alpha_{044555667}}$\\ $s_{546}s_{057}s_{645}$}};	
		     	
		     \node at ($(a1.south) + (-0.2,-0.3)$) 	{\small 
		     \parbox{1cm}{$\alpha_{012233344}$\\ $s_{342} s_{013} s_{243}$}};
		     
		      \draw [<->, dashed]
		     (a1) to[bend right=30] node[fill=white] {\small $\pi$} (a2);
		     
		     \node at (0,1) {$A_1^{(1)}|_{|\alpha|^2=6}$};
		     
		     \end{tikzpicture}
		     
		     \medskip
                     
                     $t_{\omega_{\alpha_{012233344}}}\sim t^{1/3}_{\omega_3-\omega_5}$
		     
		     \end{center}
			
			&
			
			\begin{center}
			$3A_2$ 
			
			\medskip
			      
			      \begin{tikzpicture}[elt/.style={circle,draw=black!100,thick, inner sep=0pt,minimum size=2mm},scale=0.75]
			\path (0,0) node (a1) [elt] {}
			(0,1) node (a2) [elt] {}
			(0,3) node  (a4) [elt] {}
			(0,4) node  (a0) [elt] {}
			(0,7) node  (a6) [elt] {}
			(0,6) node  (a7) [elt] {};
			\draw [black,line width=1pt] (a1) -- (a2) (a4) -- (a0) (a6) -- (a7);
		    
		    	\node at ($(a1.south) + (0,-0.2)$) 	{\small $\alpha_1$};	
		    	
		    		\node at ($(a2.north) + (0,0.2)$) 	{\small $\alpha_2$};
		    		
		    			\node at ($(a4.south) + (0,-0.2)$) 	{\small $\alpha_4$};	
		    			
		    			\node at ($(a0.north) + (0,0.2)$) 	{\small $\alpha_0$};	
		    			
		    			\node at ($(a6.north) + (0,0.2)$) 	{\small $\alpha_6$};	
		    			
		    			\node at ($(a7.south) + (0,-0.2)$) 	{\small $\alpha_7$};

			\draw[<->, dashed] (0,3.5) to[bend left=40]
			node[fill=white]{\small $s_{546}s_{057}s_{645}$} (0,6.5);
			
			\draw[<->, dashed] (0,0.5) to[bend left=40]
			node[fill=white]{\small $s_{342} s_{013} s_{243}$} (0,3.5);
			
			\draw[<->, dashed] (0,0.5) to[bend right=40]
			node[fill=white]{$\pi$} (0,6.5);
			
			\end{tikzpicture} 
			
			\end{center}

			  &
			  
			  \begin{center}
			  
			 \begin{gather*}
			      W_{A_1}^{ae} \ltimes \underbrace{C_3^3}_{triv}\\
			      \downarrow\\
			      W_{A_1}^{ae}\\
			      t^{1/3}_{\omega_3-\omega_5}\rightarrow t_{\omega_4}
			      \end{gather*}
			      
			      \medskip
			      
			      \eqref{symm_E7_E3}
			  
			  \end{center}
			  
			  &

			  \begin{center}
			
			\hspace*{-0.3cm}\begin{tikzpicture}[elt/.style={circle,draw=black!100,thick, inner sep=0pt,minimum size=1.4mm},scale=0.4]
			\path 	(-3,0) 	node 	(a1) [elt,fill] {}
			(-2,0) 	node 	(a2) [elt,fill] {}
			( -1,0) node  	(a3) [elt] {}
			( 0,0) 	node  	(a4) [elt,fill] {}
			( 1,0) 	node 	(a5) [elt] {}
			( 2,0)	node 	(a6) [elt,fill] {}
			( 3,0)	node 	(a7) [elt,fill] {}
			( 0,1)	node 	(a0) [elt,fill] {};
			\draw [black,line width=1pt ] (a1) -- (a2) -- (a3) -- (a4) -- (a5) --  (a6) -- (a7) (a4) -- (a0);
			
 			\node at ($(a1.south) + (0.2,-0.3)$) 	{\tiny $\zeta$};	
 			\node at ($(a2.south) + (0.2,-0.3)$) 	{\tiny $\zeta$};	
 			\node at ($(a6.south) + (0.2,-0.3)$) 	{\tiny $\zeta$};	
			\node at ($(a7.south) + (0.2,-0.3)$) 	{\tiny $\zeta$};
			\node at ($(a4.south) + (0.2,-0.3)$) 	{\tiny $\zeta$};	
			\node at ($(a0.north) + (0.2,0.3)$) 	{\tiny $\zeta$};

			\draw [<->, dashed] (0,-1) to node[fill=white] {\small $s_{147}$} (0,-4);
			
			\begin{scope}[yshift=-6cm]
			\path 	(-3,0) 	node 	(a1) [elt,fill] {}
			(-2,0) 	node 	(a2) [elt,fill] {}
			( -1,0) node  	(a3) [elt] {}
			( 0,0) 	node  	(a4) [elt,fill] {}
			( 1,0) 	node 	(a5) [elt] {}
			( 2,0)	node 	(a6) [elt,fill] {}
			( 3,0)	node 	(a7) [elt,fill] {}
			( 0,1)	node 	(a0) [elt,fill] {};
			\draw [black,line width=1pt ] (a1) -- (a2) -- (a3) -- (a4) -- (a5) --  (a6) -- (a7) (a4) -- (a0);
			
 			\node at ($(a1.south) + (0.2,-0.3)$) 	{\tiny $\zeta^{-1}$};	
 			\node at ($(a2.south) + (0.2,-0.3)$) 	{\tiny $\zeta^{-1}$};	
 			\node at ($(a6.south) + (0.2,-0.3)$) 	{\tiny $\zeta^{-1}$};	
			\node at ($(a7.south) + (0.2,-0.3)$) 	{\tiny $\zeta^{-1}$};
			\node at ($(a4.south) + (0.2,-0.3)$) 	{\tiny $\zeta^{-1}$};	
			\node at ($(a0.north) + (0.2,0.3)$) 	{\tiny $\zeta^{-1}$}; 
			\end{scope}

			\end{tikzpicture}  
			\end{center}

			  \\
			     \hline
			\end{tabular}
			\end{center}


        \begin{center}
			    \begin{tabular}{|m{2.5cm}|m{4cm}|m{3.5cm}|m{3cm}|m{2.5cm}|}
			    \hline
			   \begin{center} Folding  \end{center}   & \begin{center}
			     $Q^{a,w}_{E_7^{(1)}}$  
			   \end{center}  & \begin{center} $Q^{a,\perp w}$ \end{center}
			     & \begin{center} $C(w)$  \end{center} & \begin{center} $\mathcal{A}^w, \, N_{flip}$ \end{center}  \\
			    \hline


                        \includegraphics[scale=0.8]{e7_33.pdf}
			
			\medskip
			
			$w{=}s_{321} s_{765}$
			
			{\small $=t_{\omega_6}\bar{w}$
			
			$\bar{w}=s_{321} s_{\theta 65}$}
			
			&

			\begin{center}
			 
			\begin{tikzpicture}[elt/.style={circle,draw=black!100,thick, inner sep=0pt,minimum size=2mm},scale=1.25]
			\path 	(-1,0) 	node 	(a1) [elt] {}
			(1,0) 	node 	(a2) [elt] {};

		    \draw [black,line width=2.5pt ] (a1) -- (a2);
		    
		     \draw [white,line width=1.5pt ] (a1) -- (a2);
		     
		     \draw [<->,black, line width=0.5pt]
		     (a1) -- (a2);
		     
		     \node at ($(a2.north) + (0,0.1)$) 	{\small$\underline{\delta-\alpha_{0}}$};	
		     	
		     \node at ($(a1.north) + (0,0.1)$) 	{\small $\alpha_{0}$};
		     
		       \draw [<->, dashed]
		     (a1) to[bend right=40] node[fill=white]{\small $\pi s_{4354} s_{2132} \pi s_{2132} s_{4354}$}  (a2);

		     \node at (0,1) 	{$A_1^{(1)}$};
		     
		     \node at (0,-1.5) 	{$A_1^{(1)}|_{|\alpha|^2=8}$};
		     
			\begin{scope}[yshift=-2.5cm]
			\path 	(-1,0) 	node 	(a1) [elt] {}
			(1,0) 	node 	(a2) [elt] {};

		    \draw [black,line width=2.5pt ] (a1) -- (a2);
		    
		     \draw [white,line width=1.5pt ] (a1) -- (a2);
		     
		     \draw [<->,black, line width=0.5pt]
		     (a1) -- (a2);
		     
		     \node at ($(a2.south) + (0.1,-0.4)$) {\small \parbox{2cm}{$\underline{\delta+2\alpha_{0}}$\\$s_0 \pi s_{\#} \pi s_{\#}^{-1}s_0$}};	
		     	
		     \node at ($(a1.south) + (0.3,-0.4)$) {\small \parbox{2cm}{$\delta-2\alpha_{0}$\\$\pi s_{\#} \pi s_{\#}^{-1}$}};
		     
		     \draw [<->, dashed]
		     (a1) to[bend left=40] node[fill=white]{\small $s_0$}  (a2);
		     
		     \end{scope}
		     
		     \end{tikzpicture}
		     
		      \medskip
		     
		     $t_{\delta-2\alpha_0}\sim t^{1/2}_{\omega_4-2\omega_0}$

		     \end{center}

			&
			
			 \begin{center}
			     
			     $2A_3$

			      \medskip
			      
			      \begin{tikzpicture}[elt/.style={circle,draw=black!100,thick, inner sep=0pt,minimum size=2mm},scale=1]
			\path (-1.5,1) node (a3) [elt] {}
			(0,1) node  (a2) [elt] {}
			(1.5,1) node  (a1) [elt] {}
			(-1.5,-1) node (a5) [elt] {}
			(0,-1) node  (a6) [elt] {}
			(1.5,-1) node  (a7) [elt] {};
			
			\draw (a1) -- (a2) -- (a3) (a5) -- (a6) -- (a7);

		    	\node at ($(a1.north) + (0,0.3)$) 	{\small $\alpha_1$};	
		    	\node at ($(a2.south) + (0.3,-0.2)$) 	{\small $\alpha_2$};	
		    	\node at ($(a3.north) + (0,0.3)$) 	{\small $\alpha_3$};	
		    	
		    	\node at ($(a5.south) + (0,-0.3)$) 	{\small $\alpha_5$};	
		    	\node at ($(a6.north) + (0.3,0.2)$) 	{\small $\alpha_6$};	
		    	\node at ($(a7.south) + (0,-0.3)$) 	{\small $\alpha_7$};

		    	\draw[<->, dashed] (a3) to[bend left=60] node[fill=white] {\small $s_{\#} \pi s_{\#}^{-1}$} (a1);
		    	\draw[<->, dashed] (a5) to[bend right=60] node[fill=white] {\small $s_{\#} \pi s_{\#}^{-1}$} (a7);
		    	
		    	\draw[<->, dashed] (a5) to[bend right=0] node[fill=white] {\small $\pi$} (a3);
		    	\draw[<->, dashed] (a6) to[bend right=0] node[fill=white] {\small $\pi$} (a2);
		    	\draw[<->, dashed] (a7) to[bend right=0] node[fill=white] {\small $\pi$} (a1);

			\end{tikzpicture} 
			     
			     \end{center}

			& 
			
			\begin{center}
			
			\begin{gather*}
			(\!\!\!\underbrace{C_2}_{\pi s_1s_3s_5s_7}\!\!\!\!{\times}W_{A_1}^{ae}){\ltimes}\underbrace{C_4^2}_{triv} \\
			\simeq(C_2{\times}W_{A_1}^{ae}){\ltimes}\underbrace{C_4^2}_{triv}\\
			\downarrow\\
			C_2\times W_{A_1}^{ae}\\
			t^{1/2}_{\omega_4-2\omega_0}\rightarrow t^{1/2}_{-\omega_4+2\omega_0}
			\end{gather*}

			Fig. \ref{fig:sublattices_im}
			
			\end{center}
			
			& 
			
			\begin{center}
			
			\hspace*{-0.3cm}\begin{tikzpicture}[elt/.style={circle,draw=black!100,thick, inner sep=0pt,minimum size=1.4mm},scale=0.4]
			\path 	(-3,0) 	node 	(a1) [elt,fill] {}
			(-2,0) 	node 	(a2) [elt,fill] {}
			( -1,0) node  	(a3) [elt,fill] {}
			( 0,0) 	node  	(a4) [elt] {}
			( 1,0) 	node 	(a5) [elt,fill] {}
			( 2,0)	node 	(a6) [elt,fill] {}
			( 3,0)	node 	(a7) [elt,fill] {}
			( 0,1)	node 	(a0) [elt] {};
			\draw [black,line width=1pt ] (a1) -- (a2) -- (a3) -- (a4) -- (a5) --  (a6) -- (a7) (a4) -- (a0);
			
			\node at ($(a1.south) + (0,-0.4)$) 	{\tiny $\ri$};
			\node at ($(a2.south) + (0,-0.4)$) 	{\tiny $\ri$};
			\node at ($(a3.south) + (0,-0.4)$) 	{\tiny $\ri$};
			
			\node at ($(a5.south) + (0,-0.4)$) 	{\tiny $\ri$};
			\node at ($(a6.south) + (0,-0.4)$) 	{\tiny $\ri$};
			\node at ($(a7.south) + (0,-0.4)$) 	{\tiny $\ri$};
			
			\draw[<->, dashed] (0,-0.5) to
			node[fill=white]{\small $s_{1357}$} (0,-3.5);

                         \begin{scope}[yshift=-5cm]
                         
                         \path 	(-3,0) 	node 	(a1) [elt,fill] {}
			(-2,0) 	node 	(a2) [elt,fill] {}
			( -1,0) node  	(a3) [elt,fill] {}
			( 0,0) 	node  	(a4) [elt] {}
			( 1,0) 	node 	(a5) [elt,fill] {}
			( 2,0)	node 	(a6) [elt,fill] {}
			( 3,0)	node 	(a7) [elt,fill] {}
			( 0,1)	node 	(a0) [elt] {};
			\draw [black,line width=1pt ] (a1) -- (a2) -- (a3) -- (a4) -- (a5) --  (a6) -- (a7) (a4) -- (a0);
			
			\node at ($(a1.south) + (0,-0.4)$) 	{\tiny $-\ri$};
			\node at ($(a2.south) + (0,-0.4)$) 	{\tiny $-\ri$};
			\node at ($(a3.south) + (0,-0.4)$) 	{\tiny $-\ri$};
			
			\node at ($(a5.south) + (0,-0.4)$) 	{\tiny $-\ri$};
			\node at ($(a6.south) + (0,-0.4)$) 	{\tiny $-\ri$};
			\node at ($(a7.south) + (0,-0.4)$) 	{\tiny $-\ri$};
                         \end{scope}

			\end{tikzpicture}
			
			\end{center}
			
			\\
			\cline{1-1} \cline{3-5}

%

                        \includegraphics[scale=0.8]{e7_p11111.pdf}
			
			\medskip
			
			$w{=}\pi s_0 s_1 s_3$
			
			{\small $=t_{\omega_1}\bar{w}$
			
			$t_{\omega_4{-}\omega_3}sw s^{-1}$
			$=s_{321}s_{765}$
			
			$s{=}s_{430456745230423}$}

			&
			
			\begin{center}
			 
			\begin{tikzpicture}[elt/.style={circle,draw=black!100,thick, inner sep=0pt,minimum size=2mm},scale=1.25]
			\path 	(-1,0) 	node 	(a1) [elt] {}
			(1,0) 	node 	(a2) [elt] {};

		    \draw [black,line width=2.5pt ] (a1) -- (a2);
		    
		     \draw [white,line width=1.5pt ] (a1) -- (a2);
		     
		     \draw [<->,black, line width=0.5pt]
		     (a1) -- (a2);
		     
		     \node at ($(a2.north) + (0,0.1)$) 	{\small $\underline{\delta-\alpha_{03445}}$};	
		     	
		     \node at ($(a1.north) + (0,0.1)$) 	{\small $\alpha_{03445}$};
		     
		       \draw [<->, dashed]
		     (a1) to[bend right=40] node[fill=white]{\small $\pi s_{2132} s_{6576}$}  (a2);
		     
		     \node at (0,1) 	{$A_1^{(1)}$};
		     
		      \node at (0,-1.5) 	{$A_1^{(1)}|_{|\alpha|^2=8}$};

			\begin{scope}[yshift=-2cm]
			\path 	(-1,0) 	node 	(a1) [elt] {}
			(1,0) 	node 	(a2) [elt] {};

		    \draw [black,line width=2.5pt ] (a1) -- (a2);
		    
		     \draw [white,line width=1.5pt ] (a1) -- (a2);
		     
		     \draw [<->,black, line width=0.5pt]
		     (a1) -- (a2);
		     
		     \node at ($(a1.south) + (0.1,-0.3)$) 	{\small \parbox{2cm}{$\underline{\delta-2\alpha_{03445}}$\\$s_{2132}s_{6576}$}};	
		     	
		     \node at ($(a2.south) + (0,-0.3)$) 	{\small \parbox{1cm}{$2\alpha_{03445}$\\$s_0s_{\alpha_{03445}}$}};
		     
		     \end{scope}
		     
		     \end{tikzpicture}
		     
		     \medskip
		     
		     $t_{2\omega_{2\alpha_{03445}}}\sim t^{1/2}_{\omega_4-\omega_{26}}$
		     
		     \end{center}
			
			&
			
			 \begin{center}
			     
			     $\widetilde{Q^{a,\perp w}}=(A_1)_{|\alpha|^2=8} + 5A_1$

			      \medskip
			      
			      \begin{tikzpicture}[elt/.style={circle,draw=black!100,thick, inner sep=0pt,minimum size=2mm},scale=1.25]
			\path 	(0,3) 	node 	(a0) [elt] {}
		(-1,2) 	node 	(a1) [elt] {}
		(-1,1) node  	(a3) [elt] {}
		( 1,1) 	node 	(a5) [elt] {}
		( 1,2) 	node  	(a7) [elt] {}
		(0,0) 	node  	(a26) [elt] {};
		
		\node at ($(a0.north) + (0,0.2)$) 	{\small $\alpha_0$};
		\node at ($(a1.west) + (-0.2,0)$) 	{\small $\alpha_1$};	
		\node at ($(a3.west) + (-0.2,0)$) {\small $\alpha_3$};	
		\node at ($(a5.east) + (0.2,0)$) 	{\small $\alpha_5$};	
		\node at ($(a7.east) + (0.2,0)$) 	{\small $\alpha_7$};
		\node at ($(a26.south) + (0,-0.25)$) 	{\small \parbox{2cm}{$\alpha_{1223}{-}\alpha_{5667}$\\$\pi$}};

		    	\draw[<->, dashed] (a1) to[bend left=0] node[fill=white] {\small } (a5);
		    	\draw[<->, dashed] (a3) to[bend left=0] node[fill=white] {\small $\pi s_{2132}s_{6576}$} (a7);
		    	
		    	\draw[<->, dashed] (a1) to[bend left=0] node[fill=white] {\small $\pi$} (a7);
		    	\draw[<->, dashed] (a3) to[bend left=0] node[fill=white] {\small $\pi$} (a5);

			\end{tikzpicture} 
			     
			     \end{center}

			&
			
			\begin{center}
			
			\begin{gather*}
			C_2^2\ltimes (W_{A_1}\times C_2^5)\\
			\simeq (\underbrace{C_2^2}_{\pi s_0, \pi}\!\!{\times}W_{A_1}){\ltimes}\!\!\underbrace{C_2^4}_{triv}\\
			\downarrow\\
			\underbrace{C_2^2}_{\pi s_0, \pi} \times W_{A_1}\\
			t^{1/2}_{\omega_4-\omega_{26}}\rightarrow t_{\omega_6}
			\end{gather*}

			\eqref{symm_E7_E7_c3}
			
			\end{center}
			
			& 
			
			\begin{center}
			
			\hspace*{-0.3cm}\begin{tikzpicture}[elt/.style={circle,draw=black!100,thick, inner sep=0pt,minimum size=1.4mm},scale=0.4]

                         \path 	(-3,0) 	node 	(a1) [elt,fill] {}
			(-2,0) 	node 	(a2) [elt,fill] {}
			( -1,0) node  	(a3) [elt,fill] {}
			( 0,0) 	node  	(a4) [elt] {}
			( 1,0) 	node 	(a5) [elt,fill] {}
			( 2,0)	node 	(a6) [elt,fill] {}
			( 3,0)	node 	(a7) [elt,fill] {}
			( 0,1)	node 	(a0) [elt] {};
			\draw [black,line width=1pt ] (a1) -- (a2) -- (a3) -- (a4) -- (a5) --  (a6) -- (a7) (a4) -- (a0);
			
			\node at ($(a1.south) + (0,-0.4)$) 	{\tiny $-1$};
			\node at ($(a3.north) + (0,0.4)$) 	{\tiny $-1$};
			\node at ($(a0.north) + (0,0.4)$) 	{\tiny $-1$};
			\node at ($(a5.north) + (0,0.4)$) 	{\tiny $-1$};
			\node at ($(a7.south) + (0,-0.4)$) 	{\tiny $-1$};
                        
                        \draw[<->] (a2) edge[bend right=50] node[fill=white]{$=$} (a6);

			\end{tikzpicture}
			
			\end{center}
			
			\\
			 
			 \hline
        
        \end{tabular}
        
        \end{center}

	\newpage			
	
	\subsection{$E_6^{(1)}/A_2^{(1)}$}
	
	\label{e6_a}
	
	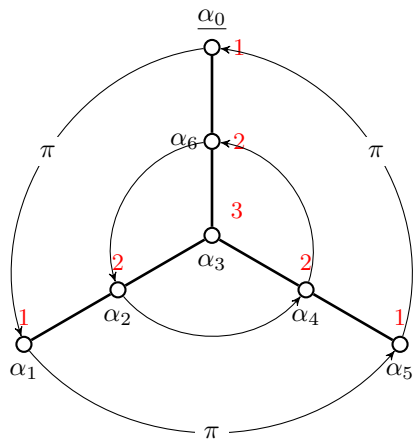
\begin{figure}[h]
	
	\begin{center}
	\begin{tikzpicture}[elt/.style={circle,draw=black!100,thick, inner sep=0pt,minimum size=2mm},scale=1.25]
		\path 	(-2,-1.16) 	node 	(a1) [elt] {}
		(-1,-0.58) 	node 	(a2) [elt] {}
		( 0,0) node  	(a3) [elt] {}
		( 1,-0.58) 	node  	(a4) [elt] {}
		( 2,-1.16) 	node 	(a5) [elt] {}
		( 0,1)	node 	(a6) [elt] {}
		( 0,2)	node 	(a0) [elt] {};
		\draw [black,line width=1pt ] (a1) -- (a2) -- (a3) -- (a4) -- (a5)   (a3) -- (a6) -- (a0);
		\node at ($(a1.south) + (0,-0.2)$) 	{$\alpha_{1}$};
		\node at ($(a2.south) + (0,-0.2)$)  {$\alpha_{2}$};
		\node at ($(a3.south) + (0,-0.2)$)  {$\alpha_{3}$};
		\node at ($(a4.south) + (0,-0.2)$)  {$\alpha_{4}$};	
		\node at ($(a5.south) + (0,-0.2)$)  {$\alpha_{5}$};		
		\node at ($(a6.west) + (-0.2,0)$) 	{$\alpha_{6}$};	
		\node at ($(a0.north) + (0,0.2)$) 	{$\underline{\alpha_{0}}$};		
		\node at ($(a0.east) + (0.2,0)$) 	{\color{red}\small$1$};
		\node at ($(a1.north) + (0,0.2)$) 	{\color{red}\small$1$};
		\node at ($(a2.north) + (0,0.2)$) 	{\color{red}\small$2$};
		\node at ($(a3.north east) + (0.2,0.2)$) 	{\color{red}\small$3$};
		\node at ($(a4.north) + (0,0.2)$) {\color{red}\small$2$};
		\node at ($(a5.north) + (0,0.2)$) 	{\color{red}\small$1$};
		\node at ($(a6.east) + (0.2,0)$) 	{\color{red}\small$2$};
		
		\draw[->] (a1) edge[bend right=50] node[fill=white]{$\pi$} (a5);
				\draw[->] (a2) edge[bend right=50] node[]{} (a4);
				\draw[->] (a5) edge[bend right=50] node[fill=white]{$\pi$} (a0);
				\draw[->] (a4) edge[bend right=50] node[]{} (a6);
				\draw[->] (a0) edge[bend right=50] node[fill=white]{$\pi$} (a1);
				\draw[->] (a6) edge[bend right=50] node[]{} (a2);
		
		\end{tikzpicture}
	
	\end{center}
	
	\caption{Root numbering and marks for $E_6^{(1)}$ Dynkin diagram}
	        
	        \label{fig:E6_rn}
                
		\end{figure}

	Affine Weyl group $W(E_6^{(1)})$ (and corresponding Dynkin diagram) has outer automorphism of order $3$.
	Diagram has also outer automorphism of order $2$ (vertical reflection), which is absent in $W(E_6^{(1)})$. 
			
		For $E_6^{(1)}$ we have $3$ foldings: 
		\begin{itemize}
		\item Two foldings of order $3$, of case $1$ and $2$, which have common (up to conjugation)
		$\bar{w}$.
		\item Folding of order $2$, of case $2$.
		\end{itemize}
	
	For folding $\pi$ we have two peculiarities:
	\begin{itemize}
	 \item 
	 First, to obtain root subsystem realization, we take index $3$ sublattice $\widetilde{Q^{a,\perp w}}$ of $Q^{a, \perp w}$
	 defined as follows.
	 $Q^{a, \perp w}$ are defined as $(\alpha=x_0 \alpha_0+x_1 \alpha_1+x_5\alpha_5+y_2\alpha_2+y_4\alpha_4+y_6\alpha_6\in Q|
	 \sum x_i=0, \, \sum y_i=0)$. 
	 Then $\widetilde{Q^{a,\perp w}}$ are defined as $(\alpha \in Q^{a,\perp w}|x_1+x_2-x_4-x_5=0 \mod 3)$.
	
	 \item After we find convenient sublattice $\widetilde{Q^{a,\perp w}}=2A_2$,
	 we discover, that its simple reflections cannot be realized inside $W^{ae}$, 
	 from $W_{A_2^2}$ we can realize only $\Omega_{A_2}^2$.
	 This realization is written at the table.
	 
	\end{itemize}

	\begin{center}
			\begin{tabular}{|m{2.5cm}|m{3.5cm}|m{3cm}|m{3cm}|m{2.5cm}|}
			\hline
			  \begin{center} Folding \end{center}    & \begin{center}
			     $Q^{a,w}$  
			   \end{center}  & \begin{center}  $Q^{a,\perp w}$ \end{center}		                                                  
			     & \begin{center} $C(w)$ \end{center} & \begin{center} $\mathcal{A}^w, \, N_{flip}$ \end{center}\\
			    \hline
			    
			    \begin{center}
		
		    \includegraphics[scale=1]{e6_p.pdf}
		
			\medskip
		
		$\pi=t_{\omega_5}\bar{\pi}$
		
		$\bar{\pi}=s s_{1245} s^{-1}$,
		
		$s=s_{123543263}$
		
		\end{center}

			    &
			    
			      \begin{center}

			      \begin{tikzpicture}[elt/.style={circle,draw=black!100,thick, inner sep=0pt,minimum size=1.4mm},scale=1]
			      \path 	
			(-1,0) 	node 	(a1) [elt] {}
			(1,0) 	node 	(a2) [elt] {}
			(0,1.73) 	node 	(a0) [elt] {};
			\draw [black,line width=1pt] (a0) -- (a1) -- (a2) -- (a0);
			
			 \node at ($(a0.north) + (0,0.3)$) 	{\small $\underline{\alpha_{0123456}}$ };	
		     	
		     \node at ($(a1.west) + (-0.2,0)$) {\small $\alpha_3$};
		     
		      \node at ($(a2.east) + (0.4,0)$) {\small $\alpha_{2346}$};
			
			\draw[<->,dashed] (a1) to[bend right=60]
			node[fill=white]{\small $s_{246}$} (a2);
			
			\draw[<->,dashed] (a2) to[bend right=60]
			node[fill=white]{\small $s_{015}$} (a0);
			
			     \node at (0, 2.75) {$A_2^{(1)}$};
			     
			     \node at (0,-1.5) {$(A_2^{(1)})_{|\alpha|^2=6}$};
			      
			      \begin{scope}[yshift=-4.5cm]
			      \path 	
			(-1,0) 	node 	(a1) [elt] {}
			(1,0) 	node 	(a2) [elt] {}
			(0,1.73) 	node 	(a0) [elt] {};
			\draw [black,line width=1pt] (a0) -- (a1) -- (a2) -- (a0);
			
			 \node at ($(a0.north) + (0.1,0.3)$) 	{\small \parbox{1cm}{$\alpha_{015}$\\$s_{015}$} };	
		     	
		     \node at ($(a1.south) + (0,-0.3)$) {\small \parbox{1cm}{$\alpha_{233346}$\\$s_{32463}$}};
		     
		      \node at ($(a2.south) + (0.1,-0.3)$) {\small \parbox{1cm}{$\alpha_{246}$\\$s_{246}$}};
			
			\draw[<->,dashed] (a1) to[bend left=30]
			node[fill=white]{\small $s_3$} (a2);
			
			    \end{scope}
			
			      \end{tikzpicture}
			      
			      \medskip
			      
			      $t_{\omega_{\alpha_{2346}}}\sim t^{1/3}_{\omega_{246}-2\omega_{15}}$
			      
			      \end{center}

			    &
			    
			   \begin{center}
			    
			    $\widetilde{Q^{a,\perp w}}=(2A_2)|_{|\alpha|^2=6}$
			    
			    	      \begin{tikzpicture}[elt/.style={circle,draw=black!100,thick, inner sep=0pt,minimum size=2mm},scale=1]
			\path (0,0) node (a1) [elt] {}
			(0,2) node (a2) [elt] {}
			(0,3) node  (a4) [elt] {}
			(0,5) node  (a5) [elt] {};
			\draw [black,line width=1pt] (a1) -- (a2) (a4) -- (a5);
		    
		    	\node at ($(a1.south) + (0,-0.2)$) 	{\small $\alpha_{14}{-}\alpha_{06}$};	
		    	
		    		\node at ($(a2.north) + (0.5,0.2)$) 	{\small $\alpha_{56}{-}\alpha_{12}$};
		    		
		    			\node at ($(a4.south) + (0.5,-0.2)$) 	{\small $\alpha_{12}{-}\alpha_{04}$};	
		    			
		    			\node at ($(a5.north) + (0,0.2)$) 	{\small $\alpha_{45}{-}\alpha_{16}$};

			\draw[<->, dashed] (0,1) to[bend left=40]
			node[fill=white]{\small $s_{246}$} (0,4);
			
			\node at (0.5,1-\l) 	{\small $\Omega_{A_2}=\langle \pi s_{015} s_{246} \rangle$};	
			
			\node at (0.5,4+\l) 	{\small $\Omega_{A_2}=\langle \pi s_{246} s_{015} \rangle$};	
			
			\end{tikzpicture}

			    \end{center}
			    
			    &
			    
			    \begin{center}
			    
			     \begin{gather*}
			     W^{ae}_{A_2}\times \underbrace{C_3}_{triv}\\
			     \simeq (C_2{\ltimes}W^a_{A_2}){\times}\underbrace{C_3}_{triv}\\
			     \downarrow\\
			     C_2 \ltimes W^a_{A_2}\\
			     t^{1/3}_{\omega_{246}-2\omega_{15}}\rightarrow t_{\omega_7}
			     \end{gather*}

			     \eqref{symm_E6_E8_c1}
			     
			     \end{center}
			    
			    & 
			    
  			      \begin{center}
  			    \begin{tikzpicture}[elt/.style={circle,draw=black!100,thick, inner sep=0pt,minimum size=1.4mm},scale=0.4]
  			\path 	(-2,-1.16) 	node 	(a1) [elt] {}
  			(-1,-0.58) 	node 	(a2) [elt] {}
  			( 0,0) node  	(a3) [elt] {}
  			( 1,-0.58) 	node  	(a4) [elt] {}
  			( 2,-1.16) 	node 	(a5) [elt] {}
  			( 0,1)	node 	(a6) [elt] {}
  			( 0,2)	node 	(a0) [elt] {};
  			\draw [black,line width=1pt ] (a1) -- (a2) -- (a3) -- (a4) -- (a5)   (a3) -- (a6) -- (a0);	
  				\draw[->] (a1) edge[bend right=40] node[fill=white]{$=$} (a5);
   				\draw[->] (a2) edge[bend right=40] node[]{} (a4);
   				\draw[->] (a5) edge[bend right=40] node[fill=white]{$=$} (a0);
   				\draw[->] (a4) edge[bend right=40] node[]{} (a6);
   				\draw[->] (a0) edge[bend right=40] node[fill=white]{$=$} (a1);
   				\draw[->] (a6) edge[bend right=40] node[]{} (a2);
  		\end{tikzpicture}
  		\end{center}
			    
			    \\
			    
			    \cline{1-1} \cline{3-5}
			     
			    \begin{center}

%

               \includegraphics[scale=1]{e6_22.pdf}
		
		\medskip
		
		$w=s_{1245}$
		
		\end{center}
			     
			    &
			    
			     \begin{center}

			      \begin{tikzpicture}[elt/.style={circle,draw=black!100,thick, inner sep=0pt,minimum size=1.4mm},scale=1]
			      \path 	
			(-1,0) 	node 	(a1) [elt] {}
			(1,0) 	node 	(a2) [elt] {}
			(0,1.73) 	node 	(a0) [elt] {};
			\draw [black,line width=1pt] (a0) -- (a1) -- (a2) -- (a0);
			
			 \node at ($(a0.north) + (0,0.3)$) 	{\small $\underline{\alpha_0}$ };	
		     	
		     \node at ($(a1.north) + (-0.2,0.3)$) 	{\small $\alpha_6$};
		     
		      \node at ($(a2.north) + (0,0.3)$) 	{\small $\alpha_{1223334456}$};

			\draw[<->,dashed] (a1) to[bend right=60]
			node[fill=white]{\small $s=s_{324} s_{135} s_{423}$} (a2);
			      
			      \node at (0, 2.75) {$A_2^{(1)}$};
			     
			     \node at (0,-1.5) {$(A_2^{(1)})_{|\alpha|^2=6}$};

			      \begin{scope}[yshift=-4.5cm]
			      \path 	
			(-1,0) 	node 	(a1) [elt] {}
			(1,0) 	node 	(a2) [elt] {}
			(0,1.73) 	node 	(a0) [elt] {};
			\draw [black,line width=1pt] (a0) -- (a1) -- (a2) -- (a0);
			
			 \node at ($(a0.north) + (0.4,0.2)$) 	{\small \parbox{2cm}{$\delta{-}\alpha_{066}$\\$s$}};	
		     	
		     \node at ($(a1.south) + (0.3,-0.3)$) 	{\small \parbox{2cm}{$\delta{-}\alpha_0{+}\alpha_6$\\$s_6 s s_6$}};
		     
		      \node at ($(a2.south) + (0.4,-0.3)$) 	{\small \parbox{2cm}{$\delta{+}\alpha_{006}$\\$s_{06} s s_{60}$}};

			\draw[<->,dashed] (a0) to[bend left=60]
			node[fill=white]{\small $s_{060}$} (a2);
			
			\draw[<->,dashed] (a1) to[bend left=40]
			node[fill=white]{\small $s_0$} (a2);
			
			   \end{scope}
			
			      \end{tikzpicture}
			      
			      \end{center}

			    &
			    
			  \begin{center}
			$2A_2$ 
			
			\medskip
			      
			      \begin{tikzpicture}[elt/.style={circle,draw=black!100,thick, inner sep=0pt,minimum size=2mm},scale=0.75]
			\path (0,0) node (a1) [elt] {}
			(0,1) node (a2) [elt] {}
			(0,3) node  (a4) [elt] {}
			(0,4) node  (a5) [elt] {};
			\draw [black,line width=1pt] (a1) -- (a2) (a4) -- (a5);
		    
		    	\node at ($(a1.south) + (0,-0.2)$) 	{\small $\alpha_1$};	
		    	
		    		\node at ($(a2.north) + (0,0.2)$) 	{\small $\alpha_2$};
		    		
		    			\node at ($(a4.south) + (0,-0.2)$) 	{\small $\alpha_4$};	
		    			
		    			\node at ($(a5.north) + (0,0.2)$) 	{\small $\alpha_5$};

			\draw[<->, dashed] (0,0.5) to[bend left=40]
			node[fill=white]{\small $s_{324} s_{135} s_{423}$} (0,3.5);
			
			\end{tikzpicture} 
			
			\end{center}

			    & 
			    
			    \begin{center}
			    
			    \begin{gather*}
			    C_2 \ltimes (W^a_{A_2} \times \underbrace{C_2^2}_{triv}) \\
			    \simeq (S_3{\ltimes}W^a_{A_2}){\ltimes}\underbrace{C_2^2}_{triv})\\
			    \downarrow\\
 			    S_3 \ltimes W^a_{A_2}\\
 			     \textrm{no proj. red.}\\
			      \textrm{in the preimage}
			    \end{gather*}

			    \eqref{symm_E6_E8_c2}
			    
			    \end{center}
			    
			    & 
			    
			    \begin{center}
			    
			      \begin{tikzpicture}[elt/.style={circle,draw=black!100,thick, inner sep=0pt,minimum size=1.4mm},scale=0.4]
			\path 	(-2,-1.16) 	node 	(a1) [elt,fill] {}
			(-1,-0.58) 	node 	(a2) [elt,fill] {}
			( 0,0) node  	(a3) [elt] {}
			( 1,-0.58) 	node  	(a4) [elt,fill] {}
			( 2,-1.16) 	node 	(a5) [elt,fill] {}
			( 0,1)	node 	(a6) [elt] {}
			( 0,2)	node 	(a0) [elt] {};
			\draw [black,line width=1pt ] (a1) -- (a2) -- (a3) -- (a4) -- (a5)   (a3) -- (a6) -- (a0);
			
			\node at ($(a1.south) + (0.2,-0.3)$) 	{\tiny $\zeta$};	
 			\node at ($(a2.south) + (0.2,-0.3)$) 	{\tiny $\zeta$};	
 			\node at ($(a4.south) + (0.2,-0.3)$) 	{\tiny $\zeta$};	
			\node at ($(a5.south) + (0.2,-0.3)$) 	{\tiny $\zeta$};

			\draw [<->,dashed] (0,-2) to node[fill=white] {\small $s_{15}$} (0,-5);
			
			\begin{scope}[yshift=-8cm]
			\path 	(-2,-1.16) 	node 	(a1) [elt,fill] {}
			(-1,-0.58) 	node 	(a2) [elt,fill] {}
			( 0,0) node  	(a3) [elt] {}
			( 1,-0.58) 	node  	(a4) [elt,fill] {}
			( 2,-1.16) 	node 	(a5) [elt,fill] {}
			( 0,1)	node 	(a6) [elt] {}
			( 0,2)	node 	(a0) [elt] {};
			\draw [black,line width=1pt ] (a1) -- (a2) -- (a3) -- (a4) -- (a5)   (a3) -- (a6) -- (a0);
			
			\node at ($(a1.south) + (0.2,-0.3)$) 	{\tiny $\zeta^{-1}$};	
 			\node at ($(a2.north) + (0.2,0.3)$) 	{\tiny $\zeta^{-1}$};	
 			\node at ($(a4.north) + (0.2,0.3)$) 	{\tiny $\zeta^{-1}$};	
			\node at ($(a5.south) + (0.2,-0.3)$) 	{\tiny $\zeta^{-1}$};	 
			\end{scope}

		\end{tikzpicture} 
		
		\end{center}

			    \\
			    
			    \hline
			    
	                 \end{tabular}
	        \end{center}

	\begin{center}
			\begin{tabular}{|m{2.5cm}|m{3.5cm}|m{4cm}|m{2.5cm}|m{2.5cm}|}
			\hline
			  \begin{center} Folding \end{center}    & \begin{center}
			     $Q^{a,w}$  
			   \end{center}  & \begin{center}  $Q^{a,\perp w}$ \end{center}
			     & \begin{center} $C(w)$ \end{center}  & \begin{center} $\mathcal{A}^w$ \end{center}\\

			     \hline
			     
			     \begin{center}


                           \includegraphics[scale=1]{e6_1111.pdf}
			     
			     \medskip
			     
			      $w=s_{0135}$
			   			      
			      \end{center}

			      &

			     \begin{center}

			      \begin{tikzpicture}[elt/.style={circle,draw=black!100,thick, inner sep=0pt,minimum size=1.4mm},scale=1]
			      \path 	
			(-1,0) 	node 	(a1) [elt] {}
			(1,0) 	node 	(a2) [elt] {}
			(0,1.73) 	node 	(a0) [elt] {};
			\draw [black,line width=1pt] (a0) -- (a1) -- (a2) -- (a0);
			
			 \node at ($(a1.south) + (0,-0.4)$) 	{\small \parbox{1cm}{$\underline{\alpha_{3660}}$\\ $s_{6036}$} };	
		     	
		     \node at ($(a0.north) + (0,0.4)$) 	{\small \parbox{1cm}{$\alpha_{1223}$\\ $s_{2132}$}};
		     
		      \node at ($(a2.south) + (0.2,-0.4)$) 	{\small \parbox{1cm}{$\alpha_{3445}$\\ $s_{4354}$}};

			\draw[<-,dashed] (a2) edge[bend right=40] node[fill=white]{\small $\pi$} (a0);
			\draw[<-,dashed] (a0) edge[bend right=40] node[fill=white]{\small $\pi$} (a1);
			\draw[<-,dashed] (a1) edge[bend right=40] node[fill=white]{\small $\pi$} (a2);

			   \node at (0,3) 	{$(A_2^{(1)})_{|\alpha|^2=4}$};
			  
			      \end{tikzpicture}
			      
			      \medskip
			      
			      $t_{\omega_{\alpha_{1223}}-\omega_{\alpha_{3445}}}\sim t^{1/2}_{\omega_2-\omega_4}$\\
			      $t_{\omega_{\alpha_{3445}}}\sim t^{1/2}_{\omega_4-\omega_6}$
			      
			      \end{center}
			      &
			      
			     \begin{center}
			     
			      $4A_1$
			      			      
			      \medskip
			      
			      \begin{tikzpicture}[elt/.style={circle,draw=black!100,thick, inner sep=0pt,minimum size=2mm},scale=0.75]
			\path (0,0) node (a3) [elt] {}
			(0,2) node (a0) [elt] {}
			(-2,-1.16) node  (a1) [elt] {}
			(2,-1.16) node  (a5) [elt] {};

		    		\node at ($(a3.west) + (-0.2,0.1)$) 	{\small $\alpha_3$};
		    		
		    		\node at ($(a0.west) + (-0.2,0.1)$) 	{\small $\alpha_0$};	
		    		
		    		\node at ($(a1.west) + (-0.3,0)$) 	{\small $\alpha_1$};	
		    		
		    			\node at ($(a5.east) + (0.3,0)$) 	{\small $\alpha_5$};	
		    			
		    	\draw[<->,dashed] (a1) to[bend right=0]
			node[fill=white]{\small $s_{2132}$} (a3);
			
			\draw[<->,dashed] (a3) to[bend right=0]
			node[fill=white]{\small $s_{4354}$} (a5);
			
			\draw[<->,dashed] (a3) to[bend right=0]
			node[fill=white]{\small $s_{6036}$} (a0);
			
			\draw[->,dashed] (a5) to[bend right=40]
			node[fill=white]{\small $\pi$} (a0);
			
			\draw[->,dashed] (a0) to[bend right=40]
			node[fill=white]{\small $\pi$} (a1);
			
			\draw[->,dashed] (a1) to[bend right=40]
			node[fill=white]{\small $\pi$} (a5);
		    
			\end{tikzpicture} 
			    
			    \end{center}
			      
			      & 
			      \begin{center}

			      \begin{gather*}
			      W_{A_2}^{ae} \ltimes \underbrace{C_2^4}_{triv}\\
			      \downarrow\\ 
			      W_{A_2}^{ae}\\
			      t^{1/2}_{\omega_2-\omega_4}\rightarrow t_{-\omega_2}\\
			      t^{1/2}_{\omega_4-\omega_6}\rightarrow t_{\omega_2-\omega_1}
			      \end{gather*}

			      \eqref{symm_E6_E3}
			      
			      \end{center}
			      
			      & 
			      
			      \begin{center}
			        \begin{tikzpicture}[elt/.style={circle,draw=black!100,thick, inner sep=0pt,minimum size=1.4mm},scale=0.4]
			\path 	(-2,-1.16) 	node 	(a1) [elt,fill] {}
			(-1,-0.58) 	node 	(a2) [elt] {}
			( 0,0) node  	(a3) [elt,fill] {}
			( 1,-0.58) 	node  	(a4) [elt] {}
			( 2,-1.16) 	node 	(a5) [elt,fill] {}
			( 0,1)	node 	(a6) [elt] {}
			( 0,2)	node 	(a0) [elt,fill] {};
			\draw [black,line width=1pt ] (a1) -- (a2) -- (a3) -- (a4) -- (a5)   (a3) -- (a6) -- (a0);	
			
			\node at ($(a1.south) + (0.2,-0.3)$) 	{\tiny $-1$};	
 			\node at ($(a3.south) + (0.1,-0.4)$) 	{\tiny $-1$};	
 			\node at ($(a0.north) + (0,0.2)$) 	{\tiny $-1$};	
			\node at ($(a5.south) + (0.2,-0.3)$) 	{\tiny $-1$};	
			
		\end{tikzpicture} 
			      \end{center}  \\
			      \hline

	\end{tabular}
	\end{center}

	\subsection{$D_5^{(1)}/A_3^{(1)}$}
	
	\label{d5_a}
	
	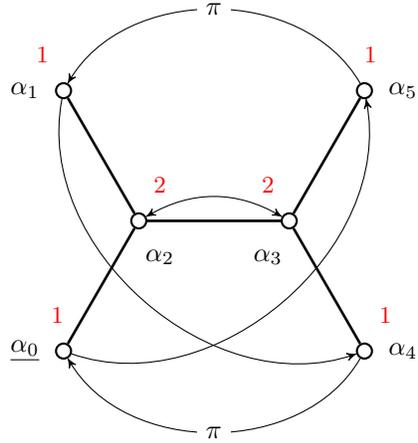
\begin{figure}[h]
	
	\begin{center}
	\begin{tikzpicture}[elt/.style={circle,draw=black!100,thick, inner sep=0pt,minimum size=2mm},scale=2]
		\path 	(-cos{60},sin{60}) 	node 	(a1) [elt] {}
		(-cos{60},-sin{60}) 	node 	(a0) [elt] {}
		( 0,0) node  	(a2) [elt] {}
		( 1,0) 	node  	(a3) [elt] {}
		( 1+cos{60},-sin{60}) 	node 	(a4) [elt] {}
		( 1+cos{60},sin{60})	node 	(a5) [elt] {};
		\draw [black,line width=1pt] (a1) -- (a2) -- (a3) -- (a4) (a3) -- (a5) (a2) -- (a0);
		\node at ($(a0.west) + (-0.2,0)$) 	{$\underline{\alpha_{0}}$};			    
		\node at ($(a1.west) + (-0.2,0)$) 	{$\alpha_{1}$};
		\node at ($(a2.south east) + (0.1,-0.2)$)  {$\alpha_{2}$};
		\node at ($(a3.south west) + (-0.1,-0.2)$)  {$\alpha_{3}$};
		\node at ($(a4.east) + (0.2,0)$)  {$\alpha_{4}$};	
		\node at ($(a5.east) + (0.2,0)$)  {$\alpha_{5}$};		
		\node at ($(a0.north west) + (-0,0.2)$) 	{\color{red}\small$1$};
		\node at ($(a1.north west) + (-0.1,0.2)$) 	{\color{red}\small$1$};
		\node at ($(a2.north east) + (0.1,0.2)$) 	{\color{red}\small$2$};
		\node at ($(a3.north west) + (-0.1,0.2)$) 	{\color{red}\small$2$};
		\node at ($(a4.north east) + (0.1,0.2)$)   {\color{red}\small$1$};
		\node at ($(a5.north east) + (0,0.2)$) 	{\color{red}\small$1$};
		\draw[->] (a0) edge[bend right=60] node[]{} (a5);
			\draw[->] (a1) edge[bend right=60] node[]{} (a4);
			\draw[->] (a4) edge[bend left=60] node[fill=white]{$\pi$} (a0);
			\draw[->] (a5) edge[bend right=60] node[fill=white]{$\pi$} (a1);
			\draw[<->] (a2) edge[bend left=30] node[]{} (a3);
		
		\end{tikzpicture}
		\end{center}
		
	  \caption{Root numbering and marks for $D_5^{(1)}$ Dynkin diagram}
	        
	        \label{fig:D5_rn}
                
		\end{figure}
			
		Affine Weyl group $W(D_5^{(1)})$ (and corresponding Dynkin diagram) has outer automorphism of order $4$.
		Diagram has also outer automorphism of order $2$ (transposition of $\alpha_0$ and $\alpha_1$), which is absent in $W(D_5^{(1)})$.

		For $D_5^{(1)}$ we have $6$ foldings: 
		\begin{itemize}
		\item Three foldings of order $4$, of case $1$, $2$ and $3$, which have common (up to conjugation)
		$\bar{w}$.
		\item Two foldings of order $2$, of case $1$, which is square of case $1$ folding of order $4$,
		and of case $2$, which is (up to conjugation) square of both other foldings of order $4$.
		$\bar{w}$.
		\item Another folding of order $2$, of case $2$.
		\end{itemize}

	\begin{center}
	\begin{tabular}{|m{2.5cm}|m{5cm}|m{2.5cm}|m{3cm}|m{2.5cm}|}
			\hline
			  \begin{center} Folding  \end{center}   & \begin{center}
			     $Q^{a,w}$  
			   \end{center}  & \begin{center}  $Q^{a,\perp w}$ \end{center}		                                                  
			     & \begin{center} $C(w)$ \end{center} & \begin{center} $\mathcal{A}^w$ \end{center}
			                                                           \\
			    \hline
			    
			    \begin{center}
%

                         \includegraphics[scale=0.8]{d5_pp.pdf}
			
			$\pi^2=t_{\omega_1}\bar{\pi}^2$
			
			{\small $\bar{\pi}^2{=}s_{123} s_4s_5 s_{123}^{-1}$}

			    \end{center}

			    &
			    
                            \begin{center}
			  
			    \begin{tikzpicture}[elt/.style={circle,draw=black!100,thick, inner sep=0pt,minimum size=1.4mm},scale=1]
			      \path 	
			 (-1,1) 	node 	(a0) [elt] {}    
			(-1,-1) 	node 	(a1) [elt] {}
			(1,-1) 	node 	(a2) [elt] {}
			(1,1) 	node 	(a3) [elt] {}; 
			
			\draw [black,line width=1pt] (a0) -- (a1) -- (a2) -- (a3) -- (a0);
			
			 \node at ($(a0.north) + (0,0.2)$) 	{\small $\underline{\alpha_{012}}$};	
		     	
		     \node at ($(a1.south) + (0,-0.2)$) 	{\small $\alpha_{3}$};
		     
		      \node at ($(a2.south) + (0,-0.2)$) 	{\small $\alpha_{2}$};
		      
		      \node at ($(a3.north) + (0,0.2)$) 	{\small $\alpha_{345}$};
		      
		       \draw[->,dashed] (a0) edge[bend left=40] node[fill=white]{\small $\pi s_0s_1$} (a3);
			  
			   \draw[->,dashed] (a3) edge[bend left=40] node[fill=white]{\small $\pi s_0s_1$} (a2);
			   
			   \draw[->,dashed] (a2) edge[bend left=40] node[fill=white]{\small $\pi s_0s_1$} (a1);
			   
			   \draw[->,dashed] (a1) edge[bend left=40] node[fill=white]{\small $\pi s_0s_1$} (a0);
			   
			    \draw[<->,dashed] (a2) edge[bend right=0] node[fill=white]{\small $s_0s_1$} (a0);
			
			      \node at (0,2) {$A_3^{(1)}$};
			      
			      \node at (0,-2.5) { $((3A_1)^{(1)})|_{|\alpha|^2=4}$};
		
		    \begin{scope}[yshift=-4.5cm]
			\path 	(2,-1) 	node 	(a1) [elt] {}
			(2,1) 	node 	(a2) [elt] {}
			(0,-1) node  	(a3) [elt] {}
			(0,1) 	node  	(a4) [elt] {}
			(-2,-1) 	node  	(a5) [elt] {}
			(-2,1) 	node  	(a6) [elt] {};
			
		    \draw [black,line width=2.5pt ] (a1) -- (a2);
		    
		     \draw [white,line width=1.5pt ] (a1) -- (a2);
		     
		     \draw [<->,black, line width=0.5pt]
		     (a1) -- (a2);
		     
		      \draw [black,line width=2.5pt ] (a3) -- (a4);
		    
		     \draw [white,line width=1.5pt ] (a3) -- (a4);
		     
		     \draw [<->,black, line width=0.5pt]
		     (a3) -- (a4);
		     
		      \draw [black,line width=2.5pt ] (a5) -- (a6);
		    
		     \draw [white,line width=1.5pt ] (a5) -- (a6);
		     
		     \draw [<->,black, line width=0.5pt]
		     (a5) -- (a6);
		     
		     	\node at ($(a1.south) + (0.1,-0.2)$) 	{\small \parbox{1cm}{$\alpha_{01}$\\$s_0s_1$}};	
		     	
		     		\node at ($(a2.north) + (0,0.1)$) 	{\small $\underline{\alpha_{223345}}$};	
		     		
		     	\node at ($(a3.south) + (0,-0.1)$) 	{\small $\alpha_{0122}$};	
		     	
		     		\node at ($(a4.north) + (0,0.1)$) 	{\small $\underline{\alpha_{3345}}$};	
		     		
		     		\node at ($(a5.south) + (0,-0.1)$) 	{\small $\underline{\alpha_{012233}}$};	
		     		
		     		\node at ($(a6.north) + (0,0.1)$) 	{\small $\alpha_{45}$};

		    \draw[<->, dashed] (a1) to [bend left=0] node[fill=white]{\small $s_2$} (a3); 
		    
		     \draw[<->, dashed] (a2) to [bend left=0] node[fill=white]{\small $s_2$} (a4); 
		    
		     \draw[<->, dashed] (a3) to [bend left=0] node[fill=white]{\small $s_3$} (a5); 
		     
		      \draw[<->, dashed] (a4) to [bend left=0] node[fill=white]{\small $s_3$} (a6);

		        \draw[<->, dashed] (a5) to [bend right=40] node[fill=white]{\small $(\pi s_3)^3$} (a6); 
		      
		        \draw[<->, dashed] (a3) to [bend left=40] node[fill=white]{\small $(\pi s_3)^3$} (a4); 
		     
		      \draw[<->, dashed] (a1) to [bend left=40] node[fill=white]{\small $(\pi s_3)^3$} (a2); 
		    
		        \end{scope}
		    
			\end{tikzpicture} 
			    
			    $t_{\omega_{\alpha_{45}}}\sim t^{1/2}_{\omega_{45}-\omega_1}$
			
			    \end{center}

			    &
			    $(2A_1)|_{|\alpha|^2=4}$
			    
			    \medskip
			    
			    \begin{center}
			      \begin{tikzpicture}[elt/.style={circle,draw=black!100,thick, inner sep=0pt,minimum size=1.4mm},scale=1]
			       \path 	
			 (0,-1) 	node 	(a4) [elt] {}    
			(0,1) 	node 	(a5) [elt] {};

			    \node at ($(a4.south)+(0,-0.4)$) {\small \parbox{1cm}{$\alpha_4-\alpha_5$\\$s_4s_5$} };
			    
			    \node at ($(a5.north)+(0,0.4)$) {\small \parbox{1cm}{$\alpha_0-\alpha_1$\\$s_0s_1$} };
			      
			    \draw[<->,dashed] (a4) edge[bend right=0] node[fill=white]{\small $\pi s_0 s_1$} (a5);

			    \end{tikzpicture}
			    
			    \medskip
			    
			    $\pi^2$ is central $-1$.
			    
			    \end{center}
						    
			    & 
			    \begin{center}
			    
			    \begin{gather*}
			     C_2\times W_{A_3}^{ae}\times \underbrace{C_2}_{triv}\\
			     \simeq {\scriptstyle ((S_3{\ltimes}C_2^3){\ltimes}W_{3A_1}^a){\times}\!\!\underbrace{\scriptstyle C_2}_{triv}}\\
			     \downarrow\\
			     (S_3{\ltimes}C_2^3){\ltimes}W_{3A_1}^a\\
			      t^{1/2}_{\omega_{45}-\omega_1}\rightarrow t_{\omega_1}
			    \end{gather*}

			    \eqref{symm_D5_E7_c1}
                             
                            \end{center}

			    & 
			       \begin{tikzpicture}[elt/.style={circle,draw=black!100,thick, inner sep=0pt,minimum size=1.4mm},scale=0.75]
			\path 	(-cos{60},-sin{60}) 	node 	(a0) [elt] {}
			(-cos{60},sin{60}) 	node 	(a1) [elt] {}
			( 0,0) node  	(a2) [elt] {}
			( 1,0) 	node  	(a3) [elt] {}
			( 1+cos{60},-sin{60}) 	node 	(a4) [elt] {}
			( 1+cos{60},sin{60})	node 	(a5) [elt] {};
			\draw [black,line width=1pt] (a1) -- (a2) -- (a3) -- (a4) (a3) -- (a5) (a2) -- (a0);
				
				\draw[<->] (a4) edge[bend right=60] node[fill=white]{\small $=$} (a5);
				\draw[<->] (a0) edge[bend left=60] node[fill=white]{\small $=$} (a1);
			\end{tikzpicture}
			    
			    \\
			    
			    \cline{1-1} \cline{3-5}
			    
			    \begin{center}

                           \includegraphics[scale=0.8]{d5_11.pdf}
			    
			    $w=\bar{w}=s_4s_5$
			    
			    \end{center}

			    &

			    \begin{center}

			    \begin{tikzpicture}[elt/.style={circle,draw=black!100,thick, inner sep=0pt,minimum size=1.4mm},scale=1]
			      \path 	
			 (-1,1) 	node 	(a0) [elt] {}    
			(-1,-1) 	node 	(a1) [elt] {}
			(1,-1) 	node 	(a2) [elt] {}
			(1,1) 	node 	(a3) [elt] {}; 
			
			\draw [black,line width=1pt] (a0) -- (a1) -- (a2) -- (a3) -- (a0);
			
			 \node at ($(a0.north) + (0,0.2)$) 	{\small $\underline{\alpha_{0}}$};	
		     	
		     \node at ($(a1.south) + (0,-0.2)$) 	{\small $\alpha_{2}$};
		     
		      \node at ($(a2.south) + (0,-0.2)$) 	{\small $\alpha_{1}$};
		      
		      \node at ($(a3.north) + (0,0.2)$) 	{\small $\alpha_{23345}$};
			
			  \draw[<->,dashed] (a1) edge[bend right=40] node[fill=white]{\small $s_{3453}$} (a3);
			  
			    \draw[<->,dashed] (a0) edge[bend left=40] node[fill=white]{\small $\pi^2$} (a2);

			     \node at (0,1.5) 	{$A_3^{(1)}$};
			     
			     \node at (0,-2) 	{$((3A_1)^{(1)})|_{|\alpha|^2=4}$};

		    \begin{scope}[yshift=-4cm]
			\path 	(2,1) 	node 	(a1) [elt] {}
			(2,-1) 	node 	(a2) [elt] {}
			(0,1) node  	(a3) [elt] {}
			(0,-1) 	node  	(a4) [elt] {}
			(-2,1) 	node  	(a5) [elt] {}
			(-2,-1) 	node  	(a6) [elt] {};
			
		    \draw [black,line width=2.5pt ] (a1) -- (a2);
		    
		     \draw [white,line width=1.5pt ] (a1) -- (a2);
		     
		     \draw [<->,black, line width=0.5pt]
		     (a1) -- (a2);
		     
		      \draw [black,line width=2.5pt ] (a3) -- (a4);
		    
		     \draw [white,line width=1.5pt ] (a3) -- (a4);
		     
		     \draw [<->,black, line width=0.5pt]
		     (a3) -- (a4);
		     
		      \draw [black,line width=2.5pt ] (a5) -- (a6);
		    
		     \draw [white,line width=1.5pt ] (a5) -- (a6);
		     
		     \draw [<->,black, line width=0.5pt]
		     (a5) -- (a6);
		     
		     	\node at ($(a1.north) + (-0.2,0.1)$) 	{\small $\alpha_{00223345}$};	
		     	
		     		\node at ($(a2.south) + (-0.2,-0.1)$) 	{\small $\alpha_{11223345}$};	
		     		
		     	\node at ($(a3.north) + (0,0.1)$) 	{\small $2\delta-\alpha_{223345}$};	
		     	
		     		\node at ($(a4.south) + (0,-0.1)$) 	{\small $\alpha_{223345}$};	
		     		
		     		\node at ($(a5.north) + (0,0.1)$) 	{\small $2\delta-\alpha_{3345}$};	
		     		
		     		\node at ($(a6.south) + (0,-0.2)$) 	{\small \parbox{1cm}{$\alpha_{3345}$\\$s_{3453}$}};

		    \draw[<->, dashed] (a1) to [bend left=0] node[fill=white]{\small $s_1$} (a3); 
		    
		     \draw[<->, dashed] (a2) to [bend left=0] node[fill=white]{\small $s_1$} (a4); 
		    
		     \draw[<->, dashed] (a3) to [bend left=0] node[fill=white]{\small $s_2$} (a5); 
		     
		      \draw[<->, dashed] (a4) to [bend left=0] node[fill=white]{\small $s_2$} (a6); 
		      
		       \draw[<->,dashed] (a1) edge[bend right=20] node[fill=white]{\small $\pi^2$} (a2);
		       
		        \draw[<->,dashed] (a3) edge[bend left=20] node[fill=white]{\small $s_1\pi^2s_1$} (a4);
		        
		         \draw[<->,dashed] (a5) edge[bend left=20] node[fill=white]{\small $s_2s_1\pi^2s_1s_2$} (a6);
		     
		       \end{scope}
		    
			\end{tikzpicture}

			    \end{center}
			    
			    &
			    
			    \begin{center}
			    
			    $2A_1$

			    \medskip
			    
			   \begin{tikzpicture}[elt/.style={circle,draw=black!100,thick, inner sep=0pt,minimum size=1.4mm},scale=1]
			       \path 	
			 (0,-1) 	node 	(a4) [elt] {}    
			(0,1) 	node 	(a5) [elt] {};

			    \node at ($(a4.south)+(0,-0.2)$) {\small $\alpha_4$ };
			    
			    \node at ($(a5.north)+(0,0.2)$) {\small $\alpha_5$ };
			      
			    \draw[<->,dashed] (a4) edge[bend right=40] node[fill=white]{\small $s_{3453}$} (a5);
			      \draw[<->,dashed] (a4) edge[bend left=40] node[fill=white]{\small $\pi^2$} (a5);

			    \end{tikzpicture}
			    
			    \end{center}
			    
			    &
			    
			    \begin{center}
			    
			    \begin{gather*}
			   C_2^2\ltimes (W_{A_3}^a \times\underbrace{C_2^2}_{triv}) \\
			   {\simeq \scriptstyle ((S_3\ltimes C_2^3){\ltimes}W_{3A_1}^a){\ltimes}\!\!\underbrace{\scriptstyle C_2^2}_{triv}} \\
			   \downarrow\\
			  (S_3\ltimes C_2^3)\ltimes W_{3A_1}^a\\
			     \textrm{no proj. red.}\\
			   \textrm{in the preimage}
			    \end{gather*}

			   \eqref{symm_D5_E7_c2} 
			   
			   \end{center}
			    
			    & 
			    
			    \begin{center}
			    \begin{tikzpicture}[elt/.style={circle,draw=black!100,thick, inner sep=0pt,minimum size=1.4mm},scale=0.75]
			\path 	(-cos{60},-sin{60}) 	node 	(a0) [elt] {}
			(-cos{60},sin{60}) 	node 	(a1) [elt] {}
			( 0,0) node  	(a2) [elt] {}
			( 1,0) 	node  	(a3) [elt] {}
			( 1+cos{60},-sin{60}) 	node 	(a4) [elt,fill] {}
			( 1+cos{60},sin{60})	node 	(a5) [elt,fill] {};
			\draw [black,line width=1pt] (a1) -- (a2) -- (a3) -- (a4) (a3) -- (a5) (a2) -- (a0);
			
			\node at ($(a4.south)+(0,-0.2)$) {\small $-1$ };
			\node at ($(a5.north)+(0,0.2)$) {\small $-1$ };
			\end{tikzpicture}  
			    \end{center}
			    
			    \\

			    \hline
			    
			    \end{tabular}
			    
			    \medskip
			    
			    For folding $\pi$ of order $4$ we have the same problems, as for case $1$ folding in $E_6^{(1)}$.
			    Namely, $Q^{a, \perp w}$ is generated by $\alpha_0-\alpha_4, \alpha_5-\alpha_0$, $\alpha_1-\alpha_5$
			    and $\alpha_2-\alpha_3$. We take index $2$ sublattice $\widetilde{Q^{a,\perp w}}$, doubling latter generator.
			    In obtained $(A_1)|_{|\alpha|^2=16}+A_3|_{|\alpha|^2=4}$ we can realize only few group elements, see below.

			    \begin{tabular}{|m{2.5cm}|m{3.5cm}|m{3.5cm}|m{3cm}|m{2.5cm}|}
			\hline
			  \begin{center} Folding \end{center}  & \begin{center}
			     $Q^{a,w}$  
			   \end{center}  &  \begin{center}  $Q^{a,\perp w}$ \end{center}
                            & \begin{center} $C(w)$ \end{center} & \begin{center} $\mathcal{A}^w, \, N_{flip}$ \end{center} \\
			     \hline
			    
			    \begin{center}

                         \medskip
                         
                         \includegraphics[scale=0.8]{d5_p.pdf}

			$w=\pi=t_{\omega_4} \bar{\pi}$
			
			{\small $\bar{\pi}=s s_{1534} s^{-1}$
			
			$s=s_{54312}$}

			    \end{center}

			    &
			    
			    \begin{center}

			\begin{tikzpicture}[elt/.style={circle,draw=black!100,thick, inner sep=0pt,minimum size=2mm},scale=1]
			\path 	(-1,0) 	node 	(a1) [elt] {}
			(1,0) 	node 	(a2) [elt] {};

		    \draw [black,line width=2.5pt ] (a1) -- (a2);
		    
		     \draw [white,line width=1.5pt ] (a1) -- (a2);
		     
		     \draw [<->,black, line width=0.5pt]
		     (a1) -- (a2);
		     
		     \node at ($(a1.south) + (0,-0.2)$) 	{\small$\alpha_{23}$};	
		     	
		     \node at ($(a2.south) + (0,-0.2)$) 	{\small $\underline{\alpha_{012345}}$};
		     
		     \draw[<->,dashed] (a1) to[bend left=40] node[fill=white]{\small $s_0s_1s_4s_5$} (a2);
		     
		      \node at (0, 1) {$A_1^{(1)}$};
		      
		       \node at (0, -1.5) {$(A_1^{(1)})_{|\alpha|^2=8}$};

		     \begin{scope}[yshift=-2.25cm]
			\path 	(-1,0) 	node 	(a1) [elt] {}
			(1,0) 	node 	(a2) [elt] {};

		    \draw [black,line width=2.5pt ] (a1) -- (a2);
		    
		     \draw [white,line width=1.5pt ] (a1) -- (a2);
		     
		     \draw [<->,black, line width=0.5pt]
		     (a1) -- (a2);
		     
		     \node at ($(a2.south) + (0.5,-0.3)$) 	{\small\parbox{2cm}{$\alpha_{0145}$\\ $s_0s_1s_4s_5$}};	
		     	
		     \node at ($(a1.south) + (0,-0.1)$) 	{\small $2\alpha_{23}$};
		     
		     \end{scope}
		     
		     \end{tikzpicture}
		     
		     \medskip
		     
		     $t_{2\omega_{2\alpha_{23}}}\sim t^{1/2}_{\omega_{23}-\omega_{145}}$

		     \end{center}
			    
			    & 
			    
			     \begin{center}   
			    
			  $\widetilde{Q^{a, \perp w}}=(A_1)_{|\alpha|^2=16}+ (A_3)_{|\alpha|^2=4}$

			      \medskip
			      
			      \begin{tikzpicture}[elt/.style={circle,draw=black!100,thick, inner sep=0pt,minimum size=2mm},scale=1.25]
			\path (0,2) node (a1) [elt] {}
			(-1,1) node (a5) [elt] {}
			(0,1) node  (a3) [elt] {}
			(1,1) node  (a4) [elt] {};

			\draw (a5) -- (a3) -- (a4);
		    
		    	\node at ($(a1.north) + (0.1,0.1)$) 	{\small \parbox{2cm}{$\alpha_{0122}{-}\alpha_{3345}$\\ $\pi$}};	
		    	
		    	\node at ($(a5.north) + (0,0.2)$) 	{\small $\alpha_0{-}\alpha_4$};	
		    	\node at ($(a3.north) + (0,0.2)$) 	{\small $\alpha_5{-}\alpha_0$};	
		    	\node at ($(a4.north) + (0,0.2)$) 	{\small $\alpha_1{-}\alpha_5$};

			\end{tikzpicture} 
			
			\medskip
			
			$\Omega_{A_3}= \langle \pi \rangle$
			
			$Z(A_3^e)=\langle s_0s_1s_4s_5\rangle$

			\end{center}

			    &
			    \begin{center}
			    
			    \begin{gather*}
			    W_{A_1}^{ae} \times \underbrace{C_4}_{triv} \\
			    \simeq W_{A_1}^a  \times \underbrace{C_4}_{triv} \\
			    \downarrow\\
			    W_{A_1}^a\\
			    t^{1/2}_{\omega_{23}-\omega_{0145}}\rightarrow t_{\omega_7}
			    \end{gather*}

			    \eqref{symm_D5_E8_c1}
			    
			    \end{center}

			    & 
			    
			      \begin{center}
			    \begin{tikzpicture}[elt/.style={circle,draw=black!100,thick, inner sep=0pt,minimum size=1.4mm},scale=0.75]
			\path 	(-cos{60},-sin{60}) 	node 	(a0) [elt] {}
			(-cos{60},sin{60}) 	node 	(a1) [elt] {}
			( 0,0) node  	(a2) [elt] {}
			( 1,0) 	node  	(a3) [elt] {}
			( 1+cos{60},-sin{60}) 	node 	(a4) [elt] {}
			( 1+cos{60},sin{60})	node 	(a5) [elt] {};
			\draw [black,line width=1pt] (a1) -- (a2) -- (a3) -- (a4) (a3) -- (a5) (a2) -- (a0);
			\draw[->] (a0) edge[bend right=60] node[]{} (a5);
			\draw[->] (a1) edge[bend right=60] node[]{} (a4);
			\draw[->] (a4) edge[bend left=60] node[fill=white]{\small $=$} (a0);
			\draw[->] (a5) edge[bend right=60] node[fill=white]{\small $=$} (a1);
			\draw[<->] (a2) edge[bend left=30] node[]{} (a3);	
			\end{tikzpicture} 
			    \end{center}
			    
			    \\
			    
			    \cline{1-1} \cline{3-5}
			    
			    \begin{center}

                         \includegraphics[scale=0.8]{d5_13.pdf}

                         \medskip
			
			$w=s_1s_{534}$
			    \end{center}

			    &

			    \begin{center}

			 \medskip

			\begin{tikzpicture}[elt/.style={circle,draw=black!100,thick, inner sep=0pt,minimum size=2mm},scale=1]

			\path 	(-1,0) 	node 	(a1) [elt] {}
			(1,0) 	node 	(a2) [elt] {};

		    \draw [black,line width=2.5pt ] (a1) -- (a2);
		    
		     \draw [white,line width=1.5pt ] (a1) -- (a2);
		     
		     \draw [<->,black, line width=0.5pt]
		     (a1) -- (a2);
		     
		     \node at ($(a1.south) + (0,-0.2)$) 	{\small$\delta-\alpha_0$};	
		     	
		     \node at ($(a2.south) + (0,-0.2)$) 	{\small $\underline{\alpha_{0}}$};
		    
		       \node at (0,0.5) 	{$A_1^{(1)}$};
			 
			\node at (0,-1.25) 	{ $(A_1^{(1)})|_{|\alpha|^2=8}$};

			\begin{scope}[yshift=-2.25cm]
			\path 	(-1,0) 	node 	(a1) [elt] {}
			(1,0) 	node 	(a2) [elt] {};

		    \draw [black,line width=2.5pt ] (a1) -- (a2);
		    
		     \draw [white,line width=1.5pt ] (a1) -- (a2);
		     
		     \draw [<->,black, line width=0.5pt]
		     (a1) -- (a2);
		     
		     \node at ($(a2.south) + (0,-0.2)$) 	{\small$2(\delta+\alpha_0)$};	
		     	
		     \node at ($(a1.south) + (0,-0.2)$) 	{\small $2(\delta-\alpha_0)$};
		     
		     \draw [<->, dashed]
		     (a1) to[bend left=40] node[fill=white] {\small $s_0$} (a2);
		     
		     \end{scope}
		     
		      \end{tikzpicture}
		     
		     \end{center}
			    
			    & 
			    
			 \begin{center}   
			    
			  $A_1+ A_3$

			      \medskip
			      
			      \begin{tikzpicture}[elt/.style={circle,draw=black!100,thick, inner sep=0pt,minimum size=2mm},scale=1.5]
			\path (0,2) node (a1) [elt] {}
			(-1,1) node (a5) [elt] {}
			(0,1) node  (a3) [elt] {}
			(1,1) node  (a4) [elt] {};

			\draw (a5) -- (a3) -- (a4);
		    
		    	\node at ($(a1.north) + (0,0.2)$) 	{\small $\alpha_1$};	
		    	
		    	\node at ($(a5.north) + (0,0.2)$) 	{\small $\alpha_5$};	
		    	\node at ($(a3.north) + (0,0.2)$) 	{\small $\alpha_3$};	
		    	\node at ($(a4.north) + (0,0.2)$) 	{\small $\alpha_4$};

			\end{tikzpicture} 
			
			\end{center} 
			    
			    &
			    
			    \begin{center}
			    \begin{gather*}
			   W_{A_1}^a \times C_2\times C_4\\ 
			   \simeq W_{A_1}^{ae}\times\underbrace{C_2\times C_4}_{triv} \\
			   \downarrow\\
			   W_{A_1}^{ae}\\
			   \textrm{no proj. red.}\\
			   \textrm{in the preimage}
			    \end{gather*}

			    \eqref{symm_D5_E8_c2}
			    
			    \end{center}
			    
			    & 
			    
			      \begin{center}
			    \begin{tikzpicture}[elt/.style={circle,draw=black!100,thick, inner sep=0pt,minimum size=1.4mm},scale=0.75]
			\path 	(-cos{60},-sin{60}) 	node 	(a0) [elt] {}
			(-cos{60},sin{60}) 	node 	(a1) [elt,fill] {}
			( 0,0) node  	(a2) [elt] {}
			( 1,0) 	node  	(a3) [elt,fill] {}
			( 1+cos{60},-sin{60}) 	node 	(a4) [elt,fill] {}
			( 1+cos{60},sin{60})	node 	(a5) [elt,fill] {};
			\draw [black,line width=1pt] (a1) -- (a2) -- (a3) -- (a4) (a3) -- (a5) (a2) -- (a0);
			
			 \node at ($(a1.west) + (-0.3,0)$) 	{\small $-1$};
			  \node at ($(a3.east) + (0.2,0)$) 	{\small $\ri$};
			   \node at ($(a4.east) + (0.2,0)$) 	{\small $\ri$};
			 \node at ($(a5.east) + (0.2,0)$) 	{\small $\ri$};
			
			\draw [<->, dashed] (0.5,-1) to[] node[fill=white] {\small $s_{45}$} (0.5,-3);

			\begin{scope}[yshift=-4cm]
			 \path 	(-cos{60},-sin{60}) 	node 	(a0) [elt] {}
			(-cos{60},sin{60}) 	node 	(a1) [elt,fill] {}
			( 0,0) node  	(a2) [elt] {}
			( 1,0) 	node  	(a3) [elt,fill] {}
			( 1+cos{60},-sin{60}) 	node 	(a4) [elt,fill] {}
			( 1+cos{60},sin{60})	node 	(a5) [elt,fill] {};
			\draw [black,line width=1pt] (a1) -- (a2) -- (a3) -- (a4) (a3) -- (a5) (a2) -- (a0);
			
			 \node at ($(a1.west) + (-0.3,0)$) 	{\small $-1$};
			  \node at ($(a3.east) + (0.2,0)$) 	{\small $-\ri$};
			   \node at ($(a4.east) + (0.2,0)$) 	{\small $-\ri$};
			 \node at ($(a5.east) + (0.2,0)$) 	{\small $-\ri$};
			\end{scope}

			\end{tikzpicture} 
			    \end{center}
			    
			    \\
			    \cline{1-1} \cline{3-5}
			    
			    \begin{center}
%
%
%

                           \includegraphics[scale=0.8]{d5_pp111.pdf}
                           
                           \medskip
	                    
	                    $w=\pi^2 s_2s_4$
	                    
	                    $\bar{w}{=}s_{123} s_1 s_{534} s_{123}^{-1}$
	                    
	                    \end{center}
	                    & 
	                    
	                      \begin{center}
			
			\begin{tikzpicture}[elt/.style={circle,draw=black!100,thick, inner sep=0pt,minimum size=2mm},scale=1]
			\path 	(-1,0) 	node 	(a1) [elt] {}
			(1,0) 	node 	(a2) [elt] {};

		    \draw [black,line width=2.5pt ] (a1) -- (a2);
		    
		     \draw [white,line width=1.5pt ] (a1) -- (a2);
		     
		     \draw [<->,black, line width=0.5pt]
		     (a1) -- (a2);
		     
		     \node at ($(a1.south) + (0,-0.1)$) 	{\small $\alpha_{23345}$};	
		     	
		     \node at ($(a2.south) + (0,-0.1)$) 	{\small $\alpha_{012}$};
		     
		     \node at (0,0.5) 	{$A_1^{(1)}$};
		     
		      \node at (0,-1.25) 	{$(A_1^{(1)})|_{|\alpha|^2=8}$};
		     
		     \begin{scope}[yshift=-2cm]
			\path 	(-1,0) 	node 	(a1) [elt] {}
			(1,0) 	node 	(a2) [elt] {};

		    \draw [black,line width=2.5pt ] (a1) -- (a2);
		    
		     \draw [white,line width=1.5pt ] (a1) -- (a2);
		     
		     \draw [<->,black, line width=0.5pt]
		     (a1) -- (a2);
		     
		     \node at ($(a1.south) + (0.2,-0.3)$) 	{\small \parbox{2cm}{$2\alpha_{23345}$\\$s_2s_{\alpha_{23345}}$}};	
		     	
		     \node at ($(a2.south) + (0,-0.1)$) 	{\small $2\alpha_{012}$};
		     
		     \end{scope}
		     
		     \end{tikzpicture}
		     
		     \end{center}
	                    
	                    &
	                    
	                     \begin{center}   
			    
			  $(A_1)|_{|\alpha|^2=4}+ 3A_1$

			      \medskip
			      
			      \begin{tikzpicture}[elt/.style={circle,draw=black!100,thick, inner sep=0pt,minimum size=2mm},scale=1]
			\path (0,-1) node (a0) [elt] {}
			(-1,0) node (a4) [elt] {}
			(1,0) node  (a5) [elt] {}
			(0,0.865) node  (a2) [elt] {};

		    	\node at ($(a0.south) + (0,-0.3)$) 	{\small \parbox{1cm}{$\alpha_0-\alpha_1$\\$\pi^2$}};	
		    	
		    	\node at ($(a4.south) + (0,-0.2)$) 	{\small $\alpha_4$};	
		    	\node at ($(a5.south) + (0,-0.2)$) 	{\small $\alpha_5$};	
		    	\node at ($(a2.north) + (0,0.2)$) 	{\small $\alpha_2$};	
		    	
		    	\draw[<->, dashed] (a4) to[] node[fill=white] {\small $\pi^2$} (a5);
		    	
			\end{tikzpicture} 
			
			\end{center} 
	                    
	                    &
	                    
	                    \begin{center}
	                    
	                    \begin{gather*}
	                     W_{A_1}^a{\times}((\underbrace{C_2{\ltimes}C_4}_{triv})\!\!{\times}\!\!\underbrace{C_2}_{s_2})\\
	                    \downarrow\\
	                     C_2 \times W_{A_1}^a\\
	                    \textrm{no proj. red.} 
	                    \end{gather*}
	                    
	                    \eqref{symm_D5_E8_c3}
	                    
	                    \end{center}

	                    &
	                    
	                     \begin{center}
			     \begin{tikzpicture}[elt/.style={circle,draw=black!100,thick, inner sep=0pt,minimum size=1.4mm},scale=0.75]
			
			 \path 	(-cos{60},sin{60}) 	node 	(a0) [elt] {}
			(-cos{60},-sin{60}) 	node 	(a1) [elt] {}
			( 0,0) node  	(a2) [elt,fill] {}
			( 1,0) 	node  	(a3) [elt] {}
			( 1+cos{60},-sin{60}) 	node 	(a4) [elt,fill] {}
			( 1+cos{60},sin{60})	node 	(a5) [elt,fill] {};
			\draw [black,line width=1pt] (a1) -- (a2) -- (a3) -- (a4) (a3) -- (a5) (a2) -- (a0);

				\draw[<->] (a0) edge[bend right=60] node[fill=white]{\small $\times -1$} (a1);
				
				\node at ($(a2.north) + (0,0.2)$) 	{\small $-1$};	
				\node at ($(a4.south) + (0,-0.2)$) 	{\small $-1$};	
				\node at ($(a5.north) + (0,0.2)$) 	{\small $-1$};	
	   
	                    \end{tikzpicture}
	                    \end{center}
	                    \\
			    
			    \hline
			    \begin{center}

                       \includegraphics[scale=0.8]{d5_1111.pdf}
                       
                       \medskip
			
			$w=s_0s_1s_4s_5$
			\end{center}
			      &
			      
			      \begin{center}
			
			\begin{tikzpicture}[elt/.style={circle,draw=black!100,thick, inner sep=0pt,minimum size=2mm},scale=1]
			\path 	(-1,0) 	node 	(a1) [elt] {}
			(1,0) 	node 	(a2) [elt] {};

		    \draw [black,line width=2.5pt ] (a1) -- (a2);
		    
		     \draw [white,line width=1.5pt ] (a1) -- (a2);
		     
		     \draw [<->,black, line width=0.5pt]
		     (a1) -- (a2);
		     
		     \node at ($(a2.south) + (0,-0.3)$) 	{\small 
		     \parbox{1cm}{$\underline{\alpha_{0221}}$\\ $s_0s_1s_{2012}$}};	
		     	
		     \node at ($(a1.south) + (0,-0.3)$) 	{\small 
		     \parbox{1cm}{$\alpha_{4335}$\\ $s_4s_5 s_{3453}$}};
		     
		     \draw[<->,dashed] (a1) to[bend left=40]
			node[fill=white]{\small $s_4s_5\pi^{-1}$} (a2);
		     
		      \node at (0,1.25) 	{$A_1^{(1)}|_{|\alpha|^2=4}$};
		     
		     \end{tikzpicture}
		     
		     \medskip
		     
		     $t_{\omega_{\alpha_{4335}}}\sim t_{\omega_3-\omega_2}^{1/2}$

		     \end{center}
			      
			      &
			      
			      \begin{center}
		
		$4A_1$

		\medskip
		   
		  \begin{tikzpicture}[elt/.style={circle,draw=black!100,thick, inner sep=0pt,minimum size=2mm},scale=1]
			\path 	(-1,-1) 	node 	(a0) [elt] {}
			(-1,1) 	node 	(a1) [elt] {}
			(1,-1) node  	(a4) [elt] {}
			(1,1) 	node  	(a5) [elt] {};
			
		        \node at ($(a0.south) + (0,-0.1)$) 	{\small $\alpha_0$};	
		        \node at ($(a1.north) + (0,0.1)$) 	{\small $\alpha_1$};	
		        \node at ($(a4.south) + (0,-0.1)$) 	{\small $\alpha_4$};	
		        \node at ($(a5.north) + (0,0.1)$) 	{\small $\alpha_5$};	
		    	
			\draw[->,dashed] (a0) to[bend right=0]
			node[fill=white]{\small $\pi$} (a5);
			\draw[->,dashed] (a5) to[bend right=0]
			node[fill=white]{\small $\pi$} (a1);
			\draw[->,dashed] (a1) to[bend right=0]
			node[fill=white]{\small $\pi$} (a4);
			\draw[->,dashed] (a4) to[bend right=0]
			node[fill=white]{\small $\pi$} (a0);
		    
		        \draw[<->,dashed] (a4) to[bend right=0]
			node[fill=white]{\small $s_{3453}$} (a5);
			\draw[<->,dashed] (a0) to[bend right=0]
			node[fill=white]{\small $s_{2012}$} (a1);
		    
			\end{tikzpicture} 
			
			\end{center}
			      
			      &
			      
			      \begin{center}
			      \begin{gather*}
			      {\scriptstyle C_4{\ltimes} (W_{A_1}^a{\ltimes} (\underbrace{\scriptstyle C_2}_{s_4s_5}\!\!{\times}
			      \!\!\underbrace{\scriptstyle C_2^3}_{triv}))} \\
			      \downarrow\\
			      (C_4\ltimes W_{A_1}^a)\times C_2\\
			      t_{\omega_3-\omega_2}^{1/2}\rightarrow t_{\omega_1}
			      \end{gather*}

			      \eqref{symm_D5_E1}
			      
			      \end{center}
			      
			      &
			      
			       \begin{center}
			      \begin{tikzpicture}[elt/.style={circle,draw=black!100,thick, inner sep=0pt,minimum size=1.4mm},scale=0.75]
			\path 	(-cos{60},sin{60}) 	node 	(a0) [elt,fill] {}
			(-cos{60},-sin{60}) 	node 	(a1) [elt,fill] {}
			( 0,0) node  	(a2) [elt] {}
			( 1,0) 	node  	(a3) [elt] {}
			( 1+cos{60},-sin{60}) 	node 	(a4) [elt,fill] {}
			( 1+cos{60},sin{60})	node 	(a5) [elt,fill] {};
			\draw [black,line width=1pt] (a1) -- (a2) -- (a3) -- (a4) (a3) -- (a5) (a2) -- (a0);
			
			\node at ($(a0.north) + (0,0.2)$) 	{\small $-1$};	
			\node at ($(a1.south) + (0,-0.2)$) 	{\small $-1$};	
			\node at ($(a4.south) + (0,-0.2)$) 	{\small $-1$};	
			\node at ($(a5.north) + (0,0.2)$) 	{\small $-1$};	
			\end{tikzpicture}
			\end{center}
			      
			      \\
			      \hline

	\end{tabular}
	\end{center}

	\subsection{$A_4^{(1)}/A_4^{(1)}$}

		\begin{figure}[h]

		\begin{center}
		\begin{tikzpicture}[elt/.style={circle,draw=black!100,thick, inner sep=0pt,minimum size=2mm},scale=2]
		\path 	(0.5,sin{72}+sin{36}) 	node 	(a0) [elt] {}
		(-cos{72},sin{72}) 	node 	(a1) [elt] {}
		( 0,0) node  	(a2) [elt] {}
		( 1,0) 	node 	(a3) [elt] {}
		( 1+cos{72},sin{72}) 	node  	(a4) [elt] {};
		\draw [black,line width=1pt] (a1) -- (a2) -- (a3) -- (a4) -- (a0) -- (a1);
		\node at ($(a0.south) + (0,-0.2)$) 	{$\underline{\alpha_{0}}$};			    
		\node at ($(a1.south west) + (-0.2,-0.2)$) 	{$\alpha_{1}$};
		\node at ($(a2.south west) + (-0.2,-0.2)$)  {$\alpha_{2}$};
		\node at ($(a3.south east) + (0.2,-0.2)$)  {$\alpha	_{3}$};
		\node at ($(a4.south east) + (0.2,-0.2)$)  {$\alpha_{4}$};	
		\node at ($(a0.north) + (-0,0.2)$) 	{\color{red}\small$1$};
		\node at ($(a1.north west) + (-0.1,0.2)$) 	{\color{red}\small$1$};
		\node at ($(a2.north east) + (0.1,0.2)$) 	{\color{red}\small$1$};
		\node at ($(a3.north west) + (-0.1,0.2)$) 	{\color{red}\small$1$};
		\node at ($(a4.north east) + (0.1,0.2)$)   {\color{red}\small$1$};
		
		\draw[->] (a4) edge[bend right=40] node[fill=white]{$\pi$} (a0);
		\draw[->] (a0) edge[bend right=40] node[fill=white]{$\pi$} (a1);
		\draw[->] (a1) edge[bend right=40] node[fill=white]{$\pi$} (a2);
		\draw[->] (a2) edge[bend right=40] node[fill=white]{$\pi$} (a3);
		\draw[->] (a3) edge[bend right=40] node[fill=white]{$\pi$} (a4);
		
		\end{tikzpicture}
		\end{center}
		
		\caption{Root numbering and marks for $E_3^{(1)}$ Dynkin diagram}
		
		\end{figure}
		

	       No true folding transformation.

	\subsection{$E_3^{(1)}/A_5^{(1)}$}
	\label{e3_a}
	
	\begin{figure}[h]
	
	\begin{center}
	\begin{tikzpicture}[elt/.style={circle,draw=black!100,thick, inner sep=0pt,minimum size=2mm},scale=1]
			\path 
			(-3,-1) 	node 	(a1) [elt] {}
			(1,-1) 	node 	(a2) [elt] {}
			(-1,2.46) 	node 	(a0) [elt] {}
			(2,0) 	node 	(a3) [elt] {}
			(6,0)	node 	(a4) [elt] {};
			\draw [black,line width=1pt] (a0) -- (a1) -- (a2) -- (a0);
			\draw[->] (a0) edge[bend right=40] node[fill=white]{$\pi$} (a1);
			\draw[->] (a1) edge[bend right=40] node[fill=white]{$\pi$} (a2);
			\draw[->] (a2) edge[bend right=40] node[fill=white]{$\pi$} (a0);
			\draw [<->,black,line width=1pt ] (a3) -- (a4);
			\draw [<->] (a3) edge[bend right=40] node[fill=white]{$\pi$} (a4);
			
			\node at ($(a0.south) + (0,-0.4)$) 	{$\underline{\alpha_{0}}$};			    
		\node at ($(a1.south west) + (-0.2,-0.2)$) 	{$\alpha_{1}$};
		\node at ($(a2.south east) + (0.2,-0.2)$)  {$\alpha_{2}$};
		\node at ($(a3.north) + (0,0.2)$)  {$\underline{\alpha_{3}}$};
		\node at ($(a4.north) + (0,0.2)$)  {$\alpha_{4}$};	
		\node at ($(a0.north) + (-0,0.2)$) 	{\color{red}\small$1$};
		\node at ($(a1.north east) + (0.3,0.2)$) 	{\color{red}\small$1$};
		\node at ($(a2.north west) + (-0.3,0.2)$) 	{\color{red}\small$1$};
		\node at ($(a3.south) + (0,-0.2)$) 	{\color{red}\small$1$};
		\node at ($(a4.south) + (0,-0.2)$)   {\color{red}\small$1$};
			
			\end{tikzpicture}
	\end{center}


	\caption{Root numbering and marks for $E_3^{(1)}$ Dynkin diagram}
	        
	        \label{fig:E3_rn}
                
		\end{figure}

	Affine Weyl group $W(E_3^{(1)})$ (and corresponding Dynkin diagram) has outer automorphism of order $6$.
	Diagram has also outer automorphism of order $2$ (vertical reflection), which is absent in $W(E_3^{(1)})$. 
			
		For $E_3^{(1)}$ we have $4$ foldings: 
		\begin{itemize}
		\item Two foldings of order $2$, of case $1$ and $2$, which have common (up to conjugation)
		$\bar{w}$.
		\item Two foldings of order $3$, of case $1$ and $2$, which have common (up to conjugation)
		$\bar{w}$.
		\end{itemize}

	\begin{center}
			\begin{tabular}{|m{3cm}|m{3cm}|m{3cm}|m{3cm}|m{2.5cm}|}
			\hline
			 \begin{center}  Folding \end{center}    & \begin{center}
			     $Q^{a,w}$  
			   \end{center}  & \begin{center}  $Q^{a,\perp w}$ \end{center}		                                                  
			     & \begin{center} $C(w)$ \end{center} & \begin{center} $\mathcal{A}^w, \, N_{flip}$ \end{center}\\
			    \hline
			    
			    		\begin{center}
%
%
%
%
%
%
%
%

             \includegraphics[scale=1]{e3_ppp.pdf}
             
             \medskip
	    
	    $w=\pi^3$
	    
	    {\small $\bar{w}=s_4$}
	    
			\end{center}

			     &
			     
			     &
			     \begin{center}
			      
			      $(A_1)|_{|\alpha|^2=8}$

			     \medskip
			     
			     \begin{tikzpicture}[elt/.style={circle,draw=black!100,thick, inner sep=0pt,minimum size=2mm},scale=1]
			\path 	
			(0,0) 	node 	(a0) [elt] {};
			
			\node at ($(a0.north) + (0,0.3)$) 	{\small \parbox{1cm}{$\alpha_3-\alpha_4$\\$\pi^3$}};
			
		       \end{tikzpicture}
			      
			     \end{center}

			      &
			      \begin{center}
			      
			      \begin{gather*}
			      W^{ae}_{A_2}\times \underbrace{C_2}_{triv}\\
			      \downarrow\\
			      W^{ae}_{A_2}\\
			       \textrm{no proj. red.}\\
			      \textrm{in the preimage}
			      \end{gather*}

			      \eqref{symm_E3_E6_c1}
			      
			      \end{center}
			      
			      & \begin{center}
			\begin{tikzpicture}[elt/.style={circle,draw=black!100,thick, inner sep=0pt,minimum size=1.4mm},scale=0.75]
			
			\begin{scope}[scale=1]
			\path 	
			(-1,0) 	node 	(a0) [elt] {}
			(1,0) 	node 	(a1) [elt] {}
			(0,1.73) 	node 	(a2) [elt] {};
			\draw [black,line width=1pt] (a0) -- (a1) -- (a2) -- (a0);
			\end{scope}
			
			\begin{scope}[xshift=0cm,yshift=-1.5cm]
			\path
			(-1,1) 	node 	(a3) [elt] {}
			(1,1)	node 	(a4) [elt] {};

			 \draw [black,line width=2.5pt ] (a3) -- (a4);
		    
		     \draw [white,line width=1.5pt ] (a3) -- (a4);
		     
		     \draw [<->,black, line width=0.5pt]
		     (a3) -- (a4);
		     
			\draw[<->] (a3) edge[bend right=50] node[fill=white]{\small $=$} (a4);
			
			\end{scope}
	    \end{tikzpicture} 
			\end{center} \\
			    
			    \cline{1-1} \cline{3-5}
			    
			    \begin{center}
%
%
%
%
%
%
%
%

        \includegraphics[scale=1]{e3_1.pdf}
        
        \medskip
	    
	    $w=s_4$
			    \end{center}

			    &

			    \vspace{-5cm}
			    
			    \begin{center}

			      \begin{tikzpicture}[elt/.style={circle,draw=black!100,thick, inner sep=0pt,minimum size=1.4mm},scale=1]
			      \path 	
			(-1,0) 	node 	(a1) [elt] {}
			(1,0) 	node 	(a2) [elt] {}
			(0,1.73) 	node 	(a0) [elt] {};
			\draw [black,line width=1pt] (a0) -- (a1) -- (a2) -- (a0);
			
			 \node at ($(a0.north) + (0,0.4)$) 	{\small $\underline{\alpha_{0}}$ };	
		     	
		     \node at ($(a1.south) + (0,-0.4)$) 	{\small $\alpha_1$};
		     
		      \node at ($(a2.south) + (0,-0.4)$) 	{\small $\alpha_2$};
			
			\draw[<-,dashed] (a2) edge[bend right=40] node[fill=white]{\small $\pi^2$} (a0);
			\draw[<-,dashed] (a0) edge[bend right=40] node[fill=white]{\small $\pi^2$} (a1);
			\draw[<-,dashed] (a1) edge[bend right=40] node[fill=white]{\small $\pi^2$} (a2);
			
			      \node at (0, 3) 	{$A_2^{(1)}$};
			       \node at (0,-1) 	{$(A_2^{(1)})_{|\alpha|^2=8}$};

			      \begin{scope}[yshift=-4cm]
			      \path 	
			(-1,0) 	node 	(a1) [elt] {}
			(1,0) 	node 	(a2) [elt] {}
			(0,1.73) 	node 	(a0) [elt] {};
			\draw [black,line width=1pt] (a0) -- (a1) -- (a2) -- (a0);
			
			 \node at ($(a0.north) + (0,0.4)$) 	{\small $\underline{2\alpha_{0}}$ };	
		     	
		     \node at ($(a1.south) + (0,-0.4)$) 	{\small $2\alpha_1$};
		     
		      \node at ($(a2.south) + (0,-0.4)$) 	{\small $2\alpha_2$};
			
			\draw[<-,dashed] (a2) edge[bend right=40] node[fill=white]{\small $\pi^2$} (a0);
			\draw[<-,dashed] (a0) edge[bend right=40] node[fill=white]{\small $\pi^2$} (a1);
			\draw[<-,dashed] (a1) edge[bend right=40] node[fill=white]{\small $\pi^2$} (a2);
			
			\end{scope}
			
			      \end{tikzpicture}

			      \end{center}
			    
			    &
			  
			  \begin{center}
			  $A_1$ 
			  
			  \medskip
			  
			     \begin{tikzpicture}[elt/.style={circle,draw=black!100,thick, inner sep=0pt,minimum size=2mm},scale=1]
			\path 	
			(0,0) 	node 	(a0) [elt] {};
			
			\node at ($(a0.north) + (0,0.2)$) 	{\small $\alpha_4$};
			
		       \end{tikzpicture}
		       \end{center}
			    
			    &
			    \begin{center}
			    
			    \begin{gather*}
			    W_{A_2}^{ae}\times \underbrace{C_2}_{triv}\\
			    \downarrow\\
			    W_{A_2}^{ae}\\
			     \textrm{no proj. red.}\\
			      \textrm{in the preimage}
			    \end{gather*}
			   
			    \eqref{symm_E3_E6_c2}
			    
			    \end{center}
			    
			    & 
			     \begin{center}
			      \begin{tikzpicture}[elt/.style={circle,draw=black!100,thick, inner sep=0pt,minimum size=1.4mm},scale=0.75]
			
			\begin{scope}[scale=1]
			\path 	
			(-1,0) 	node 	(a0) [elt] {}
			(1,0) 	node 	(a1) [elt] {}
			(0,1.73) 	node 	(a2) [elt] {};
			\draw [black,line width=1pt] (a0) -- (a1) -- (a2) -- (a0);
			\end{scope}

			\begin{scope}[xshift=0cm,yshift=-1.5cm,scale=1]
			\path
			(-1,1) 	node 	(a3) [elt] {}
			(1,1)	node 	(a4) [elt,fill] {};

			 \draw [black,line width=2.5pt ] (a3) -- (a4);
		    
		     \draw [white,line width=1.5pt ] (a3) -- (a4);
		     
		     \draw [<->,black, line width=0.5pt]
		     (a3) -- (a4);			
			
			\node at ($(a4.east) + (0.3,0)$) 	{\small $-1$};
			
			\end{scope}
	    \end{tikzpicture} 
			    \end{center}
			    
			    \\
			    
			    \hline

			    \begin{center}
%
%
%
%
%
%
%

\includegraphics[scale=1]{e3_pp.pdf}

\medskip
			
			$w=\pi^2$
			
			{\small $\bar{\pi}^2=s_{12}$
			
			$\pi^2=t_{\omega_1} \bar{\pi}^2$}
			
			\end{center}
			
			&

			&
			
			 \begin{center}
			      
			      $(A_2)_{|\alpha|^2=6}$

			     \medskip
			     
			     \begin{tikzpicture}[elt/.style={circle,draw=black!100,thick, inner sep=0pt,minimum size=2mm},scale=1]
			\path 	
			(-1,0) 	node 	(a1) [elt] {}
			(1,0) 	node 	(a2) [elt] {};
			
			\draw (a1) -- (a2);
			
			\node at ($(a1.north) + (0,0.2)$) 	{\small $\alpha_0-\alpha_2$};
			\node at ($(a2.north) + (0,0.2)$) 	{\small $\alpha_2-\alpha_1$};
			
		       \end{tikzpicture}
			      
		        \medskip	      
			      
			$\Omega_{A_2}=\langle \pi^2 \rangle$      
			      
			     \end{center}

			&
			
			\begin{center}
			
			\begin{gather*}
			W_{A_1}^{ae}\times \underbrace{C_3}_{triv}\\
			\downarrow\\
			W_{A_1}^{ae}\\
			 \textrm{no proj. red.}\\
			 \textrm{in the preimage}
			\end{gather*}

			\eqref{symm_E3_E7_c1}

			\end{center}
			
			& 
			  \begin{center}
			    \begin{tikzpicture}[elt/.style={circle,draw=black!100,thick, inner sep=0pt,minimum size=1.4mm},scale=0.75]
			\path 
			(-1,0) 	node 	(a0) [elt] {}
			(1,0) 	node 	(a1) [elt] {}
			(0,1.73) 	node 	(a2) [elt] {};
			\draw [black,line width=1pt] (a0) -- (a1) -- (a2) -- (a0);
			\draw[<-] (a0) edge[bend right=40] node[fill=white]{\small $=$} (a1);
			\draw[<-] (a1) edge[bend right=40] node[fill=white]{\small $=$} (a2);
			\draw[<-] (a2) edge[bend right=40] node[fill=white]{\small $=$} (a0);

			\begin{scope}[yshift=-2cm, scale=1]
			
 			\path
			(-1,1) 	node 	(a3) [elt] {}
			(1,1)	node 	(a4) [elt] {};
			
			 \draw [black,line width=2.5pt ] (a3) -- (a4);
		    
		     \draw [white,line width=1.5pt ] (a3) -- (a4);
		     
		     \draw [<->,black, line width=0.5pt]
		     (a3) -- (a4);
		     
		     \end{scope}
			\end{tikzpicture}
			\end{center}

			\\
			    
			    \cline{1-1} \cline{3-5}
			
			\begin{center}
%
%
%
%
%
%
	    
	    \includegraphics[scale=1]{e3_2.pdf}
	    
	    \medskip
	    
	    $w=s_{12}$
	    
			\end{center}

			    &
			    
			    \vspace{-5cm}
			    
			       \begin{center}
			
			\begin{tikzpicture}[elt/.style={circle,draw=black!100,thick, inner sep=0pt,minimum size=2mm},scale=1]
			\path 	(-1,0) 	node 	(a1) [elt] {}
			(1,0) 	node 	(a2) [elt] {};

		    \draw [black,line width=2.5pt ] (a1) -- (a2);
		    
		     \draw [white,line width=1.5pt ] (a1) -- (a2);
		     
		     \draw [<->,black, line width=0.5pt]
		     (a1) -- (a2);
		     
		     \node at ($(a1.south) + (0,-0.2)$) 	{\small$\underline{\delta-\alpha_{4}}$};	
		     	
		     \node at ($(a2.south) + (0,-0.2)$) 	{\small $\alpha_{4}$};

		     \draw[<->,dashed] (a1) to[bend left=40] node[fill=white] {\small $\pi^3$} (a2);
		     
		     \node at (0, 1) 	{$A_1^{(1)}$};
		      \node at (0, -1) 	{ $(A_1^{(1)})_{|\alpha|^2=18}$};

			\begin{scope}[yshift=-2.5cm]
			\path 	(-1,0) 	node 	(a1) [elt] {}
			(1,0) 	node 	(a2) [elt] {};

		    \draw [black,line width=2.5pt ] (a1) -- (a2);
		    
		     \draw [white,line width=1.5pt ] (a1) -- (a2);
		     
		     \draw [<->,black, line width=0.5pt]
		     (a1) -- (a2);
		     
		     \node at ($(a1.south) + (0,-0.2)$) 	{\small$\underline{3\delta-3\alpha_{4}}$};	
		     	
		     \node at ($(a2.south) + (0,-0.2)$) 	{\small $3\alpha_{4}$};
		     
		      \draw[<->,dashed] (a1) to[bend left=40] node[fill=white] {\small $\pi^3$} (a2);
		     
		     \end{scope}
		     
		     \end{tikzpicture}
		     
		     \end{center}

			  & 
			  
			   \begin{center}
			      
			      $A_2$

			     \medskip
			     
			     \begin{tikzpicture}[elt/.style={circle,draw=black!100,thick, inner sep=0pt,minimum size=2mm},scale=1]
			\path 	
			(-1,0) 	node 	(a1) [elt] {}
			(1,0) 	node 	(a2) [elt] {};
			
			\draw (a1) -- (a2);
			
			\node at ($(a1.north) + (0,0.2)$) 	{\small $\alpha_1$};
			\node at ($(a2.north) + (0,0.2)$) 	{\small $\alpha_2$};
			
		       \end{tikzpicture}
			      
			     \end{center}

			    &
			    
			    \begin{center}
			    
			    \begin{gather*}
			    W_{A_1}^{ae}\times \underbrace{C_3}_{triv}\\
			    \downarrow\\
			    W_{A_1}^{ae}\\
			     \textrm{no proj. red.}\\
			      \textrm{in the preimage}
			    \end{gather*}
                           
			    \eqref{symm_E3_E7_c2}
			    
			    \end{center}
			    
			    & 
			    	\begin{center}
			 \begin{tikzpicture}[elt/.style={circle,draw=black!100,thick, inner sep=0pt,minimum size=1.4mm},scale=0.5]
			\begin{scope}
			\path 	
			(-1,0) 	node 	(a1) [elt,fill] {}
			(1,0) 	node 	(a2) [elt,fill] {}
			(0,1.73) 	node 	(a0) [elt] {};
			\draw [black,line width=1pt] (a0) -- (a1) -- (a2) -- (a0);
			\end{scope}
			
			 \node at ($(a1.west) + (-0.3,0)$) 	{\small $\zeta$};
			  \node at ($(a2.east) + (0.3,0)$) 	{\small $\zeta$};
			
			\begin{scope}[yshift=-1.5cm,scale=1]
			\path
			(-1,1) 	node 	(a3) [elt] {}
			(1,1)	node 	(a4) [elt] {};

			 \draw [black,line width=2.5pt ] (a3) -- (a4);
		    
		     \draw [white,line width=1.5pt ] (a3) -- (a4);
		     
		     \draw [<->,black, line width=0.5pt]
		     (a3) -- (a4);			
			
			\end{scope}
			
			\draw [<->, dashed] (0,-1) -- (0,-2) 
			 node[pos=0.5,left] {\small $s_1$};
			
			\begin{scope}[yshift=-4cm]
			 \begin{scope}[scale=1]
			\path 	
			(-1,0) 	node 	(a1) [elt,fill] {}
			(1,0) 	node 	(a2) [elt,fill] {}
			(0,1.73) 	node 	(a0) [elt] {};
			\draw [black,line width=1pt] (a0) -- (a1) -- (a2) -- (a0);
			\end{scope}
			
			 \node at ($(a1.west) + (-0.4,0)$) 	{\small $\zeta^{-1}$};
			  \node at ($(a2.east) + (0.5,0)$) 	{\small $\zeta^{-1}$};
			
			\begin{scope}[yshift=-1.5cm,scale=1]
			\path
			(-1,1) 	node 	(a3) [elt] {}
			(1,1)	node 	(a4) [elt] {};

			 \draw [black,line width=2.5pt ] (a3) -- (a4);
		    
		     \draw [white,line width=1.5pt ] (a3) -- (a4);
		     
		     \draw [<->,black, line width=0.5pt]
		     (a3) -- (a4);			
			
			\end{scope}
			\end{scope}

	    \end{tikzpicture} 
			 \end{center}   
			    
			    \\
			    
			    \hline

			 \end{tabular}
	\end{center}

	\subsection{$E_2^{(1)}/A_6^{(1)}$} 
	
	\begin{figure}[h]
	
	  \begin{tikzpicture}[elt/.style={circle,draw=black!100,thick, inner sep=0pt,minimum size=1.4mm},scale=2]
			\path 	
			(-1,0) 	node 	(a0) [elt] {}
			(1,0) 	node 	(a1) [elt] {};

			 \draw [black,line width=2.5pt ] (a0) -- (a1);
		    
		     \draw [white,line width=1.5pt ] (a0) -- (a1);
		     
		     \draw [<->,black, line width=0.5pt]
		     (a0) -- (a1);
		     
		     
			\draw[<->] (a0) edge[bend right=20] node[fill=white]{\small $\pi$} (a1);
			
			 \node at ($(a0.north) + (0,0.1)$)   {$\delta-\alpha_1$};
		         \node at ($(a1.north) + (0,0.1)$)   {$\alpha_1$};
			\node at ($(a0.south) + (0,-0.2)$)   {\color{red}\small$1$};
			\node at ($(a1.south) + (0,-0.2)$)   {\color{red}\small$1$};
		
		\begin{scope}[xshift=3cm]
			
		\path 	
			(-1,0) 	node 	(a2) [elt] {}
			(1,0) 	node 	(a3) [elt] {};

			 \draw [black,line width=2.5pt ] (a2) -- (a3);
		    
		     \draw [white,line width=1.5pt ] (a2) -- (a3);
		     
		     \draw [<->,black, line width=0.5pt]
		     (a2) -- (a3);
		     
		     \node at (0,-0.25) {$\pi=t_{\omega_2}$};
		     	
		     \node at ($(a2.north) + (0,0.1)$)   {$\alpha_2$};
		     \node at ($(a3.north) + (0,0.1)$)   {$\delta-\alpha_2$};
		     \node at ($(a2.south) + (0,-0.2)$)   {\color{red}\small$1$};
		     \node at ($(a3.south) + (0,-0.2)$)   {\color{red}\small$1$};
		     	
		\end{scope}
			
	    \end{tikzpicture}
	
	\caption{Root numbering and marks for $E_2^{(1)}$ Dynkin diagram \label{Fig:geom E2}}
	
	\end{figure}
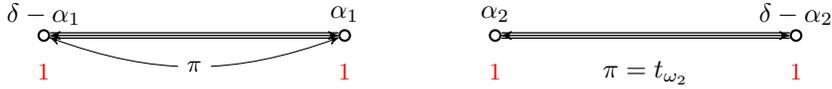
	
	It appears the symmetry lattice \(Q^a\) in this case is larger then \((A_1+A_1)^{(1)}\). Namely, if we denote by \(\alpha_1\) and \(\alpha_2\) simple roots for this 2 summands, then first \(\alpha_2^2=-8\) and second \((\alpha_1-\alpha_2)/2\in Q^a\). Due to this reason symmetry group is smaller, it is generated by \(\pi\) and \(s_1\), where \(\pi\) acts on first summand as automotphism and on second summand as translation. See Fig. \ref{Fig:geom E2}. Geometrically, this is corresponds to the fact that the surface $\mathcal{X}$ for $E_2^{(1)}/A_6^{(1)}$ depends
	on square roots of root variables, namely on $b=(a_1/a_2)^{1/2}$. See section \ref{ssec:alg E2} below.
	
	The only possibility for the folding transformation is simple reflection, say $s_1$ with \(a_1=-1\). But since \(s_1(b)=ba_1^{-1}\) this is not the case. Therefore we have no folding transformations in this case.

	\subsection{$A_1^{(1)}/A_7^{(1)'}$}
	\label{e1_a}
	
	   \begin{figure}[h]
	    
	  \begin{center}
	   \begin{tikzpicture}[elt/.style={circle,draw=black!100,thick, inner sep=0pt,minimum size=1.4mm},scale=3]
			\path 	
			(-1,0) 	node 	(a0) [elt] {}
			(1,0) 	node 	(a1) [elt] {};

			 \draw [black,line width=2.5pt ] (a0) -- (a1);
		    
		     \draw [white,line width=1.5pt ] (a0) -- (a1);
		     
		     \draw [<->,black, line width=0.5pt]
		     (a0) -- (a1);
		     
			\draw[<->] (a0) edge[bend right=20] node[fill=white]{$\pi$} (a1);
			
			\node at ($(a0.north) + (0,0.1)$) 	{$\alpha_0$};
			\node at ($(a1.north) + (0,0.1)$) 	{$\alpha_1$};
			\node at ($(a0.south) + (0,-0.2)$)   {\color{red}\small$1$};
			\node at ($(a1.south) + (0,-0.2)$)   {\color{red}\small$1$};

	    \end{tikzpicture}
	    \end{center}
	    
	    \caption{Root numbering and marks for $E_1^{(1)}$ Dynkin diagram}
	        
	        \label{fig:E1_rn}
	    
	    \end{figure}
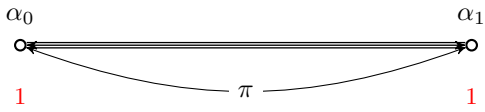
	    
	    Affine Weyl group $W(E_1^{(1)})$ has outer automorphism $\pi$ of order $2$,
	    which in $q$-Painlev\'e symmetry is extended to $C_4$ ($\pi^2$ become $\neq 1$).
	    There is also $q$-Painlev\'e (gauge) symmetry $\sigma$ of order $2$
	    \footnote{It is absent in \cite{KNY15}}.
	    These extensions act trivially on root variables, but nontriviallly on coordinates, see Subsection \ref{ssec:geom_E1}
			
		For $E_1^{(1)}$ we have $2$ foldings of order $2$, given by $\pi^2$ and $\sigma$.
		
	\begin{center}
\begin{tabular}{|m{2.5cm}|m{4cm}|m{1cm}|m{4cm}|m{2.5cm}|}
			\hline
			 \begin{center} Folding  \end{center}   & \begin{center}
			     $Q^{a,w}$  
			   \end{center}  & \begin{center}  $Q^{a,\perp w}$ \end{center}		                                                  
			     & \begin{center}  $C(w)$ \end{center} & \begin{center} $\mathcal{A}^w$ \end{center}\\
\hline
			    
			\end{tabular}
	\end{center}

	\newpage

	\section{Geometric realization: answers}
	
	\label{sec:geom_answ}

	\subsection{$E_8^{(1)}/A_0^{(1)}$}

Action on coordinates $F=f/\nu_1$, $G=g/\nu_1$  (where $f,g$ --- notations of \cite{KNY15}):

\begin{center}
              \begin{tabular}{|c|c|c|c|c|c|c|c|c|c|}
              \hline
                  coord. & $s_0$ & $s_1$ & $s_2$ & $s_3$ & $s_4$ & $s_5$ & $s_6$ & $s_7$ &
                  $s_8$\\
                  \hline
                  $F$ & $a_0F$ & $G$ & $\frac{(a_1-1)((a_{12}-1)(a_2-1)+a_0FG)-a_1(1+a_0)(a_2(F-G)-(a_1^{-1}F-G))}
                  {a_0((F-G)-a_{12}(a_1^{-1}F-G))-a_2(a_1-1)(1 + a_0)} $  & $F$ &
               $F$ & $F$ & $F$ & $F$ & $F$
                  \\
                  \hline
                  $G$ & $a_0G$ & $F$ & $a_2^{-1}G$ & $G$ & $G$ & $G$ & $G$ & $G$ & $G$
                  \\
                  \hline
                \end{tabular}
          \end{center}

Anticanonical divisor is represented by curve          
\begin{equation}\label{E8curve}
a_0a_2^{-1}(F-G) (a_1^{-1}F-G)=2-a_1-a_1^{-1}.  
\end{equation} 


	Root data and point configuration:

	\begin{center}
 %
            \end{center}

	   where $p_i=(a_{3\ldots i} a_{21}+a_{3\ldots i}^{-1} a_0^{-1},a_{3\ldots i} a_2+a^{-1}_{3\ldots i} a_0^{-1}), \, i=3\ldots 8$
	   and $(\infty,\infty)$ is self-intersection point. 
		
		\subsubsection{$3A_2$}
		
		\label{e8_222}

		%
	 
}
	    
	     \medskip

		 Under factorization we obtain already minimal surface	
		$E_6^{(1)}/A_2^{(1)}$ 	with parameters and coordinates
		
		\begin{tabular}{m{4cm}m{11cm}}
	         \begin{center} 
	         \begin{tikzpicture}[elt/.style={circle,draw=black!100,thick, inner sep=0pt,minimum size=2mm}, scale=0.75]
		\path 	(-2,-1.16) 	node 	(a1) [elt] {}
		(-1,-0.58) 	node 	(a2) [elt] {}
		( 0,0) node  	(a3) [elt] {}
		( 1,-0.58) 	node  	(a4) [elt] {}
		( 2,-1.16) 	node 	(a5) [elt] {}
		( 0,1)	node 	(a6) [elt] {}
		( 0,2)	node 	(a0) [elt] {};
		\draw [black,line width=1pt ] (a1) -- (a2) -- (a3) -- (a4) -- (a5)   (a3) -- (a6) -- (a0);
		\node at ($(a1.west) + (-0.3,0)$) 	{\small $a_0^{-1}$};
		\node at ($(a2.south) + (0,-0.3)$)  {\small $1$};
		\node at ($(a3.south) + (0,-0.3)$)  {\small $a_0$};
		\node at ($(a4.south) + (0,-0.3)$)  {\small $a_0 a_3^3$};	
		\node at ($(a5.east) + (0.3,0)$)  {\small $a_6^3$};	
		\node at ($(a6.west) + (-0.3,0)$) 	{\small $1$};	
		\node at ($(a0.north) + (0,0.3)$) 	{\small $a_0^{-1}$};
		\end{tikzpicture} 
		\end{center}
	        &
	         \begin{center}
		       \begin{equation}
		       \begin{aligned}
		              &\f=\frac{3 \zeta a_0^{-1} - (F{-}\zeta G)^2 + (1{-}\zeta) (1{+}a_0^{-1}) (F{-}\zeta G)}
		              {3 \zeta^{-1} (F{+}\zeta G) + \zeta (1{+}a_0) (F{-}G)^2 + (\zeta^{-1}{-}1) a_0 F G (F{-}G)}  \\        
		              &\g=a_0 \frac{3 \zeta a_0^{-1} - \zeta (F{-}G)^2 + (1{-}\zeta^{-1}) (1{+}a_0^{-1}) (F{-}G)}
		              {3 \zeta (F{+}G) + (1{+}a_0) (F{-}\zeta G)^2 + (\zeta{-}1) a_0 F G (F{-}\zeta G)}
		         \end{aligned}
	     \end{equation}\\
	     \end{center}
		       
	\end{tabular}	
		
           Symmetry lattice is $A_2^{(1)}$, nodal curves and lattice is $\overline{\Phi}_{\nod}=\Phi_{\nod}=2A_2$,
           symmetry group is as follows:
	     
	     \begin{tabular}{m{10cm}m{5cm}}
	     
	     \begin{center}
	       \begin{tikzpicture}[elt/.style={circle,draw=black!100,thick, inner sep=0pt,minimum size=2mm},scale=1.25]
			      \path 	
			(-1,0) 	node 	(a1) [elt] {}
			(1,0) 	node 	(a2) [elt] {}
			(0,1.73) 	node 	(a0) [elt] {};
			\draw [black,line width=1pt] (a0) -- (a1) -- (a2) -- (a0);
			
			 \node at ($(a0.north) + (0,0.2)$) 	{\small $\underline{\alpha_{012346}}$};	
		     	
		     \node at ($(a1.west) + (-0.2,0)$) 	{\small $\alpha_5$};
		     
		      \node at ($(a2.east) + (0.5,0)$) 	{\small $\alpha_{23346}$};
			
 			\draw[<->,dashed] (a2) to[bend right=40]
 			node[fill=white]{\small $s_{013}$} (a0);
			
			 
	      \begin{scope}[xshift=5cm, scale=2]
				\path 	(-1,0) 	node 	(a1) [elt] {}
					(-1,1) 	node 	(a2) [elt] {}
					( 0,0) node  	(a3) [elt] {}
					( 0,1) 	node  	(a4) [elt] {};

				\draw [black,line width=1pt] (a1) -- (a2) (a3) -- (a4);
				
				\node at ($(a1.south) + (0,-0.1)$) 	{\small $\alpha_6$};
		\node at ($(a2.north) + (0,0.1)$)  {\small $\alpha_{123}$};
		\node at ($(a3.south) + (0,-0.1)$)  {\small $\alpha_{036}$};
		\node at ($(a4.north) + (0,0.1)$)  {\small $\alpha_2$};	
		
		\draw[<->, dashed] (a1) to[] node[fill=white] {\small $s_{013}$} (a3);
		\draw[<->, dashed] (a2) to[] node[fill=white] {\small $s_{013}$} (a4);
	
	      
	    \end{scope}
		
			\end{tikzpicture}
	     \end{center}
			      
			      &
			      
			      \begin{center}
			      \begin{equation}
			      C_2 \ltimes W_{A_2}^a \label{symm_E8_E6}
			      \end{equation}
			      
			      \medskip
			      
			      No projective reduction in the image.
 
			      \end{center}

			 \end{tabular}

				\subsubsection{$4A_1$}
			
			\label{e8_1111}
		
		\begin{tabular}{m{7.5cm}m{7.5cm}}
		\begin{center}
                     \includegraphics{e8_1111.pdf}
		    \end{center}
		&
		\vspace{-3mm}
		\begin{center}
		 \begin{equation}
		 w=s_{0468}:\quad
		 \begin{aligned}
			&F\mapsto -F \\ 
			&G\mapsto -G
		  \end{aligned}	
		\end{equation}
		\end{center}
		\\
		\end{tabular}%

	    \medskip
	    
		 Under factorization we obtain already minimal surface	
		$E_7^{(1)}/A_1^{(1)}$ 
		with parameters and coordinates
		
		\begin{tabular}{m{6cm}m{9cm}}
		\begin{center}
		 \begin{tikzpicture}[elt/.style={circle,draw=black!100,thick, inner sep=0pt,minimum size=2mm},scale=0.75]
			\path 	(-3,0) 	node 	(a1) [elt] {}
			(-2,0) 	node 	(a2) [elt] {}
			( -1,0) node  	(a3) [elt] {}
			( 0,0) 	node  	(a4) [elt] {}
			( 1,0) 	node 	(a5) [elt] {}
			( 2,0)	node 	(a6) [elt] {}
			( 3,0)	node 	(a7) [elt] {}
			( 0,1)	node 	(a0) [elt] {};
			\draw [black,line width=1pt ] (a1) -- (a2) -- (a3) -- (a4) -- (a5) --  (a6) -- (a7) (a4) -- (a0);
			\node at ($(a1.south) + (0,-0.3)$) 	{\small $a_7^2$};
		\node at ($(a2.south) + (0,-0.3)$)  {\small $a_5^2$};
		\node at ($(a3.north) + (0,0.3)$)  {\small $a_{1233}$};
		\node at ($(a4.south) + (0,-0.3)$)  {\small $1$};	
		\node at ($(a5.south) + (0,-0.3)$)  {\small $a_{12}$};	
		\node at ($(a6.south) + (0,-0.3)$) 	{\small $a_1^{-1}$};
		\node at ($(a7.south) + (0,-0.3)$) 	{\small $1$};
		\node at ($(a0.east) + (0.4,0)$) 	{\small $a_{12}^{-1}$};
			\end{tikzpicture}
		\end{center}	
			&
			
		\vspace{-1cm}	
		\begin{center}
		\begin{equation}
		\begin{aligned}
		             &\f=\frac{G-a_1^{-1}F}{G-F}\\
		             &\g=\frac{(a_{12}-1)(a_1-1)+a_2^{-1}F (G-F)}{(a_{12}-1)(a_1-1)+a_1F (G-a_1^{-1}F)}
		         \end{aligned}    
		       \end{equation}
		\end{center}

		\end{tabular}   
		
		\vspace{-1mm}
		  Symmetry lattice is $D_5^{(1)}$, nodal curves and lattice is $\overline{\Phi}_{\nod}=\Phi_{\nod}=3A_1$,
           symmetry group is as follows:
		
	        \begin{tabular}{m{9cm}m{6cm}}
	       \begin{center}
	         \begin{tikzpicture}[elt/.style={circle,draw=black!100,thick, inner sep=0pt,minimum size=2mm},scale=1.25]
			\path 	(-1,0) 	node 	(a0) [elt] {}
			(0,1) 	node 	(a1) [elt] {}
			(0,0) node  	(a2) [elt] {}
			(0,-1) node  	(a3) [elt] {}
			(1,0) node  	(a4) [elt] {};
			
		     	\node at ($(a0.north) + (0,0.1)$) 	{\small$\alpha_1$};
		     	\node at ($(a1.north) + (0,0.2)$) 	{\small$\underline{\alpha_{034455667}}$};
		     	\node at ($(a2.west) + (-0.2,0.2)$) 	{\small$\alpha_2$};
		     	\node at ($(a3.south) + (0,-0.2)$) 	{\small$\alpha_{345}$};
		     	\node at ($(a4.east) + (0.3,0)$) 	{\small$\alpha_{034}$};
		     	
		    \draw[<->, dashed] (a3) to [bend right=40] node[fill=white]{\small $s_{05}$} (a4);
		     \draw[<->, dashed] (a1) to [bend left=40] node[fill=white]{\small $s_{564}s_{57}s_{465}$} (a4);
		   
		    \draw [black,line width=1pt] (a0) -- (a2) -- (a4) (a1) -- (a2) -- (a3);
		   
                 \begin{scope}[xshift=4cm,scale=1]
			\path 
			(-1,1) 	node 	(a1) [elt] {}
			(-1,-1) node  	(a2) [elt] {}
			(0,0) node  	(a3) [elt] {};

		     	\node at ($(a1.north) + (0,0.1)$) 	{\small$\alpha_{045}$};
		     	\node at ($(a2.south) + (0,-0.1)$) 	{\small$\alpha_7$};
		     	\node at ($(a3.east) + (0.2,0)$) 	{\small$\alpha_4$};

		    \draw[<->, dashed] (a1) to [bend left=0] node[fill=white]{\small $s_{05}$} (a3); 
		     \draw[<->, dashed] (a2) to [bend left=0] node[fill=white]{\small $s_{564}s_{57}s_{465}$} (a3);  
		   
		   
		   \end{scope}
		    
			\end{tikzpicture} 
               \end{center}

               &
               \begin{center}
                \begin{equation}
                 S_3\ltimes W_{D_4}^a \label{symm_E8_E7}
                \end{equation}
                
                \medskip
                
                No projective reduction in the image.
                
               \end{center}
	      \end{tabular}

	\subsubsection{$A_1+ 2A_3$}
	
	\label{e8_133}
	
	%
 
			\end{center}

	 We have $2$ fixed points under $w$ and $2$ fixed points under $w^2$(small circles)
			   
			   \medskip

			 \begin{tabular}{|c|c|c|c|}
	    \hline 
	    sing. pt. & action & except. div. & Image\\
	    \hline
	     $(0, 0)$ &  $\begin{pmatrix}
	     0 & \ri \\
	     \ri & 0
	     \end{pmatrix}$   & $A_3$ & $\e_5-\e_6, \, \h_{\g}-\e_{15}, \, \e_1-\e_2$ \\
	     \hline
	     $(\infty,\infty)$ &  $\begin{pmatrix}
	     0 & -\ri \\
	     -\ri & 0
	     \end{pmatrix}$   & $A_3$ & $\h_{\f}-\e_{56}, \, \h_{\g}-\e_{34}, \, \h_{\f}-\e_{12}$ \\
	     \hline
	     \hline
	      $(0, \infty)$ &  $w^2: (F,G^{-1}): -\mathrm{Id}_{2\times 2}$   & $A_1$ &  \\
	     \cline{1-3}
	     $(\infty,0)$ &  $w^2: (F^{-1},G): -\mathrm{Id}_{2\times 2}$   & $A_1$ & $\e_3-\e_4$ \\
	     \hline
	    \end{tabular}
	    
	    \medskip
	    
	     Preimages of $(-1)$ curves:
            
            \medskip
            
	     \begin{tabular}{|c|c|c|c|c|c|}
	     \hline
	    Curve & $(0,0)$ & $(0,\infty)$ & $(\infty,0)$ & $(\infty,\infty)$ & Image \\
	    \hline
	    $F=0 \sqcup G=0$ & $dF=0 \sqcup dG=0$ & $dF=0$ & $dG=0$ & - & $\h_{\f}-\e_3$\\
	     \hline
	    $F=\infty \sqcup G=\infty$ & - & $dG^{-1}=0$ & $dF^{-1}=0$ & $dF^{-1}=0 \sqcup dG^{-1}=0$ & $\e_4$\\
	     \hline
	    $F=G$ & $dF=dG$ & - & - & $dF^{-1}=dG^{-1}$ & $\e_6$ \\
	     \hline
	    $F=-G$ & $dF=-dG$ & - & - & $dF^{-1}=-dG^{-1}$ & $\e_2$ \\
	     \hline
	    \end{tabular}
	   
	    \medskip
	   
	    Under factorization we obtain already minimal surface $D_5^{(1)}/A_3^{(1)}$
			with parameters and coordinates

			\begin{tabular}{m{7cm}m{8cm}}
	         \begin{center} 
	         \begin{tikzpicture}[elt/.style={circle,draw=black!100,thick, inner sep=0pt,minimum size=2mm}, scale=1]
		\path 	(-cos{60},-sin{60}) 	node 	(a0) [elt] {}
			(-cos{60},sin{60}) 	node 	(a1) [elt] {}
			( 0,0) node  	(a2) [elt] {}
			( 1,0) 	node  	(a3) [elt] {}
			( 1+cos{60},-sin{60}) 	node 	(a4) [elt] {}
			( 1+cos{60},sin{60})	node 	(a5) [elt] {};
			\draw [black,line width=1pt] (a1) -- (a2) -- (a3) -- (a4) (a3) -- (a5) (a2) -- (a0);
		\node at ($(a0.west) + (-0.25,0)$) 	{\small $a_5^4$};
		\node at ($(a1.west) + (-0.3,0)$) 	{\small $1$};
		\node at ($(a2.west) + (-0.35,0)$)  {\small $-a_2^2$};
		\node at ($(a3.east) + (0.3,0)$)  {\small $1$};
		\node at ($(a4.east) + (0.3,0)$)  {\small $1$};	
		\node at ($(a5.east) + (0.3,0)$)  {\small $1$};	
		\end{tikzpicture} 
		\end{center}
	        &
	        \vspace{-1cm}
	         \begin{center}
		   \begin{equation}
		   \begin{aligned}
	  &\f=\left(\frac{F-G}{F+G}\right)^2\\
	  &\g=1-\frac{F^2-G^2}{4{\ri}a_2}
	  \end{aligned}
	 \end{equation}    
	     \end{center}
	\end{tabular}

	 Symmetry lattice is $A_1^{(1)}$, nodal curves and lattice is $\overline{\Phi}_{\nod}{=}\Phi_{\nod}{=}A_1{+}A_3$,
           symmetry group is as follows:
	   
	   \begin{tabular}{m{9cm}m{6cm}}

             \begin{center} 
	    \begin{tikzpicture}[elt/.style={circle,draw=black!100,thick, inner sep=0pt,minimum size=2mm},scale=1]
			\path 	(-1,0) 	node 	(a1) [elt] {}
			(1,0) 	node 	(a2) [elt] {};

		    \draw [black,line width=2.5pt ] (a1) -- (a2);
		    
		     \draw [white,line width=1.5pt ] (a1) -- (a2);
		     
		     \draw [<->,black, line width=0.5pt]
		     (a1) -- (a2);
		     
		     \node at ($(a2.south) + (0,-0.2)$) 	{\small$\underline{\delta-\alpha_0}$};	
		     	
		     \node at ($(a1.south) + (0,-0.2)$) 	{\small $\alpha_0$};
		     

	         \begin{scope}[xshift=4cm, scale=1]
		\path 
			(-cos{60},sin{60}) 	node 	(a1) [elt] {}
			
			( 1,0) 	node  	(a3) [elt] {}
			( 1+cos{60},-sin{60}) 	node 	(a4) [elt] {}
			( 1+cos{60},sin{60})	node 	(a5) [elt] {};
			\draw [black,line width=1pt] (a5) -- (a3) -- (a4);
		
		\node at ($(a1.west) + (-0.3,0)$) 	{\small $\alpha_1$};
		
		\node at ($(a3.east) + (0.3,0)$)  {\small $\alpha_3$};
		\node at ($(a4.east) + (0.3,0)$)  {\small $\alpha_4$};	
		\node at ($(a5.east) + (0.3,0)$)  {\small $\alpha_5$};	
		
		
		\end{scope}
		\end{tikzpicture} 
		\end{center}
	    &
	    \begin{center}
	     \begin{equation}
	      W_{A_1}^a   \label{symm_E8_D5}
	     \end{equation}
	     
	     \medskip
	     
	     No projective reduction in the image.
	     
	    \end{center}

	\end{tabular}
	
                   \newpage

		\subsection{$E_7^{(1)}/A_1^{(1)}$ surface} \label{ssec:E7/A1}

		 Action on coordinates $F=f\nu_1^{-1},G=g\nu_1$  (where $f,g$ --- notations of \cite{KNY15}):
		
          \begin{center}
              \begin{tabular}{|c|c|c|c|c|c|c|c|c|c|}
              \hline
                  coord. & $s_0$ & $s_1$ & $s_2$ & $s_3$ & $s_4$ & $s_5$ & $s_6$ & $s_7$ &
                  $\pi$\\
                  \hline
                  $F$ & $1/G$ & $F$ & $F$ & $a_3^{-1}F$ &
               $a_4^{-1}\frac{(a_0a_4-1)FG+(1-a_0)F+a_0(1-a_4)}{(1-a_4^{-1})FG+(a_0-1)G+a_4^{-1}-a_0}$   & $F$ & $F$ & $F$ & $a_0a_4/F$
                  \\
                  \hline
                  $G$ & $1/F$ & $G$ & $G$ & $a_3G$ & $a_4G$ & $G$ & $G$ & $G$ & $a_4^{-1}/G$
                  \\
                  \hline
                \end{tabular}
          \end{center}

            Anticanonical divisor has two components, represented by:
            \begin{equation}
            FG=1, \quad FG=a_0 \label{E7curve}  
            \end{equation}

            Root data and point configuration:
            
            \begin{center}
	%
            \end{center}

		\subsubsection{$\pi$}
		
		\label{e7_p}

		       \begin{tabular}{m{7.5cm}m{7.5cm}}
 		        \begin{center} 
%
                        \includegraphics{e7_p.pdf}
			\end{center}
			&
			\vspace{-3mm}
			\begin{center}
			\begin{equation}
			w=\pi:\quad
			\begin{aligned}
			  &F\mapsto a_{04}/F\\ 
			  &G\mapsto a_4^{-1}/G 
			  \end{aligned}
			\end{equation} 
			\end{center}
                       \\
		       \end{tabular}
		
		\medskip
		
		 Under $w$ we have $4$ fixed points:
			   
			   \medskip

			 \begin{tabular}{|c|c|c|c|}
	    \hline 
	    sing. pt. & action & except. div. & Image \\
	    \hline
	     $(a_{04}^{1/2},a_4^{-1/2})$ &  $-\mathrm{Id}_{2\times 2}$  & $A_1$ & $\h_{\g}-\e_{12}$ \\
	     \hline
	     $(-a_{04}^{1/2},a_4^{-1/2})$ &  $-\mathrm{Id}_{2\times 2}$   & $A_1$  & $\h_{\f}-\e_{13}$ \\
	     \hline
	      $(a_{04}^{1/2},-a_4^{-1/2})$ & $-\mathrm{Id}_{2\times 2}$   & $A_1$ & $\h_{\f}-\e_{24}$\\
	      \hline
	       $(-a_{04}^{1/2},-a_4^{-1/2})$ & $-\mathrm{Id}_{2\times 2}$   & $A_1$ & $\h_{\g}-\e_{34}$\\
	       \hline
	    \end{tabular}%
			
		\medskip	
			
		 Under factorization we obtain already minimal surface	$E_8^{(1)}/A_0^{(1)}$
		with parameters and coordinates
		
		 \begin{tabular}{m{7.5cm}m{7.5cm}}
	     \begin{center}
	      \begin{tikzpicture}[elt/.style={circle,draw=black!100,thick, inner sep=0pt,minimum size=2mm},scale=0.75]
				\path 	(-2,0) 	node 	(a1) [elt] {}
					(-1,0) 	node 	(a2) [elt] {}
					( 0,0) node  	(a3) [elt] {}
					( 1,0) 	node  	(a4) [elt] {}
					( 2,0) 	node 	(a5) [elt] {}
					( 3,0)	node 	(a6) [elt] {}
					( 4,0)	node 	(a7) [elt] {}
					( 5,0)	node 	(a8) [elt] {}
					( 0,1)	node 	(a0) [elt] {};
				\draw [black,line width=1pt ] (a1) -- (a2) -- (a3) -- (a4) -- (a5) --  (a6) -- (a7) --(a8) (a3) -- (a0);
				\node at ($(a1.south) + (0,-0.3)$) 	{\small $-1$};
		\node at ($(a2.south) + (0,-0.3)$)  {\small $1$};
		\node at ($(a3.south) + (0,-0.3)$)  {\small $-1$};
		\node at ($(a4.north) + (0,0.3)$)  {\small $a_0^{1/2}$};	
		\node at ($(a5.south) + (0,-0.3)$)  {\small $-a_4^{1/2}$};	
		\node at ($(a6.north) + (0,0.3)$) 	{\small $a_3$};
		\node at ($(a7.south) + (0,-0.3)$) 	{\small $a_2$};
		\node at ($(a8.south) + (0,-0.3)$) 	{\small $a_1$};	
		\node at ($(a0.west) + (-0.5,0)$) 	{\small $a_0^{-1/2}$};
			\end{tikzpicture}
	     \end{center}
	     &
	     \vspace{-1cm}
	     \begin{center}
	      \begin{equation}
	      \begin{aligned}
               &\f=\frac{1-a_0^{1/2}}{a_4^{1/2}}\frac{F-a_0^{1/2} a_4 G}{FG-a_0^{1/2}} \\
               &\g=\frac{1+a_0^{1/2}}{a_4^{1/2}}\frac{F+a_0^{1/2} a_4 G}{FG+a_0^{1/2}}
               \end{aligned}
	       \end{equation}\\ 
	     \end{center}  
	    \end{tabular}

	   Symmetry lattice is $D_5^{(1)}$, nodal curves and lattice is $\Phi_{\nod}=4A_1, \, \overline{\Phi}_{\nod}=D_4$,
	   $\overline{\Phi}_{\nod}/\Phi_{\nod}$ is generated by $\e_4-\e_2=-\alpha_{12}$, $\chi(-\alpha_{12})=-1$, 	      	
           symmetry group is as follows:
	      
	      %
	      
	      where $s=s_{45342103} s_{02} s_{45342103}^{-1}$. Here we obtain additional $C_2^2=\Auti^{\nod}$, generated by
	      $s_{13}$ and $s_{430434}$, which changes root branches of $a_0^{1/2}$ and $a_4^{1/2}$ respectively.

	\subsubsection{$3A_1$}
	
	\label{e7_111}	
	
	%
		  \end{center}
	    
	     Under $w$ we have non-singular pointwise invariant curve $FG=-1$ and
	     $4$ fixed points on blowups on this curve:
			   
			   \medskip 
			
			 \begin{tabular}{|c|c|c|c|}
	    \hline 
	    sing. pt. & action & except. div. & Image \\
	    \hline
	     $\left(\frac{F+a_4}{G-a_4^{-1}}=-a_4^2 ,G=a_4^{-1}\right)$ &  $-\mathrm{Id}_{2\times 2}$
	      & $A_1$ & $\e_1-\e_2$ \\
	     \hline
	      $\left(\frac{F+a_{45}}{G-a_{45}^{-1}}=-a_{45}^2 ,G=a_{45}^{-1}\right)$ &  $-\mathrm{Id}_{2\times 2}$  & $A_1$  & $\e_3-\e_4$ \\
	     \hline
	       $\left(\frac{F+a_{456}}{G-a_{456}^{-1}}=-a_{456}^2 ,G=a_{456}^{-1}\right)$ &  $-\mathrm{Id}_{2\times 2}$  & $A_1$ & $\e_5-\e_6$\\
	      \hline
	        $\left(\frac{F+a_{4567}}{G-a_{4567}^{-1}}=-a_{4567}^2 ,G=a_{4567}^{-1}\right)$ & $-\mathrm{Id}_{2\times 2}$ & $A_1$ & $\e_7-\e_8$\\
	       \hline
	    \end{tabular}
	   
	   \medskip
	  
	   Under factorization we obtain non-minimal surface.
	    Blowing down $(-1)$ curve $\e^a$ we obtain 
	    surface $E_8^{(1)}/A_0^{(1)}$ with parameters and coordinates
	      
	       \begin{tabular}{m{7.5cm}m{7.5cm}}
	     
	     \begin{center}
	      \begin{tikzpicture}[elt/.style={circle,draw=black!100,thick, inner sep=0pt,minimum size=2mm},scale=0.75]
				\path 	(-2,0) 	node 	(a1) [elt] {}
					(-1,0) 	node 	(a2) [elt] {}
					( 0,0) node  	(a3) [elt] {}
					( 1,0) 	node  	(a4) [elt] {}
					( 2,0) 	node 	(a5) [elt] {}
					( 3,0)	node 	(a6) [elt] {}
					( 4,0)	node 	(a7) [elt] {}
					( 5,0)	node 	(a8) [elt] {}
					( 0,1)	node 	(a0) [elt] {};
				\draw [black,line width=1pt ] (a1) -- (a2) -- (a3) -- (a4) -- (a5) --  (a6) -- (a7) --(a8) (a3) -- (a0);
				
				\node at ($(a1.south) + (0,-0.3)$) 	{\small $a_2^2$};
		\node at ($(a2.south) + (0,-0.3)$)  {\small $a_4^2$};
		\node at ($(a3.south) + (0,-0.3)$)  {\small $a_5$};
		\node at ($(a4.south) + (0,-0.3)$)  {\small $1$};	
		\node at ($(a5.south) + (0,-0.3)$)  {\small $a_6$};	
		\node at ($(a6.south) + (0,-0.3)$) 	{\small $1$};
		\node at ($(a7.south) + (0,-0.3)$) 	{\small $a_7$};
		\node at ($(a8.south) + (0,-0.3)$) 	{\small $1$};	
		\node at ($(a0.west) + (-0.3,0)$) 	{\small $1$};
			\end{tikzpicture}
	     \end{center}

	     &
	     
	     \vspace{-1cm}
	     \begin{center}
	      \begin{equation}
	      \begin{aligned}
               &\f=2a_4 \frac{a_2^2F-G}{FG-1}\\
               &\g=2a_4 \frac{F-G}{FG-1}
               \end{aligned}
	       \end{equation}\\ 
	     \end{center}

	    \end{tabular}

	      Symmetry lattice is $D_5^{(1)}$, nodal curves and lattice is $\overline{\Phi}_{\nod}=\Phi_{\nod}=4A_1$,
           symmetry group is as follows 
           
           \medskip
           
	      %

			\subsubsection{$4A_1$}
			
			\label{e7_1111}
			
			\begin{tabular}{m{7.5cm}m{7.5cm}}
			
			\begin{center}
			
                        \includegraphics{e7_1111.pdf}
		        \end{center}
		        &
		        \vspace{-3mm}
		        \begin{center}
		        \begin{equation}
		        w=s_1s_3s_5s_7:\quad
		       \begin{aligned}
			  &F\mapsto -F\\
			  &G\mapsto -G
			  \end{aligned}
			\end{equation}
		        \end{center}
		        \\
                       \end{tabular}

                       \medskip
                       
                        Under $w$ we have $4$ fixed points:

			   \medskip 
			    
			 \begin{tabular}{|c|c|c|c|}
	    \hline 
	    sing. pt. & action & except. div. & Image\\
	    \hline
	     $(0,0)$ &  $-\mathrm{Id}_{2\times 2}$   & $A_1$ & $\e_7-\e_8$ \\
	     \hline
	     $(0,\infty)$ &  $(F,G^{-1}):-\mathrm{Id}_{2\times 2}$
	     & $A_1$ & $\h_{\g}-\e_{34}$ \\
	     \hline
	      $(\infty,0)$ &  $(F^{-1},G):-\mathrm{Id}_{2\times 2}$ & $A_1$ & $\h_{\g}-\e_{78}$ \\
	      \hline
	       $(\infty,\infty)$ & $(F^{-1},G^{-1}): -\mathrm{Id}_{2\times 2}$ & $A_1$ & $\e_3-\e_4$ \\
	       \hline
	    \end{tabular}%
		
	      		
	      		\medskip
	      		
			{Under factorization we obtain already minimal surface $D_5^{(1)}/A_3^{(1)}$
			with parameters and coordinates

			\begin{tabular}{m{7.5cm}m{7.5cm}}
	          
	         \begin{center} 
	         \begin{tikzpicture}[elt/.style={circle,draw=black!100,thick, inner sep=0pt,minimum size=2mm}, scale=1]
		\path 	(-cos{60},-sin{60}) 	node 	(a0) [elt] {}
			(-cos{60},sin{60}) 	node 	(a1) [elt] {}
			( 0,0) node  	(a2) [elt] {}
			( 1,0) 	node  	(a3) [elt] {}
			( 1+cos{60},-sin{60}) 	node 	(a4) [elt] {}
			( 1+cos{60},sin{60})	node 	(a5) [elt] {};
			\draw [black,line width=1pt] (a1) -- (a2) -- (a3) -- (a4) (a3) -- (a5) (a2) -- (a0);
		\node at ($(a0.west) + (-0.25,0)$) 	{\small $1$};
		\node at ($(a1.west) + (-0.25,0)$) 	{\small $1$};
		\node at ($(a2.west) + (-0.3,0)$)  {\small $a_0$};
		\node at ($(a3.east) + (0.3,0)$)  {\small $a_4^2$};
		\node at ($(a4.east) + (0.3,0)$)  {\small $a_2^2$};	
		\node at ($(a5.east) + (0.3,0)$)  {\small $a_6^2$};	
		
		\end{tikzpicture} 
		\end{center}
	        &
	        \vspace{-5mm}
	         \begin{center}
		   \begin{equation}
		\begin{aligned}
		&\f=\frac{FG-a_0}{FG-1}\\
	        &\g=G^2.
	        \end{aligned}
	       \end{equation}    
	     \end{center}
		       
	\end{tabular}	
	
	}

Symmetry lattice is $A_3^{(1)}$, nodal curves and lattice is $\overline{\Phi}_{\nod}=\Phi_{\nod}=2A_1$,
        symmetry group is as follows:
			    
			 \begin{tabular} {m{8cm}m{7cm}} 
			  \begin{center}
			      
			    \begin{tikzpicture}[elt/.style={circle,draw=black!100,thick, inner sep=0pt,minimum size=2mm},scale=1]
			      \path 	
			 (-1,1) 	node 	(a0) [elt] {}    
			(-1,-1) 	node 	(a1) [elt] {}
			(1,-1) 	node 	(a2) [elt] {}
			(1,1) 	node 	(a3) [elt] {}; 
			
			\draw [black,line width=1pt] (a0) -- (a1) -- (a2) -- (a3) -- (a0);
			
			 \node at ($(a0.north) + (0,0.2)$) 	{\small $\underline{\alpha_{01223}}$};	
		     	
		     \node at ($(a1.south) + (0,-0.2)$) 	{\small $\alpha_{4}$};
		     
		      \node at ($(a2.south) + (0,-0.2)$) 	{\small $\alpha_{3}$};
		      
		      \node at ($(a3.north) + (0,0.2)$) 	{\small $\alpha_{5}$};
			
			 \draw[<->,dashed] (a1) to[bend right=20] node[fill=white]{\small $\pi^2$} (a3);
			  \draw[<->,dashed] (a2) to[bend right=20] node[fill=white]{\small $s_{2012}$} (a0);
		      	

			    \begin{scope}[xshift=3cm]
			      \path 	
			 (0,1) 	node 	(a0) [elt] {}    
			(0,-1) 	node 	(a1) [elt] {};
			
			
			 \node at ($(a0.north) + (0,0.2)$) 	{\small $\alpha_0$};	
		     	
		     \node at ($(a1.south) + (0,-0.2)$) 	{\small $\alpha_1$};
		     
		     \end{scope}
		     
			      \end{tikzpicture}
			  \end{center}
                          &
			  \begin{center}
			   \begin{equation}
			    C_2^2 \ltimes W_{A_3}^a \label{symm_E7_D5}
			   \end{equation}
			   
			   \medskip
			   
			   No projective reduction in the image.

			  \end{center}

			 \end{tabular}
			 
			 \medskip

			\subsubsection{$3A_2$}
			
			\label{e7_222}

		\begin{tabular}{m{6cm}m{9cm}}
		
		\begin{center}
		

                  \includegraphics{e7_222.pdf}

		\end{center}
		&
		\vspace{-3mm}
		\begin{center}
		\begin{equation}
		w=s_{12}s_{40}s_{76}:\quad
		\begin{aligned}
			  &F\mapsto -\frac{\zeta+F+\zeta^{-1} FG}{1+\zeta^{-1} G+\zeta F G}\\
			  &G\mapsto \zeta^{-1}/F. 
			 \end{aligned} 
			\end{equation} 
		\end{center}
                \\
		\end{tabular}%
    
	
	\medskip
	
	Under factorization we obtain already minimal surface $E_3^{(1)}/A_5^{(1)}$ (in variant (b))
	with parameters and coordinates
	
	\begin{tabular}{m{6cm}m{9cm}}
		        \begin{center}
		         \begin{tikzpicture}[elt/.style={circle,draw=black!100,thick, inner sep=0pt,minimum size=2mm},scale=1]
			
			\begin{scope}[scale=1]
			\path 	
			(-1,0) 	node 	(a1) [elt] {}
			(1,0) 	node 	(a2) [elt] {}
			(0,1.73) 	node 	(a0) [elt] {};
			\draw [black,line width=1pt] (a0) -- (a1) -- (a2) -- (a0);
			\node at ($(a0.north) + (0,0.3)$) 	{\small $a_3^{-3}a_5^{-3}$};
			\node at ($(a1.west) + (-0.3,0)$) 	{\small $1$};
			\node at ($(a2.east) + (0.3,0)$) 	{\small $1$};

			\end{scope}

			\begin{scope}[xshift=2.25cm,scale=1]
			\path
			(0,1) 	node 	(a3) [elt] {}
			(2,1)	node 	(a4) [elt] {};

			 \draw [black,line width=2.5pt ] (a3) -- (a4);
		    
		     \draw [white,line width=1.5pt ] (a3) -- (a4);
		     
		     \draw [<->,black, line width=0.5pt]
		     (a3) -- (a4);			
			\node at ($(a3.west) + (-0.3,0)$) 	{\small $a_5^{-3}$};
			\node at ($(a4.east) + (0.3,0)$) 	{\small $a_3^{-3}$};
			
			\end{scope}
	              \end{tikzpicture}
		        \end{center}
                        &
                        \vspace{-5mm}
                        \begin{center}
                        \begin{equation}
                        \begin{aligned}
		              &\f=\frac{\zeta^{-1} - \zeta}3 \frac{F G - 1}{F G - \zeta} \frac{(\zeta^{-1} - 1) (F G - 1) - 3 \zeta F}
		              {\zeta + (1 - \zeta) G - F G}\\
		              &\g=\frac{\zeta^{-1} - \zeta}3 \frac{F G - 1}{F G - \zeta} \frac{(\zeta^{-1} - 1) (F G - 1) - 3 \zeta G}
		              {\zeta + (1 - \zeta) F - F G}
		       \end{aligned}
		       \end{equation} 
                        \end{center}
                        \\
                        
		       \end{tabular}
	
	    Symmetry lattice is $A_1^{(1)}$, nodal curves and lattice is $\overline{\Phi}_{\nod}=\Phi_{\nod}=A_2$,
           symmetry group is as follows:
	   
	   \begin{tabular}{m{9cm}m{6cm}}

             \begin{center} 
	   
	    \begin{tikzpicture}[elt/.style={circle,draw=black!100,thick, inner sep=0pt,minimum size=2mm},scale=2]
			\path 	(-1,0) 	node 	(a1) [elt] {}
			(1,0) 	node 	(a2) [elt] {};

		    \draw [black,line width=2.5pt ] (a1) -- (a2);
		    
		     \draw [white,line width=1.5pt ] (a1) -- (a2);
		     
		     \draw [<->,black, line width=0.5pt]
		     (a1) -- (a2);
		     
		     \node at ($(a2.north) + (0,0.1)$) 	{\small $\alpha_4$};	
		     	
		     \node at ($(a1.north) + (0,0.1)$) 	{\small$\underline{\alpha_3}$};
		     
		     \draw[<->, dashed] (a1) edge[bend right=30] node[fill=white]{$\pi^3$} (a2);
		     

	    \begin{scope}[xshift=2cm,scale=1]
			\path 	(0,0) 	node 	(a1) [elt] {}
			(1,0) 	node 	(a2) [elt] {};
			
		     \draw [line width=1pt]
		     (a1) -- (a2);
		     
		     \node at ($(a1.south) + (0,-0.1)$) 	{\small $\alpha_1$};	
		     	
		     \node at ($(a2.south) + (0,-0.1)$) 	{\small $\alpha_2$};
		     
		     
		     \end{scope}
		     
		     \end{tikzpicture}
	    \end{center}
	    &
	    \begin{center}
	     \begin{equation}
	      W_{A_1}^{ae} \label{symm_E7_E3}
	     \end{equation}
	     
	     \medskip
	     
	     No projective reduction in the image.
	     
	    \end{center}

	\end{tabular}

		\subsubsection{$2A_3$}
			
			\label{e7_33}
		
		\begin{tabular}{m{7.5cm}m{7.5cm}}
		
		\begin{center}
                      \includegraphics{e7_33.pdf}
			\end{center}
		&
		\vspace{-3mm}
		\begin{center}
		\begin{equation}
		w=s_{321}s_{765}:\quad
		\begin{aligned}
			  &F\mapsto -{\ri} F\\
			  &G\mapsto {\ri} G
			  \end{aligned}
			\end{equation} 
		\end{center}

		\end{tabular}
		
		See Sec.~\ref{sec:geom_gs}

	\subsubsection{$\pi \ltimes  5A_1$}
	\label{e7_p11111}
	
	%
		  \end{center}
	
	  We have $2$ fixed points under $w$ and $2$ fixed points under $w^2$ (small circles on the picture)
			   
			   \medskip 
			  
			 \begin{tabular}{|c|c|c|c|}
	    \hline 
	    sing. pt. & action & except. div. & Image \\
	    \hline
	     $(0, 0)$ &  $\begin{pmatrix}
	     0 & a_4\\
	     -a_4^{-1} & 0
	     \end{pmatrix}$   & $A_3$ & $\h_{\f}-\e_{16}, \, \e_1-\e_2, \, \h_{\g}-\e_{15}$  \\
	     \hline
	     $(\infty,\infty)$ &  $\begin{pmatrix}
	     0 & a_4^{-1} \\
	     -a_4 & 0
	     \end{pmatrix}$   & $A_3$ & $\h_{\f}-\e_{35}, \, \e_3-\e_4, \, \h_{\g}-\e_{36}$ \\
	     \hline
	     \hline
	      $(0, \infty)$ &  $w^2: -\mathrm{Id}_{2\times 2}$  & $A_1$ &  \\
	     \cline{1-3}
	     $(\infty,0)$ &  $w^2: -\mathrm{Id}_{2\times 2}$ & $A_1$ &  $\h_{\f\g}-\e_{1234}$ \\
	     \hline
	    \end{tabular}
	 
	   \medskip
	   
	    Preimages of $(-1)$ curves:
	    
	    \medskip
	    
	      \begin{tabular}{|c|c|c|c|c|c|}
	     \hline
	    Curve & $(0,0)$ & $(0,\infty)$ & $(\infty,0)$ & $(\infty,\infty)$ & Image \\
	    \hline
	    $F=0 \sqcup G=0$ & $dF=0 \sqcup dG=0$ & $dF=0$ & $dG=0$ & - & $\e_2$ \\
	     \hline
	    $F=\infty \sqcup G=\infty$ & - & $dG^{-1}=0$ & $dF^{-1}=0$ & $dF^{-1}=0 \sqcup dG^{-1}=0$ & $\e_4$ \\
	     \hline
	    $F=\ri a_4 G$ & $dF=\ri a_4 dG$ & - & - & $dG^{-1}=\ri a_4 dF^{-1}$ & $\e_5$\\
	     \hline
	    $F=-\ri a_4 G$ & $dF=-\ri a_4 dG$ & - & - & $dG^{-1}=-\ri a_4 dF^{-1}$ & $\e_6$\\
	     \hline
	    \end{tabular}
	     
	     \medskip
	 
	 Under factorization we obtain already minimal surface	
		$E_7^{(1)}/A_1^{(1)}$ 
		with parameters and coordinates
		
		\begin{tabular}{m{7.5cm}m{7.5cm}}
		\begin{center}
		 \begin{tikzpicture}[elt/.style={circle,draw=black!100,thick, inner sep=0pt,minimum size=2mm},scale=0.75]
			\path 	(-3,0) 	node 	(a1) [elt] {}
			(-2,0) 	node 	(a2) [elt] {}
			( -1,0) node  	(a3) [elt] {}
			( 0,0) 	node  	(a4) [elt] {}
			( 1,0) 	node 	(a5) [elt] {}
			( 2,0)	node 	(a6) [elt] {}
			( 3,0)	node 	(a7) [elt] {}
			( 0,1)	node 	(a0) [elt] {};
			\draw [black,line width=1pt ] (a1) -- (a2) -- (a3) -- (a4) -- (a5) --  (a6) -- (a7) (a4) -- (a0);
			\node at ($(a1.south) + (0,-0.3)$) 	{\small $1$};
		\node at ($(a2.south) + (0,-0.3)$)  {\small $-1$};
		\node at ($(a3.south) + (0,-0.3)$)  {\small $1$};
		\node at ($(a4.south) + (0,-0.3)$)  {\small $1$};	
		\node at ($(a5.south) + (0,-0.3)$)  {\small $-1$};	
		\node at ($(a6.south) + (0,-0.3)$) 	{\small $-{\ri}a_4$};
		\node at ($(a7.south) + (0,-0.3)$) 	{\small $a_2^2$};
		\node at ($(a0.west) + (-0.3,0)$) 	{\small $-1$};
			\end{tikzpicture}
		\end{center}	
			&
		\vspace{-5mm}
		\begin{center}
		\begin{equation}
		\begin{aligned}
		             &\f=\frac{(F - {\ri} a_4 G) - F G (F + {\ri} a_4 G)}{(F - {\ri} a_4 G) + F G (F + {\ri} a_4 G)} \\
		             &\g=\frac{(F + {\ri} a_4 G) + F G (F - {\ri} a_4 G)}{(F + {\ri} a_4 G) - F G (F - {\ri} a_4 G)}
		      \end{aligned}
		      \end{equation}
		\end{center}

		\end{tabular}   
	
               Symmetry lattice is $A_1^{(1)}$, nodal curves and lattice is $\Phi_{\nod}=2A_3, \,\overline{\Phi}_{\nod}=D_6$,
               $\overline{\Phi}_{\nod}/\Phi_{\nod}$ is generated by $-\h_{\f\g}+\e_{2456}=-\alpha_{01233445}$, $\chi(-\alpha_{01233445})=-1$, 	      	
           symmetry group is as follows:
	         
	         %
	         
	         where $\Auti^{\nod}=C_2^2$ come from nontrivial action of $W_{\Phi^{a,\perp w}}$.

		\newpage
			
		\subsection{$E_6^{(1)}/A_2^{(1)}$ surface}
		
			Action on coordinates $F=f\nu_1^{-1},G=g\nu_1$  (where $f,g$ --- notations of \cite{KNY15}):
          
            \begin{center}
	\begin{tabular}{|c|c|c | c | c | c| c | c |c|}
\hline
	coord. & $s_0$ & $s_1$ &
	$s_2$ & $s_3$ &
	$s_4$ & $s_5$ & $s_6$& $\pi$ \\
	\hline
	$F$ & $F$ & $F$ & $F\frac{G-1}{(1-a_2)FG+a_2G-1}$ & $a_3F$ & $F$ & $F$ & $a_6^{-1}F$ & 
	$\frac{FG-1}{(1+a_3)FG-a_3F-G}$\\
	\hline
	$G$ & $G$ & $G$ & $a_2G$ &
	$a_3^{-1}G$ & $G$ & $G$ & $a_6G\frac{F-1}{F+(a_6-1)FG-a_6}$ & $a_3F$\\
	\hline
\end{tabular}
\end{center}
            
            Root data and point configuration:

         	\begin{center}
	%
		\end{center}

		\subsubsection{$\pi$}
		
		\label{e6_p}
		
		\begin{tabular}{m{7.5cm}m{7.5cm}}
		
		\begin{center}
		
		\includegraphics{e6_p.pdf}
		\end{center}
		&
		\vspace{-3mm}
		\begin{center}
			\begin{equation}
			  w=\pi:\quad
			  \begin{aligned}
			  &F\mapsto \frac{FG-1}{FG(1+a_3)-a_3F-G}\\
			  &G\mapsto a_3 F
			  \end{aligned}
			\end{equation}
				
		\end{center}
		 
		\end{tabular}
		
		\medskip
		
		Under $w$ we have $3$ fixed points:
			   
			   \medskip
			    \begin{tabular}{|c|c|c|c|}
	    \hline 
	    sing. pt. & action & except. div. & Image\\
	    \hline
	      $(F=\frac{1-\zeta a_3^{1/3}-\zeta^{-1}a_3^{-1/3}}{1+a_3},G=a_3F)$ &   $\begin{pmatrix} -1 & -a_3^{-1}\\
	     a_3 & 0
	     \end{pmatrix}$ & $A_2$ & $\h_{\f}-\e_{13}, \, \h_{\g}-\e_{46}$ \\
	     \hline
	     $(F=\frac{1-\zeta^{-1}a_3^{1/3}-\zeta a_3^{-1/3}}{1+a_3},G=a_3F)$ &  $\begin{pmatrix} -1 & -a_3^{-1}\\
	     a_3 & 0
	     \end{pmatrix}$
	     & $A_2$ & $\h_{\f}-\e_{26}, \, \h_{\g}-\e_{35}$ \\
	     \hline
	     
	      $(F=\frac{1-a_3^{1/3}-a_3^{-1/3} }{1+a_3},G=a_3F)$ &  $\begin{pmatrix} -1 & -a_3^{-1}\\
	     a_3 & 0
	     \end{pmatrix}$
	     & $A_2$ & $\h_{\f}-\e_{45}, \, \h_{\g}-\e_{12}$ \\
	     \hline

	     \end{tabular}
	     
	     \medskip
	     
		Let is introduce curves 
		\begin{align}
		\mathrm{crv}_1(a_3^{1/3})=(a_3^{2/3}F{+}G{-}a_3^{1/3})-(1{-}a_3^{1/3}{+}a_3^{2/3}) F G,\\
	        \mathrm{crv}_2(a_3^{1/3})=(a_3 F{+}a_3^{-1/3}G{-}a_3^{1/3})-(1{-}a_3^{1/3}{+}a_3^{2/3}) F G
	        \end{align}
		 
		   \begin{center}
		       
		   \end{center}
			
	    Preimages of $(-1)$ curves:    
	    
	     \medskip

		       \begin{tabular}{|c|c|c|c|c|}
	       \hline
	           Curve & sing. pt. $1$ &
	           sing. pt. $2$ & sing. pt. $3$ & Image\\
	           \hline
	          $\mathrm{crv}_1(a_3^{1/3})=0$ &
	           $dF=\zeta a_3^{-1}dG$ & $dF=\zeta^{-1} a _3^{-1}dG$  & - & $\e_6$ \\
	           \hline
	           $\mathrm{crv}_1(\zeta^{-1}a_3^{1/3})=0$ &
	            $dF=\zeta^{-1} a _3^{-1}dG$ & - & $dF=\zeta a _3^{-1}dG$ & $\e_1$ \\
	           \hline
	          $\mathrm{crv}_1(\zeta a_3^{1/3})=0$ &
	           - & $dF=\zeta a _3^{-1}dG$ &  $dF=\zeta^{-1} a _3^{-1}dG$ & $\e_5$ \\
	           \hline
	           $\mathrm{crv}_2(a_3^{1/3})=0$ &
	           $dF=\zeta^{-1} a_3^{-1}dG$  & $dF=\zeta a_3^{-1}dG$ & - & $\e_3$ \\
	           \hline
	           $\mathrm{crv}_2(\zeta^{-1}a_3^{1/3})=0$ &
	           $dF=\zeta a_3^{-1}dG$ & - & $dF=\zeta^{-1} a_3^{-1}dG$ & $\e_4$ \\
	           \hline
	           $\mathrm{crv}_2(\zeta a_3^{1/3})=0$ &
	           - & $dF=\zeta^{-1} a_3^{-1}dG$  & $dF=\zeta a_3^{-1}dG$ & $\e_2$ \\
	           \hline
	           \end{tabular}

	     \medskip

	     We already obtain minimal $E_8^{(1)}/A_0^{(1)}$ with parameters and coordinates
	    
	    \begin{tabular}{m{5.5cm}m{9.5cm}}
	     
	     \begin{center}
	      \begin{tikzpicture}[elt/.style={circle,draw=black!100,thick, inner sep=0pt,minimum size=2mm},scale=0.75]
				\path 	(-2,0) 	node 	(a1) [elt] {}
					(-1,0) 	node 	(a2) [elt] {}
					( 0,0) node  	(a3) [elt] {}
					( 1,0) 	node  	(a4) [elt] {}
					( 2,0) 	node 	(a5) [elt] {}
					( 3,0)	node 	(a6) [elt] {}
					( 4,0)	node 	(a7) [elt] {}
					( 5,0)	node 	(a8) [elt] {}
					( 0,1)	node 	(a0) [elt] {};
				\draw [black,line width=1pt ] (a1) -- (a2) -- (a3) -- (a4) -- (a5) --  (a6) -- (a7) --(a8) (a3) -- (a0);
				
				\node at ($(a1.south) + (0,-0.3)$) 	{\small $\zeta $};
		\node at ($(a2.south) + (0,-0.3)$)  {\small $1$};
		\node at ($(a3.south) + (0,-0.3)$)  {\small $\zeta^{-1}$};
		\node at ($(a4.south) + (0,-0.3)$)  {\small $\zeta^{-1}$};	
		\node at ($(a5.north) + (0,0.3)$)  {\small $a_3^{1/3}$};	
		\node at ($(a6.south) + (0,-0.3)$) 	{\small $\zeta$};
		\node at ($(a7.north) + (0,0.3)$) 	{\small $a_3^{1/3} a_6$};
		\node at ($(a8.south) + (0,-0.3)$) 	{\small $a_0$};	
		\node at ($(a0.west) + (-0.7,0)$) 	{\small $a_3^{-1/3}$};
			\end{tikzpicture}
	     \end{center}

	     &
	     
	     \begin{center}
	       \small
	      \begin{equation}
               \begin{aligned}
               &\f=a_3^{-1/3} \frac{\zeta a_3 + (a_3 F - \zeta G)^2 - (1 + a_3) (a_3 F + \zeta^{-1} G) F G}
               {a_3 F + \zeta G + (\zeta^{-1}{-}1) (a_3 F{-}\zeta^{-1} G) F G- \zeta^{-1} (1{+}a_3)  F^2 G^2}\\
               &\g=\zeta^{-1} a_3^{-1/3} \frac{\zeta^{-1} a_3 + (a_3 F - \zeta^{-1} G)^2 - (1 + a_3) (a_3 F + \zeta G) F G}
               {a_3 F + \zeta^{-1} G + (\zeta{-}1) (a_3 F{-}\zeta G) F G - \zeta (1{+}a_3) F^2 G^2}
	       \end{aligned}
	       \end{equation}\\ 
	     \end{center}

	    \end{tabular}

	   Symmetry lattice is $A_2^{(1)}$, nodal curves and lattice is  $\Phi_{\nod}=3A_2,\,\overline{\Phi}_{\nod}=E_6$,
	   $\overline{\Phi}_{\nod}/\Phi_{\nod}$ is generated by $-\h_{\f\g\g}+\e_{123456}=-\alpha_{0^212^33^44^35^26}$,
	   $\chi(-\alpha_{0^212^33^44^35^26})=\zeta^{-1}$, symmetry group is as follows:
	    
	     %
     
	    where $s=s_{0304} s_{135} s_{4030}$ generate Galois group, $s_{136} s$ changes branch of root $a_3^{1/3}$,
	    they generate additional group $S_3=\Auti^{\nod}$, corresponding to permutation of the fixed points in the preimage.
	    
			\subsubsection{$2A_2$}
			
			\label{e6_22}
			
			%
		   \end{center}
		
		\medskip
		    
		  Under $w$ we have pointwise invariant curve $G=0$ and
	     $4$ fixed singular points on blowups on this curve:
		    
		    \medskip
		    
		    \begin{tabular}{|c|c|c|c|}
	    \hline 
	    sing. pt. & action & except. div. & Image\\
	    \hline
	     $((F-a_6)/G=-a_6(a_6-1),G=0)$ &   $\begin{pmatrix}\zeta^{-1} & *\\
	     0 & \zeta
	     \end{pmatrix}$ & $A_2$ & $\e_2-\e_3, \, \h_{\g}-\e_{12}$ \\
	     \hline
	     $((F-a_{06})/G=-a_{06}(a_{06}-1),G=0)$ &  $\begin{pmatrix}\zeta^{-1} & *\\
	     0 & \zeta
	     \end{pmatrix}$
	     & $A_2$ & $\e_4-\e_5, \, \h_{\f}-\e_{14}$ \\
	     \hline
	     
	      $(0,\infty)$ &  $(F,G^{-1}):\begin{pmatrix}\zeta^{-1} & 0\\
	     0 & \zeta^{-1}
	     \end{pmatrix}$
	     & $(-3)$ & $\e^c-\e^a-\e^b\mapsto p$ \\
	     \hline
	     
	     $(1,\infty)$ &  $(F-1,G^{-1}):\begin{pmatrix}\zeta & 0\\
	     0 & \zeta^{-1}
	     \end{pmatrix}$
	     & $A_2$ & $\e_6-\e_7, \, \e_7-\e_8$ \\
	     \hline
	     
	     \end{tabular}
	     
	     \medskip
	     
	       Preimages of $(-1)$ curves:
	    
	     \medskip
		       \begin{tabular}{|c|c|c|c|c|}
	       \hline
	           Curve & sing. pt. $1$ &
	          sing. pt. $2$ & $(1,\infty)$ & Image\\
	           \hline
	           $F=1$ & - & - & $dF=0$ & $\e_8$ \\
	           \hline 
	           $G (1{+}G) (1{-}F) = a_3^3 F{-}1$ & - & -  & $dF=0$ & $\h_{\f\g}-\e_{167}$ \\
	           \hline
	           $\mathrm{crv}(a_6)=0$   &  {\small $d\frac{F{-}a_6}G{=}a_6(a_6{-}1)^2 dG$} & -  & $dG^{-1}=0$ & $\h_{\f}-\e_6$ \\
	           \hline
	           $\mathrm{crv}(a_{06})=0$  & - & {\small $d\frac{F{-}a_{06}}G{=}a_{06}(a_{06}{-}1)^2 dG$} & $dG^{-1}=0$ &  $\h_{\g}-\e_6$ \\
	           \hline
	           (*) $\f=1+\frac1{qa_{06}}$ &
	            $d\frac{F{-}a_6}G{=}a_6(a_6{-}1)^2 dG$ &  $d\frac{F{-}a_{06}}G{=}a_{06}(a_{06}{-}1)^2 dG$ & - & $\e_1$ \\
	           \hline
	           \end{tabular}
	        
	     \medskip
	     
	     where 
	     $\mathrm{crv}(a_6)=G (G+1) (F-1)^2+(a_3^3-a_6^{-1}) (a_6-1)F^2G+(F-a_6) (a_3^3 F-a_6^{-1})=0$.
	     Last curve is complicated $(3,2)$ curve, which is not drawn in the preimage (*), for its explicit formula see \eqref{fg_E6_2A_2}.
	     
	    Under factorization we obtain non-minimal surface.
	    Blowing down $(-1)$ curves $\e^a, \e^b$ and then $(-1)$ curve $\e^c$, we obtain 
	    surface $E_8^{(1)}/A_0^{(1)}$ with parameters and coordinates

	     \begin{tabular}{m{7cm}m{8cm}}
	     
	     \begin{center}
	      \begin{tikzpicture}[elt/.style={circle,draw=black!100,thick, inner sep=0pt,minimum size=2mm},scale=0.75]
				\path 	(-2,0) 	node 	(a1) [elt] {}
					(-1,0) 	node 	(a2) [elt] {}
					( 0,0) node  	(a3) [elt] {}
					( 1,0) 	node  	(a4) [elt] {}
					( 2,0) 	node 	(a5) [elt] {}
					( 3,0)	node 	(a6) [elt] {}
					( 4,0)	node 	(a7) [elt] {}
					( 5,0)	node 	(a8) [elt] {}
					( 0,1)	node 	(a0) [elt] {};
				\draw [black,line width=1pt ] (a1) -- (a2) -- (a3) -- (a4) -- (a5) --  (a6) -- (a7) --(a8) (a3) -- (a0);
				
				\node at ($(a1.south) + (0,-0.3)$) 	{\small $a_0^{-1}$};
		\node at ($(a2.south) + (0,-0.3)$)  {\small $1$};
		\node at ($(a3.south) + (0,-0.3)$)  {\small $1$};
		\node at ($(a4.south) + (0,-0.3)$)  {\small $a_0$};	
		\node at ($(a5.south) + (0,-0.3)$)  {\small $1$};	
		\node at ($(a6.north) + (0,0.3)$) 	{\small $a_{06}^{-1}$};
		\node at ($(a7.south) + (0,-0.3)$) 	{\small $1$};
		\node at ($(a8.south) + (0,-0.3)$) 	{\small $1$};	
		\node at ($(a0.west) + (-0.7,0)$) 	{\small $qa_6$};
			\end{tikzpicture}
	     \end{center}

	     &
	     \vspace{-5mm}
	     \begin{center}
	      \begin{equation}\label{fg_E6_2A_2}
	      \begin{aligned}
               &\f=a_{06}^{-1} H F + \frac1{q  H F} + \frac{(H - 1)^3}{q H F (1 - a_6^{-1} H F )}\\
                &\g=a_{6}^{-1} H F + \frac1{q  H F} + \frac{(H - 1)^3}{q H F (1 - a_{06}^{-1} H F )}
               \end{aligned}
	       \end{equation}\\ 
	     \end{center}

	    \end{tabular}
	    
	    where
	     \begin{equation}
	     H=\frac{G-1}{FG-1}, \quad q=a_0a_6^2a_3^3  
	     \end{equation}

	     Symmetry lattice is $A_2^{(1)}$, nodal curves and lattice is $\overline{\Phi}_{\nod}=\Phi_{\nod}=3A_2$,
           symmetry group is as follows:
	     
	     %

			\subsubsection{$4A_1$}	  
			
			\label{e6_1111}
			
			\begin{tabular}{m{7.5cm}m{7.5cm}}
			
			\begin{center}
			
                       \includegraphics{e6_1111.pdf}
			\end{center}
			&
			\vspace{-3mm}
			\begin{center}
			 \begin{equation}
			 w=s_0s_1s_3s_5:\quad
			 \begin{aligned}
			  &F\mapsto -F\\
			  &G\mapsto -G  
			 \end{aligned} 
			\end{equation}
			\end{center}

			\end{tabular}
			
			\medskip
			
			 Under $w$ we have $4$ fixed points:

			   \medskip 
			    
			 \begin{tabular}{|c|c|c|c|}
	    \hline 
	    sing. pt. & action & except. div. & Image \\
	    \hline
	     $(0,0)$ &  $-\mathrm{Id}_{2\times 2}$   & $A_1$ & $\h_{\f}-\e_{12}$ \\
	     \hline
	     $(0,\infty)$ &  $(F,G^{-1}):-\mathrm{Id}_{2\times 2}$
	     & $A_1$ & $\h_{\g}-\e_{24}$ \\
	     \hline
	      $(\infty,0)$ &  $(F^{-1},G):-\mathrm{Id}_{2\times 2}$ & $A_1$ & $\h_{\g}-\e_{17}$ \\
	      \hline
	       $(\infty,\infty)$ & $(F^{-1},G^{-1}):-\mathrm{Id}_{2\times 2}$ & $A_1$ & $\h_{\f}-\e_{47}$\\
	       \hline
	    \end{tabular}%
	
	    
	    \medskip
	    
		       Under factorization 
		       we obtain already minimal surface $E_3^{(1)}/A_5^{(1)}$
		       in variant (a)
		       	with parameters and coordinates (in variant a) and variant b))
		       	
		       \begin{tabular}{m{7cm}m{4cm}m{4cm}}
		        \begin{center}
		         \begin{tikzpicture}[elt/.style={circle,draw=black!100,thick, inner sep=0pt,minimum size=2mm},scale=1]
			
			\begin{scope}[scale=1]
			\path 	
			(-1,0) 	node 	(a1) [elt] {}
			(1,0) 	node 	(a2) [elt] {}
			(0,1.73) 	node 	(a0) [elt] {};
			\draw [black,line width=1pt] (a0) -- (a1) -- (a2) -- (a0);
			\node at ($(a0.west) + (-0.3,0)$) 	{\small $a_2^{-2}$};
			\node at ($(a1.west) + (-0.3,0)$) 	{\small $a_6^{-2}$};
			\node at ($(a2.east) + (0.3,0)$) 	{\small $a_4^{-2}$};

			\end{scope}

			\begin{scope}[xshift=2.5cm,scale=1]
			\path
			(0,1) 	node 	(a3) [elt] {}
			(2,1)	node 	(a4) [elt] {};

			 \draw [black,line width=2.5pt ] (a3) -- (a4);
		    
		     \draw [white,line width=1.5pt ] (a3) -- (a4);
		     
		     \draw [<->,black, line width=0.5pt]
		     (a3) -- (a4);			
			\node at ($(a3.west) + (-0.3,0)$) 	{\small $a_{246}^{-2}$};
			\node at ($(a4.east) + (0.3,0)$) 	{\small $1$};
			
			\end{scope}
	              \end{tikzpicture}
		        \end{center}
                        &
                        \vspace{-5mm}
                        \begin{center}
                        \begin{equation}
                        \begin{aligned}
		              &\f=1-\frac{1}{FG}\\
		              &\g=a_6^2 G/F
		         \end{aligned}
		       \end{equation} 
                        \end{center}
                        &
                        \vspace{-5mm}
                         \begin{center}
                        \begin{equation}
                        \begin{aligned}
		              &\f=\frac{G-F^{-1}}{G-F}\\
		              &\g=G/F
		              \end{aligned}
		       \end{equation} 
                        \end{center}
                        \\
                        
		       \end{tabular}

                       Symmetry lattice is $A_2^{(1)}$, nodal curves and lattice is $\overline{\Phi}_{\nod}=\Phi_{\nod}=A_1$,
           symmetry group is as follows:
                      
                      \begin{tabular}{m{8cm}m{7cm}}
                       \begin{center}
                         \begin{tikzpicture}[elt/.style={circle,draw=black!100,thick, inner sep=0pt,minimum size=2mm},scale=1]
			\path 	
			(-1,0) 	node 	(a1) [elt] {}
			(1,0) 	node 	(a2) [elt] {}
			(0,1.73) 	node 	(a0) [elt] {};
			\draw [black,line width=1pt] (a0) -- (a1) -- (a2) -- (a0);
			\node at ($(a0.west) + (-0.3,0.1)$) 	{\small $\alpha_0$};
			\node at ($(a1.west) + (-0.2,0)$) 	{\small $\alpha_1$};
			\node at ($(a2.east) + (0.2,0)$) 	{\small $\alpha_2$};
			\draw[<-,dashed] (a0) edge[bend right=40] node[fill=white]{$\pi^2$} (a1);
			\draw[<-,dashed] (a1) edge[bend right=40] node[fill=white]{$\pi^2$} (a2);
			\draw[<-,dashed] (a2) edge[bend right=40] node[fill=white]{$\pi^2$} (a0);

	             
                       \begin{scope}[xshift=4cm,scale=1]
			\path 	
			(0,0) 	node 	(a0) [elt] {};
			
			\node at ($(a0.south) + (0,-0.2)$) 	{\small $\alpha_4$};
			
			
			\end{scope}
			
		       \end{tikzpicture}
                     
                     \end{center}
                     
                       &
                       \begin{center}
                        \begin{equation}
                         W_{A_2}^{ae} \label{symm_E6_E3}
                        \end{equation}
                        
                        \medskip
                        
                        No projective reduction in the image.
                        
                       \end{center}

                      \end{tabular}
                      
                      \medskip
                    
                 \newpage

		\subsection{$D_5^{(1)}/A_3^{(1)}$ surface} \label{ssec:geom:D5/A3}

		Action on coordinates $F=f/\nu_3, G=g \nu_1$  (where $f,g$ --- notations of \cite{KNY15}):
		
        \begin{center}
	\begin{tabular}{|c|c|c | c | c | c | c | c|}
	\hline
	coord. & $s_0$ & $s_1$ &
	$s_2$ & $s_3$ &
	$s_4$ & $s_5$ & $\pi$ \\
	\hline
	$F$ &$F$ & $a_1 F$ & $a_2^{-1}F$ & $F \frac{G-1}{G-a_3^{-1}}$ &
	$F$ & $F$ & $1/G$\\
	\hline
	$G$& $G$ & $G$ & $G \frac{F-1}{F-a_2}$ & $a_3G$ & $a_4^{-1}G$ & $G$ & $F/a_2$\\
	\hline
	\end{tabular}
	\end{center}

       Root data and point configuration:
       
       \begin{center}
	%
            \end{center}

		\subsubsection{$\pi^2$}
		\label{d5_pp}
		
		%
		  \end{center}

		 Under factorization we obtain already minimal surface $E_7^{(1)}/A_1^{(1)}$
			with parameters and coordinates
			
		\begin{tabular}{m{7.5cm}m{7.5cm}}
		\begin{center}
		 \begin{tikzpicture}[elt/.style={circle,draw=black!100,thick, inner sep=0pt,minimum size=2mm},scale=0.75]
			\path 	(-3,0) 	node 	(a1) [elt] {}
			(-2,0) 	node 	(a2) [elt] {}
			( -1,0) node  	(a3) [elt] {}
			( 0,0) 	node  	(a4) [elt] {}
			( 1,0) 	node 	(a5) [elt] {}
			( 2,0)	node 	(a6) [elt] {}
			( 3,0)	node 	(a7) [elt] {}
			( 0,1)	node 	(a0) [elt] {};
			\draw [black,line width=1pt ] (a1) -- (a2) -- (a3) -- (a4) -- (a5) --  (a6) -- (a7) (a4) -- (a0);
			\node at ($(a1.south) + (0,-0.3)$) 	{\small $a_5^{-1}$};
		\node at ($(a2.south) + (0,-0.3)$)  {\small $a_3^{-1/2}$};
		\node at ($(a3.south) + (0,-0.3)$)  {\small $-1$};
		\node at ($(a4.south) + (0,-0.3)$)  {\small $1$};	
		\node at ($(a5.south) + (0,-0.3)$)  {\small $-1$};	
		\node at ($(a6.south) + (0,-0.3)$) 	{\small $a_2^{-1/2}$};
		\node at ($(a7.south) + (0,-0.3)$) 	{\small $a_1^{-1}$};
		\node at ($(a0.west) + (-0.3,0)$) 	{\small $-1$};
			\end{tikzpicture}
		\end{center}	
			&
		\vspace{-6mm}
		\begin{center}
		\begin{equation}
		\begin{aligned}
		             &\f=-\frac{a_2^{-1/2}a_3^{1/2}FG-1}{a_2^{-1/2}F-a_3^{1/2}G}\\
		             &\g=-\frac{a_2^{-1/2}F+a_3^{1/2}G}{a_2^{-1/2}a_3^{1/2}FG+1}
		    \end{aligned}
		       \end{equation}
		\end{center}

		\end{tabular}     	
		
		 Symmetry lattice is $(3A_1)^{(1)}$, nodal curves and lattice is  $\Phi_{\nod}=4A_1,\,\overline{\Phi}_{\nod}=D_4$,
		 $\overline{\Phi}_{\nod}/\Phi_{\nod}$ is generated by
		$-\h_{\f\g}+\e_{1256}=-\alpha_{03445}$, $\chi(-\alpha_{03445})=-1$, symmetry group is as follows:
		
		%
		
	         \vspace{-1mm}
                where $s=s_{234} s_{05} s_{432}$, $\td{s}=s_{654}s_{03}s_{456}=\pi s\pi$.
                Here we obtain additional $C_2^2=\Auti^{\nod}$, generated by
	      $s_{05}$ and $s_{03}$, which changes root branches of $a_2^{1/2}$ and $a_3^{1/2}$ respectively.

		\subsubsection{$2A_1$}
		
		\label{d5_11}

			\begin{tabular}{m{7.5cm}m{7.5cm}}
		
		\begin{center}
                \includegraphics{d5_11.pdf}
		
		\end{center}
		&
		\vspace{-3mm}
		\begin{center}
		\begin{equation}
		w=s_4s_5:\quad
		\begin{aligned}
		&F\mapsto F\\
		&G\mapsto -G 
		\end{aligned}
		\end{equation}

		\end{center}
		
	        \end{tabular}%

			  Under $w$ we have non-singular pointwise invariant curves $G=0, G=\infty$ and
	     $4$ fixed points on blowups on these curves:

			   \medskip 
			    
			 \begin{tabular}{|c|c|c|c|}
	    \hline 
	    sing. pt. & action & except. div. & Image\\
	    \hline
	     $((F-a_2)/G=0,G=0)$ & $-\mathrm{Id}_{2\times 2}$   & $A_1$ & $\e_5-\e_6$ \\
	     \hline
	     $((F-a_2a_0)/G=0,G=0)$ & $-\mathrm{Id}_{2\times 2}$
	     & $A_1$ & $\e_7-\e_8$ \\
	     \hline
	      $((F-1)/G^{-1}=0,G^{-1}=0)$ &  $-\mathrm{Id}_{2\times 2}$ & $A_1$ & $\e_1-\e_2$ \\
	      \hline
	       $((F-a_1^{-1})/G^{-1}=0,G^{-1}=0)$ & $-\mathrm{Id}_{2\times 2}$ & $A_1$ & $\e_3-\e_4$ \\
	       \hline
	    \end{tabular}
	    
	    \medskip

		     Under factorization we obtain non-minimal surface.
	             Blowing down $(-1)$ curves $\e^a, \e^b$ we obtain 
	             surface  $E_7^{(1)}/A_1^{(1)}$ with parameters and coordinates
		       	
		       \begin{tabular}{m{7.5cm}m{7.5cm}}
		\begin{center}
		 \begin{tikzpicture}[elt/.style={circle,draw=black!100,thick, inner sep=0pt,minimum size=2mm},scale=0.75]
			\path 	(-3,0) 	node 	(a1) [elt] {}
			(-2,0) 	node 	(a2) [elt] {}
			( -1,0) node  	(a3) [elt] {}
			( 0,0) 	node  	(a4) [elt] {}
			( 1,0) 	node 	(a5) [elt] {}
			( 2,0)	node 	(a6) [elt] {}
			( 3,0)	node 	(a7) [elt] {}
			( 0,1)	node 	(a0) [elt] {};
			\draw [black,line width=1pt ] (a1) -- (a2) -- (a3) -- (a4) -- (a5) --  (a6) -- (a7) (a4) -- (a0);
			\node at ($(a1.south) + (0,-0.3)$) 	{\small $1$};
		\node at ($(a2.south) + (0,-0.3)$)  {\small $a_1$};
		\node at ($(a3.south) + (0,-0.3)$)  {\small $1$};
		\node at ($(a4.south) + (0,-0.3)$)  {\small $a_2$};	
		\node at ($(a5.south) + (0,-0.3)$)  {\small $1$};	
		\node at ($(a6.south) + (0,-0.3)$) 	{\small $a_0$};
		\node at ($(a7.south) + (0,-0.3)$) 	{\small $1$};
		\node at ($(a0.west) + (-0.4,0)$) 	{\small $a_3^2$};
			\end{tikzpicture}
		\end{center}	
			&
		\vspace{-5mm}
		\begin{center}
		\begin{equation}
	        \begin{aligned}
	        &\f=F\frac{G^2-1}{G^2-a_3^{-2}}\\
	        &\g=1/F
	        \end{aligned}
	        \end{equation}
		\end{center}

		     \end{tabular}
		     
		     Symmetry lattice is $(3A_1)^{(1)}$, nodal curves and lattice is $\overline{\Phi}_{\nod}=\Phi_{\nod}=4A_1$,
           symmetry group is as follows:
		     
		     %

			\subsubsection{$\pi$}
			
			\label{d5_p}

			%
		  \end{center}
			
			Under $w$ we have $2$ fixed points with  and $2$ additional fixed points under $w^2$ (small circles).
			  
			   \medskip

			 \begin{tabular}{|c|c|c|c|}
	    \hline 
	    sing. pt. & action & except. div. & Image\\
	    \hline
	     $(a_2^{1/2},a_2^{-1/2})$ &
	     $\begin{pmatrix}
	     0 & -a_2\\
	     a_2^{-1} & 0
	     \end{pmatrix}$
	      & $A_3$ & $\h_{\f}-\e_{24}, \, \e_2-\e_3, \, \h_{\g}-\e_{12}$ \\
	     \hline
	     $(-a_2^{1/2},-a_2^{-1/2})$ &  $\begin{pmatrix}
	     0 & -a_2\\
	     a_2^{-1} & 0
	     \end{pmatrix}$   & $A_3$ & $\h_{\f}-\e_{15}, \, \e_5-\e_6, \, \h_{\g}-\e_{45}$ \\
	     \hline
	     
	     $(a_2^{1/2},-a_2^{-1/2})$ &  $w^2:-\mathrm{Id}_{2\times 2}$   & $A_1$ & $\h_{\f\g}-\e_{2356}$ \\
	     \cline{1-3}
	     
	     $(-a_2^{1/2},a_2^{-1/2})$ &  
	     $w^2:-\mathrm{Id}_{2\times 2}$ & $A_1$ & \\
	     \hline
	     \end{tabular}
		
		\medskip
		
		  Preimages of $(-1)$ curves:
	    
	    \medskip
	    
	    	 \begin{tabular}{|c|c|c|c|c|c|}
	    \hline
	    Curve & $(a_2^{1/2},a_2^{-1/2})$ & $(-a_2^{1/2},a_2^{-1/2})$ & $(a_2^{1/2},-a_2^{-1/2})$ & $(-a_2^{1/2},-a_2^{-1/2})$ & Image \\
	    \hline
	    $F=a_2^{1/2} \sqcup G=a_2^{-1/2}$ & $dF{=}0 \sqcup dG{=}0$ & $dG=0$ & $dF=0$ & - & $\e_3$ \\
	     \hline
	    $F=-a_2^{1/2} \sqcup G=-a_2^{-1/2}$ & - & $dF=0$ & $dG=0$ & $dF{=}0 \sqcup dG{=}0$ & $\e_6$\\
	     \hline
	    $FG{-}1=\ri (a_2^{-1/2}F{-}a_2^{1/2}G)$ & $dF=\ri a_2 dG$ & - & - & $dF=-\ri a_2 dG$ & $\e_1$ \\
	    \hline
	    $FG{-}1=-\ri (a_2^{-1/2}F{-}a_2^{1/2}G)$ & $dF=-\ri a_2 dG$ & - & - & $dF=\ri a_2 dG$ & $\e_4$ \\
	    \hline
	    \end{tabular}%
  	
		       	
	 Symmetry lattice is $A_1^{(1)}$, nodal curves and lattice is  $\Phi_{\nod}=A_1+2A_3,\,\overline{\Phi}_{\nod}=E_7$,
	$\overline{\Phi}_{\nod}/\Phi_{\nod}$ is generated by $-\h_{\f\g\g}+\e_{123456}=-\alpha_{0012223333444556}$,
	$\chi(-\alpha_{0012223333444556})=-\ri$, symmetry group is as follows:
	
	%
	
	where $s=s_{453}s_{046}s_{354}$ changes branch of root $a_2^{1/2}$,  $s_{031430}$ generate Galois group,
	together they generate additional group $C_2^2=\Auti^{\nod}$, coresponding to permutation of
	fixed points and invariant curves in the preimage.

	        \subsubsection{$A_3+ A_1$}	
	        \label{d5_13}

	        %
	          
	        	   Under $w$ we have non-singular pointwise invariant curves $G=0$,
	     $2$ fixed points on blowups on this curve and another $2$ fixed points; under $w^2$ we have additionally
	     pointwise invariant curve $G=\infty$ and $2$ fixed points on blowups on this curve:	  
			  
		\medskip	
			  
			 %

	 Under factorization we obtain non-minimal surface.
	 Blowing down $(-1)$ curves $\e^a, \e^b$, then $(-1)$ curves $\e^c, \e^d$ and then $(-1)$ curve $\e^{e}$,  we obtain 
	 surface  $E_8^{(1)}/A_0^{(1)}$ with parameters and coordinates
	
	 %

%
%

	  Symmetry lattice is $A_1^{(1)}$, nodal curves and lattice is $\overline{\Phi}_{\nod}=\Phi_{\nod}=A_1+2A_3$,
           symmetry group is as follows:
	  

       	
		\subsubsection{$\pi^2 \ltimes 3A_1$}	
		\label{d5_pp111}	
		
		%
	        
	   Under $w$ we have $2$ fixed points; under $w^2$ we have additionally pointwise invariant curves $G=0, \infty$
	   and one additional fixed point on blowups on each of these curves
			   
			   \medskip 
			  
			 %
	   }
	   
	   \medskip
	   
	   Under factorization we obtain non-minimal surface.
	   Blowing down $(-1)$ curve $\e^a$ we obtain 
	   surface $E_8^{(1)}/A_0^{(1)}$ with parameters and coordinates
	
	 %
		
		 Symmetry lattice is $A_1^{(1)}$, nodal curves and lattice is 
		 $\Phi_{\nod}=A_1+2A_3,\,\overline{\Phi}_{\nod}=D_6+A_1$, $\overline{\Phi}_{\nod}/\Phi_{\nod}$
		 is generated by $-\h_{\g}+\e_{36}=-\alpha_{0233456}$, $\chi(-\alpha_{0233456})=-1$, symmetry group is as follows:
	  
	   %
        
        where $s=s_{435}s_{246}s_{534}$. Here we obtain additional $C_2=\Auti^{\nod}$, which is Galois group generated by $s_{430234}$.
		
		\subsubsection{$4A_1$}
		\label{d5_1111}

		%
	    
	    \medskip
	    
		       Under factorization 
		       we obtain already minimal surface $A_1^{(1)}/A_7^{(1)'}$
		       	with parameters and coordinates
		       	
		       %
		     
		 Symmetry lattice is $A_1^{(1)}$, nodal curves and lattice is $\overline{\Phi}_{\nod}=\Phi_{\nod}=\varnothing$,
           symmetry group is as follows:
	   
	   %

                \newpage

		\subsection{$E_3^{(1)}/A_5^{(1)}$ surface}

		We have two standard realizations of $E_3^{(1)}/A_5^{(1)}$ surface,
		we will relate to them as a) and b), following \cite{KNY15}, preferring realization b).
		
		Action on coordinates $F=f/\nu_4$, $G=g\nu_1$ in realization b) (where $f,g$ --- notations of \cite{KNY15}):
		
		\begin{center}
	%
	\end{center}

		Root data and point configuration (in variant a) $F=f/\nu_4, G=g\nu_1 \nu_8/\nu_7$):
		
		\begin{center}
       %
            \end{center}

            Transformation between realizations a) and b) is given by reflection $s_{H_2-E_{16}}$. Corresponding
            transformation of coordinates is as follows:
            \begin{equation}
            F_b=F_a \frac{G_a}{G_a-a_4a_1^{-1}}, \quad G_b=a_1a_4^{-1}G_a.
            \end{equation}

		\subsubsection{$\pi^3$}
		
		\label{e3_ppp}
		
		%
	     
	    \medskip

		After factorization we already obtain surface $E_6^{(1)}/A_2^{(1)}$
		with parameters and coordinates
	
	         %

                   Symmetry lattice is $(A_2^{(1)})_{|\alpha|^2=4}$, nodal curves and lattice is 
		 $\Phi_{\nod}=4A_1,\,\overline{\Phi}_{\nod}=D_4$, $\overline{\Phi}_{\nod}/\Phi_{\nod}$
		 is generated by $\e_1-\e_2=\alpha_3$, $\chi(\alpha_3)=-1$, symmetry group is as follows:
		     
		     %
		
		 Here we obtain additional $C_2^2=\Auti^{\nod}$, generated by
	      $s_{3243}$ and $s_{3643}=s_{3243}s_{3263}$, which changes root branches of $a_2^{1/2}$ and $a_1^{1/2}$ respectively. 
		
		\subsubsection{$A_1$}
		
		
		\label{e3_1}
		
		%
	   \end{center}

		 Under factorization we obtain non-minimal surface.
	   Blowing down $(-1)$ curves $\e^a, \e^b, \e^c$ we obtain 
	   surface $E_6^{(1)}/A_2^{(1)}$ with parameters and coordinates

	%
	       
		      
		     Symmetry lattice is $(A_2^{(1)})_{|\alpha|^2=4}$, nodal curves and lattice is $\overline{\Phi}_{\nod}=\Phi_{\nod}=4A_1$,
           symmetry group is as follows:
		     
		     %

		\subsubsection{$\pi^2$}
		
		\label{e3_pp}
		
		%
	           
	           \medskip
			 
	          Last three curves are not drawn (*) in the preimage, to prevent picture overload.

		After factorization we  obtain surface $E_7^{(1)}/A_1^{(1)}$
		with parameters and coordinates
		
		%
      
		       
		 Symmetry lattice is $A_1^{(1)}|_{|\alpha|^2=6}$, nodal curves and lattice is 
		 $\Phi_{\nod}=3A_2,\, \overline{\Phi}_{\nod}=E_6$, $\overline{\Phi}_{\nod}/\Phi_{\nod}$
		 is generated by $-\h_{\f\g\g}+\e_{123567}=-\alpha_{0233444556}$, $\chi(-\alpha_{0233444556})=\zeta$, symmetry group is as follows:
		  
		%
  
	
	where 
	$s=s_{12345604521324}s_{013}s_{12345604521324}^{-1}$, $\tilde{\pi}=s_{32}\pi s_{23}$.
	Here we obtain additional $S_3=\Auti^{\nod}$, corresponding to the permutation of the fixed points in the preimage.
	It is generated by Galois generator $s_{025}$ and $s_{025}s_{036}$, which changes branches of $a_4^{1/3}$.
	
		\subsubsection{$A_2$}
		
		\label{e3_2}
	        	
	        %
	    
	    }
	    
	    \medskip
	    
	     Second and fourth curves are not drawn (*) in the preimage, to prevent picture overload.
	        
			\begin{center}
		   %
	   \end{center}

	   	 Under factorization we obtain non-minimal surface.
	   Blowing down $(-1)$ curves $\e^a, \e^b, \e^c, \e^d, \e^e, \e^f$ we obtain 
	   surface $E_7^{(1)}/A_1^{(1)}$ with parameters and coordinates
		
		%

	        Symmetry lattice is $A_1^{(1)}|_{|\alpha|^2=6}$, nodal curves and lattice is $\overline{\Phi}_{\nod}=\Phi_{\nod}=3A_2$,
           symmetry group is as follows: 
		  
		%

		\subsection{$E_2^{(1)}/A_6^{(1)}$ surface} \label{ssec:alg E2}
		Action on coordinates $F=f/\nu_4, \, G=g\nu_1$ (where $f,g$ --- notations of \cite{KNY15})
		and variable $b=(a_1/a_2)^{1/2}$ in realization b):

			\begin{center}	
			
	\begin{tabular}{|c|c | c | c |}
	\hline
	coord. & $s_0$ & $s_1$ &  $\pi$ \\
	\hline
	$F$ & $a_0^2 F \frac{(-F + F G + b G^2)^2}{(-F + F G + a_0 b G^2) (-F + a_0 F G + a_0 b G^2)}$ & $a_1^{-1} F$ & $\frac{F}{G(F+bG)}$ \\
	\hline
	$G$ & $a_0 G \frac{-F + F G + b G^2}{-F + F G + a_0 b G^2}$ & $ G\frac{F-1}{F-a_1}$ & $-b^{-1}F/G$ \\
	\hline
	$b$ & $a_0 b$ & $a_1^{-1}b$ & $b/a_1$ \\
	\hline
	\end{tabular}
	\end{center}
	
	Root data and point configuration in realization b):
	
	\begin{center}
	
	  %
		
	       \end{center}
	
	As already mentioned, geometry depends on variable $b$. It is period map
	of divisor $\frac12(\alpha_1-\alpha_2)=E_{23}-E_{14}$.
	  
	       \newpage        
		
		\subsection{$E_1^{(1)}/A_7^{(1)}$ surface}
		
		\label{ssec:geom_E1}
		
		Action on coordinates $F=f/\nu_4,\, G=g (\nu_1\nu_2\nu_3/\nu_4)^{1/2}$ (where $f,g$ --- notations of \cite{KNY15}):
		
			\begin{center}
	%
	\end{center}
		
		Root data and point configuration:
		
		\begin{center}
		%
		
	       \end{center}

		\subsubsection{$\pi^2$}
		\label{e1_pp}
		
		%
             }
             
             \medskip

		After factorization we obtain surface $D_5^{(1)}/A_3^{(1)}$
		with parameters and coordinates

		%
	 
		
	 Symmetry lattice is $A_1^{(1)}|_{|\alpha|^2=4}$, nodal curves and lattice is 
		 $\Phi_{\nod}=4A_1,\, \overline{\Phi}_{\nod}=D_4$, $\overline{\Phi}_{\nod}/\Phi_{\nod}$
		 is generated by $-\h_{\f\g}+\e_{2367}=-\alpha_{2345}$, $\chi(-\alpha_{2345})=-1$, symmetry group is as follows:
	
	%
	
	 Here we obtain additional $C_2=\Auti^{\nod}$, generated by
	      $s_{3243}$, which changes root branch of $a_1^{1/2}$.
	      
	      Branch of $a_1^{1/2} a_0^{-1/2}$ is defined by the folding action, 
	      and changing branch of $(a_1^{1/2} a_0^{-1/2})^{1/2}$ could be done by
	      $s_{3453}$, which is image of $\sigma$.
	
        \subsubsection{$\sigma$}
		\label{e1_c}
		
		%
		\end{center}

			  Under $w$ we have non-singular pointwise invariant curves 
			  $G=0, G=\infty$, $(G=0, F\sim G^2)$, $(G=\infty, F^{-1}\sim G^{-2})$
			  and $4$ fixed points on blowups on these curves:
			   
			   \medskip

			 %
	    
	    \medskip

		 Under factorization we obtain non-minimal surface.
	   Blowing down $(-1)$ curves $\e^a, \e^b, \e^c, \e^d$ we obtain 
	   surface $D_5^{(1)}/A_3^{(1)}$ with parameters and coordinates

		%
	 
		
	Symmetry lattice is $A_1^{(1)}|_{|\alpha|^2=4}$, nodal curves and lattice is $\overline{\Phi}_{\nod}=\Phi_{\nod}=4A_1$,
           symmetry group is as follows: 
	
	%

        \newpage

        \section{Further questions} \label{Sec:Further}

        \subsection{Continuous limit}
        It is interesting to study the continuous limits of the folding transformations obtained here. These limits should be compared with
        the classification of differential Painlev\'e equation foldings \cite{TOS05}.
	    This is not straightforward since 
	    there are various limits from $q$-Painlev\'e equations to differential Painlev\'e ones (especially when parameters
	    of \(q\)-equations are special). 
        
        For example, let us consider folding transformation discussed in Example \ref{ex:Intr2}. 
%
        This is folding from $D_5^{(1)}/A_3^{(1)}$ to $A_1^{(1)}/A_7^{(1)'}$, which have standard degenerations to $D_4^{(1)}/D_4^{(1)}$
        and $A_0^{(1)}/D_8^{(1)}$ correspondingly. Thus, naively one could expect that continuous limit corresponds to the folding
        between Painlev\'e VI and Painlev\'e III($D_8^{(1)}$) equations. But, according to \cite{TOS05} (see Fig. 1 in loc. cit.)
        there is no such folding.
        
        In the Introduction we obtained, that for the invariant set $a_0=a_1=a_4=a_5=-1$ of given folding $w=s_0s_1s_4s_5$
        standard Painlev\'e $D_5^{(1)}/A_3^{(1)}$ dynamics could be rewritten as ``half-shift'' equation on single function $F$ (\eqref{qPVIprr}):
        \begin{equation}
        F(q z)F(q^{-1}z)=-\frac{F(z)^2-z}{F(z)^2-1}, \label{qPVIprrz}
        \end{equation}
        where we denoted $z=a_2^2$.
        From the other hand, from Painlev\'e $A_1^{(1)}/A_7^{(1)'}$ equation \eqref{qPIIID8}, we obtain
        \begin{equation}\label{qPIIID8_f}
        \f(qz)\f(q^{-1}z)= \left(\frac{\f-z}{\f-1}\right)^2.
        \end{equation}
       
        Let us make standard limit, rescaling
        \begin{equation}
        z\mapsto \hbar^4 z, \quad F \mapsto \hbar F \Rightarrow \f \mapsto \hbar^2\f,  \quad  \textrm{where} \, \hbar=\log q     
        \end{equation}
        and tending $\hbar\rightarrow -0$.
        Namely, after such limit from \eqref{qPIIID8_f} we obtain
        \begin{equation}
        \f''=\frac{\f'^2}{\f}-\frac{\f'}z+\frac{2\f^2}{z^2}-\frac2z,
        \end{equation}
        which is Painlev\'e III($D_8^{(1)}$) equation.
        This limit of \eqref{qPVIprrz} will be 
        \begin{equation}
        F'' =\frac{F'^2}{F}-\frac{F'}z+\frac{F^3}{z^2}-\frac1{zF}
        \end{equation}
        Under substitution $z=t^2/16$, $F\mapsto F/2$
        we obtain
        \begin{equation}
        \ddot{F} =\frac{\dot{F}^2}{F}-\frac{\dot{F}}t+\frac{F^3}{t^2}-\frac1{F},
        \end{equation}
        which is Painlev\'e III ($D_6^{(1)}$) equation with $\theta_{\star}=0,\,\theta_{\ast}=1/2$.
        Finally, we have that after the limit folding $w=s_0s_1s_4s_5:\, D_5^{(1)}/A_3^{(1)} \rightarrow A_1^{(1)}/A_7^{(1)'}$
        goes to the folding 5.2 from \cite{TOS05} from $A_2^{(1)}/D_6^{(1)}$ to $A_0^{(1)}/D_8^{(1)}$
        with $\f=F^2/4, \, z=t^2/16$.
        
\subsection{Relation to known foldings of $q$-Painlev\'e equations}

        Few foldings of $q$-Painlev\'e equations were found in the paper \cite{RGT00}. The constructions in loc. cit. were based
        on two possibilities: either right side of equation becomes perfect square or depend in some variable \(x\) through its square \(x^2\).
        The relation between formulas in \cite{RGT00} and  classification, obtained in this paper is the following
        \cite{RGT00}. Here we present the relation of these results to the classification of folding transformation, obtained in our paper:        
        \begin{enumerate}
          \item Section 3.1 of loc. cit. in perfect square case corresponds to $2A_1:\, D_5^{(1)}/A_3^{(1)} \rightarrow E_7^{(1)}/A_1^{(1)}$.

	\item Section 3.1 of loc. cit. in square dependence case corresponds to $4A_1:\, E_7^{(1)}/A_1^{(1)} \rightarrow D_5^{(1)}/A_3^{(1)}$.
         
         \item Section 3.3 of loc. cit. in perfect square case corresponds to $A_1: E_3^{(1)}/A_5^{(1)} \rightarrow E_6^{(1)}/A_2^{(1)}$.

         \item Section 3.3 of loc. cit. in square dependence case corresponds to $4A_1:\, E_6^{(1)}/A_2^{(1)} \rightarrow E_3^{(1)}/A_5^{(1)}$.
         
\end{enumerate}

Note that foldings from perfect square cases and square dependence cases form pairs of opposite foldings, see Remark \ref{rem:pairs}.

	\appendix

	\newpage
	
	\section{Foldings, leaving no dynamics}
 
 \label{sec:ntrf}
 
 \begin{center}
 
 \begin{tabular}{|M{1.5cm}|M{2cm}|M{2cm}|M{0.75cm}|M{3.75cm}|M{5cm}|}
	    \hline 
	    Sym./surf. & Diagram & Name & Order & $w$ & $A_w, \, N_{flip}$ \\ 
        \hline	
 $E_8^{(1)}/A_0^{(1)}$ & 
 \begin{tikzpicture}[elt/.style={circle,draw=black!100,thick, inner sep=0pt,minimum size=1.4mm},scale=0.25]
				\path 	(-2,0) 	node 	(a1) [elt,fill] {}
					(-1,0) 	node 	(a2) [elt,fill] {}
					( 0,0) node  	(a3) [elt] {}
					( 1,0) 	node  	(a4) [elt,fill] {}
					( 2,0) 	node 	(a5) [elt,fill] {}
					( 3,0)	node 	(a6) [elt,fill] {}
					( 4,0)	node 	(a7) [elt,fill] {}
					( 5,0)	node 	(a8) [elt,fill] {}
					( 0,1)	node 	(a0) [elt,fill] {};
				\draw [black,line width=1pt ] (a1) -- (a2) -- (a3) -- (a4) -- (a5) --  (a6) -- (a7) --(a8) (a3) -- (a0);
\end{tikzpicture}
& $A_1{+}A_2{+}A_5$
& 6 & 
\vspace{1mm}
\parbox{4cm}{\centering $s_0 s_{12} s_{87654}$\\
{\small $w^2=s_{576} (s_{21}s_{54}s_{87}) s_{675}$\\
$w^3=s_{675} (s_0s_4s_6s_8) s_{576}$}}
&  $(-1,\zeta, \zeta_6), (-1,\zeta^{-1}, \zeta_6^{-1})$\\

\cline{2-6}

 &\begin{tikzpicture}[elt/.style={circle,draw=black!100,thick, inner sep=0pt,minimum size=1.4mm},scale=0.25]
				\path 	(-2,0) 	node 	(a1) [elt,fill] {}
					(-1,0) 	node 	(a2) [elt,fill] {}
					( 0,0) node  	(a3) [elt,fill] {}
					( 1,0) 	node  	(a4) [elt] {}
					( 2,0) 	node 	(a5) [elt,fill] {}
					( 3,0)	node 	(a6) [elt,fill] {}
					( 4,0)	node 	(a7) [elt,fill] {}
					( 5,0)	node 	(a8) [elt,fill] {}
					( 0,1)	node 	(a0) [elt,fill] {};
				\draw [black,line width=1pt ] (a1) -- (a2) -- (a3) -- (a4) -- (a5) --  (a6) -- (a7) --(a8) (a3) -- (a0);
\end{tikzpicture}
& $2A_4$ & 5 & $s_{0321} s_{5678}$ & $(\zeta_5^i, \zeta_5^{3i}),\, i=1,2,3,4$ \\
 \hline
 
 & 
\begin{tikzpicture}[elt/.style={circle,draw=black!100,thick, inner sep=0pt,minimum size=1.4mm},scale=0.25]

	                \path 	(-3,0) 	node 	(a1) [elt,fill] {}
			(-2,0) 	node 	(a2) [elt,fill] {}
			( -1,0) node  	(a3) [elt,fill] {}
			( 0,0) 	node  	(a4) [elt] {}
			( 1,0) 	node 	(a5) [elt,fill] {}
			( 2,0)	node 	(a6) [elt,fill] {}
			( 3,0)	node 	(a7) [elt,fill] {}
			( 0,1)	node 	(a0) [elt,fill] {};
			\draw [black,line width=1pt ] (a1) -- (a2) -- (a3) -- (a4) -- (a5) --  (a6) -- (a7) (a4) -- (a0);
 \end{tikzpicture}
&  $A_1{+}2A_3$
& 4 & \vspace{1mm} \parbox{4cm}{\centering $s_0 s_{765} s_{123}$\\{\small $w^2=s_{26} (s_1s_3s_5s_7) s_{62}$}} & $(-1,\ri,\ri), \, (-1,-\ri,-\ri)$\\ 
 \cline{2-6}
 
$E_7^{(1)}/A_1^{(1)}$ &
 \begin{tikzpicture}[elt/.style={circle,draw=black!100,thick, inner sep=0pt,minimum size=1.4mm},scale=0.25]
			
	                \path 	(-3,0) 	node 	(a1) [elt,fill] {}
			(-2,0) 	node 	(a2) [elt,fill] {}
			( -1,0) node  	(a3) [elt,fill] {}
			( 0,0) 	node  	(a4) [elt,fill] {}
			( 1,0) 	node 	(a5) [elt] {}
			( 2,0)	node 	(a6) [elt,fill] {}
			( 3,0)	node 	(a7) [elt,fill] {}
			( 0,1)	node 	(a0) [elt,fill] {};
			\draw [black,line width=1pt ] (a1) -- (a2) -- (a3) -- (a4) -- (a5) --  (a6) -- (a7) (a4) -- (a0);
\end{tikzpicture}
& $A_2{+}A_5$
& \vspace{6mm} & \vspace{1mm} \parbox{4cm}{\centering $s_{04321} s_{87}$\\{\small $w^2=s_{423} (s_{04} s_{21} s_{76}) s_{324}$\\ $w^3=s_{324} s_0s_1s_3 s_{423}$}}& $(\zeta_6,\zeta), \, (\zeta_6^{-1},\zeta^{-1})$\\
\cline{2-3} \cline{5-6}

&
\begin{tikzpicture}[elt/.style={circle,draw=black!100,thick, inner sep=0pt,minimum size=1.4mm},scale=0.25]
 \path 	(-3,0) 	node 	(a1) [elt,fill] {}
			(-2,0) 	node 	(a2) [elt,fill] {}
			( -1,0) node  	(a3) [elt] {}
			( 0,0) 	node  	(a4) [elt,fill] {}
			( 1,0) 	node 	(a5) [elt] {}
			( 2,0)	node 	(a6) [elt,fill] {}
			( 3,0)	node 	(a7) [elt,fill] {}
			( 0,1)	node 	(a0) [elt,fill] {};
			\draw [black,line width=1pt ] (a1) -- (a2) -- (a3) -- (a4) -- (a5) --  (a6) -- (a7) (a4) -- (a0);
			\draw[<->] (a1) edge[bend right=40] node[fill=white]{\small $\pi$} (a7);
			\draw[<->] (a2) edge[bend right=40] node[]{} (a6);
			\draw[<->] (a3) edge[bend right=40] node[]{} (a5); 
\end{tikzpicture}
& $\pi{\ltimes}3A_2$
& \vspace{-1cm} 6 & \vspace{1mm} \parbox{4cm}{\centering $\pi s_{12} s_{04}$\\ {\small $w^2=s_{12} s_{40} s_{76}$\\ $w^3=s_{67} \pi s_{76}$}} &
$(\zeta,\zeta, a_5{=}\zeta a_3)\, (\zeta^{-1},\zeta^{-1}, a_3{=}\zeta a_5)$ \\
        
\hline 

&
\begin{tikzpicture}[elt/.style={circle,draw=black!100,thick, inner sep=0pt,minimum size=1.4mm},scale=0.25]
  \path 	(-2,-1.16) 	node 	(a1) [elt,fill] {}
			(-1,-0.58) 	node 	(a2) [elt,fill] {}
			( 0,0) node  	(a3) [elt] {}
			( 1,-0.58) 	node  	(a4) [elt,fill] {}
			( 2,-1.16) 	node 	(a5) [elt,fill] {}
			( 0,1)	node 	(a6) [elt,fill] {}
			( 0,2)	node 	(a0) [elt,fill] {};
			\draw [black,line width=1pt ] (a1) -- (a2) -- (a3) -- (a4) -- (a5)   (a3) -- (a6) -- (a0);	
\end{tikzpicture}
& $3A_2$			
& 3 & $s_{06} s_{12} s_{54}$ & $(\zeta,\zeta,\zeta)$, \, $(\zeta^{-1},\zeta^{-1},\zeta^{-1})$ \\
\cline{2-6}

$E_6^{(1)}/A_2^{(1)}$  &
\begin{tikzpicture}[elt/.style={circle,draw=black!100,thick, inner sep=0pt,minimum size=1.4mm},scale=0.25]
   \path 	(-2,-1.16) 	node 	(a1) [elt,fill] {}
			(-1,-0.58) 	node 	(a2) [elt,fill] {}
			( 0,0) node  	(a3) [elt,fill] {}
			( 1,-0.58) 	node  	(a4) [elt,fill] {}
			( 2,-1.16) 	node 	(a5) [elt,fill] {}
			( 0,1)	node 	(a6) [elt] {}
			( 0,2)	node 	(a0) [elt,fill] {};
			\draw [black,line width=1pt ] (a1) -- (a2) -- (a3) -- (a4) -- (a5)   (a3) -- (a6) -- (a0);
\end{tikzpicture}			
& $A_1{+}A_5$		
& & \vspace{1mm}\parbox{4cm}{\centering $s_0 s_{12345}$\\{\small $w^2=s_{243} (s_{12} s_{45}) s_{342}$ \\ $w^3=s_{342} (s_0s_1s_3s_5) s_{243}$}} & $(-1,\zeta_6)\, (-1,\zeta_6^{-1})$\\
\cline{2-3} \cline{5-6}
&
\begin{tikzpicture}[elt/.style={circle,draw=black!100,thick, inner sep=0pt,minimum size=1.4mm},scale=0.25]
 \path 	(-2,-1.16) 	node 	(a1) [elt,fill] {}
			(-1,-0.58) 	node 	(a2) [elt] {}
			( 0,0) node  	(a3) [elt,fill] {}
			( 1,-0.58) 	node  	(a4) [elt] {}
			( 2,-1.16) 	node 	(a5) [elt,fill] {}
			( 0,1)	node 	(a6) [elt] {}
			( 0,2)	node 	(a0) [elt,fill] {};
			\draw [black,line width=1pt ] (a1) -- (a2) -- (a3) -- (a4) -- (a5)   (a3) -- (a6) -- (a0);	
			\draw[->] (a1) edge[bend right=40] node[fill=white]{\small $\pi$} (a5);
             		\draw[->] (a2) edge[bend right=40] node[]{} (a4);
			\draw[->] (a5) edge[bend right=40] node[fill=white]{\small $\pi$} (a0);
			\draw[->] (a4) edge[bend right=40] node[]{} (a6);
			\draw[->] (a0) edge[bend right=40] node[fill=white]{\small $\pi$} (a1);
		        \draw[->] (a6) edge[bend right=40] node[]{} (a2);
\end{tikzpicture}				
& $\pi{\ltimes}4A_1$			
& \vspace{-1cm} 6 & \vspace{1mm} \parbox{4cm}{\centering $\pi s_1 s_3$\\ {\small $w^2=s_5 (\pi^2) s_5$ \\ $w^3=s_0s_1s_3s_5$}} & $(-1,-1,-1,-1;a_4{=}a_2{=}-a_6)$\\

\hline
$A_4^{(1)}/A_4^{(1)}$ &
\begin{tikzpicture}[elt/.style={circle,draw=black!100,thick, inner sep=0pt,minimum size=1.4mm},scale=0.75]
 \path 	(0.5,sin{72}+sin{36}) 	node 	(a0) [elt] {}
		(-cos{72},sin{72}) 	node 	(a1) [elt] {}
		( 0,0) node  	(a2) [elt] {}
		( 1,0) 	node 	(a3) [elt] {}
		( 1+cos{72},sin{72}) 	node  	(a4) [elt] {}; 
		
		\draw[->] (a4) edge[bend right=50] node[fill=white]{\tiny $\pi$} (a0);
		\draw[->] (a0) edge[bend right=50] node[fill=white]{\tiny $\pi$} (a1);
		\draw[->] (a1) edge[bend right=50] node[fill=white]{\tiny $\pi$} (a2);
		\draw[->] (a2) edge[bend right=50] node[fill=white]{\tiny $\pi$} (a3);
		\draw[->] (a3) edge[bend right=50] node[fill=white]{\tiny $\pi$} (a4);
		
		\draw [black,line width=1pt] (a1) -- (a2) -- (a3) -- (a4) -- (a0) -- (a1);
		
\end{tikzpicture}
& $\pi$
& & $\pi$ & $(a_0{=}a_1{=}a_2{=}a_3{=}a_4)$ \\
\cline{2-3} \cline{5-6}
&
\begin{tikzpicture}[elt/.style={circle,draw=black!100,thick, inner sep=0pt,minimum size=1.4mm},scale=0.6]
 \path 	(0.5,sin{72}+sin{36}) 	node 	(a0) [elt] {}
		(-cos{72},sin{72}) 	node 	(a1) [elt,fill] {}
		( 0,0) node  	(a2) [elt,fill] {}
		( 1,0) 	node 	(a3) [elt,fill] {}
		( 1+cos{72},sin{72}) 	node  	(a4) [elt,fill] {}; 
		\draw [black,line width=1pt] (a1) -- (a2) -- (a3) -- (a4) -- (a0) -- (a1);
\end{tikzpicture}		
& $A_4$
& \vspace{-1cm} 5 & $s_{1234}$ & $(\zeta_5^i), \, i=1,2,3,4$ \\
\hline
 &
\begin{tikzpicture}[elt/.style={circle,draw=black!100,thick, inner sep=0pt,minimum size=1.4mm},scale=0.5]
\path 
			(-1,0) 	node 	(a0) [elt] {}
			(1,0) 	node 	(a1) [elt] {}
			(0,1.73) 	node 	(a2) [elt] {};
			
			\draw[->] (a0) edge[bend right=40] node[fill=white]{\small $\pi$} (a1);
			\draw[->] (a1) edge[bend right=40] node[fill=white]{\small $\pi$} (a2);
			\draw[->] (a2) edge[bend right=40] node[fill=white]{\small $\pi$} (a0);
			
			\draw [black,line width=1pt] (a0) -- (a1) -- (a2) -- (a0);
 			\path
			(2,0) 	node 	(a3) [elt] {}
			(2,1.73)	node 	(a4) [elt] {};
			
			 \draw [black,line width=2.5pt ] (a3) -- (a4);
		    
		     \draw [white,line width=1.5pt ] (a3) -- (a4);
		     
		     \draw [<->,black, line width=0.5pt]
		     (a3) -- (a4);
		     
		     \draw[<->] (a3) edge[bend right=60] node[fill=white]{\small $\pi$} (a4);
		   
\end{tikzpicture}		     
& $\pi$
& & \vspace{1mm} \parbox{2cm}{\centering $\pi$\\ {\small $w^2=\pi^2$\\ $w^3=\pi^3$}} & $(a_0{=}a_1{=}a_2, a_3{=}a_4)$ \\
\cline{2-3} \cline{5-6}
$E_3^{(1)}/A_5^{(1)}$ &
\begin{tikzpicture}[elt/.style={circle,draw=black!100,thick, inner sep=0pt,minimum size=1.4mm},scale=0.4]
 \path 	
			(-1,0) 	node 	(a0) [elt,fill] {}
			(1,0) 	node 	(a1) [elt,fill] {}
			(0,1.73) 	node 	(a2) [elt] {};
			\draw [black,line width=1pt] (a0) -- (a1) -- (a2) -- (a0);

        \path
			(2,0) 	node 	(a3) [elt] {}
			(1,1.73)	node 	(a4) [elt,fill] {};

			 \draw [black,line width=2.5pt ] (a3) -- (a4);
		    
		     \draw [white,line width=1.5pt ] (a3) -- (a4);
		     
		     \draw [<->,black, line width=0.5pt]
		     (a3) -- (a4); 
       
\end{tikzpicture}
& $A_2{+}A_1$        
& 6 & \vspace{1mm} \parbox{2cm}{\centering $s_{21} s_4$\\ \small{ $w^2= s_{12}$\\ $w^3=s_4$}} & $(\zeta,-1),\, (\zeta^{-1},-1)$ \\
\cline{2-3} \cline{5-6}
&
\begin{tikzpicture}[elt/.style={circle,draw=black!100,thick, inner sep=0pt,minimum size=1.4mm},scale=0.5]

\begin{scope}[scale=0.8]
\path 	
			(-1,0) 	node 	(a0) [elt,fill] {}
			(1,0) 	node 	(a1) [elt,fill] {}
			(0,1.73) 	node 	(a2) [elt] {};
			\draw [black,line width=1pt] (a0) -- (a1) -- (a2) -- (a0);
\end{scope}

        \path
			(1.6,0) 	node 	(a3) [elt] {}
			(1.6,1.73)	node 	(a4) [elt] {};

			 \draw [black,line width=2.5pt ] (a3) -- (a4);
		    
		     \draw [white,line width=1.5pt ] (a3) -- (a4);
		     
		     \draw [<->,black, line width=0.5pt]
		     (a3) -- (a4); 
		      \draw[<->] (a3) edge[bend right=80] node[fill=white]{\small $\pi^3$} (a4);
       
\end{tikzpicture}       
& $\pi^3\times A_2$
& & \vspace{1mm} \parbox{2cm}{\centering $\pi^3 s_1s_2$\\ {\small $w^2=s_2s_1$ \\ $w^3=\pi^3$}} & $(\zeta;a_3{=}a_4), \, (\zeta^{-1};a_3{=}a_4)$  \\
\cline{2-3} \cline{5-6}
&
\begin{tikzpicture}[elt/.style={circle,draw=black!100,thick, inner sep=0pt,minimum size=1.4mm},scale=0.5]
 \path 	
			(-1,0) 	node 	(a0) [elt] {}
			(1,0) 	node 	(a1) [elt] {}
			(0,1.73) 	node 	(a2) [elt] {};

			\draw[<-] (a0) edge[bend right=40] node[fill=white]{\small $\pi$} (a1);
			\draw[<-] (a1) edge[bend right=40] node[fill=white]{\small $\pi$} (a2);
			\draw[<-] (a2) edge[bend right=40] node[fill=white]{\small $\pi$} (a0);
			
			\draw [black,line width=1pt] (a0) -- (a1) -- (a2) -- (a0);
        
        \path
			(2,0) 	node 	(a3) [elt] {}
			(2,1.73)	node 	(a4) [elt,fill] {};

			 \draw [black,line width=2.5pt ] (a3) -- (a4);
		    
		     \draw [white,line width=1.5pt ] (a3) -- (a4);
		     
		     \draw [<->,black, line width=0.5pt]
		     (a3) -- (a4); 
\end{tikzpicture}        
        & $\pi^2\times A_1$ & & \vspace{1mm} \parbox{2cm}{\centering $\pi^2 s_4$\\ {\small $w^2=\pi^4$\\ $w^3=s_4$}} & $(-1;a_0{=}a_1{=}a_2)$ \\
\hline

 &
\begin{tikzpicture}[elt/.style={circle,draw=black!100,thick, inner sep=0pt,minimum size=1.4mm},scale=0.75]
\path 	
			(-1,0) 	node 	(a0) [elt] {}
			(1,0) 	node 	(a1) [elt] {};

			 \draw [black,line width=2.5pt ] (a0) -- (a1);
		    
		     \draw [white,line width=1.5pt ] (a0) -- (a1);
		     
		     \draw [<->,black, line width=0.5pt]
		     (a0) -- (a1);

	             \draw[<->] (a0) edge[bend right=50] node[fill=white]{\small $\pi$} (a1); 
\end{tikzpicture}	             
& $\pi$             
& 4 & $\pi$ & $(a_0{=}a_1)$ \\
\cline{2-6}
$E_1^{(1)}/A_7^{(1)}$ &
\begin{tikzpicture}[elt/.style={circle,draw=black!100,thick, inner sep=0pt,minimum size=1.4mm},scale=0.75]
 \path
			(-1,1) 	node 	(a3) [elt] {}
			(1,1)	node 	(a4) [elt,fill] {};

			 \draw [black,line width=2.5pt ] (a3) -- (a4);
		    
		     \draw [white,line width=1.5pt ] (a3) -- (a4);
		     
		     \draw [<->,black, line width=0.5pt]
		     (a3) -- (a4); 
\end{tikzpicture}
& $A_1$
& & $s_1$ & $(-1)$ \\
\cline{2-3} \cline{5-6}
&
\begin{tikzpicture}[elt/.style={circle,draw=black!100,thick, inner sep=0pt,minimum size=1.4mm},scale=0.75]
\path 	
			(-1,0) 	node 	(a0) [elt] {}
			(1,0) 	node 	(a1) [elt,fill] {};

			 \draw [black,line width=2.5pt ] (a0) -- (a1);
		    
		     \draw [white,line width=1.5pt ] (a0) -- (a1);
		     
		     \draw [<->,black, line width=0.5pt]
		     (a0) -- (a1);
		     
		     \draw (0,0.5)  node[fill=white]{\small $\pi^2$};
\end{tikzpicture}		     
& $\pi^2\times A_1$	     
& \vspace{-8mm} 2 & $\pi^2 s_1$ & $(-1)$ \\
\hline

 \end{tabular}       
 
 \end{center}
 

 \begin{tabular}{|M{1.5cm}|M{2.5cm}|M{2cm}|M{4.5cm}|M{4cm}|}
	    \hline 
	    Sym./surf. & Diagram & Name & Group & $\mathcal{A}_w$ \\ 
        \hline
      $E_7^{(1)}/A_1^{(1)}$ & \begin{tikzpicture}[elt/.style={circle,draw=black!100,thick, inner sep=0pt,minimum size=1.4mm},scale=0.375]
			\path 	(-3,0) 	node 	(a1) [elt,fill] {}
			(-2,0) 	node 	(a2) [elt,fill] {}
			( -1,0) node  	(a3) [elt,fill] {}
			( 0,0) 	node  	(a4) [elt] {}
			( 1,0) 	node 	(a5) [elt,fill] {}
			( 2,0)	node 	(a6) [elt,fill] {}
			( 3,0)	node 	(a7) [elt,fill] {}
			( 0,1)	node 	(a0) [elt,fill] {};
			\draw [black,line width=1pt ] (a1) -- (a2) -- (a3) -- (a4) -- (a5) --  (a6) -- (a7) (a4) -- (a0);
			\draw[<->] (a1) edge[bend right=40] node[fill=white]{\small $\pi$} (a7);
			\draw[<->] (a2) edge[bend right=40] node[]{} (a6);
			\draw[<->] (a3) edge[bend right=40] node[]{} (a5);
			
			\end{tikzpicture}
			& $\pi{\ltimes}(A_1{+}A_3)$
			&$\underbrace{C_2}_{\pi}\ltimes (\underbrace{C_2}_{s_0 s_{123}^2} \times \underbrace{C_4}_{s_0 s_{123}s_{765}})$ & $(-1,\ri,\ri), \, (-1,-\ri,-\ri)$\\
			
        \hline	   
        
  $E_7^{(1)}/A_1^{(1)}$ & \begin{tikzpicture}[elt/.style={circle,draw=black!100,thick, inner sep=0pt,minimum size=1.4mm},scale=0.375]
			\path 	(-3,0) 	node 	(a1) [elt,fill] {}
			(-2,0) 	node 	(a2) [elt,fill] {}
			( -1,0) node  	(a3) [elt,fill] {}
			( 0,0) 	node  	(a4) [elt] {}
			( 1,0) 	node 	(a5) [elt,fill] {}
			( 2,0)	node 	(a6) [elt,fill] {}
			( 3,0)	node 	(a7) [elt,fill] {}
			( 0,1)	node 	(a0) [elt,fill] {};
			\draw [black,line width=1pt ] (a1) -- (a2) -- (a3) -- (a4) -- (a5) --  (a6) -- (a7) (a4) -- (a0);
			\draw[<->] (a1) edge[bend right=40] node[fill=white]{\small $\pi$} (a7);
			\draw[<->] (a2) edge[bend right=40] node[]{} (a6);
			\draw[<->] (a3) edge[bend right=40] node[]{} (a5);
			
			\end{tikzpicture}
			& $\pi{\ltimes}(A_1{+}A_3)$
			&$\underbrace{C_2}_{\pi}\times \underbrace{C_4}_{s_0 s_{123}s_{765}}\subset C_2\ltimes (C_2\times C_4)$ & $(-1,\ri,\ri), \, (-1,-\ri,-\ri)$\\
			
        \hline	
   &
 \begin{tikzpicture}[elt/.style={circle,draw=black!100,thick, inner sep=0pt,minimum size=1.4mm},scale=0.375]
 \path 	(-2,-1.16) 	node 	(a1) [elt,fill] {}
			(-1,-0.58) 	node 	(a2) [elt,fill] {}
			( 0,0) node  	(a3) [elt] {}
			( 1,-0.58) 	node  	(a4) [elt,fill] {}
			( 2,-1.16) 	node 	(a5) [elt,fill] {}
			( 0,1)	node 	(a6) [elt,fill] {}
			( 0,2)	node 	(a0) [elt,fill] {};
			\draw [black,line width=1pt ] (a1) -- (a2) -- (a3) -- (a4) -- (a5)   (a3) -- (a6) -- (a0);	
				\draw[->] (a1) edge[bend right=40] node[fill=white]{\small $\pi$} (a5);
				\draw[->] (a2) edge[bend right=40] node[]{} (a4);
				\draw[->] (a5) edge[bend right=40] node[fill=white]{\small $\pi$} (a0);
				\draw[->] (a4) edge[bend right=40] node[]{} (a6);
				\draw[->] (a0) edge[bend right=40] node[fill=white]{\small $\pi$} (a1);
				\draw[->] (a6) edge[bend right=40] node[]{} (a2);
\end{tikzpicture}	 
  & $\pi{\ltimes}3A_2$
  & $\underbrace{C_3}_{\pi} \ltimes (\underbrace{C_3}_{s_{12}s_{45}}\times \underbrace{C_3}_{s_{06}s_{45}})$ &
  $(\zeta,\zeta,\zeta), \, (\zeta^{-1},\zeta^{-1},\zeta^{-1})$ \\
  \cline{2-5}
 $E_6^{(1)}/A_2^{(1)}$ &
  \begin{tikzpicture}[elt/.style={circle,draw=black!100,thick, inner sep=0pt,minimum size=1.4mm},scale=0.375]
 \path 	(-2,-1.16) 	node 	(a1) [elt,fill] {}
			(-1,-0.58) 	node 	(a2) [elt,fill] {}
			( 0,0) node  	(a3) [elt] {}
			( 1,-0.58) 	node  	(a4) [elt,fill] {}
			( 2,-1.16) 	node 	(a5) [elt,fill] {}
			( 0,1)	node 	(a6) [elt,fill] {}
			( 0,2)	node 	(a0) [elt,fill] {};
			\draw [black,line width=1pt ] (a1) -- (a2) -- (a3) -- (a4) -- (a5)   (a3) -- (a6) -- (a0);	
				\draw[->] (a1) edge[bend right=40] node[fill=white]{\small $\pi$} (a5);
				\draw[->] (a2) edge[bend right=40] node[]{} (a4);
				\draw[->] (a5) edge[bend right=40] node[fill=white]{\small $\pi$} (a0);
				\draw[->] (a4) edge[bend right=40] node[]{} (a6);
				\draw[->] (a0) edge[bend right=40] node[fill=white]{\small $\pi$} (a1);
				\draw[->] (a6) edge[bend right=40] node[]{} (a2);
\end{tikzpicture}
& $\pi{\ltimes}3A_2$
  & $\underbrace{C_3}_{\pi}\times \underbrace{C_3}_{s_{06} s_{12} s_{54}}\subset C_3\ltimes C_3^2$  & $(\zeta,\zeta,\zeta), \, (\zeta^{-1},\zeta^{-1},\zeta^{-1})$\\
  \cline{2-5}
  &
   \begin{tikzpicture}[elt/.style={circle,draw=black!100,thick, inner sep=0pt,minimum size=1.4mm},scale=0.375]
 \path 	(-2,-1.16) 	node 	(a1) [elt,fill] {}
			(-1,-0.58) 	node 	(a2) [elt,fill] {}
			( 0,0) node  	(a3) [elt] {}
			( 1,-0.58) 	node  	(a4) [elt,fill] {}
			( 2,-1.16) 	node 	(a5) [elt,fill] {}
			( 0,1)	node 	(a6) [elt,fill] {}
			( 0,2)	node 	(a0) [elt,fill] {};
			\draw [black,line width=1pt ] (a1) -- (a2) -- (a3) -- (a4) -- (a5)   (a3) -- (a6) -- (a0);	
				
\end{tikzpicture}
  & $3A_2$
  &  $\underbrace{C_3}_{s_{12}s_{45}}\times \underbrace{C_3}_{s_{06} s_{12} s_{54}}\subset C_3^3$  & $(\zeta,\zeta,\zeta), \, (\zeta^{-1},\zeta^{-1},\zeta^{-1})$\\
  \hline
  &
  \begin{tikzpicture}[elt/.style={circle,draw=black!100,thick, inner sep=0pt,minimum size=1.4mm},scale=0.75]
			\path 	(-cos{60},-sin{60}) 	node 	(a0) [elt,fill] {}
			(-cos{60},sin{60}) 	node 	(a1) [elt,fill] {}
			( 0,0) node  	(a2) [elt] {}
			( 1,0) 	node  	(a3) [elt] {}
			( 1+cos{60},-sin{60}) 	node 	(a4) [elt,fill] {}
			( 1+cos{60},sin{60})	node 	(a5) [elt,fill] {};
			\draw [black,line width=1pt] (a1) -- (a2) -- (a3) -- (a4) (a3) -- (a5) (a2) -- (a0);
			\draw[->] (a4) edge[bend left=30] node[fill=white]{\small $\pi$} (a0);
			\draw[->] (a5) edge[bend right=30] node[fill=white]{\small $\pi$} (a1);
			\draw[->] (a0) edge[bend right=60] node[]{} (a5);
			\draw[->] (a1) edge[bend right=60] node[]{} (a4);
			\draw[<->] (a2) edge[bend left=30] node[]{} (a3);	
			\end{tikzpicture}
&  $\pi{\ltimes}4A_1$			
  & $\underbrace{C_4}_{\pi}\ltimes (\underbrace{C_2}_{s_0s_1}\times \underbrace{C_2}_{s_4s_5})$ & $(-1,-1,-1,-1;a_2=a_3)$\\
  \cline{2-5}
  $D_5^{(1)}/A_3^{(1)}$ & \begin{tikzpicture}[elt/.style={circle,draw=black!100,thick, inner sep=0pt,minimum size=1.4mm},scale=0.7]
			\path 	(-cos{60},-sin{60}) 	node 	(a0) [elt,fill] {}
			(-cos{60},sin{60}) 	node 	(a1) [elt,fill] {}
			( 0,0) node  	(a2) [elt] {}
			( 1,0) 	node  	(a3) [elt,fill] {}
			( 1+cos{60},-sin{60}) 	node 	(a4) [elt,fill] {}
			( 1+cos{60},sin{60})	node 	(a5) [elt,fill] {};
			\draw [black,line width=1pt] (a1) -- (a2) -- (a3) -- (a4) (a3) -- (a5) (a2) -- (a0);
				
				\draw[<->] (a4) edge[bend right=60] node[fill=white]{\small $\pi^2$} (a5);
				\draw[<->] (a0) edge[bend left=60] node[fill=white]{\small $\pi^2$} (a1);
			\end{tikzpicture}
	& $\pi^2{\ltimes}(2A_1{+}A_3)$		
   & $\underbrace{C_2}_{\pi^2}\ltimes (\underbrace{C_2}_{s_0s_1}\times  \underbrace{C_4}_{s_1s_{534}})$ & $(-1,-1,\ri),\, (-1,-1,-\ri)$  \\	
   \cline{2-5}
   &
    \begin{tikzpicture}[elt/.style={circle,draw=black!100,thick, inner sep=0pt,minimum size=1.4mm},scale=0.75]
			\path 	(-cos{60},-sin{60}) 	node 	(a0) [elt,fill] {}
			(-cos{60},sin{60}) 	node 	(a1) [elt,fill] {}
			( 0,0) node  	(a2) [elt] {}
			( 1,0) 	node  	(a3) [elt] {}
			( 1+cos{60},-sin{60}) 	node 	(a4) [elt,fill] {}
			( 1+cos{60},sin{60})	node 	(a5) [elt,fill] {};
			\draw [black,line width=1pt] (a1) -- (a2) -- (a3) -- (a4) (a3) -- (a5) (a2) -- (a0);
			\draw[->] (a4) edge[bend left=30] node[fill=white]{\small $\pi$} (a0);
			\draw[->] (a5) edge[bend right=30] node[fill=white]{\small $\pi$} (a1);
			\draw[->] (a0) edge[bend right=60] node[]{} (a5);
			\draw[->] (a1) edge[bend right=60] node[]{} (a4);
			\draw[<->] (a2) edge[bend left=30] node[]{} (a3);	
			\end{tikzpicture}   
	& $\pi{\ltimes}4A_1$	
   & $\underbrace{C_4}_{\pi}\times\!\!\underbrace{C_2}_{s_0s_1s_4s_5}\!\!\subset C_4\ltimes C_2^2$ & $(-1,-1,-1,-1;a_2=a_3)$ \\
   \cline{2-5}
    & \begin{tikzpicture}[elt/.style={circle,draw=black!100,thick, inner sep=0pt,minimum size=1.4mm},scale=0.55]
			\path 	(-cos{60},-sin{60}) 	node 	(a0) [elt,fill] {}
			(-cos{60},sin{60}) 	node 	(a1) [elt,fill] {}
			( 0,0) node  	(a2) [elt] {}
			( 1,0) 	node  	(a3) [elt,fill] {}
			( 1+cos{60},-sin{60}) 	node 	(a4) [elt,fill] {}
			( 1+cos{60},sin{60})	node 	(a5) [elt,fill] {};
			\draw [black,line width=1pt] (a1) -- (a2) -- (a3) -- (a4) (a3) -- (a5) (a2) -- (a0);
				
			\end{tikzpicture}
	& $2A_1{+}A_3$		
   & $\underbrace{C_2}_{s_0s_1}\times \underbrace{C_4}_{s_1s_{534}}\!\!\subset C_2{\ltimes}(C_2{\times}C_4)$ & $(-1,-1,\ri),\, (-1,-1,-\ri)$  \\
\hline
 & \begin{tikzpicture}[elt/.style={circle,draw=black!100,thick, inner sep=0pt,minimum size=1.4mm},scale=0.75]
\path 	
			(-1,0) 	node 	(a0) [elt] {}
			(1,0) 	node 	(a1) [elt,fill] {};

			 \draw [black,line width=2.5pt ] (a0) -- (a1);
		    
		     \draw [white,line width=1.5pt ] (a0) -- (a1);
		     
		     \draw [<->,black, line width=0.5pt]
		     (a0) -- (a1);
		     
		     \draw (0,-0.5)  node[fill=white]{\small $\pi^2$};
\end{tikzpicture}
& $\pi^2{\times}A_1$  & 
$\underbrace{C_2}_{\pi^2}\times \underbrace{C_2}_{s_1}\times \underbrace{C_2}_{\sigma}$
& $(-1)$\\
\cline{2-5}
$E_1^{(1)}/A_7^{(1)}$ & \begin{tikzpicture}[elt/.style={circle,draw=black!100,thick, inner sep=0pt,minimum size=1.4mm},scale=0.75]
\path 	
			(-1,0) 	node 	(a0) [elt] {}
			(1,0) 	node 	(a1) [elt] {};

			 \draw [black,line width=2.5pt ] (a0) -- (a1);
		    
		     \draw [white,line width=1.5pt ] (a0) -- (a1);
		     
		     \draw [<->,black, line width=0.5pt]
		     (a0) -- (a1);

	             \draw[<->] (a0) edge[bend right=50] node[fill=white]{\small $\pi$} (a1); 
	             \node at (0,0.25) {\small $\sigma$};
\end{tikzpicture}
& $\pi{\times}\sigma$
& $\underbrace{C_4}_{\pi}\times \underbrace{C_2}_{\sigma}$ &  $(a_0=a_1)$\\
\cline{2-5}
& \begin{tikzpicture}[elt/.style={circle,draw=black!100,thick, inner sep=0pt,minimum size=1.4mm},scale=0.75]
\path 	
			(-1,0) 	node 	(a0) [elt] {}
			(1,0) 	node 	(a1) [elt,fill] {};

			 \draw [black,line width=2.5pt ] (a0) -- (a1);
		    
		     \draw [white,line width=1.5pt ] (a0) -- (a1);
		     
		     \draw [<->,black, line width=0.5pt]
		     (a0) -- (a1);
		     
		    \draw (0,-0.5)  node[fill=white]{\small $\pi^2$};
\end{tikzpicture} &
$\pi^2{\times}A_1$ & $\underbrace{C_2}_{\pi^2}\times \underbrace{C_2}_{s_1}$ & $(-1)$\\
\cline{2-5}
& \begin{tikzpicture}[elt/.style={circle,draw=black!100,thick, inner sep=0pt,minimum size=1.4mm},scale=0.75]
\path 	
			(-1,0) 	node 	(a0) [elt] {}
			(1,0) 	node 	(a1) [elt,fill] {};

			 \draw [black,line width=2.5pt ] (a0) -- (a1);
		    
		     \draw [white,line width=1.5pt ] (a0) -- (a1);
		     
		     \draw [<->,black, line width=0.5pt]
		     (a0) -- (a1);
		     
		     \node at (0,-0.25) {\small $\sigma$};

\end{tikzpicture} & $\sigma{\times} A_1$ & $\underbrace{C_2}_{s_1}\times \underbrace{C_2}_{\sigma}$ & $(-1)$\\
\hline

  \end{tabular} 
	
\newpage
	
	\bibliographystyle{alpha}
	\bibliography{bibtex}

\end{document}